\documentclass[prl,twocolumn,superscriptaddress,longbibliography]{revtex4-2}

\usepackage{amsmath,amssymb}
\usepackage{graphicx}
\usepackage{color}
\usepackage[colorlinks=true,citecolor=blue,linkcolor=blue,urlcolor=blue]{hyperref}

\renewcommand{\eqref}[1]{\textup{(\ref{#1})}}
\definecolor{sigbg}{rgb}{0.91,0.94,0.97}

\begin{document}

\title{Twisted R\'enyi Negativity as a Reliable Proxy for Mixed-State Entanglement in Fermionic Systems}

\author{Fo-Hong Wang}
\affiliation{Tsung-Dao Lee Institute, Key Laboratory of Artificial Structures and Quantum Control (Ministry of Education), School of Physics and Astronomy, Shanghai Jiao Tong University, Shanghai 200240, China}

\author{Xiao Yan Xu}
\email{xiaoyanxu@sjtu.edu.cn}
\affiliation{Tsung-Dao Lee Institute, Key Laboratory of Artificial Structures and Quantum Control (Ministry of Education), School of Physics and Astronomy, Shanghai Jiao Tong University, Shanghai 200240, China}
\affiliation{Hefei National Laboratory, Hefei 230088, China}

\date{\today}

\begin{abstract}
Characterizing mixed-state entanglement in fermionic quantum matter represents a fundamental challenge at the intersection of condensed matter physics and quantum information. Unlike pure states, mixed states lack a universally accepted and easily computable entanglement measure, a difficulty compounded in fermionic systems by the ambiguity of the partial transpose operation.
Here, motivated by the Hermiticity of the partially transposed density matrix and the analytic continuation of the R\'enyi negativity to logarithmic negativity, we address this puzzle by demonstrating the ``twisted'' R\'enyi negativity as a physically consistent proxy for mixed-state entanglement in the Hubbard and spinless $t$-$V$ models.
Through large-scale quantum Monte Carlo simulations of these paradigmatic correlated systems, we show that this indicator captures essential physical expectations---including the area law and monotonic thermal suppression---while excluding anomalous behavior present in alternative definitions.
Our findings not only provide a robust theoretical framework for quantifying entanglement in strongly correlated electronic systems but also offer a practical tool for experimental detection in rapidly developing platforms, including ultracold atomic gases and programmable quantum simulators.
\end{abstract}

\maketitle

\noindent\colorbox{sigbg}{\parbox{\dimexpr\linewidth-2\fboxsep\relax}{\small
\textbf{Significance statement}\par\smallskip
Entanglement is a defining feature of quantum mechanics, yet quantifying it in realistic, finite-temperature environments (mixed states) remains a formidable challenge, particularly for fermionic systems like electrons in solids. A long-standing ambiguity exists regarding how to correctly adapt standard entanglement measures to fermions, leading to a plethora of proposed definitions with unclear physical implications.
We address this uncertainty by demonstrating that the ``twisted'' R\'enyi negativity serves as a physically consistent proxy in two paradigmatic models, where alternative definitions exhibit anomalous, nonmonotonic behavior. This advance suggests a theoretical diagnostic for characterizing mixed-state entanglement in condensed-matter systems, and a possible experimentally accessible target for quantum simulation platforms such as ultracold atoms and quantum processors.}}

\section{Introduction}
Quantum entanglement plays a key role in understanding exotic phases and quantum criticality in many-body systems~\cite{Phys.Rev.Lett.2003Vidal,Rev.Mod.Phys.2008Amico,PhysicsReports2016Laflorencie,QuantumInfo.Comput.2007Plenio,Rev.Mod.Phys.2009Horodecki,J.Stat.Mech.Theor.Exp.2004Calabrese,Phys.Rev.Lett.2006Gioev,Phys.Rev.Lett.2006Kitaev,Phys.Rev.Lett.2006Levin,Phys.Rev.Lett.2006Fradkin,Phys.Rev.Lett.2008Wolf,J.Phys.AMath.Theor.2009Calabrese,Phys.Rev.A2009Wichterich,Phys.Rev.Lett.2010Hastings,J.Stat.Mech.2013Calabresea,Phys.Rev.Lett.2013Grover,ArXiv2015Metlitski,Phys.Rev.B2016Ruggiero,J.Phys.AMath.Theor.2020Gruber,Phys.Rev.B2020Wybo,Phys.Rev.Lett.2018Gray,Phys.Rev.A2019Cornfeld,Phys.Rev.Lett.2020Elben,npjQuantumInf2021Neven,Phys.Rev.Lett.2022Zhao,J.Stat.Mech.2019Shapourian,SciPostPhys.2023Alba,Phys.Rev.B2024Choi,Phys.Rev.Lett.2024DEmidio,npjQuantumInf2025Liao,Phys.Rev.Res.2024Parez,Phys.Rev.B2016Chang,PRB2025JQ.Fang,Nat.Commun.2025Wang}. While entanglement entropy effectively characterizes pure-state entanglement, it fails to distinguish quantum correlations from classical ones in mixed states. Mixed-state entanglement is ubiquitous in nature: it arises at finite temperatures, in open quantum systems, and after measurements in quantum computing platforms. Quantifying such entanglement remains a central challenge in quantum information science.

For mixed states, logarithmic negativity (LN)~\cite{Phys.Rev.A1998Zyczkowski,J.Mod.Opt.1999Eisert,Phys.Rev.A2002Vidal,Phys.Rev.Lett.2005Plenio}, based on the positive partial transpose criterion~\cite{Phys.Rev.Lett.1996Peres,PhysicsLettersA1996Horodecki}, provides a computable entanglement measure. However, direct computation of LN in large-scale interacting systems remains challenging. This computational bottleneck has limited our understanding of mixed-state entanglement in strongly correlated quantum matter, where interactions and thermal fluctuations compete.

For bosonic systems, the moments of the partially transposed density matrix, termed R\'enyi negativity, provide a computationally accessible alternative~\cite{Phys.Rev.Lett.2012Calabrese,J.Stat.Mech.2013Calabresea,J.Stat.Mech.2013Calabrese}. Previous studies show that the deviation of R\'enyi negativity from R\'enyi entropy of the same rank---the R\'enyi negativity ratio (RNR)---captures entanglement properties~\cite{J.Stat.Mech.2013Alba,Phys.Rev.B2014Chunga,Phys.Rev.Lett.2020Wu,Phys.Rev.B2025Ding,J.Stat.Mech.2014Coser,Phys.Rev.B2020Wybo,Murciano2022QuenchDynamicsRenyi}. RNR exhibits singular behavior at phase transitions and obeys area-law scaling, serving as a reliable proxy for LN, analogous to how R\'enyi entanglement entropy plays as a proxy for entanglement entropy in pure states.

Fermionic systems are more complex due to the fermionic partial transpose~\cite{Phys.Rev.B2017Shapourian,Phys.Rev.B2018Shiozaki,Phys.Rev.A2019Shapourian}, which respects anticommutation statistics. Two fermionic partial transpose definitions exist: \emph{untwisted} and \emph{twisted}~\cite{SciPostPhys.2019Shapourian}. While both yield identical LN, they produce distinct R\'enyi negativities. The untwisted partially transposed density matrix is generally non-Hermitian, whereas the twisted counterpart is Hermitian and the twisted R\'enyi negativity can be analytically continued to LN via $\mathcal{E}=\lim_{r\to 1/2}(1-2r)\tilde{\mathcal{E}}_{2r}$~\cite{SciPostPhys.2019Shapourian}. A key question is then which R\'enyi negativity variant correctly characterizes mixed-state entanglement in fermionic systems.

In this work, we develop stable determinant quantum Monte Carlo (DQMC) algorithms to compute high-rank untwisted and twisted R\'enyi negativities in interacting fermionic systems. Numerical studies of the Hubbard model and spinless $t$-$V$ model reveal key differences between the two variants.
Within the temperature ranges and system sizes studied, the twisted RNR exhibits monotonic temperature dependence and area-law scaling, whereas untwisted RNRs can show nonmonotonic temperature dependence and beyond-area-law behavior near replicated critical points.
Together with the Hermiticity of the twisted partially transposed density matrix and the analytic continuation of the twisted R\'enyi negativity to LN, these results support the twisted RNR as a promising proxy and motivate further applications in quantum materials, ultracold atoms, and quantum simulators.

\section{Untwisted and twisted R\'enyi negativity}
We consider a bipartite geometry $A=A_1\cup A_2$ within a many-body system, where $A$ may represent a Gibbs state at finite temperature, or a subsystem embedded in environment $\bar{A}$ as shown in Fig.~\ref{fig:overview}(b). We quantify entanglement between $A_1$ and $A_2$ via the fermionic partial transpose. Unlike the bosonic case, the fermionic partial transpose must respect both fermionic anticommutation relations and parity conservation. The untwisted fermionic partial transpose with respect to $A_2$ is defined in the Fock space as~\cite{Phys.Rev.B2017Shapourian}
\begin{equation}
\begin{aligned}
&(|n_{A_1},n_{A_2}\rangle\langle\bar{n}_{A_1},\bar{n}_{A_2}|)^{T_2^f}\\
&= (-1)^{\phi(n,\bar{n})}|n_{A_1},\bar{n}_{A_2}\rangle\langle\bar{n}_{A_1},n_{A_2}|,
\end{aligned}
\end{equation}
where the phase $\phi(n,\bar{n})$ distinguishes the fermionic partial transpose from the bosonic one. The untwisted partially transposed density matrix is defined as $\rho^{{T}_2^f}$, whereas its twisted counterpart is defined as $\rho^{\tilde{T}_2^f}\equiv\rho^{T_2^f}X_2$, with $\rho$ being the density matrix for system $A$ and $X_2=(-1)^{\sum_{j\in A_2}c^\dagger_j c_j}$ being a disorder operator.
Throughout this manuscript, a tilde denotes the twisted counterpart of the corresponding untwisted quantity.
Crucially, the untwisted partially transposed density matrix $\rho^{T_2^f}$ is generally non-Hermitian, while its twisted counterpart $\rho^{\tilde{T}_2^f}$ is Hermitian. The LN is defined as trace norm of the partially transposed density matrix, and both untwisted and twisted fermionic partial transposes yield the same LN but produce different R\'enyi negativities.
The rank-$r$ untwisted R\'enyi negativity is defined as
\begin{equation}
\mathcal{E}_r = \frac{1}{1-r}\ln\mathrm{Tr}[(\rho^{T_2^f})^r],
\end{equation}
and similarly $\tilde{\mathcal{E}}_r$ for the twisted partial transpose $\rho^{\tilde{T}_2^f}$.
To isolate entanglement from thermal entropy, the RNR is defined as $R_r = \mathcal{E}_r - S_r^{\mathrm{th}}$ (similarly for $\tilde{R}_r$), where $S_r^{\mathrm{th}}$ is the thermal R\'enyi entropy. For bosonic systems, RNR faithfully captures entanglement structure~\cite{Phys.Rev.Lett.2020Wu,Phys.Rev.B2025Ding}; our goal is to identify which fermionic variant plays this role in fermionic systems.

As illustrated in Fig.~\ref{fig:overview}(c), the fermionic partial transpose preserves Gaussianity of free-fermion Gaussian states, analogous to partial trace~\cite{Phys.Rev.B2017Shapourian}. The untwisted partially transposed Green's function reads
\begin{equation}
G^{T_2^f} = \begin{pmatrix} G_{11} & \mathrm{i}G_{12} \\ \mathrm{i}G_{21} & I_2-G_{22} \end{pmatrix},
\end{equation}
while the twisted counterpart $G^{\tilde{T}_2^f}$ requires additional transformations~\cite{SciPostPhys.2019Shapourian},
\begin{equation}
\label{equ:twistedPTGreen}
        G^{\tilde{T}_{2}^{f}}=\left(\begin{array}{cc}
        I_{1} & 2G_{12}^{T_{2}^{f}}\\
        0 & 2G_{22}^{T_{2}^{f}}-I_{2}
        \end{array}\right)^{-1}G^{T_{2}^{f}}.
\end{equation}

For interacting fermions, we employ DQMC, where Hubbard-Stratonovich decoupling yields $\rho=\sum_{\mathbf{s}}p_{\mathbf{s}}\rho_{\mathbf{s}}$ over auxiliary field configurations $\mathbf{s}$, with each $\rho_{\mathbf{s}}$ being Gaussian. Exploiting Gaussianity preservation under the fermionic partial transpose, rank-$r$ R\'enyi negativities become
\begin{subequations}
\label{equ:expRenyi}
\begin{equation}\label{equ:expRenyiNk_untwisted_DQMC}
    e^{\left(1-r\right)\mathcal{E}_{r}}=\sum_{\mathbf{s}^{(1)}\cdots\mathbf{s}^{(r)}}\left(\prod_{i=1}^{r}p_{\mathbf{s}^{(i)}}\right)\det g_{r,\bar{\boldsymbol{s}}},
\end{equation}
and
\begin{equation}\label{equ:expRenyiNk_twisted_DQMC}
        e^{\left(1-r\right)\tilde{\mathcal{E}}_{r}}=\sum_{\mathbf{s}^{(1)}\cdots\mathbf{s}^{(r)}}\left(\prod_{i=1}^{r}p_{\mathbf{s}^{(i)}}\right)\left(\prod_{i=1}^{r}\mathcal{Z}_{\tilde{T}_{2}^{f},\mathbf{s}^{(i)}}\right)\det\tilde{g}_{r,\bar{\boldsymbol{s}}},
\end{equation}
\end{subequations}
where $g_{r,\bar{\boldsymbol{s}}}$ and $\tilde{g}_{r,\bar{\boldsymbol{s}}}$ are the Grover matrices involving products of partially transposed Green's functions,
\begin{subequations}\label{equ:gk_DQMC}
    \begin{equation}\label{equ:gk_untwisted_DQMC}
        g_{r,\bar{\boldsymbol{s}}}\equiv\left[\prod_{i=r}^{1}G_{\mathbf{s}^{(i)}}^{T_{2}^{f}}\right]\left[I+\prod_{i=r}^{1}\left(\left(G_{\mathbf{s}^{(i)}}^{T_{2}^{f}}\right)^{-1}-I\right)\right],
    \end{equation}
    and
    \begin{equation}\label{equ:gk_twisted_DQMC}
        \tilde{g}_{r,\bar{\boldsymbol{s}}}\equiv\left[\prod_{i=r}^{1}G_{\mathbf{s}^{(i)}}^{\tilde{T}_{2}^{f}}\right]\left[I+\prod_{i=r}^{1}\left(\left(G_{\mathbf{s}^{(i)}}^{\tilde{T}_{2}^{f}}\right)^{-1}-I\right)\right].
    \end{equation}
\end{subequations}
Here $\bar{\boldsymbol{s}}\equiv\{\mathbf{s}^{(1)},\ldots,\mathbf{s}^{(r)}\}$ denotes $r$ replica configurations collectively. The twisted R\'enyi negativity includes additional normalization factors $\mathcal{Z}_{\tilde{T}_{2}^{f},\mathbf{s}}=\det(I_{2}-2G_{\mathbf{s},22})$. Importantly, rank-2 twisted R\'enyi negativity trivially equals rank-2 R\'enyi entropy, making rank 3 the smallest nontrivial rank---analogous to bosonic systems~\cite{Phys.Rev.Lett.2020Wu}. This necessitates studying high-rank R\'enyi negativities to identify the proper LN proxy.
Among these, we focus on even ranks, as odd-rank twisted R\'enyi negativity can encounter additional sign problem for specific bipartition geometries (see Appendix~D of the Supplementary Information).

\begin{figure}[t]
    \centering
    \includegraphics[width=\linewidth]{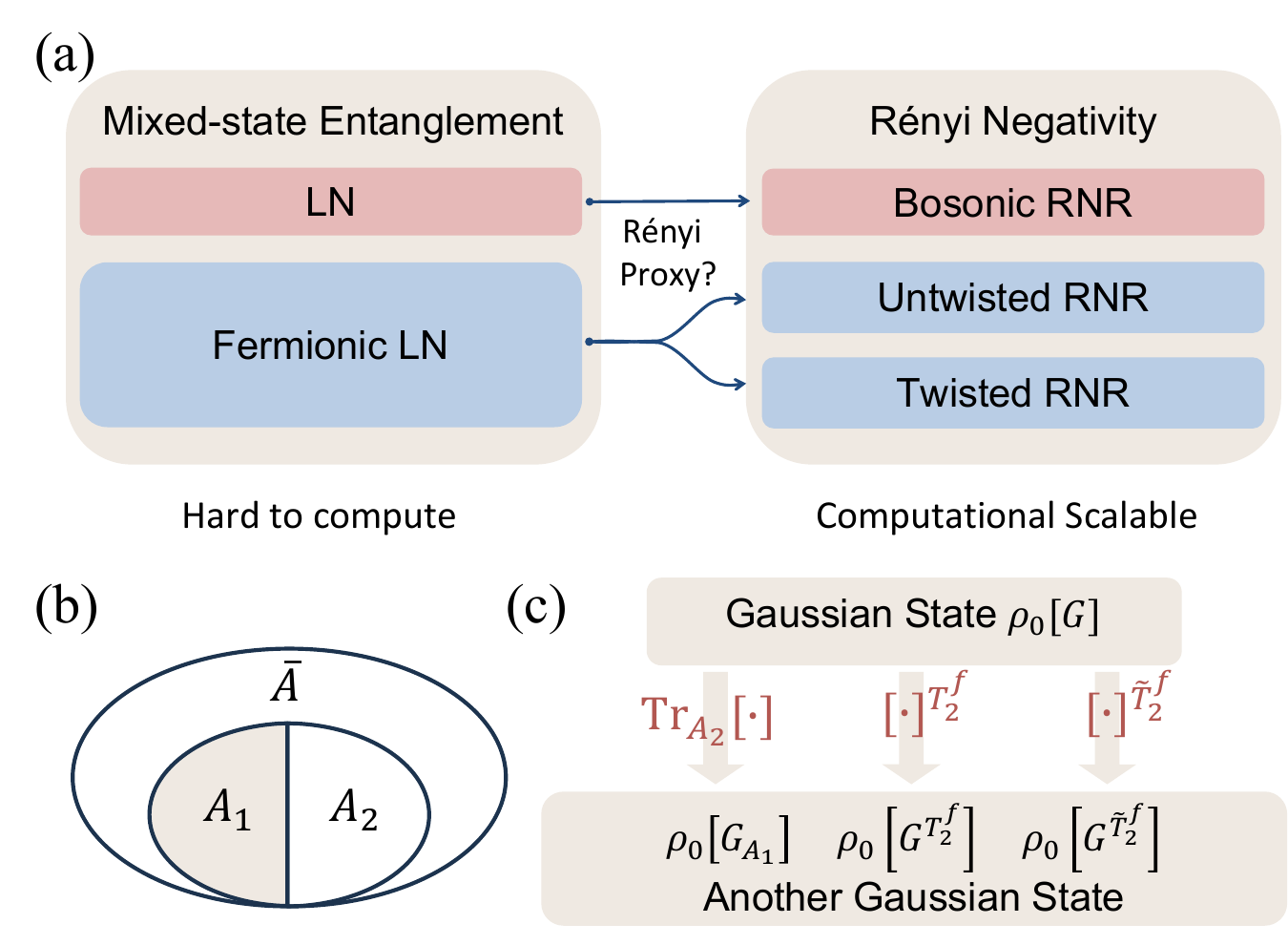}
    \caption{(a) Overview of various negativity definitions. While LN and its fermionic analog are computationally challenging, RNRs are scalable with quantum Monte Carlo algorithms. This work seeks to identify a R\'{e}nyi proxy for LN in fermionic systems. 
    (b) Tripartition geometry: system $A=A_{1}\cup A_{2}$ is embedded in environment $\bar A$, and we examine entanglement between subsystems $A_1$ and $A_2$.
    (c) Foundation underlying DQMC algorithms for R\'enyi entanglement entropy and RNRs. Operations $\mathrm{Tr}_{A_{2}}[\cdot]$ (partial trace), $[\cdot]^{T^{f}_{2}}$ (\textit{untwisted} fermionic partial transpose) and $[\cdot]^{\tilde{T}^{f}_{2}}$ (\textit{twisted} fermionic partial transpose) transform Gaussian states with modified Green's functions.}
    \label{fig:overview}
\end{figure}

\section{Computational challenges and solutions}
Computing high-rank R\'enyi negativities via Eqs.~\eqref{equ:expRenyi} and \eqref{equ:gk_DQMC} faces two critical challenges: (i) numerical instability from inverting partially transposed Green's functions, and (ii) exponentially large variance in Monte Carlo sampling.
We address them using stabilization schemes and incremental algorithms.

Challenge (i) originates from the singularities that the partially transposed Green's functions $G^{T_2^f}$ and $G^{\tilde{T}_2^f}$ inherit from the standard Green's function $G$.
As temperature decreases, their smallest singular values decay exponentially, causing matrix inversion to become unstable when approaching machine precision.
The severity depends on the model, the Hubbard-Stratonovich transformation, the bipartition geometry (with equal bipartition being the most stable), and the temperature.

We develop two complementary solutions. First, we employ Loh's stabilization scheme~\cite{Loh1989StableMatrixMultiplicationAlgorithms} with singular value decomposition, separating exponentially large and small scales. This remains stable until the breakdown temperature where singular values approach $\sim 10^{-16}$. To extend calculations to lower temperatures, we regularize vanishing singular values to around $10^{-13}$, ensuring systematic errors remain below statistical uncertainties. Second, we implement Drut's reconstruction~\cite{Phys.Rev.E2016Drut} that avoids matrix inversion entirely. Both methods yield consistent results across all temperatures.
The origin of this instability, the stabilization formulas, and free-fermion benchmarks are detailed in Sec.~S4 of the Supplementary Information.

For challenge (ii), direct Monte Carlo sampling of Eq.~\eqref{equ:expRenyi} suffers from exponentially growing variance as rank $r$ or system size increases. We develop incremental algorithms~\cite{Phys.Rev.Lett.2020DEmidio,npjQuantumMater.2022Zhao,Phys.Rev.Lett.2022Zhao,Phys.Rev.B2023Pan,npjQuantumInf2025Liao,ArXiv2023Liao,Phys.Rev.B2024Zhang,Phys.Rev.Lett.2024DEmidio,Stat.Comput.2001Neal,NewJ.Phys.2008Pollet,Phys.Rev.E2017Alba,JHEP2023Bulgarelli,Phys.Rev.B2025Dai,Phys.Rev.B2024Ding} decomposing the calculation into $N_{\mathrm{inc}}$ factors:
\begin{equation}
e^{(1-r)\mathcal{E}_r} = \prod_{k=0}^{N_{\mathrm{inc}}-1}\frac{Z_{k+1}}{Z_k}.
\end{equation}
Each factor remains $\mathcal{O}(1)$ and this method hence exponentially reduces variance.
We derive low-rank fast updates for the modified Monte Carlo weight incorporating the Grover determinant $\det g_{r,\bar{\boldsymbol{s}}}$, achieving $\mathcal{O}(N^{2})$ cost per local update, where $N$ is the number of lattice sites.
Each incremental process thus retains the per-sweep scaling $\mathcal{O}(\beta N^{3})$ with $\beta$ the inverse temperature, and the full calculation costs $\mathcal{O}(N_{\mathrm{inc}}\,\beta\,N^{3})$ over $N_{\mathrm{inc}}$ embarrassingly parallel processes.
A complete algorithmic description, including update formulas, stabilization routines, and practical implementation, is given in Sec.~S5 of the Supplementary Information.

These algorithmic advances enable a systematic study of high-rank R\'enyi negativities in interacting fermionic systems.
In practice, we deploy them differently across dimensions and relevant system sizes.
For the 1D Hubbard chain, conventional DQMC with the new R\'enyi-negativity measurements in Eq.~\eqref{equ:expRenyi}, combined with the stabilization formulas for challenge~(i), can reach arbitrarily low temperatures, including the ground state within the projector formalism.
For the 2D spinless $t$-$V$ model, we use the incremental algorithm to accurately compute the R\'enyi negativity ratios, whose area-law growth makes the corresponding Grover determinant exponentially small.

\section{Quantum-classical crossover in the Hubbard model}
We first examine the half-filled Hubbard chain with Hamiltonian $\hat{H}=-t\sum_{i,\sigma}(c^\dagger_{i,\sigma}c_{i+1,\sigma}+\text{H.c.})+\frac{U}{2}\sum_i(\hat{n}_i-1)^2$. We set $U/t=1$ and study system sizes $L=6,10,14$ with equal bipartition $L_{A_1}=L/2$. Figure~\ref{fig:hubbard} compares the twisted and untwisted RNRs with the rank-2 R\'enyi entanglement entropy across different temperatures.
Another finite-temperature reference, the R\'enyi mutual information, is compared with both RNRs in the Discussion and the Supplementary Information.

\begin{figure}[t]
\centering
\includegraphics[width=0.48\textwidth]{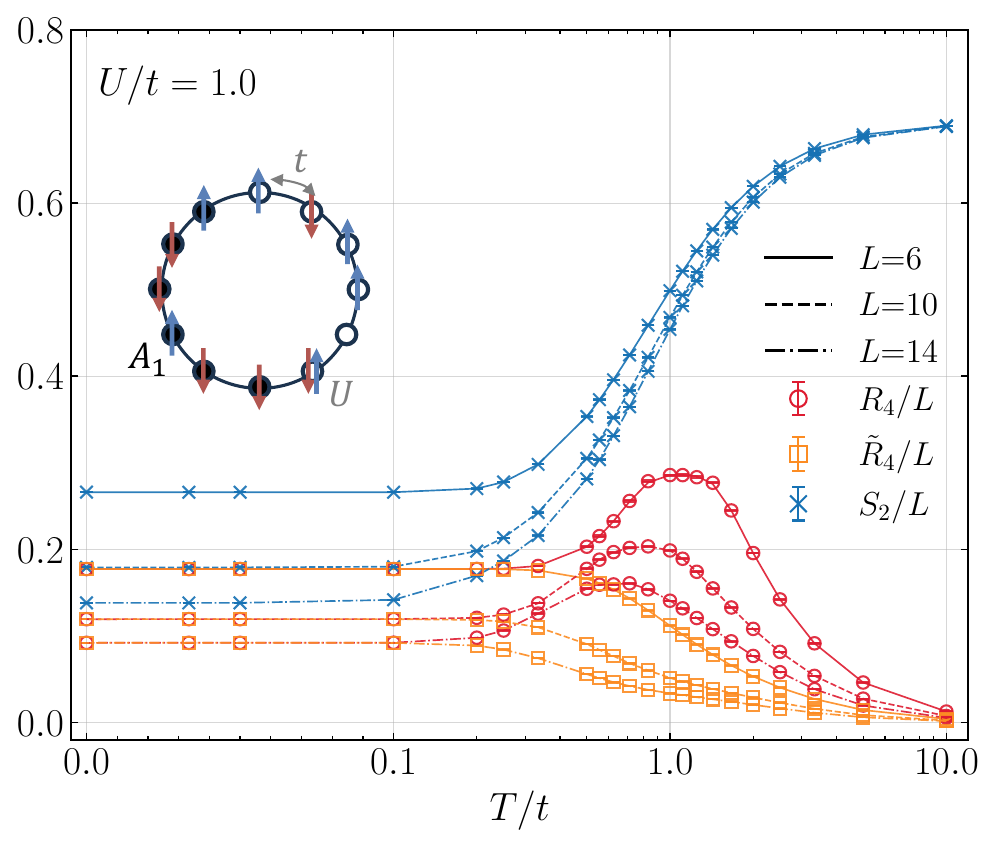}
\caption{Quantum-classical crossover in the Hubbard chain ($U/t=1$, half-filling). The temperature dependence of rank-4 untwisted RNR ($R_{4}/L$) and rank-4 twisted RNR ($\tilde{R}_4/L$) are compared with rank-2 R\'enyi entanglement entropy ($S_2/L$). While $\tilde{R}_4$ decreases monotonically with temperature (capturing the crossover), $R_4$ exhibits anomalous nonmonotonic behavior with local maxima. The subsystem length is $L_{A_1}=L/2$.
All symbols are DQMC data (projector DQMC at $T=0$, finite-temperature DQMC otherwise), with error bars smaller than symbol size.
The solid line ($L=6$) is the exact-diagonalization reference, while the dashed/dash-dot lines ($L=10,14$) connect the DQMC points as guides to the eye.}
\label{fig:hubbard}
\end{figure}

The R\'enyi entanglement entropy fails to distinguish quantum from classical correlations at finite temperatures. As shown in Fig.~\ref{fig:hubbard}, $S_2/L$ increases monotonically with temperature, reaching the maximum entropy per site $\ln 2$ in the high temperature limit, where quantum entanglement actually vanishes and $S_2$ follows a volume law. This inherent deficiency requires alternative measures such as LN or its proxy for mixed-state entanglement.

To this end, the twisted RNR successfully serves as a proper proxy for LN (see Figs.~S6 and S7 in the Supplementary Information for exact-diagonalization validation and multirank analysis). As shown in Fig.~\ref{fig:hubbard}, the rank-4 twisted RNR ($\tilde{R}_4/L$) decreases monotonically with increasing temperature, correctly capturing the quantum-classical crossover.
At low temperatures, quantum entanglement dominates and RNR persists in a plateau set by the finite-size gap ($\sim 1/L$), consistent with the plateau narrowing with increasing system length.
As temperature rises, thermal fluctuations suppress entanglement, and thus RNR decreases, which is consistent with physical intuition. Quantitatively, $\tilde{R}_4/L$ drops from $\sim 0.1$ at $T/t=0.2$ to nearly zero at $T/t=10.0$ for $L=14$.

By contrast, the rank-4 untwisted RNR ($R_4/L$) exhibits nonmonotonic temperature dependence with local maxima at intermediate temperatures. Specifically, $R_4/L$ peaks at $T/t\sim 0.8$ for $L=14$, before decreasing at higher temperatures. This anomalous behavior confirms that untwisted RNR fails to serve as a reliable entanglement measure---a failure rooted in the non-Hermitian nature of the untwisted partially transposed density matrix and its lack of analytical continuation to LN.

\begin{figure}[t]
\centering
\includegraphics[width=0.48\textwidth]{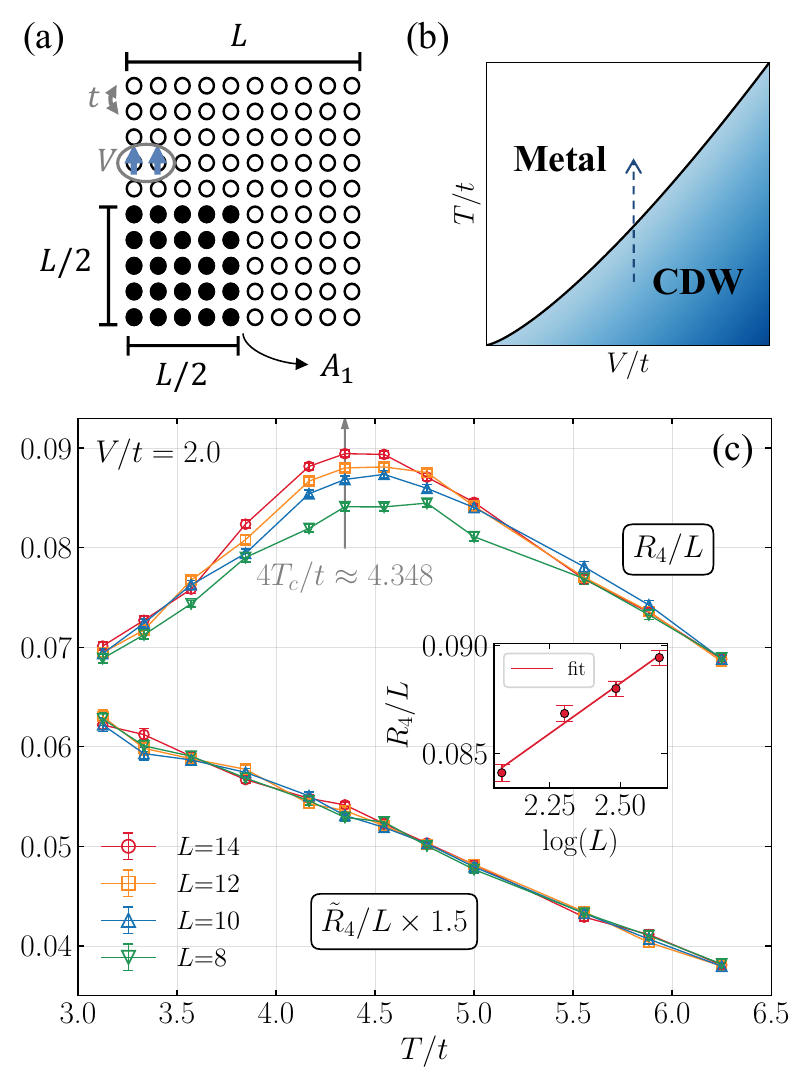}
\caption{Finite-temperature phase transition in the 2D spinless $t$-$V$ model ($V/t=2$, half-filling). (a) Model illustration and bipartition geometry. (b) Illustrative phase diagram with a dashed arrow indicating the cut shown in (c).
(c) Temperature dependence of area-law coefficients for rank-4 untwisted RNRs ($R_4/L$) and twisted RNR ($\tilde{R}_4/L$, scaled by 1.5). Untwisted RNRs peak at replicated critical temperatures $4T_c$ (indicated by a vertical gray arrow) with beyond-area-law scaling (insets show logarithmic divergence with system size $L$). Twisted RNR decreases monotonically and obeys area law throughout, consistent with bosonic systems~\cite{Phys.Rev.Lett.2020Wu}.}
\label{fig:tv}
\end{figure}

\section{Critical behavior in the $t$-$V$ model}
We next investigate the spinless $t$-$V$ model on a square lattice, with Hamiltonian $H=-t\sum_{\langle i,j\rangle}(c_i^\dagger c_j+\text{H.c.})+V\sum_{\langle i,j\rangle}(n_i-\frac{1}{2})(n_j-\frac{1}{2})$, illustrated in Fig.~\ref{fig:tv}(a).
At $V/t=2$, this model exhibits a finite-temperature phase transition from charge density wave to metallic phase at $T_c/t\approx 1.0$ [Fig.~\ref{fig:tv}(b)], belonging to the 2D Ising universality class~\cite{Phys.Rev.B1985Gubernatis}.
We study square lattices with linear size $L=8,10,12,14$ and periodic boundary conditions, selecting subsystem $A_1$ as the lower-left corner with boundary length $L/2$ [Fig.~\ref{fig:tv}(a)].

Figure~\ref{fig:tv}(c) reveals profound differences between untwisted and twisted RNRs near the critical point. The rank-4 untwisted RNR shows nonmonotonic temperature dependence, displaying a pronounced peak at replicated critical temperature $T=4T_c$.
In general, the singularity of the rank-$r$ untwisted RNR is anticipated at $T=rT_c$ with $T_c$ being the physical transition temperature. This expectation arises because $R_r$, as indicated by $\mathrm{Tr}[\rho^r]\propto\mathrm{Tr}[e^{-r\beta H}]$ in its definition, effectively probes the system at an effective inverse temperature $r\beta$.
The peak height increases systematically with system size. Crucially, finite-size scaling [Fig.~\ref{fig:tv}(c) inset] shows that untwisted RNR exhibits \emph{beyond-area-law} scaling with logarithmic divergence $R_4/L\sim\ln L$ at $4T_c$.
This singular behavior also occurs for other ranks of untwisted RNRs~\cite{Nat.Commun.2025Wanga} (see also Fig.~S8 in the Supplementary Information).
Away from $4T_c$, the untwisted RNR recovers area-law scaling, with $R_4/L$ approaching size-independent values within error bars at both low ($T/t\lesssim 3.5$) and high ($T/t\gtrsim 5.5$) temperatures [Fig.~\ref{fig:tv}(c)].

In contrast, the rank-4 twisted RNR $\tilde{R}_4$ decreases monotonically with temperature and obeys the area law throughout. Quantitatively, $\tilde{R}_4/L$ decreases smoothly, with all system sizes collapsing onto a single curve within error bars [Fig.~\ref{fig:tv}(c)]. This behavior precisely matches bosonic RNR in the 2D transverse field Ising model~\cite{Phys.Rev.Lett.2020Wu,Phys.Rev.Research2020Lu}, where the singularity appears only in the temperature derivative. We note that precise measurement of the temperature derivative of twisted RNR itself is also technically challenging; it requires a tailored incremental method, and we leave it as future work.
The twisted RNR's consistency with bosonic systems---both defined via Hermitian partially transposed density matrices with analytical continuation to LN---strongly supports its physical reliability.
Moreover, the absence of beyond-area-law scaling in twisted RNR indicates it captures genuine bipartite entanglement rather than replica-induced artifacts.

A complementary fixed-temperature interaction scan, corresponding to a horizontal cut of the phase diagram in Fig.~\ref{fig:tv}(b), is provided in Fig.~S9 of the Supplementary Information. Along this cut, the twisted RNR remains area-law-like across the phase boundary.

\section{Discussion}
Our investigation reveals fundamental distinctions between untwisted and twisted RNRs as entanglement proxies in fermionic systems. The twisted RNR demonstrates three essential advantages, systematically compared with untwisted RNR across models, temperature regimes, and multiple ranks in Table~S2 of the Supplementary Information. First, the Hermiticity of the twisted partially transposed density matrix $\rho^{\tilde{T}_2^f}$ guarantees real eigenvalues and well-defined moments, contrasting with the pseudo-Hermitian nature of $\rho^{T_2^f}$. Second, even-rank twisted R\'enyi negativities admit analytical continuation to logarithmic negativity through $\mathcal{E}=\lim_{r\rightarrow 1/2}(1-2r)\tilde{\mathcal{E}}_{2r}$, providing a direct connection to the entanglement measure. Third, in the models studied, twisted RNR exhibits physically consistent behavior---monotonic temperature dependence and area-law scaling.

The anomalous features of untwisted RNR observed here stem from the non-Hermitian structure of $\rho^{T_2^f}$. The nonmonotonic temperature dependence in the Hubbard chain contradicts the thermal suppression of entanglement expected in this model. The beyond-area-law scaling at finite-temperature critical points, though capturing correlation structure, deviates from the area-law scaling expected of an entanglement proxy. Quantitatively, untwisted RNRs can exhibit local maxima at intermediate temperatures for all system sizes studied, whereas the twisted RNR decreases monotonically over the temperature range examined.

The critical behavior in the spinless $t$-$V$ model further illuminates these distinctions. Untwisted RNRs display pronounced peaks at replicated critical temperatures $T=rT_c$, with finite-size scaling revealing logarithmic divergence $R_r/L\sim\ln L$. Away from $rT_c$, untwisted RNR recovers area-law scaling. In contrast, the twisted RNR decreases monotonically across the transition, maintaining area-law scaling.

These findings parallel the relationship between R\'enyi entanglement entropy and entanglement entropy for pure states. Just as R\'enyi entanglement entropy provides a computationally tractable proxy for entanglement entropy, the twisted RNR serves as a promising proxy for LN in fermionic mixed states.

Finally, we compare the twisted RNR with the R\'enyi mutual information, a well-established probe of finite-temperature correlations.
In the Hubbard chain, the rank-2 R\'enyi mutual information $I_2$ decreases monotonically with temperature in the grand-canonical ensemble~\cite{PRB2013L.Bonnes}, a behavior also exhibited by the twisted RNR (see Fig.~S7(b) in the Supplementary Information).
Near the finite-temperature Ising transition in the two-dimensional $t$-$V$ model, however, the diagnostics diverge.
The area-law coefficient of $I_2$ exhibits a finite-size crossing pattern driven by sign changes in the constant term~\cite{PRL2011RRP.Singh,PRB2013J.Iaconis,Phys.Rev.Lett.2014Wang}, while the twisted RNR remains area-law-like and thus features expected behavior for quantum entanglement in this model~\cite{Phys.Rev.Lett.2020Wu,Phys.Rev.Research2020Lu}.
The shared and contrasting behaviors across both models are summarized in Table~S3 of the Supplementary Information.

\section{Conclusions and Outlook}
In this work, we have addressed the ambiguity surrounding the characterization of mixed-state entanglement in fermionic systems by demonstrating the twisted RNR as a physically consistent proxy for logarithmic negativity in the Hubbard and spinless $t$-$V$ models. Our systematic quantum Monte Carlo simulations of these models demonstrate that the twisted RNR captures essential physical features---including monotonic thermal suppression and area-law scaling---whereas the untwisted counterpart can suffer from anomalous divergences and nonmonotonicity. We attribute this distinction to the fundamental mathematical structure: the twisted partial time-reversal defines a Hermitian operator analytically connected to the logarithmic negativity, while the untwisted definition lacks this crucial property.

While our results are underpinned by advances in DQMC algorithms, the implications extend far beyond specific numerical techniques. The twisted RNR is defined independently of any specific numerical technique and can be computed by any numerical method that provides access to the partially transposed density matrix. Consequently, our work provides a clear protocol for future investigations using other powerful methods, such as exact diagonalization, density matrix renormalization group, and tensor networks, particularly for models where the sign problem precludes Monte Carlo approaches.

Looking forward, establishing the range of applicability of the twisted RNR calls for systematic tests across additional fermionic models, bipartitions, ranks, and parameter regimes. Theoretically, such tests could include precise comparisons with conformal field theory predictions, potentially revealing universal scaling properties of fermionic entanglement near quantum critical points. Experimentally, our findings suggest a possible route toward detecting mixed-state entanglement in programmable quantum simulators and ultracold atomic gases. By building upon existing measurement protocols for R\'enyi quantities~\cite{Phys.Rev.Lett.2018Gray,Phys.Rev.A2019Cornfeld,Phys.Rev.Lett.2020Elben,npjQuantumInf2021Neven}, the twisted RNR can serve as a practical and robust observable, bridging the gap between abstract entanglement theory and laboratory realization in strongly correlated quantum matter.

\section*{Data, Materials, and Software Availability}
Monte Carlo raw data supporting the findings of this study have been deposited in Shanghai Jiao Tong University Scidata~\cite{SJTUScidata2026Wang}. The algorithms are described in detail in Secs.~S4 and S5 of the Supplementary Information. All other study data are included in the article and/or the Supplementary Information.

\section*{Acknowledgments}
\begin{acknowledgments}
This work was supported by the National Key R\&D Program of China (Grant No. 2022YFA1402702, No. 2021YFA1401400), the National Natural Science Foundation of China (Grants No. 12447103, No. 12274289, No. 125B2077), the Innovation Program for Quantum Science and Technology (under Grant No. 2021ZD0301902), Yangyang Development Fund, Shanghai Jiao Tong University 2030 Initiative, and startup funds from Shanghai Jiao Tong University. F.-H.W. is supported by T.D. Lee scholarship. Computations were performed on Siyuan-1 and $\pi$ 2.0 clusters at Shanghai Jiao Tong University.
\end{acknowledgments}

\bibliography{../reference_increm}

\begin{thebibliography}{105}%
\makeatletter
\providecommand \@ifxundefined [1]{%
 \@ifx{#1\undefined}
}%
\providecommand \@ifnum [1]{%
 \ifnum #1\expandafter \@firstoftwo
 \else \expandafter \@secondoftwo
 \fi
}%
\providecommand \@ifx [1]{%
 \ifx #1\expandafter \@firstoftwo
 \else \expandafter \@secondoftwo
 \fi
}%
\providecommand \natexlab [1]{#1}%
\providecommand \enquote  [1]{``#1''}%
\providecommand \bibnamefont  [1]{#1}%
\providecommand \bibfnamefont [1]{#1}%
\providecommand \citenamefont [1]{#1}%
\providecommand \href@noop [0]{\@secondoftwo}%
\providecommand \href [0]{\begingroup \@sanitize@url \@href}%
\providecommand \@href[1]{\@@startlink{#1}\@@href}%
\providecommand \@@href[1]{\endgroup#1\@@endlink}%
\providecommand \@sanitize@url [0]{\catcode `\\12\catcode `\$12\catcode
  `\&12\catcode `\#12\catcode `\^12\catcode `\_12\catcode `\%12\relax}%
\providecommand \@@startlink[1]{}%
\providecommand \@@endlink[0]{}%
\providecommand \url  [0]{\begingroup\@sanitize@url \@url }%
\providecommand \@url [1]{\endgroup\@href {#1}{\urlprefix }}%
\providecommand \urlprefix  [0]{URL }%
\providecommand \Eprint [0]{\href }%
\providecommand \doibase [0]{https://doi.org/}%
\providecommand \selectlanguage [0]{\@gobble}%
\providecommand \bibinfo  [0]{\@secondoftwo}%
\providecommand \bibfield  [0]{\@secondoftwo}%
\providecommand \translation [1]{[#1]}%
\providecommand \BibitemOpen [0]{}%
\providecommand \bibitemStop [0]{}%
\providecommand \bibitemNoStop [0]{.\EOS\space}%
\providecommand \EOS [0]{\spacefactor3000\relax}%
\providecommand \BibitemShut  [1]{\csname bibitem#1\endcsname}%
\let\auto@bib@innerbib\@empty
%</preamble>
\bibitem [{\citenamefont {Vidal}\ \emph {et~al.}(2003)\citenamefont {Vidal},
  \citenamefont {Latorre}, \citenamefont {Rico},\ and\ \citenamefont
  {Kitaev}}]{Phys.Rev.Lett.2003Vidal}%
  \BibitemOpen
  \bibfield  {author} {\bibinfo {author} {\bibfnamefont {G.}~\bibnamefont
  {Vidal}}, \bibinfo {author} {\bibfnamefont {J.~I.}\ \bibnamefont {Latorre}},
  \bibinfo {author} {\bibfnamefont {E.}~\bibnamefont {Rico}},\ and\ \bibinfo
  {author} {\bibfnamefont {A.}~\bibnamefont {Kitaev}},\ }\bibfield  {title}
  {\bibinfo {title} {Entanglement in quantum critical phenomena},\ }\href
  {https://doi.org/10.1103/PhysRevLett.90.227902} {\bibfield  {journal}
  {\bibinfo  {journal} {Physical Review Letters}\ }\textbf {\bibinfo {volume}
  {90}},\ \bibinfo {pages} {227902} (\bibinfo {year} {2003})}\BibitemShut
  {NoStop}%
\bibitem [{\citenamefont {Amico}\ \emph {et~al.}(2008)\citenamefont {Amico},
  \citenamefont {Fazio}, \citenamefont {Osterloh},\ and\ \citenamefont
  {Vedral}}]{Rev.Mod.Phys.2008Amico}%
  \BibitemOpen
  \bibfield  {author} {\bibinfo {author} {\bibfnamefont {L.}~\bibnamefont
  {Amico}}, \bibinfo {author} {\bibfnamefont {R.}~\bibnamefont {Fazio}},
  \bibinfo {author} {\bibfnamefont {A.}~\bibnamefont {Osterloh}},\ and\
  \bibinfo {author} {\bibfnamefont {V.}~\bibnamefont {Vedral}},\ }\bibfield
  {title} {\bibinfo {title} {Entanglement in many-body systems},\ }\href
  {https://doi.org/10.1103/RevModPhys.80.517} {\bibfield  {journal} {\bibinfo
  {journal} {Reviews of Modern Physics}\ }\textbf {\bibinfo {volume} {80}},\
  \bibinfo {pages} {517} (\bibinfo {year} {2008})}\BibitemShut {NoStop}%
\bibitem [{\citenamefont {Laflorencie}(2016)}]{PhysicsReports2016Laflorencie}%
  \BibitemOpen
  \bibfield  {author} {\bibinfo {author} {\bibfnamefont {N.}~\bibnamefont
  {Laflorencie}},\ }\bibfield  {title} {\bibinfo {title} {Quantum entanglement
  in condensed matter systems},\ }\href
  {https://doi.org/10.1016/j.physrep.2016.06.008} {\bibfield  {journal}
  {\bibinfo  {journal} {Physics Reports}\ }\textbf {\bibinfo {volume} {646}},\
  \bibinfo {pages} {1} (\bibinfo {year} {2016})}\BibitemShut {NoStop}%
\bibitem [{\citenamefont {Plenio}\ and\ \citenamefont
  {Virmani}(2007)}]{QuantumInfo.Comput.2007Plenio}%
  \BibitemOpen
  \bibfield  {author} {\bibinfo {author} {\bibfnamefont {M.~B.}\ \bibnamefont
  {Plenio}}\ and\ \bibinfo {author} {\bibfnamefont {S.}~\bibnamefont
  {Virmani}},\ }\bibfield  {title} {\bibinfo {title} {An introduction to
  entanglement measures},\ }\href@noop {} {\bibfield  {journal} {\bibinfo
  {journal} {Quantum Information \& Computation}\ }\textbf {\bibinfo {volume}
  {7}},\ \bibinfo {pages} {1} (\bibinfo {year} {2007})}\BibitemShut {NoStop}%
\bibitem [{\citenamefont {Horodecki}\ \emph {et~al.}(2009)\citenamefont
  {Horodecki}, \citenamefont {Horodecki}, \citenamefont {Horodecki},\ and\
  \citenamefont {Horodecki}}]{Rev.Mod.Phys.2009Horodecki}%
  \BibitemOpen
  \bibfield  {author} {\bibinfo {author} {\bibfnamefont {R.}~\bibnamefont
  {Horodecki}}, \bibinfo {author} {\bibfnamefont {P.}~\bibnamefont
  {Horodecki}}, \bibinfo {author} {\bibfnamefont {M.}~\bibnamefont
  {Horodecki}},\ and\ \bibinfo {author} {\bibfnamefont {K.}~\bibnamefont
  {Horodecki}},\ }\bibfield  {title} {\bibinfo {title} {Quantum entanglement},\
  }\href {https://doi.org/10.1103/RevModPhys.81.865} {\bibfield  {journal}
  {\bibinfo  {journal} {Reviews of Modern Physics}\ }\textbf {\bibinfo {volume}
  {81}},\ \bibinfo {pages} {865} (\bibinfo {year} {2009})}\BibitemShut
  {NoStop}%
\bibitem [{\citenamefont {Calabrese}\ and\ \citenamefont
  {Cardy}(2004)}]{J.Stat.Mech.Theor.Exp.2004Calabrese}%
  \BibitemOpen
  \bibfield  {author} {\bibinfo {author} {\bibfnamefont {P.}~\bibnamefont
  {Calabrese}}\ and\ \bibinfo {author} {\bibfnamefont {J.}~\bibnamefont
  {Cardy}},\ }\bibfield  {title} {\bibinfo {title} {Entanglement entropy and
  quantum field theory},\ }\href
  {https://doi.org/10.1088/1742-5468/2004/06/P06002} {\bibfield  {journal}
  {\bibinfo  {journal} {Journal of Statistical Mechanics: Theory and
  Experiment}\ }\textbf {\bibinfo {volume} {2004}},\ \bibinfo {pages} {P06002}
  (\bibinfo {year} {2004})}\BibitemShut {NoStop}%
\bibitem [{\citenamefont {Gioev}\ and\ \citenamefont
  {Klich}(2006)}]{Phys.Rev.Lett.2006Gioev}%
  \BibitemOpen
  \bibfield  {author} {\bibinfo {author} {\bibfnamefont {D.}~\bibnamefont
  {Gioev}}\ and\ \bibinfo {author} {\bibfnamefont {I.}~\bibnamefont {Klich}},\
  }\bibfield  {title} {\bibinfo {title} {Entanglement entropy of fermions in
  any dimension and the {{Widom}} conjecture},\ }\href
  {https://doi.org/10.1103/PhysRevLett.96.100503} {\bibfield  {journal}
  {\bibinfo  {journal} {Physical Review Letters}\ }\textbf {\bibinfo {volume}
  {96}},\ \bibinfo {pages} {100503} (\bibinfo {year} {2006})}\BibitemShut
  {NoStop}%
\bibitem [{\citenamefont {Kitaev}\ and\ \citenamefont
  {Preskill}(2006)}]{Phys.Rev.Lett.2006Kitaev}%
  \BibitemOpen
  \bibfield  {author} {\bibinfo {author} {\bibfnamefont {A.}~\bibnamefont
  {Kitaev}}\ and\ \bibinfo {author} {\bibfnamefont {J.}~\bibnamefont
  {Preskill}},\ }\bibfield  {title} {\bibinfo {title} {Topological entanglement
  entropy},\ }\href {https://doi.org/10.1103/PhysRevLett.96.110404} {\bibfield
  {journal} {\bibinfo  {journal} {Physical Review Letters}\ }\textbf {\bibinfo
  {volume} {96}},\ \bibinfo {pages} {110404} (\bibinfo {year}
  {2006})}\BibitemShut {NoStop}%
\bibitem [{\citenamefont {Levin}\ and\ \citenamefont
  {Wen}(2006)}]{Phys.Rev.Lett.2006Levin}%
  \BibitemOpen
  \bibfield  {author} {\bibinfo {author} {\bibfnamefont {M.}~\bibnamefont
  {Levin}}\ and\ \bibinfo {author} {\bibfnamefont {X.-G.}\ \bibnamefont
  {Wen}},\ }\bibfield  {title} {\bibinfo {title} {Detecting topological order
  in a ground state wave function},\ }\href
  {https://doi.org/10.1103/PhysRevLett.96.110405} {\bibfield  {journal}
  {\bibinfo  {journal} {Physical Review Letters}\ }\textbf {\bibinfo {volume}
  {96}},\ \bibinfo {pages} {110405} (\bibinfo {year} {2006})}\BibitemShut
  {NoStop}%
\bibitem [{\citenamefont {Fradkin}\ and\ \citenamefont
  {Moore}(2006)}]{Phys.Rev.Lett.2006Fradkin}%
  \BibitemOpen
  \bibfield  {author} {\bibinfo {author} {\bibfnamefont {E.}~\bibnamefont
  {Fradkin}}\ and\ \bibinfo {author} {\bibfnamefont {J.~E.}\ \bibnamefont
  {Moore}},\ }\bibfield  {title} {\bibinfo {title} {Entanglement entropy of
  {{2D}} conformal quantum critical points: Hearing the shape of a quantum
  drum},\ }\href {https://doi.org/10.1103/PhysRevLett.97.050404} {\bibfield
  {journal} {\bibinfo  {journal} {Physical Review Letters}\ }\textbf {\bibinfo
  {volume} {97}},\ \bibinfo {pages} {050404} (\bibinfo {year}
  {2006})}\BibitemShut {NoStop}%
\bibitem [{\citenamefont {Wolf}\ \emph {et~al.}(2008)\citenamefont {Wolf},
  \citenamefont {Verstraete}, \citenamefont {Hastings},\ and\ \citenamefont
  {Cirac}}]{Phys.Rev.Lett.2008Wolf}%
  \BibitemOpen
  \bibfield  {author} {\bibinfo {author} {\bibfnamefont {M.~M.}\ \bibnamefont
  {Wolf}}, \bibinfo {author} {\bibfnamefont {F.}~\bibnamefont {Verstraete}},
  \bibinfo {author} {\bibfnamefont {M.~B.}\ \bibnamefont {Hastings}},\ and\
  \bibinfo {author} {\bibfnamefont {J.~I.}\ \bibnamefont {Cirac}},\ }\bibfield
  {title} {\bibinfo {title} {Area laws in quantum systems: Mutual information
  and correlations},\ }\href {https://doi.org/10.1103/PhysRevLett.100.070502}
  {\bibfield  {journal} {\bibinfo  {journal} {Physical Review Letters}\
  }\textbf {\bibinfo {volume} {100}},\ \bibinfo {pages} {070502} (\bibinfo
  {year} {2008})}\BibitemShut {NoStop}%
\bibitem [{\citenamefont {Calabrese}\ \emph {et~al.}(2009)\citenamefont
  {Calabrese}, \citenamefont {Cardy},\ and\ \citenamefont
  {Doyon}}]{J.Phys.AMath.Theor.2009Calabrese}%
  \BibitemOpen
  \bibfield  {author} {\bibinfo {author} {\bibfnamefont {P.}~\bibnamefont
  {Calabrese}}, \bibinfo {author} {\bibfnamefont {J.}~\bibnamefont {Cardy}},\
  and\ \bibinfo {author} {\bibfnamefont {B.}~\bibnamefont {Doyon}},\ }\bibfield
   {title} {\bibinfo {title} {Entanglement entropy in extended quantum
  systems},\ }\href {https://doi.org/10.1088/1751-8121/42/50/500301} {\bibfield
   {journal} {\bibinfo  {journal} {Journal of Physics A: Mathematical and
  Theoretical}\ }\textbf {\bibinfo {volume} {42}},\ \bibinfo {pages} {500301}
  (\bibinfo {year} {2009})}\BibitemShut {NoStop}%
\bibitem [{\citenamefont {Hastings}\ \emph {et~al.}(2010)\citenamefont
  {Hastings}, \citenamefont {Gonz{\'a}lez}, \citenamefont {Kallin},\ and\
  \citenamefont {Melko}}]{Phys.Rev.Lett.2010Hastings}%
  \BibitemOpen
  \bibfield  {author} {\bibinfo {author} {\bibfnamefont {M.~B.}\ \bibnamefont
  {Hastings}}, \bibinfo {author} {\bibfnamefont {I.}~\bibnamefont
  {Gonz{\'a}lez}}, \bibinfo {author} {\bibfnamefont {A.~B.}\ \bibnamefont
  {Kallin}},\ and\ \bibinfo {author} {\bibfnamefont {R.~G.}\ \bibnamefont
  {Melko}},\ }\bibfield  {title} {\bibinfo {title} {Measuring {{Renyi}}
  entanglement entropy in quantum {{Monte Carlo}} simulations},\ }\href
  {https://doi.org/10.1103/PhysRevLett.104.157201} {\bibfield  {journal}
  {\bibinfo  {journal} {Physical Review Letters}\ }\textbf {\bibinfo {volume}
  {104}},\ \bibinfo {pages} {157201} (\bibinfo {year} {2010})}\BibitemShut
  {NoStop}%
\bibitem [{\citenamefont {Grover}(2013)}]{Phys.Rev.Lett.2013Grover}%
  \BibitemOpen
  \bibfield  {author} {\bibinfo {author} {\bibfnamefont {T.}~\bibnamefont
  {Grover}},\ }\bibfield  {title} {\bibinfo {title} {Entanglement of
  interacting fermions in quantum {{Monte Carlo}} calculations},\ }\href
  {https://doi.org/10.1103/PhysRevLett.111.130402} {\bibfield  {journal}
  {\bibinfo  {journal} {Physical Review Letters}\ }\textbf {\bibinfo {volume}
  {111}},\ \bibinfo {pages} {130402} (\bibinfo {year} {2013})}\BibitemShut
  {NoStop}%
\bibitem [{\citenamefont {Metlitski}\ and\ \citenamefont
  {Grover}(2015)}]{ArXiv2015Metlitski}%
  \BibitemOpen
  \bibfield  {author} {\bibinfo {author} {\bibfnamefont {M.~A.}\ \bibnamefont
  {Metlitski}}\ and\ \bibinfo {author} {\bibfnamefont {T.}~\bibnamefont
  {Grover}},\ }\href {https://doi.org/10.48550/arXiv.1112.5166} {\bibinfo
  {title} {Entanglement {{Entropy}} of {{Systems}} with {{Spontaneously Broken
  Continuous Symmetry}}}},\ \bibinfo {howpublished} {arXiv:1112.5166} (\bibinfo
  {year} {2015}),\ \Eprint {https://arxiv.org/abs/1112.5166} {arXiv:1112.5166
  [cond-mat, physics:hep-th, physics:quant-ph]} \BibitemShut {NoStop}%
\bibitem [{\citenamefont {Zhao}\ \emph
  {et~al.}(2022{\natexlab{a}})\citenamefont {Zhao}, \citenamefont {Wang},
  \citenamefont {Yan}, \citenamefont {Cheng},\ and\ \citenamefont
  {Meng}}]{Phys.Rev.Lett.2022Zhao}%
  \BibitemOpen
  \bibfield  {author} {\bibinfo {author} {\bibfnamefont {J.}~\bibnamefont
  {Zhao}}, \bibinfo {author} {\bibfnamefont {Y.-C.}\ \bibnamefont {Wang}},
  \bibinfo {author} {\bibfnamefont {Z.}~\bibnamefont {Yan}}, \bibinfo {author}
  {\bibfnamefont {M.}~\bibnamefont {Cheng}},\ and\ \bibinfo {author}
  {\bibfnamefont {Z.~Y.}\ \bibnamefont {Meng}},\ }\bibfield  {title} {\bibinfo
  {title} {Scaling of entanglement entropy at deconfined quantum criticality},\
  }\href {https://doi.org/10.1103/PhysRevLett.128.010601} {\bibfield  {journal}
  {\bibinfo  {journal} {Physical Review Letters}\ }\textbf {\bibinfo {volume}
  {128}},\ \bibinfo {pages} {010601} (\bibinfo {year}
  {2022}{\natexlab{a}})}\BibitemShut {NoStop}%
\bibitem [{\citenamefont {D'Emidio}\ \emph {et~al.}(2024)\citenamefont
  {D'Emidio}, \citenamefont {Or{\'u}s}, \citenamefont {Laflorencie},\ and\
  \citenamefont {De~Juan}}]{Phys.Rev.Lett.2024DEmidio}%
  \BibitemOpen
  \bibfield  {author} {\bibinfo {author} {\bibfnamefont {J.}~\bibnamefont
  {D'Emidio}}, \bibinfo {author} {\bibfnamefont {R.}~\bibnamefont {Or{\'u}s}},
  \bibinfo {author} {\bibfnamefont {N.}~\bibnamefont {Laflorencie}},\ and\
  \bibinfo {author} {\bibfnamefont {F.}~\bibnamefont {De~Juan}},\ }\bibfield
  {title} {\bibinfo {title} {Universal features of entanglement entropy in the
  honeycomb {{Hubbard}} model},\ }\href
  {https://doi.org/10.1103/PhysRevLett.132.076502} {\bibfield  {journal}
  {\bibinfo  {journal} {Physical Review Letters}\ }\textbf {\bibinfo {volume}
  {132}},\ \bibinfo {pages} {076502} (\bibinfo {year} {2024})}\BibitemShut
  {NoStop}%
\bibitem [{\citenamefont {Liao}(2025)}]{npjQuantumInf2025Liao}%
  \BibitemOpen
  \bibfield  {author} {\bibinfo {author} {\bibfnamefont {Y.~D.}\ \bibnamefont
  {Liao}},\ }\bibfield  {title} {\bibinfo {title} {Universal term of
  entanglement entropy in the {$\pi$}-flux {{Hubbard}} model},\ }\href
  {https://doi.org/10.1038/s41534-025-01010-3} {\bibfield  {journal} {\bibinfo
  {journal} {npj Quantum Information}\ }\textbf {\bibinfo {volume} {11}},\
  \bibinfo {pages} {64} (\bibinfo {year} {2025})},\ \Eprint
  {https://arxiv.org/abs/2307.10602} {arXiv:2307.10602 [cond-mat]} \BibitemShut
  {NoStop}%
\bibitem [{\citenamefont {{\.Z}yczkowski}\ \emph {et~al.}(1998)\citenamefont
  {{\.Z}yczkowski}, \citenamefont {Horodecki}, \citenamefont {Sanpera},\ and\
  \citenamefont {Lewenstein}}]{Phys.Rev.A1998Zyczkowski}%
  \BibitemOpen
  \bibfield  {author} {\bibinfo {author} {\bibfnamefont {K.}~\bibnamefont
  {{\.Z}yczkowski}}, \bibinfo {author} {\bibfnamefont {P.}~\bibnamefont
  {Horodecki}}, \bibinfo {author} {\bibfnamefont {A.}~\bibnamefont {Sanpera}},\
  and\ \bibinfo {author} {\bibfnamefont {M.}~\bibnamefont {Lewenstein}},\
  }\bibfield  {title} {\bibinfo {title} {Volume of the set of separable
  states},\ }\href {https://doi.org/10.1103/PhysRevA.58.883} {\bibfield
  {journal} {\bibinfo  {journal} {Physical Review A}\ }\textbf {\bibinfo
  {volume} {58}},\ \bibinfo {pages} {883} (\bibinfo {year} {1998})}\BibitemShut
  {NoStop}%
\bibitem [{\citenamefont {Eisert}\ and\ \citenamefont
  {Plenio}(1999)}]{J.Mod.Opt.1999Eisert}%
  \BibitemOpen
  \bibfield  {author} {\bibinfo {author} {\bibfnamefont {J.}~\bibnamefont
  {Eisert}}\ and\ \bibinfo {author} {\bibfnamefont {M.~B.}\ \bibnamefont
  {Plenio}},\ }\bibfield  {title} {\bibinfo {title} {A comparison of
  entanglement measures},\ }\href {https://doi.org/10.1080/09500349908231260}
  {\bibfield  {journal} {\bibinfo  {journal} {Journal of Modern Optics}\
  }\textbf {\bibinfo {volume} {46}},\ \bibinfo {pages} {145} (\bibinfo {year}
  {1999})}\BibitemShut {NoStop}%
\bibitem [{\citenamefont {Vidal}\ and\ \citenamefont
  {Werner}(2002)}]{Phys.Rev.A2002Vidal}%
  \BibitemOpen
  \bibfield  {author} {\bibinfo {author} {\bibfnamefont {G.}~\bibnamefont
  {Vidal}}\ and\ \bibinfo {author} {\bibfnamefont {R.~F.}\ \bibnamefont
  {Werner}},\ }\bibfield  {title} {\bibinfo {title} {Computable measure of
  entanglement},\ }\href {https://doi.org/10.1103/PhysRevA.65.032314}
  {\bibfield  {journal} {\bibinfo  {journal} {Physical Review A}\ }\textbf
  {\bibinfo {volume} {65}},\ \bibinfo {pages} {032314} (\bibinfo {year}
  {2002})}\BibitemShut {NoStop}%
\bibitem [{\citenamefont {Plenio}(2005)}]{Phys.Rev.Lett.2005Plenio}%
  \BibitemOpen
  \bibfield  {author} {\bibinfo {author} {\bibfnamefont {M.~B.}\ \bibnamefont
  {Plenio}},\ }\bibfield  {title} {\bibinfo {title} {Logarithmic negativity: A
  full entanglement monotone that is not convex},\ }\href
  {https://doi.org/10.1103/PhysRevLett.95.090503} {\bibfield  {journal}
  {\bibinfo  {journal} {Physical Review Letters}\ }\textbf {\bibinfo {volume}
  {95}},\ \bibinfo {pages} {090503} (\bibinfo {year} {2005})}\BibitemShut
  {NoStop}%
\bibitem [{\citenamefont {Peres}(1996)}]{Phys.Rev.Lett.1996Peres}%
  \BibitemOpen
  \bibfield  {author} {\bibinfo {author} {\bibfnamefont {A.}~\bibnamefont
  {Peres}},\ }\bibfield  {title} {\bibinfo {title} {Separability criterion for
  density matrices},\ }\href {https://doi.org/10.1103/PhysRevLett.77.1413}
  {\bibfield  {journal} {\bibinfo  {journal} {Physical Review Letters}\
  }\textbf {\bibinfo {volume} {77}},\ \bibinfo {pages} {1413} (\bibinfo {year}
  {1996})}\BibitemShut {NoStop}%
\bibitem [{\citenamefont {Horodecki}\ \emph {et~al.}(1996)\citenamefont
  {Horodecki}, \citenamefont {Horodecki},\ and\ \citenamefont
  {Horodecki}}]{PhysicsLettersA1996Horodecki}%
  \BibitemOpen
  \bibfield  {author} {\bibinfo {author} {\bibfnamefont {M.}~\bibnamefont
  {Horodecki}}, \bibinfo {author} {\bibfnamefont {P.}~\bibnamefont
  {Horodecki}},\ and\ \bibinfo {author} {\bibfnamefont {R.}~\bibnamefont
  {Horodecki}},\ }\bibfield  {title} {\bibinfo {title} {Separability of mixed
  states: Necessary and sufficient conditions},\ }\href
  {https://doi.org/10.1016/S0375-9601(96)00706-2} {\bibfield  {journal}
  {\bibinfo  {journal} {Physics Letters A}\ }\textbf {\bibinfo {volume}
  {223}},\ \bibinfo {pages} {1} (\bibinfo {year} {1996})}\BibitemShut {NoStop}%
\bibitem [{\citenamefont {Shapourian}\ \emph {et~al.}(2017)\citenamefont
  {Shapourian}, \citenamefont {Shiozaki},\ and\ \citenamefont
  {Ryu}}]{Phys.Rev.B2017Shapourian}%
  \BibitemOpen
  \bibfield  {author} {\bibinfo {author} {\bibfnamefont {H.}~\bibnamefont
  {Shapourian}}, \bibinfo {author} {\bibfnamefont {K.}~\bibnamefont
  {Shiozaki}},\ and\ \bibinfo {author} {\bibfnamefont {S.}~\bibnamefont
  {Ryu}},\ }\bibfield  {title} {\bibinfo {title} {Partial time-reversal
  transformation and entanglement negativity in fermionic systems},\ }\href
  {https://doi.org/10.1103/PhysRevB.95.165101} {\bibfield  {journal} {\bibinfo
  {journal} {Physical Review B}\ }\textbf {\bibinfo {volume} {95}},\ \bibinfo
  {pages} {165101} (\bibinfo {year} {2017})}\BibitemShut {NoStop}%
\bibitem [{\citenamefont {Shiozaki}\ \emph {et~al.}(2018)\citenamefont
  {Shiozaki}, \citenamefont {Shapourian}, \citenamefont {Gomi},\ and\
  \citenamefont {Ryu}}]{Phys.Rev.B2018Shiozaki}%
  \BibitemOpen
  \bibfield  {author} {\bibinfo {author} {\bibfnamefont {K.}~\bibnamefont
  {Shiozaki}}, \bibinfo {author} {\bibfnamefont {H.}~\bibnamefont
  {Shapourian}}, \bibinfo {author} {\bibfnamefont {K.}~\bibnamefont {Gomi}},\
  and\ \bibinfo {author} {\bibfnamefont {S.}~\bibnamefont {Ryu}},\ }\bibfield
  {title} {\bibinfo {title} {Many-body topological invariants for fermionic
  short-range entangled topological phases protected by antiunitary
  symmetries},\ }\href {https://doi.org/10.1103/PhysRevB.98.035151} {\bibfield
  {journal} {\bibinfo  {journal} {Physical Review B}\ }\textbf {\bibinfo
  {volume} {98}},\ \bibinfo {pages} {035151} (\bibinfo {year}
  {2018})}\BibitemShut {NoStop}%
\bibitem [{\citenamefont {Shapourian}\ and\ \citenamefont
  {Ryu}(2019{\natexlab{a}})}]{Phys.Rev.A2019Shapourian}%
  \BibitemOpen
  \bibfield  {author} {\bibinfo {author} {\bibfnamefont {H.}~\bibnamefont
  {Shapourian}}\ and\ \bibinfo {author} {\bibfnamefont {S.}~\bibnamefont
  {Ryu}},\ }\bibfield  {title} {\bibinfo {title} {Entanglement negativity of
  fermions: Monotonicity, separability criterion, and classification of
  few-mode states},\ }\href {https://doi.org/10.1103/PhysRevA.99.022310}
  {\bibfield  {journal} {\bibinfo  {journal} {Physical Review A}\ }\textbf
  {\bibinfo {volume} {99}},\ \bibinfo {pages} {022310} (\bibinfo {year}
  {2019}{\natexlab{a}})}\BibitemShut {NoStop}%
\bibitem [{\citenamefont {Shapourian}\ \emph {et~al.}(2019)\citenamefont
  {Shapourian}, \citenamefont {Ruggiero}, \citenamefont {Ryu},\ and\
  \citenamefont {Calabrese}}]{SciPostPhys.2019Shapourian}%
  \BibitemOpen
  \bibfield  {author} {\bibinfo {author} {\bibfnamefont {H.}~\bibnamefont
  {Shapourian}}, \bibinfo {author} {\bibfnamefont {P.}~\bibnamefont
  {Ruggiero}}, \bibinfo {author} {\bibfnamefont {S.}~\bibnamefont {Ryu}},\ and\
  \bibinfo {author} {\bibfnamefont {P.}~\bibnamefont {Calabrese}},\ }\bibfield
  {title} {\bibinfo {title} {Twisted and untwisted negativity spectrum of free
  fermions},\ }\href {https://doi.org/10.21468/SciPostPhys.7.3.037} {\bibfield
  {journal} {\bibinfo  {journal} {SciPost Physics}\ }\textbf {\bibinfo {volume}
  {7}},\ \bibinfo {pages} {037} (\bibinfo {year} {2019})}\BibitemShut {NoStop}%
\bibitem [{\citenamefont {Calabrese}\ \emph
  {et~al.}(2013{\natexlab{a}})\citenamefont {Calabrese}, \citenamefont
  {Tagliacozzo},\ and\ \citenamefont {Tonni}}]{J.Stat.Mech.2013Calabresea}%
  \BibitemOpen
  \bibfield  {author} {\bibinfo {author} {\bibfnamefont {P.}~\bibnamefont
  {Calabrese}}, \bibinfo {author} {\bibfnamefont {L.}~\bibnamefont
  {Tagliacozzo}},\ and\ \bibinfo {author} {\bibfnamefont {E.}~\bibnamefont
  {Tonni}},\ }\bibfield  {title} {\bibinfo {title} {Entanglement negativity in
  the critical {{Ising}} chain},\ }\href
  {https://doi.org/10.1088/1742-5468/2013/05/P05002} {\bibfield  {journal}
  {\bibinfo  {journal} {Journal of Statistical Mechanics: Theory and
  Experiment}\ }\textbf {\bibinfo {volume} {2013}},\ \bibinfo {pages} {P05002}
  (\bibinfo {year} {2013}{\natexlab{a}})}\BibitemShut {NoStop}%
\bibitem [{\citenamefont {Ruggiero}\ \emph {et~al.}(2016)\citenamefont
  {Ruggiero}, \citenamefont {Alba},\ and\ \citenamefont
  {Calabrese}}]{Phys.Rev.B2016Ruggiero}%
  \BibitemOpen
  \bibfield  {author} {\bibinfo {author} {\bibfnamefont {P.}~\bibnamefont
  {Ruggiero}}, \bibinfo {author} {\bibfnamefont {V.}~\bibnamefont {Alba}},\
  and\ \bibinfo {author} {\bibfnamefont {P.}~\bibnamefont {Calabrese}},\
  }\bibfield  {title} {\bibinfo {title} {Entanglement negativity in random spin
  chains},\ }\href {https://doi.org/10.1103/PhysRevB.94.035152} {\bibfield
  {journal} {\bibinfo  {journal} {Physical Review B}\ }\textbf {\bibinfo
  {volume} {94}},\ \bibinfo {pages} {035152} (\bibinfo {year}
  {2016})}\BibitemShut {NoStop}%
\bibitem [{\citenamefont {Gruber}\ and\ \citenamefont
  {Eisler}(2020)}]{J.Phys.AMath.Theor.2020Gruber}%
  \BibitemOpen
  \bibfield  {author} {\bibinfo {author} {\bibfnamefont {M.}~\bibnamefont
  {Gruber}}\ and\ \bibinfo {author} {\bibfnamefont {V.}~\bibnamefont
  {Eisler}},\ }\bibfield  {title} {\bibinfo {title} {Time evolution of
  entanglement negativity across a defect},\ }\href
  {https://doi.org/10.1088/1751-8121/ab831c} {\bibfield  {journal} {\bibinfo
  {journal} {Journal of Physics A: Mathematical and Theoretical}\ }\textbf
  {\bibinfo {volume} {53}},\ \bibinfo {pages} {205301} (\bibinfo {year}
  {2020})}\BibitemShut {NoStop}%
\bibitem [{\citenamefont {Wichterich}\ \emph {et~al.}(2009)\citenamefont
  {Wichterich}, \citenamefont {{Molina-Vilaplana}},\ and\ \citenamefont
  {Bose}}]{Phys.Rev.A2009Wichterich}%
  \BibitemOpen
  \bibfield  {author} {\bibinfo {author} {\bibfnamefont {H.}~\bibnamefont
  {Wichterich}}, \bibinfo {author} {\bibfnamefont {J.}~\bibnamefont
  {{Molina-Vilaplana}}},\ and\ \bibinfo {author} {\bibfnamefont
  {S.}~\bibnamefont {Bose}},\ }\bibfield  {title} {\bibinfo {title} {Scaling of
  entanglement between separated blocks in spin chains at criticality},\ }\href
  {https://doi.org/10.1103/PhysRevA.80.010304} {\bibfield  {journal} {\bibinfo
  {journal} {Physical Review A}\ }\textbf {\bibinfo {volume} {80}},\ \bibinfo
  {pages} {010304} (\bibinfo {year} {2009})}\BibitemShut {NoStop}%
\bibitem [{\citenamefont {Wybo}\ \emph {et~al.}(2020)\citenamefont {Wybo},
  \citenamefont {Knap},\ and\ \citenamefont {Pollmann}}]{Phys.Rev.B2020Wybo}%
  \BibitemOpen
  \bibfield  {author} {\bibinfo {author} {\bibfnamefont {E.}~\bibnamefont
  {Wybo}}, \bibinfo {author} {\bibfnamefont {M.}~\bibnamefont {Knap}},\ and\
  \bibinfo {author} {\bibfnamefont {F.}~\bibnamefont {Pollmann}},\ }\bibfield
  {title} {\bibinfo {title} {Entanglement dynamics of a many-body localized
  system coupled to a bath},\ }\href
  {https://doi.org/10.1103/PhysRevB.102.064304} {\bibfield  {journal} {\bibinfo
   {journal} {Physical Review B}\ }\textbf {\bibinfo {volume} {102}},\ \bibinfo
  {pages} {064304} (\bibinfo {year} {2020})}\BibitemShut {NoStop}%
\bibitem [{\citenamefont {Calabrese}\ \emph {et~al.}(2012)\citenamefont
  {Calabrese}, \citenamefont {Cardy},\ and\ \citenamefont
  {Tonni}}]{Phys.Rev.Lett.2012Calabrese}%
  \BibitemOpen
  \bibfield  {author} {\bibinfo {author} {\bibfnamefont {P.}~\bibnamefont
  {Calabrese}}, \bibinfo {author} {\bibfnamefont {J.}~\bibnamefont {Cardy}},\
  and\ \bibinfo {author} {\bibfnamefont {E.}~\bibnamefont {Tonni}},\ }\bibfield
   {title} {\bibinfo {title} {Entanglement negativity in quantum field
  theory},\ }\href {https://doi.org/10.1103/PhysRevLett.109.130502} {\bibfield
  {journal} {\bibinfo  {journal} {Physical Review Letters}\ }\textbf {\bibinfo
  {volume} {109}},\ \bibinfo {pages} {130502} (\bibinfo {year}
  {2012})}\BibitemShut {NoStop}%
\bibitem [{\citenamefont {Calabrese}\ \emph
  {et~al.}(2013{\natexlab{b}})\citenamefont {Calabrese}, \citenamefont
  {Cardy},\ and\ \citenamefont {Tonni}}]{J.Stat.Mech.2013Calabrese}%
  \BibitemOpen
  \bibfield  {author} {\bibinfo {author} {\bibfnamefont {P.}~\bibnamefont
  {Calabrese}}, \bibinfo {author} {\bibfnamefont {J.}~\bibnamefont {Cardy}},\
  and\ \bibinfo {author} {\bibfnamefont {E.}~\bibnamefont {Tonni}},\ }\bibfield
   {title} {\bibinfo {title} {Entanglement negativity in extended systems: A
  field theoretical approach},\ }\href
  {https://doi.org/10.1088/1742-5468/2013/02/P02008} {\bibfield  {journal}
  {\bibinfo  {journal} {Journal of Statistical Mechanics: Theory and
  Experiment}\ }\textbf {\bibinfo {volume} {2013}},\ \bibinfo {pages} {P02008}
  (\bibinfo {year} {2013}{\natexlab{b}})}\BibitemShut {NoStop}%
\bibitem [{\citenamefont {Alba}(2013)}]{J.Stat.Mech.2013Alba}%
  \BibitemOpen
  \bibfield  {author} {\bibinfo {author} {\bibfnamefont {V.}~\bibnamefont
  {Alba}},\ }\bibfield  {title} {\bibinfo {title} {Entanglement negativity and
  conformal field theory: A {{Monte Carlo}} study},\ }\href
  {https://doi.org/10.1088/1742-5468/2013/05/P05013} {\bibfield  {journal}
  {\bibinfo  {journal} {Journal of Statistical Mechanics: Theory and
  Experiment}\ }\textbf {\bibinfo {volume} {2013}},\ \bibinfo {pages} {P05013}
  (\bibinfo {year} {2013})}\BibitemShut {NoStop}%
\bibitem [{\citenamefont {Chung}\ \emph {et~al.}(2014)\citenamefont {Chung},
  \citenamefont {Alba}, \citenamefont {Bonnes}, \citenamefont {Chen},\ and\
  \citenamefont {L{\"a}uchli}}]{Phys.Rev.B2014Chunga}%
  \BibitemOpen
  \bibfield  {author} {\bibinfo {author} {\bibfnamefont {C.-M.}\ \bibnamefont
  {Chung}}, \bibinfo {author} {\bibfnamefont {V.}~\bibnamefont {Alba}},
  \bibinfo {author} {\bibfnamefont {L.}~\bibnamefont {Bonnes}}, \bibinfo
  {author} {\bibfnamefont {P.}~\bibnamefont {Chen}},\ and\ \bibinfo {author}
  {\bibfnamefont {A.~M.}\ \bibnamefont {L{\"a}uchli}},\ }\bibfield  {title}
  {\bibinfo {title} {Entanglement negativity via the replica trick: A quantum
  {{Monte Carlo}} approach},\ }\href
  {https://doi.org/10.1103/PhysRevB.90.064401} {\bibfield  {journal} {\bibinfo
  {journal} {Physical Review B}\ }\textbf {\bibinfo {volume} {90}},\ \bibinfo
  {pages} {064401} (\bibinfo {year} {2014})}\BibitemShut {NoStop}%
\bibitem [{\citenamefont {Wu}\ \emph {et~al.}(2020)\citenamefont {Wu},
  \citenamefont {Lu}, \citenamefont {Chung}, \citenamefont {Kao},\ and\
  \citenamefont {Grover}}]{Phys.Rev.Lett.2020Wu}%
  \BibitemOpen
  \bibfield  {author} {\bibinfo {author} {\bibfnamefont {K.-H.}\ \bibnamefont
  {Wu}}, \bibinfo {author} {\bibfnamefont {T.-C.}\ \bibnamefont {Lu}}, \bibinfo
  {author} {\bibfnamefont {C.-M.}\ \bibnamefont {Chung}}, \bibinfo {author}
  {\bibfnamefont {Y.-J.}\ \bibnamefont {Kao}},\ and\ \bibinfo {author}
  {\bibfnamefont {T.}~\bibnamefont {Grover}},\ }\bibfield  {title} {\bibinfo
  {title} {Entanglement {{Renyi}} negativity across a finite temperature
  transition: A {{Monte Carlo}} study},\ }\href
  {https://doi.org/10.1103/PhysRevLett.125.140603} {\bibfield  {journal}
  {\bibinfo  {journal} {Physical Review Letters}\ }\textbf {\bibinfo {volume}
  {125}},\ \bibinfo {pages} {140603} (\bibinfo {year} {2020})}\BibitemShut
  {NoStop}%
\bibitem [{\citenamefont {Ding}\ \emph {et~al.}(2025)\citenamefont {Ding},
  \citenamefont {Tang}, \citenamefont {Wang}, \citenamefont {Wang},
  \citenamefont {Mao},\ and\ \citenamefont {Yan}}]{Phys.Rev.B2025Ding}%
  \BibitemOpen
  \bibfield  {author} {\bibinfo {author} {\bibfnamefont {Y.-M.}\ \bibnamefont
  {Ding}}, \bibinfo {author} {\bibfnamefont {Y.}~\bibnamefont {Tang}}, \bibinfo
  {author} {\bibfnamefont {Z.}~\bibnamefont {Wang}}, \bibinfo {author}
  {\bibfnamefont {Z.}~\bibnamefont {Wang}}, \bibinfo {author} {\bibfnamefont
  {B.-B.}\ \bibnamefont {Mao}},\ and\ \bibinfo {author} {\bibfnamefont
  {Z.}~\bibnamefont {Yan}},\ }\bibfield  {title} {\bibinfo {title} {Tracking
  the variation of entanglement {{R\'enyi}} negativity: {{A}} quantum {{Monte
  Carlo}} study},\ }\href {https://doi.org/10.1103/PhysRevB.111.L241108}
  {\bibfield  {journal} {\bibinfo  {journal} {Physical Review B}\ }\textbf
  {\bibinfo {volume} {111}},\ \bibinfo {pages} {L241108} (\bibinfo {year}
  {2025})},\ \Eprint {https://arxiv.org/abs/2409.10273} {arXiv:2409.10273
  [cond-mat, physics:quant-ph]} \BibitemShut {NoStop}%
\bibitem [{\citenamefont {Gray}\ \emph {et~al.}(2018)\citenamefont {Gray},
  \citenamefont {Banchi}, \citenamefont {Bayat},\ and\ \citenamefont
  {Bose}}]{Phys.Rev.Lett.2018Gray}%
  \BibitemOpen
  \bibfield  {author} {\bibinfo {author} {\bibfnamefont {J.}~\bibnamefont
  {Gray}}, \bibinfo {author} {\bibfnamefont {L.}~\bibnamefont {Banchi}},
  \bibinfo {author} {\bibfnamefont {A.}~\bibnamefont {Bayat}},\ and\ \bibinfo
  {author} {\bibfnamefont {S.}~\bibnamefont {Bose}},\ }\bibfield  {title}
  {\bibinfo {title} {Machine-{{Learning-Assisted Many-Body Entanglement
  Measurement}}},\ }\href {https://doi.org/10.1103/PhysRevLett.121.150503}
  {\bibfield  {journal} {\bibinfo  {journal} {Physical Review Letters}\
  }\textbf {\bibinfo {volume} {121}},\ \bibinfo {pages} {150503} (\bibinfo
  {year} {2018})}\BibitemShut {NoStop}%
\bibitem [{\citenamefont {Cornfeld}\ \emph {et~al.}(2019)\citenamefont
  {Cornfeld}, \citenamefont {Sela},\ and\ \citenamefont
  {Goldstein}}]{Phys.Rev.A2019Cornfeld}%
  \BibitemOpen
  \bibfield  {author} {\bibinfo {author} {\bibfnamefont {E.}~\bibnamefont
  {Cornfeld}}, \bibinfo {author} {\bibfnamefont {E.}~\bibnamefont {Sela}},\
  and\ \bibinfo {author} {\bibfnamefont {M.}~\bibnamefont {Goldstein}},\
  }\bibfield  {title} {\bibinfo {title} {Measuring fermionic entanglement:
  {{Entropy}}, negativity, and spin structure},\ }\href
  {https://doi.org/10.1103/PhysRevA.99.062309} {\bibfield  {journal} {\bibinfo
  {journal} {Physical Review A}\ }\textbf {\bibinfo {volume} {99}},\ \bibinfo
  {pages} {062309} (\bibinfo {year} {2019})}\BibitemShut {NoStop}%
\bibitem [{\citenamefont {Elben}\ \emph {et~al.}(2020)\citenamefont {Elben},
  \citenamefont {Kueng}, \citenamefont {Huang}, \citenamefont {{van Bijnen}},
  \citenamefont {Kokail}, \citenamefont {Dalmonte}, \citenamefont {Calabrese},
  \citenamefont {Kraus}, \citenamefont {Preskill}, \citenamefont {Zoller},\
  and\ \citenamefont {Vermersch}}]{Phys.Rev.Lett.2020Elben}%
  \BibitemOpen
  \bibfield  {author} {\bibinfo {author} {\bibfnamefont {A.}~\bibnamefont
  {Elben}}, \bibinfo {author} {\bibfnamefont {R.}~\bibnamefont {Kueng}},
  \bibinfo {author} {\bibfnamefont {H.-Y.~R.}\ \bibnamefont {Huang}}, \bibinfo
  {author} {\bibfnamefont {R.}~\bibnamefont {{van Bijnen}}}, \bibinfo {author}
  {\bibfnamefont {C.}~\bibnamefont {Kokail}}, \bibinfo {author} {\bibfnamefont
  {M.}~\bibnamefont {Dalmonte}}, \bibinfo {author} {\bibfnamefont
  {P.}~\bibnamefont {Calabrese}}, \bibinfo {author} {\bibfnamefont
  {B.}~\bibnamefont {Kraus}}, \bibinfo {author} {\bibfnamefont
  {J.}~\bibnamefont {Preskill}}, \bibinfo {author} {\bibfnamefont
  {P.}~\bibnamefont {Zoller}},\ and\ \bibinfo {author} {\bibfnamefont
  {B.}~\bibnamefont {Vermersch}},\ }\bibfield  {title} {\bibinfo {title}
  {Mixed-{{State Entanglement}} from {{Local Randomized Measurements}}},\
  }\href {https://doi.org/10.1103/PhysRevLett.125.200501} {\bibfield  {journal}
  {\bibinfo  {journal} {Physical Review Letters}\ }\textbf {\bibinfo {volume}
  {125}},\ \bibinfo {pages} {200501} (\bibinfo {year} {2020})}\BibitemShut
  {NoStop}%
\bibitem [{\citenamefont {Neven}\ \emph {et~al.}(2021)\citenamefont {Neven},
  \citenamefont {Carrasco}, \citenamefont {Vitale}, \citenamefont {Kokail},
  \citenamefont {Elben}, \citenamefont {Dalmonte}, \citenamefont {Calabrese},
  \citenamefont {Zoller}, \citenamefont {Vermersch}, \citenamefont {Kueng},\
  and\ \citenamefont {Kraus}}]{npjQuantumInf2021Neven}%
  \BibitemOpen
  \bibfield  {author} {\bibinfo {author} {\bibfnamefont {A.}~\bibnamefont
  {Neven}}, \bibinfo {author} {\bibfnamefont {J.}~\bibnamefont {Carrasco}},
  \bibinfo {author} {\bibfnamefont {V.}~\bibnamefont {Vitale}}, \bibinfo
  {author} {\bibfnamefont {C.}~\bibnamefont {Kokail}}, \bibinfo {author}
  {\bibfnamefont {A.}~\bibnamefont {Elben}}, \bibinfo {author} {\bibfnamefont
  {M.}~\bibnamefont {Dalmonte}}, \bibinfo {author} {\bibfnamefont
  {P.}~\bibnamefont {Calabrese}}, \bibinfo {author} {\bibfnamefont
  {P.}~\bibnamefont {Zoller}}, \bibinfo {author} {\bibfnamefont
  {B.}~\bibnamefont {Vermersch}}, \bibinfo {author} {\bibfnamefont
  {R.}~\bibnamefont {Kueng}},\ and\ \bibinfo {author} {\bibfnamefont
  {B.}~\bibnamefont {Kraus}},\ }\bibfield  {title} {\bibinfo {title}
  {Symmetry-resolved entanglement detection using partial transpose moments},\
  }\href {https://doi.org/10.1038/s41534-021-00487-y} {\bibfield  {journal}
  {\bibinfo  {journal} {npj Quantum Information}\ }\textbf {\bibinfo {volume}
  {7}},\ \bibinfo {pages} {152} (\bibinfo {year} {2021})}\BibitemShut {NoStop}%
\bibitem [{\citenamefont {Coser}\ \emph {et~al.}(2014)\citenamefont {Coser},
  \citenamefont {Tonni},\ and\ \citenamefont
  {Calabrese}}]{J.Stat.Mech.2014Coser}%
  \BibitemOpen
  \bibfield  {author} {\bibinfo {author} {\bibfnamefont {A.}~\bibnamefont
  {Coser}}, \bibinfo {author} {\bibfnamefont {E.}~\bibnamefont {Tonni}},\ and\
  \bibinfo {author} {\bibfnamefont {P.}~\bibnamefont {Calabrese}},\ }\bibfield
  {title} {\bibinfo {title} {Entanglement negativity after a global quantum
  quench},\ }\href {https://doi.org/10.1088/1742-5468/2014/12/P12017}
  {\bibfield  {journal} {\bibinfo  {journal} {Journal of Statistical Mechanics:
  Theory and Experiment}\ }\textbf {\bibinfo {volume} {2014}},\ \bibinfo
  {pages} {P12017} (\bibinfo {year} {2014})}\BibitemShut {NoStop}%
\bibitem [{\citenamefont {Murciano}\ \emph {et~al.}(2022)\citenamefont
  {Murciano}, \citenamefont {Alba},\ and\ \citenamefont
  {Calabrese}}]{Murciano2022QuenchDynamicsRenyi}%
  \BibitemOpen
  \bibfield  {author} {\bibinfo {author} {\bibfnamefont {S.}~\bibnamefont
  {Murciano}}, \bibinfo {author} {\bibfnamefont {V.}~\bibnamefont {Alba}},\
  and\ \bibinfo {author} {\bibfnamefont {P.}~\bibnamefont {Calabrese}},\
  }\bibfield  {title} {\bibinfo {title} {Quench {{Dynamics}} of {{R\'enyi
  Negativities}} and the {{Quasiparticle Picture}}},\ }in\ \href
  {https://doi.org/10.1007/978-3-031-03998-0_14} {\emph {\bibinfo {booktitle}
  {Entanglement in {{Spin Chains}}: {{From Theory}} to {{Quantum Technology
  Applications}}}}},\ \bibinfo {editor} {edited by\ \bibinfo {editor}
  {\bibfnamefont {A.}~\bibnamefont {Bayat}}, \bibinfo {editor} {\bibfnamefont
  {S.}~\bibnamefont {Bose}},\ and\ \bibinfo {editor} {\bibfnamefont
  {H.}~\bibnamefont {Johannesson}}}\ (\bibinfo  {publisher} {Springer
  International Publishing},\ \bibinfo {address} {Cham},\ \bibinfo {year}
  {2022})\ pp.\ \bibinfo {pages} {397--424}\BibitemShut {NoStop}%
\bibitem [{\citenamefont {Parez}\ and\ \citenamefont
  {{Witczak-Krempa}}(2024)}]{Phys.Rev.Res.2024Parez}%
  \BibitemOpen
  \bibfield  {author} {\bibinfo {author} {\bibfnamefont {G.}~\bibnamefont
  {Parez}}\ and\ \bibinfo {author} {\bibfnamefont {W.}~\bibnamefont
  {{Witczak-Krempa}}},\ }\bibfield  {title} {\bibinfo {title} {Entanglement
  negativity between separated regions in quantum critical systems},\ }\href
  {https://doi.org/10.1103/PhysRevResearch.6.023125} {\bibfield  {journal}
  {\bibinfo  {journal} {Phys. Rev. Res.}\ }\textbf {\bibinfo {volume} {6}},\
  \bibinfo {pages} {023125} (\bibinfo {year} {2024})},\ \Eprint
  {https://arxiv.org/abs/2310.15273} {arXiv:2310.15273 [cond-mat.str-el]}
  \BibitemShut {NoStop}%
\bibitem [{\citenamefont {Wang}\ \emph {et~al.}(2025)\citenamefont {Wang},
  \citenamefont {Song}, \citenamefont {Lyu}, \citenamefont {{Witczak-Krempa}},\
  and\ \citenamefont {Meng}}]{Nat.Commun.2025Wang}%
  \BibitemOpen
  \bibfield  {author} {\bibinfo {author} {\bibfnamefont {T.-T.}\ \bibnamefont
  {Wang}}, \bibinfo {author} {\bibfnamefont {M.}~\bibnamefont {Song}}, \bibinfo
  {author} {\bibfnamefont {L.}~\bibnamefont {Lyu}}, \bibinfo {author}
  {\bibfnamefont {W.}~\bibnamefont {{Witczak-Krempa}}},\ and\ \bibinfo {author}
  {\bibfnamefont {Z.~Y.}\ \bibnamefont {Meng}},\ }\bibfield  {title} {\bibinfo
  {title} {Entanglement microscopy and tomography in many-body systems},\
  }\href {https://doi.org/10.1038/s41467-024-55354-z} {\bibfield  {journal}
  {\bibinfo  {journal} {Nature Communications}\ }\textbf {\bibinfo {volume}
  {16}},\ \bibinfo {pages} {96} (\bibinfo {year} {2025})}\BibitemShut {NoStop}%
\bibitem [{\citenamefont {Shapourian}\ and\ \citenamefont
  {Ryu}(2019{\natexlab{b}})}]{J.Stat.Mech.2019Shapourian}%
  \BibitemOpen
  \bibfield  {author} {\bibinfo {author} {\bibfnamefont {H.}~\bibnamefont
  {Shapourian}}\ and\ \bibinfo {author} {\bibfnamefont {S.}~\bibnamefont
  {Ryu}},\ }\bibfield  {title} {\bibinfo {title} {Finite-temperature
  entanglement negativity of free fermions},\ }\href
  {https://doi.org/10.1088/1742-5468/ab11e0} {\bibfield  {journal} {\bibinfo
  {journal} {Journal of Statistical Mechanics: Theory and Experiment}\ }\textbf
  {\bibinfo {volume} {2019}},\ \bibinfo {pages} {043106} (\bibinfo {year}
  {2019}{\natexlab{b}})}\BibitemShut {NoStop}%
\bibitem [{\citenamefont {Alba}\ and\ \citenamefont
  {Carollo}(2023)}]{SciPostPhys.2023Alba}%
  \BibitemOpen
  \bibfield  {author} {\bibinfo {author} {\bibfnamefont {V.}~\bibnamefont
  {Alba}}\ and\ \bibinfo {author} {\bibfnamefont {F.}~\bibnamefont {Carollo}},\
  }\bibfield  {title} {\bibinfo {title} {Logarithmic negativity in
  out-of-equilibrium open free-fermion chains: An exactly solvable case},\
  }\href {https://doi.org/10.21468/SciPostPhys.15.3.124} {\bibfield  {journal}
  {\bibinfo  {journal} {SciPost Physics}\ }\textbf {\bibinfo {volume} {15}},\
  \bibinfo {pages} {124} (\bibinfo {year} {2023})},\ \Eprint
  {https://arxiv.org/abs/2205.02139} {arXiv:2205.02139 [cond-mat,
  physics:hep-th, physics:quant-ph]} \BibitemShut {NoStop}%
\bibitem [{\citenamefont {Choi}\ \emph {et~al.}(2024)\citenamefont {Choi},
  \citenamefont {Knap},\ and\ \citenamefont {Pollmann}}]{Phys.Rev.B2024Choi}%
  \BibitemOpen
  \bibfield  {author} {\bibinfo {author} {\bibfnamefont {W.}~\bibnamefont
  {Choi}}, \bibinfo {author} {\bibfnamefont {M.}~\bibnamefont {Knap}},\ and\
  \bibinfo {author} {\bibfnamefont {F.}~\bibnamefont {Pollmann}},\ }\bibfield
  {title} {\bibinfo {title} {Finite-temperature entanglement negativity of
  fermionic symmetry-protected topological phases and quantum critical points
  in one dimension},\ }\href {https://doi.org/10.1103/PhysRevB.109.115132}
  {\bibfield  {journal} {\bibinfo  {journal} {Physical Review B}\ }\textbf
  {\bibinfo {volume} {109}},\ \bibinfo {pages} {115132} (\bibinfo {year}
  {2024})}\BibitemShut {NoStop}%
\bibitem [{\citenamefont {Wang}\ and\ \citenamefont
  {Xu}(2025)}]{Nat.Commun.2025Wanga}%
  \BibitemOpen
  \bibfield  {author} {\bibinfo {author} {\bibfnamefont {F.-H.}\ \bibnamefont
  {Wang}}\ and\ \bibinfo {author} {\bibfnamefont {X.~Y.}\ \bibnamefont {Xu}},\
  }\bibfield  {title} {\bibinfo {title} {Entanglement {{R\'enyi}} negativity of
  interacting fermions from quantum {{Monte Carlo}} simulations},\ }\href
  {https://doi.org/10.1038/s41467-025-57971-8} {\bibfield  {journal} {\bibinfo
  {journal} {Nature Communications}\ }\textbf {\bibinfo {volume} {16}},\
  \bibinfo {pages} {2637} (\bibinfo {year} {2025})},\ \Eprint
  {https://arxiv.org/abs/2312.14155} {arXiv:2312.14155} \BibitemShut {NoStop}%
\bibitem [{\citenamefont {D'Emidio}(2020)}]{Phys.Rev.Lett.2020DEmidio}%
  \BibitemOpen
  \bibfield  {author} {\bibinfo {author} {\bibfnamefont {J.}~\bibnamefont
  {D'Emidio}},\ }\bibfield  {title} {\bibinfo {title} {Entanglement {{Entropy}}
  from {{Nonequilibrium Work}}},\ }\href
  {https://doi.org/10.1103/PhysRevLett.124.110602} {\bibfield  {journal}
  {\bibinfo  {journal} {Physical Review Letters}\ }\textbf {\bibinfo {volume}
  {124}},\ \bibinfo {pages} {110602} (\bibinfo {year} {2020})}\BibitemShut
  {NoStop}%
\bibitem [{\citenamefont {Zhao}\ \emph
  {et~al.}(2022{\natexlab{b}})\citenamefont {Zhao}, \citenamefont {Chen},
  \citenamefont {Wang}, \citenamefont {Yan}, \citenamefont {Cheng},\ and\
  \citenamefont {Meng}}]{npjQuantumMater.2022Zhao}%
  \BibitemOpen
  \bibfield  {author} {\bibinfo {author} {\bibfnamefont {J.}~\bibnamefont
  {Zhao}}, \bibinfo {author} {\bibfnamefont {B.-B.}\ \bibnamefont {Chen}},
  \bibinfo {author} {\bibfnamefont {Y.-C.}\ \bibnamefont {Wang}}, \bibinfo
  {author} {\bibfnamefont {Z.}~\bibnamefont {Yan}}, \bibinfo {author}
  {\bibfnamefont {M.}~\bibnamefont {Cheng}},\ and\ \bibinfo {author}
  {\bibfnamefont {Z.~Y.}\ \bibnamefont {Meng}},\ }\bibfield  {title} {\bibinfo
  {title} {Measuring {{R\'enyi}} entanglement entropy with high efficiency and
  precision in quantum {{Monte Carlo}} simulations},\ }\href
  {https://doi.org/10.1038/s41535-022-00476-0} {\bibfield  {journal} {\bibinfo
  {journal} {npj Quantum Materials}\ }\textbf {\bibinfo {volume} {7}},\
  \bibinfo {pages} {69} (\bibinfo {year} {2022}{\natexlab{b}})}\BibitemShut
  {NoStop}%
\bibitem [{\citenamefont {Pan}\ \emph {et~al.}(2023)\citenamefont {Pan},
  \citenamefont {Da~Liao}, \citenamefont {Jiang}, \citenamefont {D'Emidio},
  \citenamefont {Qi},\ and\ \citenamefont {Meng}}]{Phys.Rev.B2023Pan}%
  \BibitemOpen
  \bibfield  {author} {\bibinfo {author} {\bibfnamefont {G.}~\bibnamefont
  {Pan}}, \bibinfo {author} {\bibfnamefont {Y.}~\bibnamefont {Da~Liao}},
  \bibinfo {author} {\bibfnamefont {W.}~\bibnamefont {Jiang}}, \bibinfo
  {author} {\bibfnamefont {J.}~\bibnamefont {D'Emidio}}, \bibinfo {author}
  {\bibfnamefont {Y.}~\bibnamefont {Qi}},\ and\ \bibinfo {author}
  {\bibfnamefont {Z.~Y.}\ \bibnamefont {Meng}},\ }\bibfield  {title} {\bibinfo
  {title} {Stable computation of entanglement entropy for {{2D}} interacting
  fermion systems},\ }\href {https://doi.org/10.1103/PhysRevB.108.L081123}
  {\bibfield  {journal} {\bibinfo  {journal} {Physical Review B}\ }\textbf
  {\bibinfo {volume} {108}},\ \bibinfo {pages} {L081123} (\bibinfo {year}
  {2023})}\BibitemShut {NoStop}%
\bibitem [{\citenamefont {Liao}\ \emph {et~al.}(2023)\citenamefont {Liao},
  \citenamefont {Pan}, \citenamefont {Jiang}, \citenamefont {Qi},\ and\
  \citenamefont {Meng}}]{ArXiv2023Liao}%
  \BibitemOpen
  \bibfield  {author} {\bibinfo {author} {\bibfnamefont {Y.~D.}\ \bibnamefont
  {Liao}}, \bibinfo {author} {\bibfnamefont {G.}~\bibnamefont {Pan}}, \bibinfo
  {author} {\bibfnamefont {W.}~\bibnamefont {Jiang}}, \bibinfo {author}
  {\bibfnamefont {Y.}~\bibnamefont {Qi}},\ and\ \bibinfo {author}
  {\bibfnamefont {Z.~Y.}\ \bibnamefont {Meng}},\ }\href@noop {} {\bibinfo
  {title} {The teaching from entanglement: {{2D SU}}(2) antiferromagnet to
  valence bond solid deconfined quantum critical points are not conformal}},\
  \bibinfo {howpublished} {arXiv:2302.11742} (\bibinfo {year} {2023}),\ \Eprint
  {https://arxiv.org/abs/2302.11742} {arXiv:2302.11742 [cond-mat,
  physics:math-ph, physics:physics, physics:quant-ph]} \BibitemShut {NoStop}%
\bibitem [{\citenamefont {Zhang}\ \emph {et~al.}(2024)\citenamefont {Zhang},
  \citenamefont {Pan}, \citenamefont {Chen}, \citenamefont {Sun},\ and\
  \citenamefont {Meng}}]{Phys.Rev.B2024Zhang}%
  \BibitemOpen
  \bibfield  {author} {\bibinfo {author} {\bibfnamefont {X.}~\bibnamefont
  {Zhang}}, \bibinfo {author} {\bibfnamefont {G.}~\bibnamefont {Pan}}, \bibinfo
  {author} {\bibfnamefont {B.-B.}\ \bibnamefont {Chen}}, \bibinfo {author}
  {\bibfnamefont {K.}~\bibnamefont {Sun}},\ and\ \bibinfo {author}
  {\bibfnamefont {Z.~Y.}\ \bibnamefont {Meng}},\ }\bibfield  {title} {\bibinfo
  {title} {Integral algorithm of exponential observables for interacting
  fermions in quantum {{Monte Carlo}} simulations},\ }\href
  {https://doi.org/10.1103/PhysRevB.109.205147} {\bibfield  {journal} {\bibinfo
   {journal} {Physical Review B}\ }\textbf {\bibinfo {volume} {109}},\ \bibinfo
  {pages} {205147} (\bibinfo {year} {2024})},\ \Eprint
  {https://arxiv.org/abs/2311.03448} {arXiv:2311.03448 [cond-mat,
  physics:quant-ph]} \BibitemShut {NoStop}%
\bibitem [{\citenamefont {Neal}(2001)}]{Stat.Comput.2001Neal}%
  \BibitemOpen
  \bibfield  {author} {\bibinfo {author} {\bibfnamefont {R.~M.}\ \bibnamefont
  {Neal}},\ }\bibfield  {title} {\bibinfo {title} {Annealed importance
  sampling},\ }\href {https://doi.org/10.1023/A:1008923215028} {\bibfield
  {journal} {\bibinfo  {journal} {Statistics and Computing}\ }\textbf {\bibinfo
  {volume} {11}},\ \bibinfo {pages} {125} (\bibinfo {year} {2001})}\BibitemShut
  {NoStop}%
\bibitem [{\citenamefont {Pollet}\ \emph {et~al.}(2008)\citenamefont {Pollet},
  \citenamefont {Kollath}, \citenamefont {Van~Houcke},\ and\ \citenamefont
  {Troyer}}]{NewJ.Phys.2008Pollet}%
  \BibitemOpen
  \bibfield  {author} {\bibinfo {author} {\bibfnamefont {L.}~\bibnamefont
  {Pollet}}, \bibinfo {author} {\bibfnamefont {C.}~\bibnamefont {Kollath}},
  \bibinfo {author} {\bibfnamefont {K.}~\bibnamefont {Van~Houcke}},\ and\
  \bibinfo {author} {\bibfnamefont {M.}~\bibnamefont {Troyer}},\ }\bibfield
  {title} {\bibinfo {title} {Temperature changes when adiabatically ramping up
  an optical lattice},\ }\href {https://doi.org/10.1088/1367-2630/10/6/065001}
  {\bibfield  {journal} {\bibinfo  {journal} {New Journal of Physics}\ }\textbf
  {\bibinfo {volume} {10}},\ \bibinfo {pages} {065001} (\bibinfo {year}
  {2008})}\BibitemShut {NoStop}%
\bibitem [{\citenamefont {Alba}(2017)}]{Phys.Rev.E2017Alba}%
  \BibitemOpen
  \bibfield  {author} {\bibinfo {author} {\bibfnamefont {V.}~\bibnamefont
  {Alba}},\ }\bibfield  {title} {\bibinfo {title} {Out-of-equilibrium protocol
  for {R{\'e}nyi} entropies via the {{Jarzynski}} equality},\ }\href
  {https://doi.org/10.1103/PhysRevE.95.062132} {\bibfield  {journal} {\bibinfo
  {journal} {Physical Review E}\ }\textbf {\bibinfo {volume} {95}},\ \bibinfo
  {pages} {062132} (\bibinfo {year} {2017})}\BibitemShut {NoStop}%
\bibitem [{\citenamefont {Bulgarelli}\ and\ \citenamefont
  {Panero}(2023)}]{JHEP2023Bulgarelli}%
  \BibitemOpen
  \bibfield  {author} {\bibinfo {author} {\bibfnamefont {A.}~\bibnamefont
  {Bulgarelli}}\ and\ \bibinfo {author} {\bibfnamefont {M.}~\bibnamefont
  {Panero}},\ }\bibfield  {title} {\bibinfo {title} {Entanglement entropy from
  non-equilibrium {{Monte Carlo}} simulations},\ }\href
  {https://doi.org/10.1007/JHEP06(2023)030} {\bibfield  {journal} {\bibinfo
  {journal} {JHEP}\ }\textbf {\bibinfo {volume} {06}},\ \bibinfo {pages}
  {030}}\BibitemShut {NoStop}%
\bibitem [{\citenamefont {Dai}\ and\ \citenamefont
  {Xu}(2025)}]{Phys.Rev.B2025Dai}%
  \BibitemOpen
  \bibfield  {author} {\bibinfo {author} {\bibfnamefont {Z.}~\bibnamefont
  {Dai}}\ and\ \bibinfo {author} {\bibfnamefont {X.~Y.}\ \bibnamefont {Xu}},\
  }\bibfield  {title} {\bibinfo {title} {Residual entropy from the temperature
  incremental {{Monte Carlo}} method},\ }\href
  {https://doi.org/10.1103/PhysRevB.111.L081108} {\bibfield  {journal}
  {\bibinfo  {journal} {Physical Review B}\ }\textbf {\bibinfo {volume}
  {111}},\ \bibinfo {pages} {L081108} (\bibinfo {year} {2025})},\ \Eprint
  {https://arxiv.org/abs/2402.17827} {arXiv:2402.17827 [cond-mat]} \BibitemShut
  {NoStop}%
\bibitem [{\citenamefont {Ding}\ \emph {et~al.}(2024)\citenamefont {Ding},
  \citenamefont {Sun}, \citenamefont {Ma}, \citenamefont {Pan}, \citenamefont
  {Cheng},\ and\ \citenamefont {Yan}}]{Phys.Rev.B2024Ding}%
  \BibitemOpen
  \bibfield  {author} {\bibinfo {author} {\bibfnamefont {Y.-M.}\ \bibnamefont
  {Ding}}, \bibinfo {author} {\bibfnamefont {J.-S.}\ \bibnamefont {Sun}},
  \bibinfo {author} {\bibfnamefont {N.}~\bibnamefont {Ma}}, \bibinfo {author}
  {\bibfnamefont {G.}~\bibnamefont {Pan}}, \bibinfo {author} {\bibfnamefont
  {C.}~\bibnamefont {Cheng}},\ and\ \bibinfo {author} {\bibfnamefont
  {Z.}~\bibnamefont {Yan}},\ }\bibfield  {title} {\bibinfo {title}
  {Reweight-annealing method for evaluating the partition function via quantum
  {{Monte Carlo}} calculations},\ }\href
  {https://doi.org/10.1103/PhysRevB.110.165152} {\bibfield  {journal} {\bibinfo
   {journal} {Physical Review B}\ }\textbf {\bibinfo {volume} {110}},\ \bibinfo
  {pages} {165152} (\bibinfo {year} {2024})},\ \Eprint
  {https://arxiv.org/abs/2403.08642} {arXiv:2403.08642 [cond-mat,
  physics:quant-ph]} \BibitemShut {NoStop}%
\bibitem [{\citenamefont {Ripka}\ \emph {et~al.}(1986)\citenamefont {Ripka},
  \citenamefont {Blaizot},\ and\ \citenamefont
  {Ripka}}]{Ripka1986QuantumTheoryFinite}%
  \BibitemOpen
  \bibfield  {author} {\bibinfo {author} {\bibfnamefont {S.~R. P.~G.}\
  \bibnamefont {Ripka}}, \bibinfo {author} {\bibfnamefont {J.-P.}\ \bibnamefont
  {Blaizot}},\ and\ \bibinfo {author} {\bibfnamefont {G.}~\bibnamefont
  {Ripka}},\ }\href@noop {} {\emph {\bibinfo {title} {Quantum {{Theory}} of
  {{Finite Systems}}}}}\ (\bibinfo  {publisher} {MIT Press},\ \bibinfo {year}
  {1986})\BibitemShut {NoStop}%
\bibitem [{\citenamefont {{de
  Torres-Solanot}}(2017)}]{deTorres-Solanot2017TimedependentGaussianVariational}%
  \BibitemOpen
  \bibfield  {author} {\bibinfo {author} {\bibfnamefont {P.~S.}\ \bibnamefont
  {{de Torres-Solanot}}},\ }\emph {\bibinfo {title} {A Time-Dependent
  {{Gaussian}} Variational Description of {{Lattice Gauge Theories}}}},\
  \href@noop {} {Master's thesis} (\bibinfo {year} {2017})\BibitemShut
  {NoStop}%
\bibitem [{\citenamefont {Wick}\ \emph {et~al.}(1952)\citenamefont {Wick},
  \citenamefont {Wightman},\ and\ \citenamefont {Wigner}}]{Phys.Rev.1952Wick}%
  \BibitemOpen
  \bibfield  {author} {\bibinfo {author} {\bibfnamefont {G.~C.}\ \bibnamefont
  {Wick}}, \bibinfo {author} {\bibfnamefont {A.~S.}\ \bibnamefont {Wightman}},\
  and\ \bibinfo {author} {\bibfnamefont {E.~P.}\ \bibnamefont {Wigner}},\
  }\bibfield  {title} {\bibinfo {title} {The {{Intrinsic Parity}} of
  {{Elementary Particles}}},\ }\href {https://doi.org/10.1103/PhysRev.88.101}
  {\bibfield  {journal} {\bibinfo  {journal} {Physical Review}\ }\textbf
  {\bibinfo {volume} {88}},\ \bibinfo {pages} {101} (\bibinfo {year}
  {1952})}\BibitemShut {NoStop}%
\bibitem [{\citenamefont {Aharonov}\ and\ \citenamefont
  {Susskind}(1967)}]{Phys.Rev.1967Aharonov}%
  \BibitemOpen
  \bibfield  {author} {\bibinfo {author} {\bibfnamefont {Y.}~\bibnamefont
  {Aharonov}}\ and\ \bibinfo {author} {\bibfnamefont {L.}~\bibnamefont
  {Susskind}},\ }\bibfield  {title} {\bibinfo {title} {Charge {{Superselection
  Rule}}},\ }\href {https://doi.org/10.1103/PhysRev.155.1428} {\bibfield
  {journal} {\bibinfo  {journal} {Physical Review}\ }\textbf {\bibinfo {volume}
  {155}},\ \bibinfo {pages} {1428} (\bibinfo {year} {1967})}\BibitemShut
  {NoStop}%
\bibitem [{\citenamefont
  {Kraus}(2009)}]{Kraus2009QuantumInformationPerspective}%
  \BibitemOpen
  \bibfield  {author} {\bibinfo {author} {\bibfnamefont {C.~V.}\ \bibnamefont
  {Kraus}},\ }\emph {\bibinfo {title} {A {{Quantum Information Perspective}} of
  {{Fermionic Quantum Many-Body Systems}}}},\ \href@noop {} {Ph.D. thesis}
  (\bibinfo {year} {2009})\BibitemShut {NoStop}%
\bibitem [{\citenamefont {Klich}(2014)}]{J.Stat.Mech.2014Klich}%
  \BibitemOpen
  \bibfield  {author} {\bibinfo {author} {\bibfnamefont {I.}~\bibnamefont
  {Klich}},\ }\bibfield  {title} {\bibinfo {title} {A note on the full counting
  statistics of paired fermions},\ }\href
  {https://doi.org/10.1088/1742-5468/2014/11/P11006} {\bibfield  {journal}
  {\bibinfo  {journal} {Journal of Statistical Mechanics: Theory and
  Experiment}\ }\textbf {\bibinfo {volume} {2014}},\ \bibinfo {pages} {P11006}
  (\bibinfo {year} {2014})}\BibitemShut {NoStop}%
\bibitem [{\citenamefont {Chang}\ and\ \citenamefont
  {Wen}(2016)}]{Phys.Rev.B2016Chang}%
  \BibitemOpen
  \bibfield  {author} {\bibinfo {author} {\bibfnamefont {P.-Y.}\ \bibnamefont
  {Chang}}\ and\ \bibinfo {author} {\bibfnamefont {X.}~\bibnamefont {Wen}},\
  }\bibfield  {title} {\bibinfo {title} {Entanglement negativity in
  free-fermion systems: An overlap matrix approach},\ }\href
  {https://doi.org/10.1103/PhysRevB.93.195140} {\bibfield  {journal} {\bibinfo
  {journal} {Physical Review B}\ }\textbf {\bibinfo {volume} {93}},\ \bibinfo
  {pages} {195140} (\bibinfo {year} {2016})}\BibitemShut {NoStop}%
\bibitem [{\citenamefont {Fang}\ and\ \citenamefont
  {Xu}(2025)}]{PRB2025JQ.Fang}%
  \BibitemOpen
  \bibfield  {author} {\bibinfo {author} {\bibfnamefont {J.~Q.}\ \bibnamefont
  {Fang}}\ and\ \bibinfo {author} {\bibfnamefont {X.~Y.}\ \bibnamefont {Xu}},\
  }\bibfield  {title} {\bibinfo {title} {Fermionic partial transpose in the
  overlap matrix framework for entanglement negativity},\ }\href
  {https://doi.org/10.1103/rq4y-35rv} {\bibfield  {journal} {\bibinfo
  {journal} {Physical Review B}\ }\textbf {\bibinfo {volume} {112}},\ \bibinfo
  {pages} {075154} (\bibinfo {year} {2025})},\ \Eprint
  {https://arxiv.org/abs/2503.07742} {arXiv:2503.07742 [quant-ph]} \BibitemShut
  {NoStop}%
\bibitem [{\citenamefont
  {Bravyi}(2005{\natexlab{a}})}]{Quant.Inf.Comput.2005Bravyi}%
  \BibitemOpen
  \bibfield  {author} {\bibinfo {author} {\bibfnamefont {S.}~\bibnamefont
  {Bravyi}},\ }\bibfield  {title} {\bibinfo {title} {Lagrangian representation
  for fermionic linear optics},\ }\href {https://doi.org/10.26421/QIC5.3-3}
  {\bibfield  {journal} {\bibinfo  {journal} {Quant. Inf. Comput.}\ }\textbf
  {\bibinfo {volume} {5}},\ \bibinfo {pages} {216} (\bibinfo {year}
  {2005}{\natexlab{a}})}\BibitemShut {NoStop}%
\bibitem [{\citenamefont {Bravyi}(2005{\natexlab{b}})}]{ArXiv2005Bravyi}%
  \BibitemOpen
  \bibfield  {author} {\bibinfo {author} {\bibfnamefont {S.}~\bibnamefont
  {Bravyi}},\ }\href {https://doi.org/10.48550/arXiv.quant-ph/0507282}
  {\bibinfo {title} {Classical capacity of fermionic product channels}},\
  \bibinfo {howpublished} {arXiv:quant-ph/0507282} (\bibinfo {year}
  {2005}{\natexlab{b}}),\ \Eprint {https://arxiv.org/abs/quant-ph/0507282}
  {arXiv:quant-ph/0507282} \BibitemShut {NoStop}%
\bibitem [{\citenamefont {Fagotti}\ and\ \citenamefont
  {Calabrese}(2010)}]{J.Stat.Mech.2010Fagotti}%
  \BibitemOpen
  \bibfield  {author} {\bibinfo {author} {\bibfnamefont {M.}~\bibnamefont
  {Fagotti}}\ and\ \bibinfo {author} {\bibfnamefont {P.}~\bibnamefont
  {Calabrese}},\ }\bibfield  {title} {\bibinfo {title} {Entanglement entropy of
  two disjoint blocks in {{XY}} chains},\ }\href
  {https://doi.org/10.1088/1742-5468/2010/04/P04016} {\bibfield  {journal}
  {\bibinfo  {journal} {Journal of Statistical Mechanics: Theory and
  Experiment}\ }\textbf {\bibinfo {volume} {2010}},\ \bibinfo {pages} {P04016}
  (\bibinfo {year} {2010})}\BibitemShut {NoStop}%
\bibitem [{\citenamefont {Peschel}(2003)}]{J.Phys.AMath.Gen.2003Peschel}%
  \BibitemOpen
  \bibfield  {author} {\bibinfo {author} {\bibfnamefont {I.}~\bibnamefont
  {Peschel}},\ }\bibfield  {title} {\bibinfo {title} {Calculation of reduced
  density matrices from correlation functions},\ }\href
  {https://doi.org/10.1088/0305-4470/36/14/101} {\bibfield  {journal} {\bibinfo
   {journal} {Journal of Physics A: Mathematical and General}\ }\textbf
  {\bibinfo {volume} {36}},\ \bibinfo {pages} {L205} (\bibinfo {year}
  {2003})}\BibitemShut {NoStop}%
\bibitem [{\citenamefont {Cheong}\ and\ \citenamefont
  {Henley}(2004)}]{Phys.Rev.B2004Cheong}%
  \BibitemOpen
  \bibfield  {author} {\bibinfo {author} {\bibfnamefont {S.-A.}\ \bibnamefont
  {Cheong}}\ and\ \bibinfo {author} {\bibfnamefont {C.~L.}\ \bibnamefont
  {Henley}},\ }\bibfield  {title} {\bibinfo {title} {Many-body density matrices
  for free fermions},\ }\href {https://doi.org/10.1103/PhysRevB.69.075111}
  {\bibfield  {journal} {\bibinfo  {journal} {Physical Review B}\ }\textbf
  {\bibinfo {volume} {69}},\ \bibinfo {pages} {075111} (\bibinfo {year}
  {2004})}\BibitemShut {NoStop}%
\bibitem [{\citenamefont {Peschel}\ and\ \citenamefont
  {Eisler}(2009)}]{J.Phys.AMath.Theor.2009Peschel}%
  \BibitemOpen
  \bibfield  {author} {\bibinfo {author} {\bibfnamefont {I.}~\bibnamefont
  {Peschel}}\ and\ \bibinfo {author} {\bibfnamefont {V.}~\bibnamefont
  {Eisler}},\ }\bibfield  {title} {\bibinfo {title} {Reduced density matrices
  and entanglement entropy in free lattice models},\ }\href
  {https://doi.org/10.1088/1751-8113/42/50/504003} {\bibfield  {journal}
  {\bibinfo  {journal} {Journal of Physics A: Mathematical and Theoretical}\
  }\textbf {\bibinfo {volume} {42}},\ \bibinfo {pages} {504003} (\bibinfo
  {year} {2009})}\BibitemShut {NoStop}%
\bibitem [{\citenamefont {Eisler}\ and\ \citenamefont
  {Zimbor{\'a}s}(2015)}]{NewJ.Phys.2015Eisler}%
  \BibitemOpen
  \bibfield  {author} {\bibinfo {author} {\bibfnamefont {V.}~\bibnamefont
  {Eisler}}\ and\ \bibinfo {author} {\bibfnamefont {Z.}~\bibnamefont
  {Zimbor{\'a}s}},\ }\bibfield  {title} {\bibinfo {title} {On the partial
  transpose of fermionic {{Gaussian}} states},\ }\href
  {https://doi.org/10.1088/1367-2630/17/5/053048} {\bibfield  {journal}
  {\bibinfo  {journal} {New Journal of Physics}\ }\textbf {\bibinfo {volume}
  {17}},\ \bibinfo {pages} {053048} (\bibinfo {year} {2015})}\BibitemShut
  {NoStop}%
\bibitem [{\citenamefont {Li}\ \emph {et~al.}(2015)\citenamefont {Li},
  \citenamefont {Jiang},\ and\ \citenamefont {Yao}}]{Phys.Rev.B2015Li}%
  \BibitemOpen
  \bibfield  {author} {\bibinfo {author} {\bibfnamefont {Z.-X.}\ \bibnamefont
  {Li}}, \bibinfo {author} {\bibfnamefont {Y.-F.}\ \bibnamefont {Jiang}},\ and\
  \bibinfo {author} {\bibfnamefont {H.}~\bibnamefont {Yao}},\ }\bibfield
  {title} {\bibinfo {title} {Solving the fermion sign problem in quantum
  {{Monte Carlo}} simulations by {{Majorana}} representation},\ }\href
  {https://doi.org/10.1103/PhysRevB.91.241117} {\bibfield  {journal} {\bibinfo
  {journal} {Physical Review B}\ }\textbf {\bibinfo {volume} {91}},\ \bibinfo
  {pages} {241117} (\bibinfo {year} {2015})}\BibitemShut {NoStop}%
\bibitem [{\citenamefont {Han}\ \emph {et~al.}(2024)\citenamefont {Han},
  \citenamefont {Wan},\ and\ \citenamefont {Yao}}]{ArXiv2024Han}%
  \BibitemOpen
  \bibfield  {author} {\bibinfo {author} {\bibfnamefont {Z.-Y.}\ \bibnamefont
  {Han}}, \bibinfo {author} {\bibfnamefont {Z.-Q.}\ \bibnamefont {Wan}},\ and\
  \bibinfo {author} {\bibfnamefont {H.}~\bibnamefont {Yao}},\ }\href
  {https://doi.org/10.48550/arXiv.2408.10311} {\bibinfo {title} {Pfaffian
  quantum {{Monte Carlo}}: Solution to {{Majorana}} sign ambiguity and
  applications}},\ \bibinfo {howpublished} {arXiv:2408.10311} (\bibinfo {year}
  {2024}),\ \Eprint {https://arxiv.org/abs/2408.10311} {arXiv:2408.10311
  [cond-mat, physics:physics, physics:quant-ph]} \BibitemShut {NoStop}%
\bibitem [{\citenamefont {Blankenbecler}\ \emph {et~al.}(1981)\citenamefont
  {Blankenbecler}, \citenamefont {Scalapino},\ and\ \citenamefont
  {Sugar}}]{Phys.Rev.D1981Blankenbecler}%
  \BibitemOpen
  \bibfield  {author} {\bibinfo {author} {\bibfnamefont {R.}~\bibnamefont
  {Blankenbecler}}, \bibinfo {author} {\bibfnamefont {D.~J.}\ \bibnamefont
  {Scalapino}},\ and\ \bibinfo {author} {\bibfnamefont {R.~L.}\ \bibnamefont
  {Sugar}},\ }\bibfield  {title} {\bibinfo {title} {Monte {{Carlo}}
  calculations of coupled boson-fermion systems. {{I}}},\ }\href
  {https://doi.org/10.1103/PhysRevD.24.2278} {\bibfield  {journal} {\bibinfo
  {journal} {Physical Review D}\ }\textbf {\bibinfo {volume} {24}},\ \bibinfo
  {pages} {2278} (\bibinfo {year} {1981})}\BibitemShut {NoStop}%
\bibitem [{\citenamefont {Scalapino}\ and\ \citenamefont
  {Sugar}(1981)}]{Phys.Rev.B1981Scalapino}%
  \BibitemOpen
  \bibfield  {author} {\bibinfo {author} {\bibfnamefont {D.~J.}\ \bibnamefont
  {Scalapino}}\ and\ \bibinfo {author} {\bibfnamefont {R.~L.}\ \bibnamefont
  {Sugar}},\ }\bibfield  {title} {\bibinfo {title} {Monte {{Carlo}}
  calculations of coupled boson-fermion systems. {{II}}},\ }\href
  {https://doi.org/10.1103/PhysRevB.24.4295} {\bibfield  {journal} {\bibinfo
  {journal} {Physical Review B}\ }\textbf {\bibinfo {volume} {24}},\ \bibinfo
  {pages} {4295} (\bibinfo {year} {1981})}\BibitemShut {NoStop}%
\bibitem [{\citenamefont {Hirsch}(1985)}]{Phys.Rev.B1985Hirsch}%
  \BibitemOpen
  \bibfield  {author} {\bibinfo {author} {\bibfnamefont {J.~E.}\ \bibnamefont
  {Hirsch}},\ }\bibfield  {title} {\bibinfo {title} {Two-dimensional
  {{Hubbard}} model: Numerical simulation study},\ }\href
  {https://doi.org/10.1103/PhysRevB.31.4403} {\bibfield  {journal} {\bibinfo
  {journal} {Physical Review B}\ }\textbf {\bibinfo {volume} {31}},\ \bibinfo
  {pages} {4403} (\bibinfo {year} {1985})}\BibitemShut {NoStop}%
\bibitem [{\citenamefont {Assaad}\ and\ \citenamefont
  {Evertz}(2008)}]{Assaad2008WorldlineDeterminantalQuantum}%
  \BibitemOpen
  \bibfield  {author} {\bibinfo {author} {\bibfnamefont {F.}~\bibnamefont
  {Assaad}}\ and\ \bibinfo {author} {\bibfnamefont {H.}~\bibnamefont
  {Evertz}},\ }\bibfield  {title} {\bibinfo {title} {World-line and
  {{Determinantal Quantum Monte Carlo Methods}} for {{Spins}}, {{Phonons}} and
  {{Electrons}}},\ }in\ \href {https://doi.org/10.1007/978-3-540-74686-7_10}
  {\emph {\bibinfo {booktitle} {Computational {{Many-Particle Physics}}}}},\
  \bibinfo {series and number} {Lecture {{Notes}} in {{Physics}}},\ \bibinfo
  {editor} {edited by\ \bibinfo {editor} {\bibfnamefont {H.}~\bibnamefont
  {Fehske}}, \bibinfo {editor} {\bibfnamefont {R.}~\bibnamefont {Schneider}},\
  and\ \bibinfo {editor} {\bibfnamefont {A.}~\bibnamefont {Wei{\ss}e}}}\
  (\bibinfo  {publisher} {Springer},\ \bibinfo {address} {Berlin, Heidelberg},\
  \bibinfo {year} {2008})\ pp.\ \bibinfo {pages} {277--356}\BibitemShut
  {NoStop}%
\bibitem [{\citenamefont {Loh}\ \emph {et~al.}(1989)\citenamefont {Loh},
  \citenamefont {Gubernatis}, \citenamefont {Scalettar}, \citenamefont
  {Sugar},\ and\ \citenamefont
  {White}}]{Loh1989StableMatrixMultiplicationAlgorithms}%
  \BibitemOpen
  \bibfield  {author} {\bibinfo {author} {\bibfnamefont {E.~Y.}\ \bibnamefont
  {Loh}}, \bibinfo {author} {\bibfnamefont {J.~E.}\ \bibnamefont {Gubernatis}},
  \bibinfo {author} {\bibfnamefont {R.~T.}\ \bibnamefont {Scalettar}}, \bibinfo
  {author} {\bibfnamefont {R.~L.}\ \bibnamefont {Sugar}},\ and\ \bibinfo
  {author} {\bibfnamefont {S.~R.}\ \bibnamefont {White}},\ }\bibfield  {title}
  {\bibinfo {title} {Stable {{Matrix-Multiplication Algorithms}} for
  {{Low-Temperature Numerical Simulations}} of {{Fermions}}},\ }in\ \href
  {https://doi.org/10.1007/978-1-4613-0565-1_8} {\emph {\bibinfo {booktitle}
  {Interacting {{Electrons}} in {{Reduced Dimensions}}}}},\ \bibinfo {editor}
  {edited by\ \bibinfo {editor} {\bibfnamefont {D.}~\bibnamefont {Baeriswyl}}\
  and\ \bibinfo {editor} {\bibfnamefont {D.~K.}\ \bibnamefont {Campbell}}}\
  (\bibinfo  {publisher} {Springer US},\ \bibinfo {address} {Boston, MA},\
  \bibinfo {year} {1989})\ pp.\ \bibinfo {pages} {55--60}\BibitemShut {NoStop}%
\bibitem [{\citenamefont {Bauer}(2020)}]{SciPostPhys.Core2020Bauer}%
  \BibitemOpen
  \bibfield  {author} {\bibinfo {author} {\bibfnamefont {C.}~\bibnamefont
  {Bauer}},\ }\bibfield  {title} {\bibinfo {title} {Fast and stable determinant
  quantum {{Monte Carlo}}},\ }\href
  {https://doi.org/10.21468/SciPostPhysCore.2.2.011} {\bibfield  {journal}
  {\bibinfo  {journal} {SciPost Physics Core}\ }\textbf {\bibinfo {volume}
  {2}},\ \bibinfo {pages} {011} (\bibinfo {year} {2020})}\BibitemShut {NoStop}%
\bibitem [{\citenamefont {Assaad}\ \emph {et~al.}(2014)\citenamefont {Assaad},
  \citenamefont {Lang},\ and\ \citenamefont
  {Parisen~Toldin}}]{Phys.Rev.B2014Assaad}%
  \BibitemOpen
  \bibfield  {author} {\bibinfo {author} {\bibfnamefont {F.~F.}\ \bibnamefont
  {Assaad}}, \bibinfo {author} {\bibfnamefont {T.~C.}\ \bibnamefont {Lang}},\
  and\ \bibinfo {author} {\bibfnamefont {F.}~\bibnamefont {Parisen~Toldin}},\
  }\bibfield  {title} {\bibinfo {title} {Entanglement spectra of interacting
  fermions in quantum {{Monte Carlo}} simulations},\ }\href
  {https://doi.org/10.1103/PhysRevB.89.125121} {\bibfield  {journal} {\bibinfo
  {journal} {Physical Review B}\ }\textbf {\bibinfo {volume} {89}},\ \bibinfo
  {pages} {125121} (\bibinfo {year} {2014})}\BibitemShut {NoStop}%
\bibitem [{\citenamefont {Sugiyama}\ and\ \citenamefont
  {Koonin}(1986)}]{Ann.Phys.1986Sugiyama}%
  \BibitemOpen
  \bibfield  {author} {\bibinfo {author} {\bibfnamefont {G.}~\bibnamefont
  {Sugiyama}}\ and\ \bibinfo {author} {\bibfnamefont {S.~E.}\ \bibnamefont
  {Koonin}},\ }\bibfield  {title} {\bibinfo {title} {Auxiliary field
  {{Monte-Carlo}} for quantum many-body ground states},\ }\href
  {https://doi.org/10.1016/0003-4916(86)90107-7} {\bibfield  {journal}
  {\bibinfo  {journal} {Annals of Physics}\ }\textbf {\bibinfo {volume}
  {168}},\ \bibinfo {pages} {1} (\bibinfo {year} {1986})}\BibitemShut {NoStop}%
\bibitem [{\citenamefont {Sorella}\ \emph
  {et~al.}(1989{\natexlab{a}})\citenamefont {Sorella}, \citenamefont {Baroni},
  \citenamefont {Car},\ and\ \citenamefont
  {Parrinello}}]{Europhys.Lett.1989Sorella}%
  \BibitemOpen
  \bibfield  {author} {\bibinfo {author} {\bibfnamefont {S.}~\bibnamefont
  {Sorella}}, \bibinfo {author} {\bibfnamefont {S.}~\bibnamefont {Baroni}},
  \bibinfo {author} {\bibfnamefont {R.}~\bibnamefont {Car}},\ and\ \bibinfo
  {author} {\bibfnamefont {M.}~\bibnamefont {Parrinello}},\ }\bibfield  {title}
  {\bibinfo {title} {A novel technique for the simulation of interacting
  fermion systems},\ }\href {https://doi.org/10.1209/0295-5075/8/7/014}
  {\bibfield  {journal} {\bibinfo  {journal} {Europhysics Letters (EPL)}\
  }\textbf {\bibinfo {volume} {8}},\ \bibinfo {pages} {663} (\bibinfo {year}
  {1989}{\natexlab{a}})}\BibitemShut {NoStop}%
\bibitem [{\citenamefont {Sorella}\ \emph
  {et~al.}(1989{\natexlab{b}})\citenamefont {Sorella}, \citenamefont {Parola},
  \citenamefont {Parrinello},\ and\ \citenamefont
  {Tosatti}}]{Int.J.ModernPhys.B1989Sorella}%
  \BibitemOpen
  \bibfield  {author} {\bibinfo {author} {\bibfnamefont {S.}~\bibnamefont
  {Sorella}}, \bibinfo {author} {\bibfnamefont {A.}~\bibnamefont {Parola}},
  \bibinfo {author} {\bibfnamefont {M.}~\bibnamefont {Parrinello}},\ and\
  \bibinfo {author} {\bibfnamefont {E.}~\bibnamefont {Tosatti}},\ }\bibfield
  {title} {\bibinfo {title} {{{NUMERICAl STUDY OF THE 2D HUBBARD MODEL AT HALF
  FILLING}}},\ }\href {https://doi.org/10.1142/S0217979289001214} {\bibfield
  {journal} {\bibinfo  {journal} {International Journal of Modern Physics B}\
  }\textbf {\bibinfo {volume} {3}},\ \bibinfo {pages} {1875} (\bibinfo {year}
  {1989}{\natexlab{b}})}\BibitemShut {NoStop}%
\bibitem [{\citenamefont {Drut}\ and\ \citenamefont
  {Porter}(2016)}]{Phys.Rev.E2016Drut}%
  \BibitemOpen
  \bibfield  {author} {\bibinfo {author} {\bibfnamefont {J.~E.}\ \bibnamefont
  {Drut}}\ and\ \bibinfo {author} {\bibfnamefont {W.~J.}\ \bibnamefont
  {Porter}},\ }\bibfield  {title} {\bibinfo {title} {Entanglement, noise, and
  the cumulant expansion},\ }\href {https://doi.org/10.1103/PhysRevE.93.043301}
  {\bibfield  {journal} {\bibinfo  {journal} {Physical Review E}\ }\textbf
  {\bibinfo {volume} {93}},\ \bibinfo {pages} {043301} (\bibinfo {year}
  {2016})}\BibitemShut {NoStop}%
\bibitem [{\citenamefont {Shi}\ and\ \citenamefont
  {Zhang}(2016)}]{Phys.Rev.E2016Shi}%
  \BibitemOpen
  \bibfield  {author} {\bibinfo {author} {\bibfnamefont {H.}~\bibnamefont
  {Shi}}\ and\ \bibinfo {author} {\bibfnamefont {S.}~\bibnamefont {Zhang}},\
  }\bibfield  {title} {\bibinfo {title} {Infinite variance in fermion quantum
  {{Monte Carlo}} calculations},\ }\href
  {https://doi.org/10.1103/PhysRevE.93.033303} {\bibfield  {journal} {\bibinfo
  {journal} {Physical Review E}\ }\textbf {\bibinfo {volume} {93}},\ \bibinfo
  {pages} {033303} (\bibinfo {year} {2016})}\BibitemShut {NoStop}%
\bibitem [{\citenamefont {Scalapino}\ \emph {et~al.}(1984)\citenamefont
  {Scalapino}, \citenamefont {Sugar},\ and\ \citenamefont
  {Toussaint}}]{Phys.Rev.B1984Scalapino}%
  \BibitemOpen
  \bibfield  {author} {\bibinfo {author} {\bibfnamefont {D.~J.}\ \bibnamefont
  {Scalapino}}, \bibinfo {author} {\bibfnamefont {R.~L.}\ \bibnamefont
  {Sugar}},\ and\ \bibinfo {author} {\bibfnamefont {W.~D.}\ \bibnamefont
  {Toussaint}},\ }\bibfield  {title} {\bibinfo {title} {Monte {{Carlo}} study
  of a two-dimensional spin-polarized fermion lattice gas},\ }\href
  {https://doi.org/10.1103/PhysRevB.29.5253} {\bibfield  {journal} {\bibinfo
  {journal} {Physical Review B}\ }\textbf {\bibinfo {volume} {29}},\ \bibinfo
  {pages} {5253} (\bibinfo {year} {1984})}\BibitemShut {NoStop}%
\bibitem [{\citenamefont {Gubernatis}\ \emph {et~al.}(1985)\citenamefont
  {Gubernatis}, \citenamefont {Scalapino}, \citenamefont {Sugar},\ and\
  \citenamefont {Toussaint}}]{Phys.Rev.B1985Gubernatis}%
  \BibitemOpen
  \bibfield  {author} {\bibinfo {author} {\bibfnamefont {J.~E.}\ \bibnamefont
  {Gubernatis}}, \bibinfo {author} {\bibfnamefont {D.~J.}\ \bibnamefont
  {Scalapino}}, \bibinfo {author} {\bibfnamefont {R.~L.}\ \bibnamefont
  {Sugar}},\ and\ \bibinfo {author} {\bibfnamefont {W.~D.}\ \bibnamefont
  {Toussaint}},\ }\bibfield  {title} {\bibinfo {title} {Two-dimensional
  spin-polarized fermion lattice gases},\ }\href
  {https://doi.org/10.1103/PhysRevB.32.103} {\bibfield  {journal} {\bibinfo
  {journal} {Physical Review B}\ }\textbf {\bibinfo {volume} {32}},\ \bibinfo
  {pages} {103} (\bibinfo {year} {1985})}\BibitemShut {NoStop}%
\bibitem [{\citenamefont {Wang}\ and\ \citenamefont
  {Troyer}(2014)}]{Phys.Rev.Lett.2014Wang}%
  \BibitemOpen
  \bibfield  {author} {\bibinfo {author} {\bibfnamefont {L.}~\bibnamefont
  {Wang}}\ and\ \bibinfo {author} {\bibfnamefont {M.}~\bibnamefont {Troyer}},\
  }\bibfield  {title} {\bibinfo {title} {Renyi entanglement entropy of
  interacting fermions calculated using the continuous-time quantum {{Monte
  Carlo}} method},\ }\href {https://doi.org/10.1103/PhysRevLett.113.110401}
  {\bibfield  {journal} {\bibinfo  {journal} {Physical Review Letters}\
  }\textbf {\bibinfo {volume} {113}},\ \bibinfo {pages} {110401} (\bibinfo
  {year} {2014})}\BibitemShut {NoStop}%
\bibitem [{\citenamefont {Huffman}\ and\ \citenamefont
  {Chandrasekharan}(2014)}]{Phys.Rev.B2014Huffman}%
  \BibitemOpen
  \bibfield  {author} {\bibinfo {author} {\bibfnamefont {E.~F.}\ \bibnamefont
  {Huffman}}\ and\ \bibinfo {author} {\bibfnamefont {S.}~\bibnamefont
  {Chandrasekharan}},\ }\bibfield  {title} {\bibinfo {title} {Solution to sign
  problems in half-filled spin-polarized electronic systems},\ }\href
  {https://doi.org/10.1103/PhysRevB.89.111101} {\bibfield  {journal} {\bibinfo
  {journal} {Physical Review B}\ }\textbf {\bibinfo {volume} {89}},\ \bibinfo
  {pages} {111101} (\bibinfo {year} {2014})}\BibitemShut {NoStop}%
\bibitem [{\citenamefont {Wang}\ \emph {et~al.}(2015)\citenamefont {Wang},
  \citenamefont {Liu}, \citenamefont {Iazzi}, \citenamefont {Troyer},\ and\
  \citenamefont {Harcos}}]{Phys.Rev.Lett.2015Wang}%
  \BibitemOpen
  \bibfield  {author} {\bibinfo {author} {\bibfnamefont {L.}~\bibnamefont
  {Wang}}, \bibinfo {author} {\bibfnamefont {Y.-H.}\ \bibnamefont {Liu}},
  \bibinfo {author} {\bibfnamefont {M.}~\bibnamefont {Iazzi}}, \bibinfo
  {author} {\bibfnamefont {M.}~\bibnamefont {Troyer}},\ and\ \bibinfo {author}
  {\bibfnamefont {G.}~\bibnamefont {Harcos}},\ }\bibfield  {title} {\bibinfo
  {title} {Split orthogonal group: A guiding principle for sign-problem-free
  fermionic simulations},\ }\href
  {https://doi.org/10.1103/PhysRevLett.115.250601} {\bibfield  {journal}
  {\bibinfo  {journal} {Physical Review Letters}\ }\textbf {\bibinfo {volume}
  {115}},\ \bibinfo {pages} {250601} (\bibinfo {year} {2015})}\BibitemShut
  {NoStop}%
\bibitem [{\citenamefont {Wei}\ \emph {et~al.}(2016)\citenamefont {Wei},
  \citenamefont {Wu}, \citenamefont {Li}, \citenamefont {Zhang},\ and\
  \citenamefont {Xiang}}]{Phys.Rev.Lett.2016Wei}%
  \BibitemOpen
  \bibfield  {author} {\bibinfo {author} {\bibfnamefont {Z.~C.}\ \bibnamefont
  {Wei}}, \bibinfo {author} {\bibfnamefont {C.}~\bibnamefont {Wu}}, \bibinfo
  {author} {\bibfnamefont {Y.}~\bibnamefont {Li}}, \bibinfo {author}
  {\bibfnamefont {S.}~\bibnamefont {Zhang}},\ and\ \bibinfo {author}
  {\bibfnamefont {T.}~\bibnamefont {Xiang}},\ }\bibfield  {title} {\bibinfo
  {title} {Majorana positivity and the fermion sign problem of quantum {{Monte
  Carlo}} simulations},\ }\href
  {https://doi.org/10.1103/PhysRevLett.116.250601} {\bibfield  {journal}
  {\bibinfo  {journal} {Physical Review Letters}\ }\textbf {\bibinfo {volume}
  {116}},\ \bibinfo {pages} {250601} (\bibinfo {year} {2016})}\BibitemShut
  {NoStop}%
\bibitem [{\citenamefont {Li}\ \emph {et~al.}(2016)\citenamefont {Li},
  \citenamefont {Jiang},\ and\ \citenamefont {Yao}}]{Phys.Rev.Lett.2016Li}%
  \BibitemOpen
  \bibfield  {author} {\bibinfo {author} {\bibfnamefont {Z.-X.}\ \bibnamefont
  {Li}}, \bibinfo {author} {\bibfnamefont {Y.-F.}\ \bibnamefont {Jiang}},\ and\
  \bibinfo {author} {\bibfnamefont {H.}~\bibnamefont {Yao}},\ }\bibfield
  {title} {\bibinfo {title} {Majorana-time-reversal symmetries: A fundamental
  principle for sign-problem-free quantum {{Monte Carlo}} simulations},\ }\href
  {https://doi.org/10.1103/PhysRevLett.117.267002} {\bibfield  {journal}
  {\bibinfo  {journal} {Physical Review Letters}\ }\textbf {\bibinfo {volume}
  {117}},\ \bibinfo {pages} {267002} (\bibinfo {year} {2016})}\BibitemShut
  {NoStop}%
\bibitem [{\citenamefont {Hesselmann}\ and\ \citenamefont
  {Wessel}(2016)}]{Phys.Rev.B2016Hesselmann}%
  \BibitemOpen
  \bibfield  {author} {\bibinfo {author} {\bibfnamefont {S.}~\bibnamefont
  {Hesselmann}}\ and\ \bibinfo {author} {\bibfnamefont {S.}~\bibnamefont
  {Wessel}},\ }\bibfield  {title} {\bibinfo {title} {Thermal {{Ising}}
  transitions in the vicinity of two-dimensional quantum critical points},\
  }\href {https://doi.org/10.1103/PhysRevB.93.155157} {\bibfield  {journal}
  {\bibinfo  {journal} {Physical Review B}\ }\textbf {\bibinfo {volume} {93}},\
  \bibinfo {pages} {155157} (\bibinfo {year} {2016})}\BibitemShut {NoStop}%
\bibitem [{\citenamefont {Lu}\ and\ \citenamefont
  {Grover}(2019)}]{Phys.Rev.B2019Lu}%
  \BibitemOpen
  \bibfield  {author} {\bibinfo {author} {\bibfnamefont {T.-C.}\ \bibnamefont
  {Lu}}\ and\ \bibinfo {author} {\bibfnamefont {T.}~\bibnamefont {Grover}},\
  }\bibfield  {title} {\bibinfo {title} {Singularity in entanglement negativity
  across finite-temperature phase transitions},\ }\href
  {https://doi.org/10.1103/PhysRevB.99.075157} {\bibfield  {journal} {\bibinfo
  {journal} {Physical Review B}\ }\textbf {\bibinfo {volume} {99}},\ \bibinfo
  {pages} {075157} (\bibinfo {year} {2019})}\BibitemShut {NoStop}%
\bibitem [{\citenamefont {Lu}\ and\ \citenamefont
  {Grover}(2020)}]{Phys.Rev.Research2020Lu}%
  \BibitemOpen
  \bibfield  {author} {\bibinfo {author} {\bibfnamefont {T.-C.}\ \bibnamefont
  {Lu}}\ and\ \bibinfo {author} {\bibfnamefont {T.}~\bibnamefont {Grover}},\
  }\bibfield  {title} {\bibinfo {title} {Structure of quantum entanglement at a
  finite temperature critical point},\ }\href
  {https://doi.org/10.1103/PhysRevResearch.2.043345} {\bibfield  {journal}
  {\bibinfo  {journal} {Physical Review Research}\ }\textbf {\bibinfo {volume}
  {2}},\ \bibinfo {pages} {043345} (\bibinfo {year} {2020})}\BibitemShut
  {NoStop}%
\bibitem [{\citenamefont {Bonnes}\ \emph {et~al.}(2013)\citenamefont {Bonnes},
  \citenamefont {Pichler},\ and\ \citenamefont
  {L{\"a}uchli}}]{PRB2013L.Bonnes}%
  \BibitemOpen
  \bibfield  {author} {\bibinfo {author} {\bibfnamefont {L.}~\bibnamefont
  {Bonnes}}, \bibinfo {author} {\bibfnamefont {H.}~\bibnamefont {Pichler}},\
  and\ \bibinfo {author} {\bibfnamefont {A.~M.}\ \bibnamefont {L{\"a}uchli}},\
  }\bibfield  {title} {\bibinfo {title} {Entropy perspective on the thermal
  crossover in a fermionic {{Hubbard}} chain},\ }\href
  {https://doi.org/10.1103/PhysRevB.88.155103} {\bibfield  {journal} {\bibinfo
  {journal} {Physical Review B}\ }\textbf {\bibinfo {volume} {88}},\ \bibinfo
  {pages} {155103} (\bibinfo {year} {2013})}\BibitemShut {NoStop}%
\bibitem [{\citenamefont {Singh}\ \emph {et~al.}(2011)\citenamefont {Singh},
  \citenamefont {Hastings}, \citenamefont {Kallin},\ and\ \citenamefont
  {Melko}}]{PRL2011RRP.Singh}%
  \BibitemOpen
  \bibfield  {author} {\bibinfo {author} {\bibfnamefont {R.~R.~P.}\
  \bibnamefont {Singh}}, \bibinfo {author} {\bibfnamefont {M.~B.}\ \bibnamefont
  {Hastings}}, \bibinfo {author} {\bibfnamefont {A.~B.}\ \bibnamefont
  {Kallin}},\ and\ \bibinfo {author} {\bibfnamefont {R.~G.}\ \bibnamefont
  {Melko}},\ }\bibfield  {title} {\bibinfo {title} {Finite-temperature critical
  behavior of mutual information},\ }\href
  {https://doi.org/10.1103/PhysRevLett.106.135701} {\bibfield  {journal}
  {\bibinfo  {journal} {Physical Review Letters}\ }\textbf {\bibinfo {volume}
  {106}},\ \bibinfo {pages} {135701} (\bibinfo {year} {2011})}\BibitemShut
  {NoStop}%
\bibitem [{\citenamefont {Iaconis}\ \emph {et~al.}(2013)\citenamefont
  {Iaconis}, \citenamefont {Inglis}, \citenamefont {Kallin},\ and\
  \citenamefont {Melko}}]{PRB2013J.Iaconis}%
  \BibitemOpen
  \bibfield  {author} {\bibinfo {author} {\bibfnamefont {J.}~\bibnamefont
  {Iaconis}}, \bibinfo {author} {\bibfnamefont {S.}~\bibnamefont {Inglis}},
  \bibinfo {author} {\bibfnamefont {A.~B.}\ \bibnamefont {Kallin}},\ and\
  \bibinfo {author} {\bibfnamefont {R.~G.}\ \bibnamefont {Melko}},\ }\bibfield
  {title} {\bibinfo {title} {Detecting classical phase transitions with
  {{Renyi}} mutual information},\ }\href
  {https://doi.org/10.1103/PhysRevB.87.195134} {\bibfield  {journal} {\bibinfo
  {journal} {Physical Review B}\ }\textbf {\bibinfo {volume} {87}},\ \bibinfo
  {pages} {195134} (\bibinfo {year} {2013})}\BibitemShut {NoStop}%
\bibitem [{\citenamefont {Horn}\ and\ \citenamefont
  {Johnson}(1985)}]{Horn1985MatrixAnalysis}%
  \BibitemOpen
  \bibfield  {author} {\bibinfo {author} {\bibfnamefont {R.~A.}\ \bibnamefont
  {Horn}}\ and\ \bibinfo {author} {\bibfnamefont {C.~R.}\ \bibnamefont
  {Johnson}},\ }\href {https://doi.org/10.1017/CBO9780511810817} {\emph
  {\bibinfo {title} {Matrix {{Analysis}}}}}\ (\bibinfo  {publisher} {Cambridge
  University Press},\ \bibinfo {address} {Cambridge},\ \bibinfo {year}
  {1985})\BibitemShut {NoStop}%
\end{thebibliography}%



\end{document}

\title{Supplementary Information for ``Twisted R\'enyi Negativity as a Reliable Proxy for Mixed-State Entanglement in Fermionic Systems''}

\author{Fo-Hong Wang}
\affiliation{Tsung-Dao Lee Institute, Key Laboratory of Artificial Structures and Quantum Control (Ministry of Education), School of Physics and Astronomy, Shanghai Jiao Tong University, Shanghai 200240, China}

\author{Xiao Yan Xu}
\email{xiaoyanxu@sjtu.edu.cn}
\affiliation{Tsung-Dao Lee Institute, Key Laboratory of Artificial Structures and Quantum Control (Ministry of Education), School of Physics and Astronomy, Shanghai Jiao Tong University, Shanghai 200240, China}
\affiliation{Hefei National Laboratory, Hefei 230088, China}

\date{\today}

\maketitle

%%%%%%%%%%%%%%%%%%%%%%%%%%%%%%%%%%%%%%%%%%%%%%%%%%%%%
\section{Introduction}
%%%%%%%%%%%%%%%%%%%%%%%%%%%%%%%%%%%%%%%%%%%%%%%%%%%%%
Quantum entanglement in many-body systems aids in identifying exotic phases and quantum criticality~\cite{Phys.Rev.Lett.2003Vidal,Rev.Mod.Phys.2008Amico,PhysicsReports2016Laflorencie}. 
Among the various entanglement measures~\cite{QuantumInfo.Comput.2007Plenio,Rev.Mod.Phys.2009Horodecki}, the entanglement entropy (EE) has garnered significant interest over recent decades~\cite{J.Stat.Mech.Theor.Exp.2004Calabrese,Phys.Rev.Lett.2006Gioev,Phys.Rev.Lett.2006Kitaev,Phys.Rev.Lett.2006Levin,Phys.Rev.Lett.2006Fradkin,Phys.Rev.Lett.2008Wolf,J.Phys.AMath.Theor.2009Calabrese,Phys.Rev.Lett.2010Hastings,Phys.Rev.Lett.2013Grover,ArXiv2015Metlitski,Phys.Rev.Lett.2022Zhao,Phys.Rev.Lett.2024DEmidio,npjQuantumInf2025Liao}.
EE is a fundamental measure of quantum entanglement for pure states, but for large-scale many-body systems, Rényi entanglement entropy (REE) is much more computationally accessible. 
Although EE is effective for identifying bipartite entanglement for pure states, it is not a reliable measure for mixed-state entanglement due to its inability to distinguish quantum entanglement from classical correlations. 
Instead, the logarithmic negativity (LN)~\cite{Phys.Rev.A1998Zyczkowski,J.Mod.Opt.1999Eisert,Phys.Rev.A2002Vidal,Phys.Rev.Lett.2005Plenio}, designed based on positive partial transpose criterion for the inseparability of density matrices~\cite{Phys.Rev.Lett.1996Peres,PhysicsLettersA1996Horodecki}, can effectively detect quantum correlations in mixed states.  
The definition of LN hinges on the partial transpose of density matrices. 
For fermionic systems, additional considerations are necessary to respect their anticommuting statistical properties. 
A modified definition, known as the fermionic partial transpose (FPT), was proposed~\cite{Phys.Rev.B2017Shapourian,Phys.Rev.B2018Shiozaki,Phys.Rev.A2019Shapourian}.
There are two types of FPT, characterized by different boundary conditions in the spacetime picture but yielding the same LN: the \textit{untwisted} and \textit{twisted} FPTs~\cite{SciPostPhys.2019Shapourian}. 

Although being a computable entanglement measure without invoking any optimization, LN is still hard to compute in interacting many-body systems. 
For bosonic cases, certain tensor network techniques including tree tensor networks~\cite{J.Stat.Mech.2013Calabresea} and matrix product states/operators~\cite{Phys.Rev.B2016Ruggiero,J.Phys.AMath.Theor.2020Gruber} were utilized to represent the partially transposed density matrix (PTDM) and thus obtain LN. 
There also exists early study using traditional density matrix renormalization group~\cite{Phys.Rev.A2009Wichterich}. 
However, while the matrix product states/operators approach is universal, it is highly demanding in computational resources~\cite{Phys.Rev.B2016Ruggiero,J.Phys.AMath.Theor.2020Gruber,Phys.Rev.B2020Wybo}. 
The moments of PTDM (referred to as the R\'{e}nyi negativity (RN)), originally defined via the replica trick in the context of conformal field theory~\cite{Phys.Rev.Lett.2012Calabrese,J.Stat.Mech.2013Calabresea,J.Stat.Mech.2013Calabrese}, are easier to compute using various many-body algorithms, including tensor network~\cite{Phys.Rev.B2020Wybo} and Monte Carlo~\cite{J.Stat.Mech.2013Alba,Phys.Rev.B2014Chunga,Phys.Rev.Lett.2020Wu,Phys.Rev.B2025Ding}. 
There are also feasible experimental protocols for measuring RN in artificial quantum systems~\cite{Phys.Rev.Lett.2018Gray,Phys.Rev.A2019Cornfeld,Phys.Rev.Lett.2020Elben,npjQuantumInf2021Neven}. 
Although the RN has no direct relation to either entanglement or correlation, both equilibrium~\cite{J.Stat.Mech.2013Alba,Phys.Rev.B2014Chunga,Phys.Rev.Lett.2020Wu,Phys.Rev.B2025Ding} and non-equilibrium~\cite{J.Stat.Mech.2014Coser,Phys.Rev.B2020Wybo,Murciano2022QuenchDynamicsRenyi} studies indicate that its deviation from R\'{e}nyi entropy (moments of the full density matrices) of the same moment rank, called R\'{e}nyi negativity ratio (RNR), measures either entanglement or correlation. 

Computing the LN or RN in interacting fermionic  systems is even more difficult due to the appearance of extra phases beyond conventional partial transpose~\cite{Phys.Rev.B2017Shapourian,Phys.Rev.B2018Shiozaki,Phys.Rev.A2019Shapourian}. 
In particular, it is challenging to implement the FPT for matrix product states.
When the regions of interest comprise only a few points in total, one can first perform tomography of the reduced density matrix (e.g., using correlation functions to represent each of its elements~\cite{Phys.Rev.Res.2024Parez,Nat.Commun.2025Wang}), and then apply a partial transpose.
However, for large regions, the size of the reduced density matrices increases exponentially. 
For general system sizes, since the FPT of fermionic Gaussian states remains Gaussian, it enables the calculation of fermionic RN and LN through the Green's function or covariance matrix for free-fermion systems~\cite{Phys.Rev.B2017Shapourian,J.Stat.Mech.2019Shapourian,SciPostPhys.2023Alba,Phys.Rev.B2024Choi}. 
The two definitions of FPTs, namely, the untwisted and twisted FPTs, yield the same LN but different RNs in general.
Recently, in previous work~\cite{Nat.Commun.2025Wanga}, we demonstrated that the untwisted RN of interacting systems can be feasibly calculated within the framework of determinantal quantum Monte Carlo (DQMC), provided the models are sign-problem-free. 
We computed the rank-2 untwisted RN for two representative models in the realm of strongly-correlated systems, namely, the Hubbard model and the $t$-$V$ model.

For interacting fermionic systems with large system sizes, employing DQMC to compute exponential observables, such as REE and RN, encounters an inaccurate sampling issue due to the exponentially large invariance. 
Since Grover's pioneering work~\cite{Phys.Rev.Lett.2013Grover}, many efforts have been devoted to calculating the REE of interacting fermionic systems via DQMC. 
However, for large system sizes, Grover's original method becomes less accurate, prompting the development of several incremental algorithms to address this issue. 
There are two main approaches to designing incremental algorithms. 
One approach involves intermediate processes without physical meaning~\cite{Phys.Rev.Lett.2020DEmidio,npjQuantumMater.2022Zhao,Phys.Rev.Lett.2022Zhao,Phys.Rev.B2023Pan,npjQuantumInf2025Liao,ArXiv2023Liao,Phys.Rev.B2024Zhang,Phys.Rev.Lett.2024DEmidio}. 
The other approach is grounded in the concept of annealing, where each incremental process is assigned a physical interpretation, such as temperature~\cite{Stat.Comput.2001Neal,NewJ.Phys.2008Pollet,Phys.Rev.E2017Alba,JHEP2023Bulgarelli,Phys.Rev.B2025Dai,Phys.Rev.B2024Ding}. 
Fermionic RN also suffers from this inaccurate sampling problem as the system size increases. 
In previous work~\cite{Nat.Commun.2025Wanga}, we developed an incremental algorithm for rank-2 untwisted RN and unveiled its behavior across a finite-temperature phase transition point. 

In addition to the rank-2 untwisted RN examined in Ref.~\cite{Nat.Commun.2025Wanga}, higher-order moments of the PTDM, including both untwisted and twisted RNs of arbitrary ranks, warrant further investigation. 
Given that the untwisted PTDM is generally non-Hermitian and only the twisted RN can be analytically continued to LN by definition~\cite{SciPostPhys.2019Shapourian}, twisted RN is more analogous to the RN in bosonic systems~\cite{Phys.Rev.Lett.2012Calabrese,J.Stat.Mech.2013Calabresea,J.Stat.Mech.2013Calabrese,Phys.Rev.Lett.2020Wu}. 
Indeed, the rank-2 twisted RN is trivially the rank-2 R\'{e}nyi entropy~\cite{SciPostPhys.2019Shapourian}, making rank 3 the smallest non-trivial rank, analogous to the conventional bosonic RN~\cite{Phys.Rev.Lett.2020Wu}.
This highlights the need to explore the computation of high-rank RNs (i.e., ranks greater than 2) for fermionic systems, a largely unexplored area that is the focus of this paper. 

In this work, we address two primary challenges in computing high-rank RNs of Gibbs states using DQMC: (i) the numerical instability associated with inverting the partially transposed Green's function, and (ii) the exponentially large variance in Monte Carlo sampling.
We present a comprehensive technical analysis of the numerical instability and introduce stable update schemes for the incremental measurement of RNs.
These algorithms are applied to analyze the RNs in the Hubbard model and spinless $t$-$V$ model.
We demonstrate the distinctions between untwisted and twisted RNs, as well as the variations among RNs of different ranks.
Based on the numerical results, we discuss which version of RN, untwisted or twisted, is more appropriate as a R\'{e}nyi proxy for LN.
This work provides robust technical support for the stable QMC computation of exponential observables with analogous mathematical structures, including R\'{e}nyi entropy and REE.

\begin{figure}[t]
    \centering
    \includegraphics[width=\linewidth]{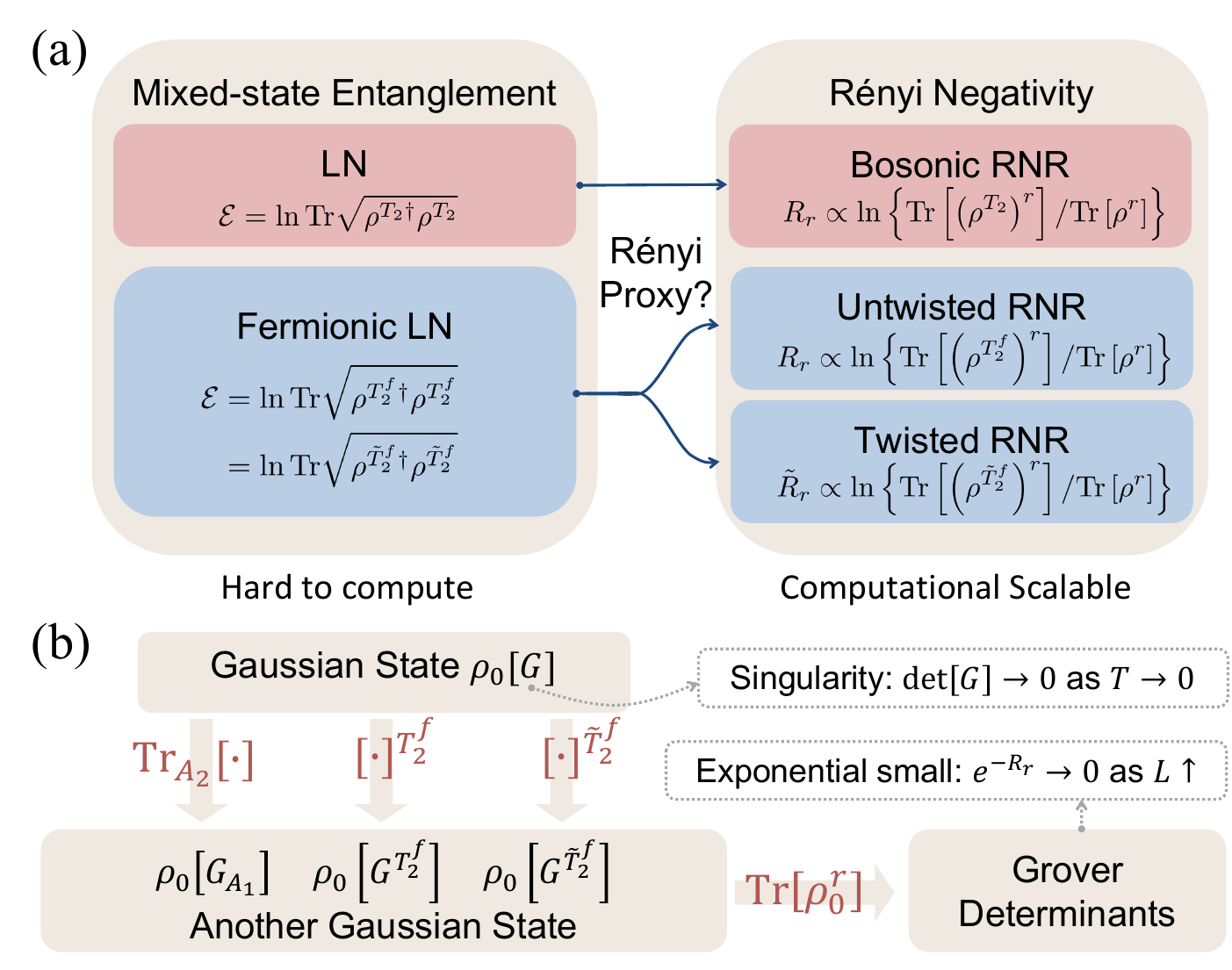}
    \caption{
        (a) Overview of various negativity definitions. While LN and its fermionic analog are computationally challenging, RNRs are scalable with QMC algorithms. This work seeks to identify a R\'{e}nyi proxy for LN in fermionic systems. 
        (b) Essential foundation underlying DQMC algorithms for REE and RNRs. 
        As reviewed in Sec.~\ref{sec:Prelimiary}, the operations $\mathrm{Tr}_{A_2}[\cdot]$ (partial trace), $[\cdot]^{T^{f}_{2}}$ (untwisted FPT), and $[\cdot]^{\tilde T^{f}_{2}}$ (twisted FPT) map a Gaussian state $\rho_{0}[G]$ to Gaussian states with modified Green's functions $G_{A_1}$, $G^{T^{f}_{2}}$, and $G^{\tilde T^{f}_{2}}$, respectively. 
        The Rényi traces $\mathrm{Tr}[\rho_{0}^{r}]$ then produce Grover determinants defined in Sec.~\ref{sec:Renyi_negativity_and_Grover_determinant}. 
        The dashed callouts mark two practical difficulties in computing RNRs, and we address them in Sec.~\ref{sec:numerical_stable} and Sec.~\ref{sec:incremental}. 
    }
    \label{fig:overview}
\end{figure}

The remainder of this paper is organized as follows. 
In Sec.~\ref{sec:Prelimiary}, we set the stage by reviewing the FPT of both general states and Gaussian states, subsequently transitioning to interacting fermionic states within the framework of DQMC. 
Various definitions of R\'{e}nyi negativities and their Monte Carlo estimators, termed as Grover determinants, are introduced in Sec.~\ref{sec:Renyi_negativity_and_Grover_determinant}.
Readers less versed in technical details may wish to bypass Secs.~\ref{sec:numerical_stable} and \ref{sec:incremental}. 
Sec.~\ref{sec:numerical_stable} addresses the first key challenge in calculating high-rank RNs using DQMC: the numerical instability arising from inverting partially transposed Green's functions. 
In Sec.~\ref{sec:incremental}, we tackle the second challenge, namely the inaccurate sampling issue, and present a fast update scheme sufficiently stable for incremental measurement of RNs. 
Sec.~\ref{sec:examples} exhibits and discusses the simulation results for the Hubbard chain and the spinless $t$-$V$ model. 
Finally, we encapsulate our work and provide a perspective on potential future avenues in Sec.~\ref{sec:conclusion}.

%%%%%%%%%%%%%%%%%%%%%%%%%%%%%%%%%%%%%%%%%%%%%%%%%%%%%
\section{Fermionic partial transpose: From Gaussian states to interacting fermions}\label{sec:Prelimiary}
%%%%%%%%%%%%%%%%%%%%%%%%%%%%%%%%%%%%%%%%%%%%%%%%%%%%%
In this section, we review the definition of FPT in the Fock space and how it acts on Gaussian states. 
Subsequently, leveraging the FPT of Gaussian states, we develop a formulation for the FPT of interacting fermionic Gibbs states within the framework of DQMC.

In this paper, we employ two types of fermionic operators. 
For systems that conserve particle number and do not include pairing terms in the Hamiltonian, the formalism is fully captured using the complex fermion basis, or Dirac basis. 
This is characterized by the creation and annihilation operators $c_j^\dagger, c_j$ ($j=1,2,\dots,N$, where $N$ denotes the number of modes,  such as the number of sites for spinless fermions), which adhere to the anti-commutation relations $\{c_j, c_{j'}^\dagger\} = \delta_{jj'}$. 
In more general scenarios, it is necessary to work in the Majorana basis\footnote{Alternatively, one can employ the Nambu vector $\Psi\equiv\{c_1,\dots,c_N,c^\dagger_1,\dots,c^\dagger_N\}$~\cite{Ripka1986QuantumTheoryFinite,deTorres-Solanot2017TimedependentGaussianVariational}.}  
with the Majorana operators defined as:
\begin{equation}\label{equ:MajoranaOp}
    \gamma_{j}^{\left(1\right)}=\gamma_{2j-1}=c_{j}+c_{j}^{\dagger},\gamma_{j}^{\left(2\right)}=\gamma_{2j}={\rm i}\left(c_{j}^{\dagger}-c_{j}\right).
\end{equation} 
Given that it is more convenient to present and derive certain formulas in the Majorana basis, we will primarily utilize this basis in this section.  
If necessary for potential applications, we will reduce some results to the Dirac basis. 

\subsection{Fermionic partial transpose}
We begin by revisiting the definition of FPT within the Fock space, where each basis state $|n\rangle$ is represented by a bit string $n=(n_1n_2\cdots n_N)$ that encodes the occupation numbers of each mode. 
Along the way, we establish several intriguing general trace properties of both the untwisted and twisted PTDM, some of which will be corroborated by numerical simulations in the subsequent sections.

Under a bipartite geometry ($A=A_1\cup A_2$), where the modes are divided into two parts, $A_1$ with $N_1$ modes and $A_2$ with $N_2$ modes ($N_1+N_2=N$), the partial transpose of a density matrix $\rho$ with respect to the subsystem $A_2$ is defined as~\cite{Phys.Rev.A2002Vidal}
\begin{equation}\label{equ:BPT}
    \left(\left|n_{A_{1}},n_{A_{2}}\right\rangle \left\langle \bar{n}_{A_{1}},\bar{n}_{A_{2}}\right|\right)^{T_{2}}=\left|n_{A_{1}},\bar{n}_{A_{2}}\right\rangle \left\langle \bar{n}_{A_{1}},n_{A_{2}}\right|,
\end{equation}
where $n_{A_b}$ denotes the (sub-)bit string of length $N_b$ comprising the occupation numbers of the modes within subsystem $A_b$ (with $b=1,2$ denoting distinct subsystem blocks).
A more suitable definition for fermions, which respects the fermionic anti-commutation statistics and preserves fermionic parity, was proposed~\cite{Phys.Rev.B2017Shapourian,J.Stat.Mech.2019Shapourian}:
\begin{equation}\label{equ:FPT}
    \begin{aligned}
        &\left(\left|n_{A_{1}},n_{A_{2}}\right\rangle \left\langle \bar{n}_{A_{1}},\bar{n}_{A_{2}}\right|\right)^{T_{2}^{f}}\\&=\left(-1\right)^{\phi\left(n,\overline{n}\right)}\left|n_{A_{1}},\bar{n}_{A_{2}}\right\rangle \left\langle \bar{n}_{A_{1}},n_{A_{2}}\right|,
    \end{aligned}
\end{equation}
where
\begin{equation}\label{equ:FPT_phase}
    {\phi\left(n,\overline{n}\right)}=\frac{\left[\left(\tau_{2}+\bar{\tau}_{2}\right)\bmod2\right]}{2}+\left(\tau_{1}+\bar{\tau}_{1}\right)\left(\tau_{2}+\bar{\tau}_{2}\right),
\end{equation}
with $\tau_b=\sum_{j\in A_b}n_j$ being the norm of the bit string $n_{A_b}$ ($b=1,2$). 
This is referred to as the \textit{untwisted} FPT~\cite{SciPostPhys.2019Shapourian}. 
In general, the untwisted FPT of a density matrix is pseudo-Hermitian. The untwisted PTDM satisfies $\rho^{T_{2}^{f}\dagger}=X_{2}\rho^{T_{2}^{f}}X_{2}$, where ${X}_{2}=(-1)^{\sum_{j\in A_2}c^\dagger_jc_j}$ is the disorder operator acting on the subsystem $A_2$. 
In contrast, the \textit{twisted} PTDM, defined as $\rho^{\tilde{T}_{2}^{f}}\equiv\rho^{T_{2}^{f}}X_{2}$, is Hermitian~\cite{SciPostPhys.2019Shapourian}.
In the Fock space, the twisted FPT is similarly given by Eq.~\eqref{equ:FPT} but the extra phase is modified to $\tilde{\phi}\left(n,\overline{n}\right)=\phi\left(n,\overline{n}\right)+\tau_{2}$. 

Utilizing the definitions in Eqs.~\eqref{equ:FPT} and \eqref{equ:FPT_phase}, we have proven some intriguing trace properties of the PTDM (see Appendix \ref{app:prood_in_Fock_space} for detailed proofs). 
Firstly, the moment of untwisted PTDM is expressed as
\begin{equation}\label{equ:moment_untwisted}
    \begin{aligned}
        &\mathrm{Tr}\left[\left(\rho^{T_{2}^{f}}\right)^{r}\right]=\sum_{n^{1},n^{2},\dots,n^{r}}\\&\left\langle n_{A_{1}}^{1},n_{A_{2}}^{3}\left|\rho\right|n^{2}\right\rangle \cdots\left\langle n_{A_{1}}^{i},n_{A_{2}}^{i+2}\left|\rho\right|n^{i+1}\right\rangle \cdots\\&\times\left\langle n_{A_{1}}^{r},n_{A_{2}}^{2}\left|\rho\right|n^{1}\right\rangle \\&\times\left(-1\right)^{\phi\left(n^{1},n^{2}\right)}\cdots\left(-1\right)^{\phi\left(n^{i},n^{i+1}\right)}\cdots\left(-1\right)^{\phi\left(n^{r},n^{1}\right)}.
    \end{aligned}
\end{equation}
In particular, for $r=1$, we find that $\mathrm{Tr}[\rho^{T_2^{^{f}}}	]=\mathrm{Tr}\lbrack\rho^{}\rbrack$, while for $r=2$, we have
\begin{equation}\label{equ:untwisted_moment_2}
    {\rm Tr}\left[\left(\rho^{T_{2}^{f}}\right)^{2}\right]={\rm Tr}\left(\rho{X}_{2}\rho{X}_{2}\right).
\end{equation}
The second relation concerns the twisted moments: 
\begin{equation}\label{equ:moment_twisted}
    \begin{aligned}
        &{\rm Tr}\left[\left(\rho^{\tilde{T}_{2}^{f}}\right)^{r}\right]={\rm Tr}\left[\left(\rho^{T_{2}^{f}}X_{2}\right)^{r}\right]=\sum_{n^{1},n^{2},\dots,n^{r}}\\&\left\langle n_{A_{1}}^{1},n_{A_{2}}^{3}\left|\rho\right|n^{2}\right\rangle \left\langle n_{A_{1}}^{2},n_{A_{2}}^{4}\left|\rho\right|n^{3}\right\rangle \cdots\left\langle n_{A_{1}}^{r},n_{A_{2}}^{2}\left|\rho\right|n^{1}\right\rangle \\&\times\left(-1\right)^{\tilde{\phi}\left(n^{1},n^{2}\right)}\left(-1\right)^{\tilde{\phi}\left(n^{2},n^{3}\right)}\cdots\left(-1\right)^{\tilde{\phi}\left(n^{r},n^{1}\right)}.
    \end{aligned}
\end{equation}
In particular, in the case of $r=2$ we have ${\rm Tr}[\rho^{T_{2}^{f}}\rho^{T_{2}^{f}\dagger}]=\mathrm{Tr}[\rho^{2}]$. 
Finally, for a density matrix with a well-defined parity (as required by the parity superselection rule~\cite{Phys.Rev.1952Wick,Phys.Rev.1967Aharonov}), the parity constrain on Eqs.~\eqref{equ:moment_untwisted} and \eqref{equ:moment_twisted} leads to further simplification. 
On the one hand, we prove that
\begin{equation}\label{equ:moment_4k_with_parity}
    {\rm Tr}\left[\left(\rho^{\tilde{T}_{2}^{f}}\right)^{4r}\right]={\rm Tr}\left[\left(\rho^{T_{2}^{f}}\right)^{4r}\right],r\in\mathbb{Z}.
\end{equation}
On the other hand, the phase chain in the third line of Eq.~\eqref{equ:moment_untwisted} can also be simplified to
\begin{equation}\label{equ:phase_chain_with_parity}
    \begin{aligned}
        &\left(-1\right)^{\phi\left(n^{1},n^{2}\right)}\cdots\left(-1\right)^{\phi\left(n^{i},n^{i+1}\right)}\cdots\left(-1\right)^{\phi\left(n^{r},n^{1}\right)}\\=&\left(-1\right)^{\sum_{i=1}^{r}\left(\tau_{2}^{i}\right)+\sum_{i=1}^{r}\left(\tau_{2}^{i}\tau_{2}^{\left(i+1\right)\text{ mod }r}\right)}.
    \end{aligned}
\end{equation}
These relations are applicable to arbitrary fermionic states and are not restricted to Gaussian states (the last two relations hold for states with well-defined parity).

\subsection{Gaussian states}

Free fermionic systems, characterized by quadratic Hamiltonians $\hat{H}_0$, are typically found in \textit{Gaussian states} when in equilibrium. 
These states include the thermal Gibbs states, expressed as $\rho_T={e^{-\beta \hat{H}_0}}/{\mathrm{Tr}[e^{-\beta \hat{H}_0}]}$, at finite temperature $T=1/\beta$. 
Additionally, the ground states can be derived by taking the zero-temperature limit, given by $\rho_{\mathrm{GS}}=\lim_{T\rightarrow0}\rho_T$\footnote{Indeed, every pure Gaussian state can be viewed as the ground state of some quadratic Hamiltonian. Furthermore, under an appropriate basis transformation, it can be reformulated into a BCS-like form~\cite{Kraus2009QuantumInformationPerspective}.}. 
Formally, the Gaussian states are defined as the exponential of quadratic forms: 
\begin{equation}\label{equ:Gaussian_state_Majorana}
    \rho_{0}=\frac{1}{Z}\exp\left(\frac{1}{4}\sum_{k,l}W_{kl}\gamma_{k}\gamma_{l}\right)=\frac{1}{Z}\exp\left(\frac{1}{4}\boldsymbol{\gamma}^{T}W\boldsymbol{\gamma}\right). 
\end{equation}
In this expression, the matrix $W$ defining the Gaussian state is antisymmetric due to the anti-commutation relations between Majorana operators. 
For Hermitian systems, $W$ is also purely imaginary and can be expressed as $W=\mathrm{i}W_0$ with $W_0$ a real antisymmetric matrix. 
Hereafter, we consider a general case where $W$ could be non-Hermitian. 
The normalization factor $Z$ can be determined in terms of the eigenvalues of the $W$ matrix and is given by $\det[\sqrt{I+e^W}]^{1/2}$ up to a sign ambiguity~\cite{J.Stat.Mech.2014Klich}. 
The (fermionic) partial transpose of pure Gaussian states can be expressed within the overlap matrix framework, enabling analytical calculations and highly efficient numerical evaluations of LN~\cite{Phys.Rev.B2016Chang,PRB2025JQ.Fang}.
However, for general mixed Gaussian states, the Green’s function method to be introduced is more applicable.

One fundamental attribute of Gaussian states is that the covariance matrix, defined as $\Gamma_{ij}=\frac{1}{2}\mathrm{Tr}(\rho_0[\gamma_i,\gamma_j])$, fully characterizes a Gaussian state. 
The covariance matrix $\Gamma$ is related to the matrix $W$ via
\begin{equation}\label{equ:Gamma_W}
    \Gamma=-\tanh\left(\frac{W}{2}\right)\text{ and }W=\ln\left[(I+\Gamma)^{-1}(I-\Gamma)\right].
\end{equation}
Ref.~\cite{Nat.Commun.2025Wanga} presents a proof utilizing the trace formula~\cite{J.Stat.Mech.2014Klich} which applies to cases where $Z=\sqrt{\det(I+e^{W})}$ without sign ambiguity, a condition typically met when $W$ is Hermitian. 
A more comprehensive proof employs the product form of Gaussian states~\cite{Quant.Inf.Comput.2005Bravyi,ArXiv2005Bravyi,J.Stat.Mech.2010Fagotti} (see Appendix \ref{app:Gamma-W-relation} for detailed exposition). 
All high-order correlations can be derived using Wick's theorem, indicating that the entirety of information pertaining to these Gaussian states is encapsulated within $\Gamma$. 
It is convenient to represent Gaussian states in a functional form as $\rho_0[\Gamma]$, where $\Gamma$ determines $W$ via Eq.~\eqref{equ:Gamma_W}, and subsequently, $W$ determines $\rho_0$ through Eq.~\eqref{equ:Gaussian_state_Majorana}. 

In addition to considering temperature, a mixed Gaussian state can also be obtained by tracing out certain degrees of freedom from a pure Gaussian state. 
The partial trace of a Gaussian state remains a Gaussian state~\cite{J.Phys.AMath.Gen.2003Peschel,Phys.Rev.B2004Cheong,J.Phys.AMath.Theor.2009Peschel}. 
Remarkably, the FPT of a Gaussian state is also a Gaussian state~\cite{Phys.Rev.B2017Shapourian}, and it can be efficiently calculated through the ``partial transpose'' of its covariance matrix. 
For instance, with respect to subsystem $A_2$, this is given by
\begin{equation}\label{equ:PT_Gamma}
    \Gamma^{T_{2}^{f}}=\left(\begin{array}{cc}
        \Gamma_{11} & \text{i}\Gamma_{12}\\
        \text{i}\Gamma_{21} & -\Gamma_{22}
        \end{array}\right),
\end{equation} 
where $\Gamma_{b,b'}$ denotes the matrix block of $\Gamma$, with rows corresponding to the $A_b$ region and columns corresponding to the $A_{b'}$ region. This notation will be consistently used throughout the paper.
Using the functional form of Gaussian states, the untwisted PTDM can be expressed as $\rho_0^{T_2^f}[\Gamma]=\rho_0[\Gamma^{T_2^f}]$. 
This fact was proved by utilizing Wick's theorem~\cite{NewJ.Phys.2015Eisler,Nat.Commun.2025Wanga}. 
The twisted PTDM, defined as $\rho^{\tilde{T}_2^f}_0\equiv\rho^{{T}_2^f}_0 X_2$, is also a Gaussian state~\cite{SciPostPhys.2019Shapourian}, owing to the product rule of Gaussian states \cite{J.Stat.Mech.2010Fagotti} (see Appendix~\ref{app:product_rule_Gaussian} for further details). 
By employing Eqs.~\eqref{equ:prodrule_W} and \eqref{equ:Gamma_W}, we can express the covariance matrix of $\rho_0^{\tilde{T}_2^f}$ in terms of the blocks of $\Gamma^{T_{2}^{f}}$~\cite{SciPostPhys.2019Shapourian}: 
\begin{equation}\label{equ:TPT_Gamma}
    \Gamma^{\tilde{T}_{2}^{f}}=\left(\begin{array}{cc}
        \Gamma_{11}-\Gamma_{12}\Gamma_{22}^{-1}\Gamma_{21} & \mathrm{i}\Gamma_{12}\Gamma_{22}^{-1}\\
        -\mathrm{i}\Gamma_{22}^{-1}\Gamma_{21} & -\Gamma_{22}^{-1}
        \end{array}\right).
\end{equation}
It is important to note that $\rho^{\tilde{T}_2^f}_0$ may be not normalized, but $\Gamma^{\tilde{T}_{2}^{f}}$ actually corresponds to the normalized Gaussian state $\rho^{\tilde{T}_2^{f}}_0/\mathcal{Z}_{\tilde{T}_2^{f}}$ with $\mathcal{Z}_{\tilde{T}_{2}^{f}}\equiv\mathrm{Tr}[\rho_0^{T_{2}^{f}}X_2]=\pm\sqrt{\det(\Gamma_{22})}$~\cite{SciPostPhys.2019Shapourian}. 

For particle-number-conserving fermionic systems, their Gaussian states can be expressed in terms of Dirac operators and are completely characterized by the Green's function $G_{ij}=\mathrm{Tr}(\rho_0 c_i c_j^\dagger)$, as given by~\cite{Phys.Rev.B2004Cheong}
\begin{equation}\label{equ:Gaussian_state_Dirac}
    \rho_{0}[G]=\det G\exp\left[\mathbf{c}^{\dagger}\ln\left(G^{-1}-I\right)\mathbf{c}\right]. 
\end{equation}
Its FPT is given by the same functional form $\rho_0[\cdot]$, but with the input of partially transposed Green's function, 
\begin{equation}\label{equ:untwistedPTGreen}
    G^{T_{2}^{f}}=\left(\begin{array}{cc}
        G_{11} & \text{i}G_{12}\\
        \text{i}G_{21} & I_{2}-G_{22}
        \end{array}\right), 
\end{equation}
which satisfies $\rho^{T_2^f}_0[G]=\rho_0[G^{T_2^f}]$. 
By applying the product rule formulated in Eq.~\eqref{equ:prodrule_Green}, we obtain $[(G^{\tilde{T}_{2}^{f}})^{-1}-I]=[(G^{T_{2}^{f}})^{-1}-I]U_{2}$ with $U_2\equiv I_1\oplus (-I_2)$, where $I_{b}$ is the diagonal block of the identity matrix $I$, corresponding to rows and columns in region $A_b$. Consequently, the twisted partially transposed Green's function is derived as follows:
\begin{equation}\label{equ:twistedPTGreen}
    \begin{aligned}
        G^{\tilde{T}_{2}^{f}}&=\left(\begin{array}{cc}
        I_{1} & 2G_{12}^{T_{2}^{f}}\\
        0 & 2G_{22}^{T_{2}^{f}}-I_{2}
        \end{array}\right)^{-1}G^{T_{2}^{f}}\\&=\left(\begin{array}{cc}
        I_{1} & -2\text{i}G_{12}\left(I_{2}-2G_{22}\right)^{-1}\\
        0 & \left(I_{2}-2G_{22}\right)^{-1}
        \end{array}\right)G^{T_{2}^{f}},
        \end{aligned}
\end{equation}
which is related to the twisted PTDM via $\mathrm{Tr}(\rho_0^{\tilde{T}_2^f}c_i c^\dagger_j)=\mathcal{Z}_{\tilde{T}_2^f}G^{\tilde{T}_2^f}_{ij}$ with ${\mathcal{Z}}_{\tilde{T}_{2}^{f}}=\det\left(I_{2}-2G_{22}\right)$. 
The determinants of the untwisted and twisted Green's functions are directly related by $\det[\mathcal{Z}_{\tilde{T}_2^f}G^{\tilde{T}_2^f}]=\det[G^{T_2^f}]$. 
We note that the Gaussian-preserving property of untwisted and twisted FPTs is also shared by partial trace, as illustrated collectively in Fig.~\ref{fig:overview}(b). 

\subsection{Interacting fermions}
Before we start, we emphasize that the use of a hat, such as in $\hat{H}$, denotes an operator, while the same symbol without a hat represents its corresponding matrix representation. 
In the Majorana basis, the operator is expressed as $\hat{H} \equiv \frac{1}{4}\boldsymbol{\gamma}^\dagger H \boldsymbol{\gamma}$, whereas in the Dirac basis, it is given by $\hat{H} \equiv \mathbf{c}^\dagger H \mathbf{c}$. 
Note that the matrix representations of the same operator differ in different bases, even though the same symbols are used.

The Gibbs states of interacting fermionic systems with a Hamiltonian $\hat{H}=\hat{H}_0+\hat{H}_I$, where $\hat{H}_I$ includes a quartic term, are inherently non-Gaussian. 
However, by successively employing the Trotter-Suzuki decomposition and Hubbard-Stratonovich (HS) transformations, we can systematically decouple the interaction terms into quadratic forms. 
This process leads to the expression:
$e^{-\Delta_\tau \hat{H}}=\sum_{s}\alpha[s]e^{\hat{K}}e^{\hat{V}[s]} e^{\hat{K}}+O(\Delta_\tau^3)$ 
where $\hat{K}\equiv-\Delta_\tau \hat{H}_0/2$ simply incorporates the prefactor, while the quadratic interaction term $\hat{V}[s]$ along with the coupling coefficient $\alpha[s]$ depend on the HS auxiliary fields $s$. 
For an interacting Gibbs state at inverse temperature $\beta=L_\tau \Delta_\tau$, the partition function is given by
\begin{equation}\label{equ:Z_DQMC}
    Z=\mathrm{Tr}\left[e^{-\beta\hat{H}}\right]=\sum_{\left\{ s_{l}\right\} }\mathrm{Tr}\left[\prod_{l=1}^{L_{\tau}}\left(\alpha\left[s_{l}\right]e^{\hat{K}_{l}}e^{\hat{V}\left[s_{l}\right]}e^{\hat{K}_{l}}\right)\right].
\end{equation}
In what follows, we adopt a shorthand notation to represent the complete set of all spacetime-dependent auxiliary fields, $\mathbf{s}\equiv\{s_l| l=1,\dots,L_\tau\}$.
The partition function is thus a summation over the weights of all possible $\mathbf{s}$-configurations, given by $Z=\sum_{\mathbf{s}}w_{\mathbf{s}}$.

The density matrix $\rho=e^{-\beta \hat{H}}/Z$ can be expressed as a sum of Gaussian states by utilizing the product rules for Gaussian states (refer to Appendix~\ref{app:product_rule_Gaussian} for details). 
Indeed, for each $\mathbf{s}$-configuration, the operators within the square bracket in the final expression of Eq.~\eqref{equ:Z_DQMC} combine to form a single Gaussian state, represented as $e^{\hat{W}_\mathbf{s}}=e^{\frac14\boldsymbol{\gamma}^TW_\mathbf{s}\boldsymbol{\gamma}}\equiv \prod_l e^{\hat{h}_l[\mathbf{s}]}$. 
Each $e^{\hat{h}_l[\mathbf{s}]}$ is further decomposed into a product of three Gaussian states, $e^{\hat{h}_l[\mathbf{s}]}\equiv e^{\hat{K}_l}e^{\hat{V}_l[\mathbf{s}]}e^{\hat{K}_l}$. 
Incorporating normalization factors, the density matrix can be expressed as
\begin{equation}\label{equ:rho_MQMC}
    \rho=\sum_{{\bf s}}p_{{\bf s}}\rho_{{\bf s}}\text{ with } \rho_{\bf{s}}=\det\left[I+e^{W_\mathbf{s}}\right]^{-1/2}e^{\frac{1}{4}{\boldsymbol\gamma}^{T}W_\mathbf{s}{\boldsymbol\gamma}}, 
\end{equation}
where $p_\mathbf{s}=w_\mathbf{s}/Z$ is the probability of $\mathbf{s}$-configuration. 
Due to the linearity of the FPT, we can partially transpose the Gaussian states corresponding to different $\mathbf{s}$-configurations independently, and then sum them to obtain
\begin{equation}
    \rho^{T_{2}^{f}}=\sum_{\mathbf{s}}p_{\mathbf{s}}\rho_{\mathbf{s}}^{T_{2}^{f}}=\sum_{\mathbf{s}}p_{\mathbf{s}}\rho_{0}\left[\Gamma_{\mathbf{s}}^{T_{2}^{f}}\right],
\end{equation}
where $\Gamma_\mathbf{s}=\frac12\mathrm{Tr}(\rho_\mathbf{s}[\gamma_i,\gamma_j])$ represents the covariance matrix of the Gaussian state for configuration $\mathbf{s}$, and $\Gamma^{T_2^f}_\mathbf{s}$ is obtained from Eq.~\eqref{equ:PT_Gamma}, with explicit dependence on the auxiliary fields $\mathbf{s}$. 
Finally, the twisted PTDM is then given by
\begin{equation}
    \rho^{\tilde{T}_{2}^{f}}=\sum_{\mathbf{s}}p_{\mathbf{s}}\rho_{\mathbf{s}}^{T_{2}^{f}}X_{2}=\sum_{\mathbf{s}}p_{\mathbf{s}}\mathcal{Z}_{\tilde{T}_{2}^{f},\mathbf{s}}\rho_{0}\left[\Gamma^{\tilde{T}_{2}^{f}}_\mathbf{s}\right], 
\end{equation}
where $\mathcal{Z}_{\tilde{T}_{2}^{f},\mathbf{s}}=\pm\sqrt{\det\left(\Gamma_{\mathbf{s},22}\right)}$ is a normalization factor and $\Gamma^{\tilde{T}_{2}^{f}}_\mathbf{s}$ is obtained from Eq.~\eqref{equ:TPT_Gamma}, again explicitly dependent on the auxiliary fields $\mathbf{s}$. 

If the weight in Eq.~\eqref{equ:rho_MQMC} is consistently real and positive, the summation can be efficiently evaluated using Monte Carlo sampling over the configuration space of $\mathbf{s}$. 
For simplicity, we have disregarded the sign ambiguity of the normalization factor in Eq.~\eqref{equ:rho_MQMC}. 
When the sign is determined by symmetry arguments, Majorana QMC~\cite{Phys.Rev.B2015Li} can be employed for calculations. 
In cases where sign problems arise, one may turn to Pfaffian QMC~\cite{ArXiv2024Han} as an alternative approach.
By focusing exclusively on particle-number-conserving systems, we can express the equations in terms of Dirac operators. 
Specifically, we have $e^{\hat{W}_\mathbf{s}}=e^{\mathbf{c}^\dagger W_\mathbf{s} \mathbf{c}}\equiv\prod_l e^{\hat{h}_l[\mathbf{s}]}$, which transforms Eq.~\eqref{equ:rho_MQMC} into the following form~\cite{Phys.Rev.Lett.2013Grover}:
\begin{equation}\label{equ:rho_DQMC}
    \rho=\sum_{\mathbf{s}}p_{\mathbf{s}}\rho_{\mathbf{s}}\text{ with } \rho_{\mathbf{s}}=\det\left[G_{\mathbf{s}}\right]e^{\mathbf{c}^{\dagger}\ln\left[G_{{\bf s}}^{-1}-I\right]\mathbf{c}}, 
\end{equation}
which aligns with the framework of the celebrated DQMC~\cite{Phys.Rev.D1981Blankenbecler,Phys.Rev.B1981Scalapino,Phys.Rev.B1985Hirsch,Assaad2008WorldlineDeterminantalQuantum}. 
After performing the untwisted and twisted FPTs, we derive the following expressions:
\begin{equation}\label{equ:rho_untwisted_DQMC}
    \rho^{T_{2}^{f}}=\sum_{\mathbf{s}}p_{\mathbf{s}}\rho_{\mathbf{s}}^{T_{2}^{f}}=\sum_{\mathbf{s}}p_{\mathbf{s}}\rho_{0}\left[G_{\mathbf{s}}^{T_{2}^{f}}\right],
\end{equation}
and
\begin{equation}\label{equ:rho_twisted_DQMC}
    \rho^{\tilde{T}_{2}^{f}}=\sum_{\mathbf{s}}p_{\mathbf{s}}\rho_{\mathbf{s}}^{T_{2}^{f}}X_{2}=\sum_{\mathbf{s}}p_{\mathbf{s}}\mathcal{Z}_{\tilde{T}_{2}^{f},\mathbf{s}}\rho_{0}\left[G_{\mathbf{s}}^{\tilde{T}_{2}^{f}}\right],
\end{equation}
where $\mathcal{Z}_{\tilde{T}_{2}^{f},\mathbf{s}}=\det\left(I_{2}-2G_{\mathbf{s},22}\right)$, and the matrices $G_{\mathbf{s}}^{{T}_{2}^{f}}$ and $G_{\mathbf{s}}^{\tilde{T}_{2}^{f}}$ are given by Eqs.~\eqref{equ:untwistedPTGreen} and \eqref{equ:twistedPTGreen}, respectively.

The Green's function $G_\mathbf{s}$ appearing above is given by $G_{\mathbf{s},ij}\equiv\mathrm{Tr}[\rho_{\mathbf{s}}c_{i}c_{j}^{\dagger}]=(I+e^{W_{\mathbf{s}}})^{-1}$. 
We can extend it to a more general construct, specifically the equal-time Green's function for the total system governed by $\hat{H}$.
This function is defined as $G_{ij}(\tau)\equiv\mathrm{Tr}[e^{-\beta \hat{H}}c_i(\tau)c^\dagger_j(\tau)]$, evaluated at a particular imaginary time $\tau\in[0,\beta]$. 
For a Hermitian $\hat{H}$, the value of $G$ remains invariant across different $\tau$. 
However, since $\hat{W}_\mathbf{s}$ is typically non-Hermitian, translation symmetry along the imaginary-time direction is not preserved for a specific $\mathbf{s}$. 
We define intermediate imaginary-time evolution operators by decomposing $e^{\hat{W}_\mathbf{s}}$ as follows:
\begin{equation}\label{equ:DQMC_B}
    e^{\hat{W}_{{\bf s}}}\equiv\prod_{l=1}^{L_{\tau}}e^{\hat{h}_{l}\left[\mathbf{s}\right]}\equiv\hat{B}_{{\bf s}}\left(\beta,0\right)=\hat{B}_{{\bf s}}\left(\beta,\tau\right)\hat{B}_{{\bf s}}\left(\tau,0\right),
\end{equation}
where $\hat{B}_\mathbf{s}\left(\tau_{1}=l_{1}\Delta_{\tau},\tau_{2}=l_{2}\Delta_{\tau}\right)=\prod_{l=l_{1}}^{l_{2}+1}e^{\hat{h}_{l}\left[\mathbf{s}\right]}$ can be interpreted as un-normalized intermediate Gaussian states. 
Consequently, we have~\cite{Assaad2008WorldlineDeterminantalQuantum}
\begin{equation}\label{equ:DQMC_equaltimeG}
    \begin{aligned}
    G_{\mathbf{s},ij}\left(\tau,\tau\right)&=\frac{\mathrm{Tr}\left[\hat{B}_{{\bf s}}\left(\beta,\tau\right)c_{i}c_{j}^{\dagger}\hat{B}_{{\bf s}}\left(\tau,0\right)\right]}{\mathrm{Tr}\left[\hat{B}_{{\bf s}}\left(\beta,0\right)\right]}
    \\&=\left(I+B_{{\bf s}}\left(\tau,0\right)B_{{\bf s}}\left(\beta,\tau\right)\right)^{-1},
    \end{aligned}
\end{equation}
where $B_\mathbf{s}\left(\tau_{1}=l_{1}\Delta_{\tau},\tau_{2}=l_{2}\Delta_{\tau}\right)=\prod_{l=l_{1}}^{l_{2}+1}e^{K_{l}}e^{V_{l}\left[\mathbf{s}\right]}e^{K_{l}}$. 
The imaginary-time translation symmetry is restored after summing over all $\mathbf{s}$-configurations. 
By averaging these $\mathbf{s}$-dependent Green's functions over different $\tau$, we obtain the overall Green's function $G=\frac{1}{L_\tau}\sum_{\mathbf{s},\tau}p_{\mathbf{s}}G_{\mathbf{s}}(\tau,\tau)$. 
All equal-time measurements follow a similar sampling process, which also applies to the RN introduced in the next section. 

%%%%%%%%%%%%%%%%%%%%%%%%%%%%%%%%%%%%%%%%%%%%%%%%%%%%%
\section{R\'{e}nyi negativity (ratio), Grover matrix and Grover determinant}\label{sec:Renyi_negativity_and_Grover_determinant}
%%%%%%%%%%%%%%%%%%%%%%%%%%%%%%%%%%%%%%%%%%%%%%%%%%%%%
In this section, we introduce variations of the LN, namely, the untwisted and twisted RNs, and next section will discuss methods for stably calculating their sampled value for a specific $\mathbf{s}$-configuration. 

Based on the bosonic and fermionic partial transposes (see Eqs.~\eqref{equ:BPT} and \eqref{equ:FPT}), the LN is defined as
\begin{equation}
    \mathcal{E} = \ln\mathrm{Tr}\sqrt{\rho^{T_2^{(f)}\dagger}\rho^{T_2^{(f)}}}.
\end{equation}
This quantity has been shown to be an entanglement monotone for both bosonic~\cite{Phys.Rev.A2002Vidal} and fermionic systems~\cite{Phys.Rev.A2019Shapourian}. 
It is noteworthy that since $\rho^{T_{2}^f\dagger}\rho^{T_{2}^f}=\rho^{\tilde{T}_{2}^f\dagger}\rho^{\tilde{T}_{2}^f}=(\rho^{\tilde{T}_{2}^f})^{2}$, both untwisted and twisted FPTs yield the same LN. 
However, due to challenges associated with evaluating the trace norm in both analytical and numerical contexts, some researchers have turned to a variation of LN, known as the RN or negativity moments~\cite{Phys.Rev.Lett.2012Calabrese,J.Stat.Mech.2013Calabrese,SciPostPhys.2019Shapourian,Phys.Rev.Lett.2020Wu}. 
The untwisted and twisted RN are defined as follows:
\begin{equation}\label{equ:def_Renyinegativity}
    \begin{aligned}
        \mathcal{E}_r&=\frac{1}{1-r}\ln\mathrm{Tr}\left[\left(\rho^{T_2^f}\right)^r\right]\quad\text{and}\\\tilde{\mathcal{E}}_r&=\frac{1}{1-r}\ln\mathrm{Tr}\left[\left(\rho^{\tilde{T}_2^f}\right)^r\right],
    \end{aligned}
\end{equation}
respectively. 
Here, we define the RNs in a way analogous to the R\'{e}nyi entropy, incorporating the prefactor $1/(1-r)$ to ensure a positive value. 
The twisted RN is analytically continued to the LN via $\mathcal{E}=\lim_{r\rightarrow 1/2}(1-2r)\tilde{\mathcal{E}}_{2r}$~\cite{SciPostPhys.2019Shapourian}. 
It is important to note that untwisted and twisted FPTs of the same rank can yield different values of RN. 
Specifically, while $\mathcal{E}_2=-\ln\mathrm{Tr}[(\rho X_2)^2]$ using Eq.~\eqref{equ:untwisted_moment_2}, $\tilde{\mathcal{E}}_2=-\ln\mathrm{Tr}[\rho^2]$ is trivially related to the R\'{e}nyi entropy. 
Therefore, there is generally a greater interest in high-rank RN, which is the primary focus of this work. 
As discussed in the introduction, the deviation of RN from the R\'{e}nyi entropy of the same rank may provide meaningful insights related to either entanglement or correlation. 
We define this deviation as the R\'{e}nyi negativity ratio (RNR):
\begin{equation}\label{equ:def_RNR}
    R_{r}=\frac{1}{1-r}\ln\frac{\text{Tr}\left[\left(\rho^{T_{2}^{f}}\right)^{r}\right]}{\text{Tr}\left[\left(\rho\right)^{r}\right]}=\mathcal{E}_{r}-S_{r}^{\mathrm{th}},
\end{equation}
where the thermal R\'{e}nyi entropy is denoted by $S_{r}^{\mathrm{th}}=\ln\mathrm{Tr}\left[\rho^{r}\right]/\left(1-r\right)$. 
The twisted RNR is similarly defined as $\tilde{R}_{r}=\tilde{\mathcal{E}}_{r}-S_{r}^{\mathrm{th}}$. 

By substituting the decomposition forms from Eqs.~\eqref{equ:rho_untwisted_DQMC} and \eqref{equ:rho_twisted_DQMC} into the definition of rank-$r$ RN in Eq.~\eqref{equ:def_Renyinegativity}, we obtain
\begin{subequations}
\begin{equation}\label{equ:expRenyiNk_untwisted_DQMC}
    e^{\left(1-r\right)\mathcal{E}_{r}}=\text{Tr}\left[\left(\rho^{T_{2}^{f}}\right)^{r}\right]=\sum_{\mathbf{s}^{(1)}\cdots\mathbf{s}^{(r)}}\left(\prod_{i=1}^{r}p_{\mathbf{s}^{(i)}}\right)\det g_{r,\bar{\boldsymbol{s}}},
\end{equation}
and
\begin{equation}\label{equ:expRenyiNk_twisted_DQMC}
    \begin{aligned}
        e^{\left(1-r\right)\mathcal{\tilde{E}}_{r}}&=\text{Tr}\left[\left(\rho^{\tilde{T}_{2}^{f}}\right)^{r}\right]\\&=\sum_{\mathbf{s}^{(1)}\cdots\mathbf{s}^{(r)}}\left(\prod_{i=1}^{r}p_{\mathbf{s}^{(i)}}\right)\left(\prod_{i=1}^{r}\mathcal{Z}_{\tilde{T}_{2}^{f},\mathbf{s}^{(i)}}\right)\det\tilde{g}_{r,\bar{\boldsymbol{s}}},
        \end{aligned}
\end{equation}
\end{subequations}
where ${\mathcal{Z}}_{\tilde{T}_{2}^{f},\mathbf{s}}=\det\left(I_{2}-2G_{\mathbf{s},22}\right)$, and $\bar{\boldsymbol{s}}$ is a shorthand notation for the replica collection $\mathbf{s}^{(1)}\cdots\mathbf{s}^{(r)}$. 
The rank-$r$ \textit{Grover matrices} are defined as
\begin{subequations}\label{equ:gk_DQMC}
    \begin{equation}\label{equ:gk_untwisted_DQMC}
        g_{r,\bar{\boldsymbol{s}}}\equiv\left[\prod_{i=r}^{1}G_{\mathbf{s}^{(i)}}^{T_{2}^{f}}\right]\left[I+\prod_{i=r}^{1}\left(\left(G_{\mathbf{s}^{(i)}}^{T_{2}^{f}}\right)^{-1}-I\right)\right],
    \end{equation}
    and
    \begin{equation}\label{equ:gk_twisted_DQMC}
        \tilde{g}_{r,\bar{\boldsymbol{s}}}\equiv\left[\prod_{i=r}^{1}G_{\mathbf{s}^{(i)}}^{\tilde{T}_{2}^{f}}\right]\left[I+\prod_{i=r}^{1}\left(\left(G_{\mathbf{s}^{(i)}}^{\tilde{T}_{2}^{f}}\right)^{-1}-I\right)\right].
    \end{equation}
\end{subequations}
The untwisted and twisted Green's functions, $G^{T_2^f}_\mathbf{s}$ and $G^{\tilde{T}_2^f}_\mathbf{s}$, are provided in Eqs.~\eqref{equ:untwistedPTGreen} and \eqref{equ:twistedPTGreen}, respectively. 
The determinants of Grover matrices are referred to as the \textit{Grover determinants}. The term ``Grover'' is used because these formulas are analogous to Grover's formulas for calculating REE~\cite{Phys.Rev.Lett.2013Grover}.
The twisted Grover determinant can be combined with the normalization factors via $\det[\mathcal{Z}_{\tilde{T}_2^f}G^{\tilde{T}_2^f}]=\det[G^{T_2^f}]$ (see the discussion below Eq.~\eqref{equ:twistedPTGreen}). 
This allows the calculation of twisted RN using untwisted Green's function, as given by the formula:
\begin{equation}\label{equ:expRenyiNk_twisted_DQMC2}
    e^{\left(1-r\right)\mathcal{\tilde{E}}_{r}}=\sum_{\mathbf{s}^{(1)}\cdots\mathbf{s}^{(r)}}\left(\prod_{i=1}^{r}p_{\mathbf{s}^{(i)}}\right)\det\underline{\tilde{g}}_{r,\bar{\boldsymbol{s}}}, 
\end{equation}
where a modified Grover matrix is defined as:
\begin{equation}\label{equ:gkbar_twisted_DQMC}
    \underline{\tilde{g}}_{r,\bar{\boldsymbol{s}}}\equiv\left[\prod_{i=r}^{1}G_{\mathbf{s}^{(i)}}^{T_{2}^{f}}\right]\left[I+\prod_{i=r}^{1}\left[\left(\left(G_{\mathbf{s}^{(i)}}^{T_{2}^{f}}\right)^{-1}-I\right)U_{2}\right]\right],
\end{equation}
which satisfies $\det\underline{\tilde{g}}_{r,\bar{\boldsymbol{s}}}=\left(\prod_{i=1}^{r}\mathcal{Z}_{\tilde{T}_{2}^{f},\mathbf{s}^{(i)}}\right)\det\tilde{g}_{r,\bar{\boldsymbol{s}}}$. 
Eq.~\eqref{equ:expRenyiNk_twisted_DQMC2} is advantageous over Eq.~\eqref{equ:expRenyiNk_twisted_DQMC} when the matrix block $(I_2-2G_{22})$ in Eq.~\eqref{equ:twistedPTGreen} becomes singular. 
This singularity leads to a nearly zero value for $\mathcal{Z}_{\tilde{T}_2^f}$ and results in an unstable $G^{\tilde{T}_2^f}$. 

%%%%%%%%%%%%%%%%%%%%%%%%%%%%%%%%%%%%%%%%%%%%%%%%%%%%%
\section{Numerically stable calculation of rank-$r$ Grover determinants}\label{sec:numerical_stable}
%%%%%%%%%%%%%%%%%%%%%%%%%%%%%%%%%%%%%%%%%%%%%%%%%%%%%
In this section, we address the first challenge in computing high-rank RNs via DQMC: the instability arising from inverting partially transposed Green's functions, which appears necessary in Eqs.~\eqref{equ:gk_DQMC} and \eqref{equ:gkbar_twisted_DQMC}. 

The direct computation of high-rank Grover determinants is plagued by numerical instability due to the singular nature of the partially transposed Green's functions, which itself stems from the singularity of the Green's function $G$. 
For free-fermion systems, where the HS transformation and Monte Carlo sampling are unnecessary, we have $G_{\mathbf{s}^{(i)}}^{T_{2}^{f}/\tilde{T}_2^f}=G^{T_{2}^{f}/\tilde{T}_2^f}$ for all $i=1,2,\dots,r$. Therefore, stable formulas for RNs can be derived without invoking the inversion of Green's functions:
\begin{subequations}\label{equ:expRenyiNn_free_fermion}
    \begin{equation}\label{equ:expRenyiNn_untwisted_free_fermion}
        e^{\left(1-r\right)\mathcal{E}_{r}}=\det\left[\left(G^{T_{2}^{f}}\right)^{r}+\left(I-G^{T_{2}^{f}}\right)^{r}\right],
    \end{equation}
    \begin{equation}\label{equ:expRenyiNn_twisted_free_fermion}
        e^{\left(1-r\right)\mathcal{\tilde{E}}_{r}}=\mathcal{Z}_{\tilde{T}_{2}^{f}}^{r}\det\left[\left(G^{\tilde{T}_{2}^{f}}\right)^{r}+\left(I-G^{\tilde{T}_{2}^{f}}\right)^{r}\right].
    \end{equation}
\end{subequations}
In these expressions, one can diagonalize the matrices $G^{T_2^f}$ and $G^{\tilde{T}_2^f}$, allowing the untwisted and twisted RNs to be computed from their eigenvalues, respectively. 
In interacting systems, rank-2 Grover determinants maintain stability. 
For instance, the untwisted Grover determinant is expressed as $\det g_{2,\bar{\boldsymbol{s}}}=\det[G_{\mathbf{s}^{(1)}}^{T_{2}^{f}}G_{\mathbf{s}^{(2)}}^{T_{2}^{f}}+(I-G_{\mathbf{s}^{(1)}}^{T_{2}^{f}})(I-G_{\mathbf{s}^{(2)}}^{T_{2}^{f}})]$. 
However, as the rank increases beyond 2, it becomes unavoidable to deal with the inverse of the partially transposed Green's functions, leading to potential numerical difficulties. 

\begin{figure}
    \center
    \includegraphics[width=0.35\textwidth]{./fig1.pdf}
    \caption{
        The variations of untwisted and twisted rank-4 RNs of a free-fermion chain with $\mu/t=0.5$ are shown as functions of subsystem length $L_{A_1}$ for two temperatures (finite-temperature and ground-state).  
        The chain length is $L_A=50$ and the inset illustrates the bipartite geometry. 
        All the points are calculated using Eq.~\eqref{equ:expRenyiNn_free_fermion}, except for the ground-state $\tilde{\mathcal{E}}_4$ at $L_{A_1}=L_A/2$. 
        At this specific point, the matrix block $(I_2-2G_{22})$ in Eq.~\eqref{equ:twistedPTGreen} becomes singular and we utilize Eq.~\eqref{equ:detGroverbar_Drut_method} instead. 
    }
    \label{fig:free_fermion_standard}
\end{figure}

Given that Eq.~\eqref{equ:expRenyiNn_free_fermion} provides standard results suitable for benchmarking, we demonstrate the issue of numerical instability of Grover determinants in the context of free fermions. 
The Hamiltonian under consideration, $\hat{H}_0=-t\sum_{i=1}^{N}(c^\dagger_ic_{i+1}+\mathrm{H.c.})-\mu\sum_{i}c^\dagger_ic_i$, describes a simple tight-binding periodic chain. 
Utilizing Eq.~\eqref{equ:expRenyiNn_free_fermion}, we calculate the rank-4 untwisted and twisted RNs as functions of subsystem length $L_{A_1}$, as depicted in Fig.~\ref{fig:free_fermion_standard}. 
At finite temperatures, the untwisted negativity differs from the twisted negativity, whereas they coincide when considering ground state averages. 
These observations align with our theoretical proof of Eq.~\eqref{equ:moment_4k_with_parity}. 
At finite temperatures, Fock states with varying total particle numbers mix to form a Gibbs state, which lacks a well-defined parity. 
In contrast, the ground state, along with all other eigenstates of $\hat{H}_0$, possesses a well-defined parity due to the parity symmetry preserved by the Hamiltonian. 

In the following discussion, we seek stable approaches to compute the Grover determinants without relying on Eq.~\eqref{equ:expRenyiNn_free_fermion}, as this formula is impractical for interacting systems, where the matrices $G^{T_2^f/\tilde{T}_2^f}_{\mathbf{s}^{(i)}}$ for different $i$ do not commute. 
We will begin by demonstrating the singularity of the partially transposed Green's function matrices and then introduce two solutions to stabilize the computation. 

% Fig 2 is a two-column figure
\begin{figure*}[ht]
    \center
    \includegraphics[width=\textwidth]{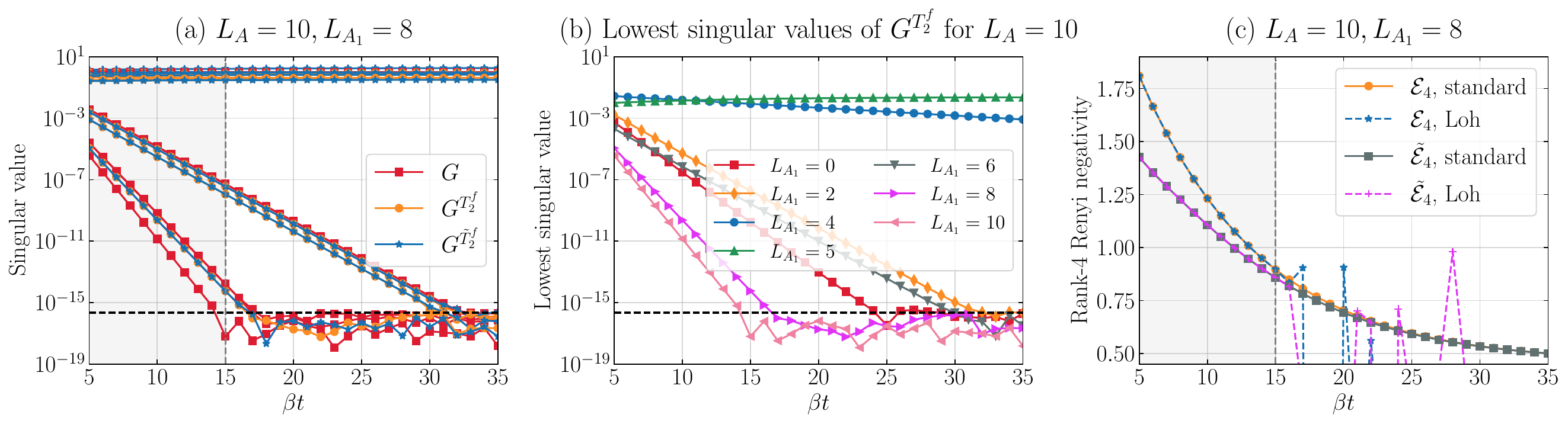}
    \caption{
        The numerical instability of high-rank Grover determinants and its stabilization in the high-temperature regime, demonstrated in the context of a $L_A=10$ free-fermion chain with $\mu/t=0.5$. 
        The bipartition geometry is consistent with that shown in Fig.~\ref{fig:free_fermion_standard}. 
        (a) All the singular values of matrices $G$ (red squares), $G^{T_2^f}$ (orange circles), and $G^{\tilde{T}_2^f}$ (blue stars) are plotted against the dimensionless inverse temperature $\beta t$, and for the latter two matrices, the subsystem length $L_{A_1}$ is $8$. 
        As $\beta$ increases, several singular values of the three matrices decrease exponentially. 
        The matrix inversion becomes unstable when the smallest singular value approaches machine precision, indicated by the black dashed horizontal line, which is approximately $2.22\times 10^{-16}$ for double-precision floating-point numbers. 
        (b) The lowest singular values of $G^{T_2^f}$ are shown as functions of inverse temperature $\beta$ for various subsystem lengths $L_{A_1}$. 
        (c) The rank-4 RNs, calculated using both the standard formula in Eq.~\eqref{equ:expRenyiNn_free_fermion} (origin circles and grey squares) and Loh's stable inversion formula in Eq.~\eqref{equ:logdetGrover_loh} (blue stars and purple plus signs), are plotted as functions of $\beta t$.
        Loh's stable inversion breaks down beyond $\beta t=15$, as indicated by the grey dashed vertical line, which aligns with (a). 
        Within the gray-filled region, the stable inverse formula in Eq.~\eqref{equ:logdetGrover_loh} for Grover determinants remains effective. 
    }
    \label{fig:numerical_instability}
\end{figure*}

\subsection{Singularity of the Green's function matrix}

If one naively performs the calculation following the formula, say, for the untwisted case, 
\begin{equation}\label{equ:detGrover_general}
    \det g_{r,\bar{\boldsymbol{s}}}=\det\left[\prod_{i=r}^{1}G_{\mathbf{s}^{(i)}}^{T_{2}^{f}}\left[I+\prod_{i=r}^{1}\left(\left(G_{\mathbf{s}^{(i)}}^{T_{2}^{f}}\right)^{-1}-I\right)\right]\right],
\end{equation}
a nonsense, exponentially large value would result. 
The issue arises from the unstable inverse of the partially transposed Green's functions.
In Fig.~\ref{fig:numerical_instability}(a), we plot the singular values of $G$, $G^{T_2^f}$, and $G^{\tilde{T}_2^f}$ against the inverse temperature $\beta$. 
While generally less severe, the partially transposed Green's functions ($G^{T_2^f}$ and $G^{\tilde{T}_2^f}$) indeed inherit singularities from $G$. 
The degree of singularity exhibits an evident dependence on temperature. Specifically, as the temperature decreases, the smallest few singular values of $G^{T_2^f/\tilde{T}_2^f}$ decrease exponentially. 
When the smallest singular value approaches machine precision (below the so-called ``breakdown temperature''), the matrix inversion becomes numerically unstable.

Additionally, the degree of instability is strongly influenced by the specific bipartite geometry.
In fact, Fig.~\ref{fig:numerical_instability}(a) only exhibits a particularly unstable case for $L_A=10$ (namely, $L_{A_1}=8$). 
In Fig.~\ref{fig:numerical_instability}(b), we present the lowest singular values of $G^{T_2^f}$ for various subsystem lengths $L_{A_1}$, revealing that the breakdown temperature varies significantly with $L_{A_1}$. 
For most values of $L_{A_1}$, the lowest singular values exponentially decay with increasing $\beta$. 
Remarkably, exactly at the equal-bipartition point ($L_{A_1}=L_A/2=5$), the lowest singular value is insensitive to temperature, allowing $G^{T_2^f}$ to be considered as a non-singular matrix. 
This permits the direct evaluation of Grover determinants using Eq.~\eqref{equ:detGrover_general} for equal-bipartition geometry, at very low temperatures and even in ground-state calculations. 
A general conclusion is that, as $L_{A_1}$ approaches the equal-bipartition point, the lowest singular value of $G^{T_2^f/\tilde{T}_2^f}$ decreases more slowly as $\beta$ increases. 

These observations—namely, that $G^{T_2^f/\tilde{T}_2^f}$ exhibits increasingly singular with decreasing temperature and as the geometry deviates from equal bipartition—remain valid for two-dimensional free systems and for interacting systems after HS transformation, based on practical experience. 
It is unsurprising that the overall singularity is significantly affected by the model under consideration and its specific parameters. 
In the context of interacting systems, within the framework of DQMC, additional factors such as the HS transformation employed and the specific auxiliary-field configuration also play crucial roles.
We conclude that the severity of the singularity is contingent upon three factors: (i) the model and the HS transformation employed, (ii) the bipartition geometry, and (iii) the temperature. 

\subsection{Solution 1: Loh's stable inversion formula and regularization of Green's functions}

Our first solution insists on working with the inversion of the partially transposed Green's functions and aims at stabilizing the Grover determinant calculation in Eq.~\eqref{equ:detGrover_general}. 
We first show that, prior to the breakdown point where the lowest singular values are close to machine precision, the Grover determinant can still be computed stably using a stabilization scheme based on matrix decomposition~\cite{Loh1989StableMatrixMultiplicationAlgorithms,SciPostPhys.Core2020Bauer}. 
Subsequently, we introduce a regularization scheme~\cite{Phys.Rev.B2014Assaad} for the partially transposed Green's functions to extend the stable calculation to lower temperatures. 

The stable formula bears a close resemblance to the methodologies employed for the stable computation of weights in DQMC (refer to Appendix \ref{app:numerical_stabilization_DQMC} for a brief overview of numerical stabilization in DQMC). 
Indeed, the term $(G_{\mathbf{s}}^{T_{2}^{f}})^{-1}-I$ is reminiscent of the matrix $B_\mathbf{s}(\beta,0)=(G_\mathbf{s}^{-1}-I)$ in Eq.~\eqref{equ:DQMC_B}, and they should exhibit similar mathematical behavior, characterized by the presence of exponentially large and small scales, leading to an exponentially large condition number~\cite{Loh1989StableMatrixMultiplicationAlgorithms,SciPostPhys.Core2020Bauer}. 
For convenience, we denote $\mathbb{B}_i \equiv (G_{\mathbf{s}^{(i)}}^{T_{2}^{f}})^{-1} - I$, and sometimes use $\mathbb{B}_{\mathbf{s}^{(i)}}$ to emphasize its dependence on $\mathbf{s}^{(i)}$. 
However, a more challenging situation arises here since $\mathbb{B}_i$ is not known from a priori (unlike $B(\beta,0)$, which can be calculated using the decoupled Hamiltonians $e^{K_l}$ and $e^{V_l[\mathbf{s}]}$). $\mathbb{B}_i$ must be determined by inverting $G^{T_2^f}_{\mathbf{s}^{(i)}}$. 
After obtaining $\mathbb{B}_i$, we perform UDV decomposition\footnote{Specifically, it can either be a SVD decomposition, or a QR decomposition followed by extracting the diagonal elements of R as D. See Appendix \ref{app:numerical_stabilization_DQMC} for details. } on $\mathbb{B}_i$, 
and then utilize stable multiplication and inversion formulas to separate large and small scales in the matrices as much as possible.
Analogous to the $B(\tau,\tau^\prime)$ matrices in Eq.~\eqref{equ:DQMC_B}, we define
\begin{equation}\label{equ:mathbbBii}
    \mathbb{B}\left(i,i^\prime\right)\equiv \prod_{j=i}^{i^\prime+1}\mathbb{B}_j=\mathbb{B}_i\mathbb{B}_{i-1}\cdots \mathbb{B}_{i^\prime+1},
\end{equation}
where $0\leq i^\prime\leq i\leq r$. It also satisfies $\mathbb{B}(i,i)=I$.  
After performing matrix decomposition $\mathbb{B}\left(r,0\right)=UDV$, we have\footnote{One could also split the $\mathbb{B}$ chain into two UDV decompositions and then use the corresponding formula for a stable inverse of the form, say, $(I+UDVUDV)^{-1}$. This approach is used in the incremental algorithm (see Appendix~\ref{app:schemes12}).}
\begin{equation}\label{equ:Getg00}
    \mathbb{G}\equiv\left[I+\mathbb{B}\left(r,0\right)\right]^{-1}=\left[D_{+}^{-1}U^{\dagger}+D_{-}V\right]^{-1}D_{+}^{-1}U^{\dagger},
\end{equation}
where $D_{+}=\max (D, 1)$ and $D_{-}=\min (D, 1)$. 
Subsequently, we obtain the stable formula for the Grover determinant:
\begin{equation}\label{equ:logdetGrover_loh}
    \begin{aligned}
        \ln\det g_{r,\bar{\boldsymbol{s}}}=&\sum_{i=1}^{r}\ln\det G_{\mathbf{s}^{(i)}}^{T_{2}^{f}}+\ln\det U+\ln\det D_{+}\\&+\ln\det\left(D_{+}^{-1}U^{\dagger}+D_{-}V\right).
    \end{aligned}
\end{equation}
A similar formula can be obtained for the twisted case. 
As depicted in Fig.~\ref{fig:numerical_instability}(b), when the temperature is relatively moderate and within the grey-shaded region, the stabilized formula (Eq.~\eqref{equ:logdetGrover_loh}) works well, yielding results consistent with those obtained from Eq.~\eqref{equ:expRenyiNn_free_fermion}. 
However, as the temperature decreases further, Eq.~\eqref{equ:logdetGrover_loh} fails to benchmark against Eq.~\eqref{equ:expRenyiNn_free_fermion}. 
This breakdown occurs at a temperature around $\beta t=15$, where the lowest singular value of $G^{T_2^f}$ or $G^{\tilde{T}_2^f}$ approaches machine precision, as illustrated in Fig.~\ref{fig:numerical_instability}(a). 

In summary, we have demonstrated that the singularity of partially transposed Green's functions leads to numerical instability in calculating high-rank Grover determinants.
The workaround proposed in Eq.~\eqref{equ:logdetGrover_loh} borrows the idea of numerical stabilization in DQMC~\cite{Loh1989StableMatrixMultiplicationAlgorithms,SciPostPhys.Core2020Bauer} and was also mentioned in Ref.~\cite{Phys.Rev.B2024Zhang} for computing high-rank REE. 
However, this formula is only effective when the smallest singular value of the partially transposed Green's functions is not close to machine precision. 
This obstacle prevents us from stably calculating high-rank RNs for generic bipartition at very low temperatures, including the zero-temperature or ground-state RNs, which corresponds to the projective QMC algorithm~\cite{Ann.Phys.1986Sugiyama,Europhys.Lett.1989Sorella,Int.J.ModernPhys.B1989Sorella}. 

To stabilize the computation of high-rank ground-state REE, Assaad \textit{et al.} proposed a regularization scheme for the Green's function matrices~\cite{Phys.Rev.B2014Assaad}. 
Before reducing the Green's function matrix to a subsystem, they regularized the matrix by adding a small positive constant $\Lambda$ to the zero eigenvalues and subtracting $\Lambda$ from other eigenvalues. 
Since they used the naive formula (similar to Eq.~\eqref{equ:detGrover_general}) to calculate high-rank REE, the regularization parameter $\Lambda$ should be large enough to stabilize the calculation.   
Additionally, the physical results should be tested to ensure they converge as $\Lambda$ decrease ($\Lambda\sim 10^{-5}$ is a good choice for the Kane-Mele Hubbard model they studied~\cite{Phys.Rev.B2014Assaad}). 
Motivated by these findings, we propose a new regularization scheme for partially transposed Green's function matrices, which can harbor complex eigenvalues.  
Instead of adding $\Lambda$ to the matrices' zero eigenvalues, we incorporate $\Lambda$ into their (nearly) vanishing singular values. 
Since Eq.~\eqref{equ:logdetGrover_loh} remains numerically stable when the smallest singular value is not near machine precision, a very small $\Lambda$ (approximately $10^{-13}$) suffices to stabilize the calculation. 
It is anticipated that the physical results should not be significantly affected by such a small $\Lambda$. 
Indeed, as shown in Fig.~\ref{fig:free_fermion_loh_regularization_Drut}, we applied this regularization scheme to the same free-fermion chain as in Fig.~\ref{fig:numerical_instability} and found that the results closely match those obtained from Eq.~\eqref{equ:expRenyiNn_free_fermion} across all temperatures. 
For interacting systems, this method is sufficiently accurate when the systematic error introduced by $\Lambda$ is negligible compared to the statistical error from Monte Carlo sampling. 

\begin{figure}
    \center
    \includegraphics[width=0.4\textwidth]{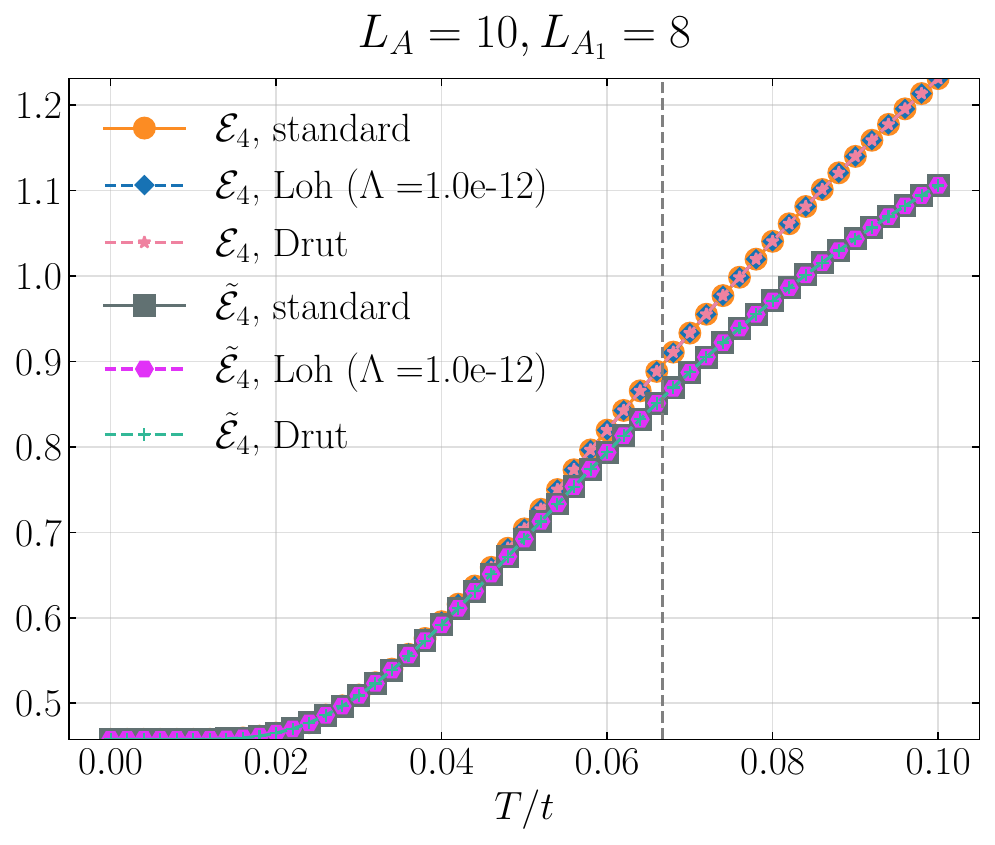}
    \caption{
        The rank-4 RNs of the same free-fermion chain studied in Fig.~\ref{fig:numerical_instability}(c) are calculated using the standard formula in Eq.~\eqref{equ:expRenyiNn_free_fermion}, Loh's stable inversion formula in Eq.~\eqref{equ:logdetGrover_loh} with regularized $G^{{T}_2^f}$ and $G^{\tilde{T}_2^f}$, and Drut's formula in Eqs.~\eqref{equ:detGrover_Drut_method} and \eqref{equ:detGroverbar_Drut_method}. 
        These results are plotted as functions of temperature $T/t$. 
        The grey dashed vertical line indicates the breakdown temperature of Loh's formula without regularization (see Fig.~\ref{fig:numerical_instability}). 
    }
    \label{fig:free_fermion_loh_regularization_Drut}
\end{figure}

\subsection{Solution 2: Drut's reconstruction of the Grover matrix}
The second solution is quite elegant and simple. Drut and Porter found a way to reconstruct high-rank Grover matrices in the context of REE, such that one can avoid the inversion of reduced Green's function matrices~\cite{Phys.Rev.E2016Drut}. 
Let us write the corresponding formulas for RNs. First we consider the untwisted case: 
\begin{subequations}\label{equ:detGrover_Drut_method}
    \begin{equation}\label{equ:detGrover_Drut}
        \det g_{r,\bar{\boldsymbol{s}}}=\det\mathcal{G}_{1\rightarrow r,\bar{\boldsymbol{s}}}^{T_{2}^{f}}\det K_{r,\bar{\boldsymbol{s}}}=\det T_{r,\bar{\boldsymbol{s}}},
    \end{equation}
where $\mathcal{G}_{1\rightarrow r,\bar{\boldsymbol{s}}}^{T_{2}^{f}}=\mathrm{diag}[G_{\mathbf{s}^{(1)}}^{T_{2}^{f}},\dots,G_{\mathbf{s}^{(r)}}^{T_{2}^{f}}]$ is a block-diagonal matrix with $G_{\mathbf{s}^{(i)}}^{T_{2}^{f}}$ as the $i$-th block, $K_{r,\bar{\boldsymbol{s}}}$ is given by
    \begin{equation}\label{equ:Kmat}
        K_{r,\bar{\boldsymbol{s}}}\equiv\left(\begin{array}{ccccc}
            I & 0 & \cdots & 0 & \mathbb{B}_{\mathbf{s}^{(r)}}\\
            -\mathbb{B}_{\mathbf{s}^{(1)}} & I & \cdots & 0 & 0\\
            0 & -\mathbb{B}_{\mathbf{s}^{(2)}} & \ddots & 0 & 0\\
            \vdots & \vdots & \ddots & \vdots & \vdots\\
            0 & 0 & \cdots & -\mathbb{B}_{\mathbf{s}^{(r-1)}} & I
            \end{array}\right),
    \end{equation}
and thus $T_{r,\bar{\boldsymbol{s}}}$ is given by
    \begin{equation}\label{equ:Tmat}
        \begin{aligned} & T_{r,\bar{\boldsymbol{s}}}\equiv K_{r,\bar{\boldsymbol{s}}}\mathcal{G}_{1\rightarrow r,\bar{\boldsymbol{s}}}^{T_{2}^{f}}\\
 & =\left(\begin{array}{ccccc}
G_{\mathbf{s}^{(1)}}^{T_{2}^{f}} & 0 & \cdots & 0 & I-G_{\mathbf{s}^{(r)}}^{T_{2}^{f}}\\
G_{\mathbf{s}^{(1)}}^{T_{2}^{f}}-I & G_{\mathbf{s}^{(2)}}^{T_{2}^{f}} & \cdots & 0 & 0\\
\vdots & \vdots & \ddots & \vdots & \vdots\\
0 & 0 & \cdots & G_{\mathbf{s}^{(r-1)}}^{T_{2}^{f}} & 0\\
0 & 0 & \cdots & G_{\mathbf{s}^{(r-1)}}^{T_{2}^{f}}-I & G_{\mathbf{s}^{(r)}}^{T_{2}^{f}}
\end{array}\right).
\end{aligned}
    \end{equation}
\end{subequations}
To calculate the twisted Grover determinants, one can similarly define $\tilde{T}_{r,\bar{\boldsymbol{s}}}$, satisfying $\det \tilde{g}_{r,\bar{\boldsymbol{s}}}=\det \tilde{T}_{r,\bar{\boldsymbol{s}}}$, by replacing $G_{\mathbf{s}^{(i)}}^{T_{2}^{f}}$ with $G_{\mathbf{s}^{(i)}}^{\tilde{T}_{2}^{f}}$ in Eq.~\eqref{equ:Tmat}. 
Additionally, the determinant of the modified Grover matrix $\underline{\tilde{g}}_{r,\bar{\boldsymbol{s}}}$ defined in Eq.~\eqref{equ:gkbar_twisted_DQMC} can also be reconstructed following a similar procedure. 
One can show that
\begin{subequations}\label{equ:detGroverbar_Drut_method}
\begin{equation}\label{equ:detGroverbar_Drut}
    \det\underline{\tilde{g}}_{r,\bar{\boldsymbol{s}}}=\det\mathcal{G}_{r\rightarrow 1,\bar{\boldsymbol{s}}}^{T_{2}^{f}}\det\underline{\tilde{K}}_{r,\bar{\boldsymbol{s}}}=\det\underline{\tilde{T}}_{r,\bar{\boldsymbol{s}}}, 
\end{equation}
where $\mathcal{G}_{r\rightarrow 1,\bar{\boldsymbol{s}}}^{T_{2}^{f}}=\mathrm{diag}[G_{\mathbf{s}^{(r)}}^{T_{2}^{f}},\dots,G_{\mathbf{s}^{(1)}}^{T_{2}^{f}}]$, $\underline{\tilde{K}}_{r,\bar{\boldsymbol{s}}}$ is given by
\begin{equation}\label{equ:Kbarmat}
    \underline{\tilde{K}}_{r,\bar{\boldsymbol{s}}}=\left(\begin{array}{ccccc}
I & -\mathbb{B}_{\mathbf{s}^{(r)}}U_{2} & \cdots & 0 & 0\\
0 & I & \cdots & 0 & 0\\
\vdots & \vdots & \ddots & \vdots & \vdots\\
0 & 0 & \cdots & I & -\mathbb{B}_{\mathbf{s}^{(2)}}U_{2}\\
\mathbb{B}_{\mathbf{s}^{(1)}}U_{2} & 0 & \cdots & 0 & I
\end{array}\right),
\end{equation}
and 
\begin{equation}\label{equ:Tbarmat}
    \begin{aligned}
        &\underline{\tilde{T}}_{r,\bar{\boldsymbol{s}}}\equiv\mathcal{G}_{r\rightarrow 1,\bar{\boldsymbol{s}}}^{T_{2}^{f}}\underline{\tilde{K}}_{r,\bar{\boldsymbol{s}}}\\&=\left(\begin{array}{cccc}
        G_{\mathbf{s}^{(r)}}^{T_{2}^{f}} & \left(G_{\mathbf{s}^{(r-1)}}^{T_{2}^{f}}-I\right)U_{2} & \cdots & 0\\
        0 & G_{\mathbf{s}^{(r-1)}}^{T_{2}^{f}} & \cdots & 0\\
        \vdots & \vdots & \ddots & \vdots\\
        \left(I-G_{\mathbf{s}^{(1)}}^{T_{2}^{f}}\right)U_{2} & 0 & \cdots & G_{\mathbf{s}^{(1)}}^{T_{2}^{f}}
        \end{array}\right).
    \end{aligned}
\end{equation}
\end{subequations}
Here, the matrix $U_2$ is defined above Eq.~\eqref{equ:twistedPTGreen}. 
Using Eqs.~\eqref{equ:detGrover_Drut_method} and \eqref{equ:detGroverbar_Drut_method}, we revisit the computation of RNs in Fig.~\ref{fig:free_fermion_loh_regularization_Drut}. The results are consistent with those obtained from both Eq.~\eqref{equ:expRenyiNn_free_fermion} and Eq.~\eqref{equ:logdetGrover_loh} with regularized Green's function matrices for all temperatures.  

%%%%%%%%%%%%%%%%%%%%%%%%%%%%%%%%%%%%%%%%%%%%%%%%%%%%%
\section{Incremental algorithm for rank-$r$ R\'{e}nyi negativity}\label{sec:incremental}
%%%%%%%%%%%%%%%%%%%%%%%%%%%%%%%%%%%%%%%%%%%%%%%%%%%%%
As system size $L$ or temperature increases, the RN in general become larger, and the associated Grover determinant becomes exponentially small, which causes the second challenge in calculating high-rank RN. 
This challenge is well-documented in the literature on various quantum Monte Carlo methods aimed at calculating REE, and several incremental algorithms have been proposed to address this issue~\cite{Phys.Rev.Lett.2020DEmidio,npjQuantumMater.2022Zhao,Phys.Rev.Lett.2022Zhao,Phys.Rev.B2023Pan,npjQuantumInf2025Liao,ArXiv2023Liao,Phys.Rev.B2024Zhang,Phys.Rev.Lett.2024DEmidio}. 

In this section, we develop numerically stable incremental algorithms for calculating both untwisted and twisted RNs.
For each variant, we first establish the theoretical framework of the incremental algorithm, then propose specific, computationally stable update scheme.
Finally, we demonstrate that the RNR can be efficiently computed within a single QMC run, eliminating the need for separate calculations of RN and thermal R\'{e}nyi entropy $S_{r}^{\mathrm{th}}$.

\subsection{Incremental algorithm for untwisted R\'{e}nyi negativity}
Let us first rewrite the DQMC formula for untwisted RN in Eq.~\eqref{equ:expRenyiNk_untwisted_DQMC}, 
\begin{equation}\label{equ:expRenyiNk-naive}
    e^{\left(1-r\right)\mathcal{E}_{r}}=\frac{\sum_{\mathbf{s}^{(1)}\cdots\mathbf{s}^{(r)}}W_{\mathbf{s}^{(1)}\cdots\mathbf{s}^{(r)}}\det g_{r,\bar{\boldsymbol{s}}}}{\sum_{\mathbf{s}^{(1)}\cdots\mathbf{s}^{(r)}}W_{\mathbf{s}^{(1)}\cdots\mathbf{s}^{(r)}}},
\end{equation}
where $W_{\mathbf{s}^{(1)}\cdots\mathbf{s}^{(r)}}=\prod_{i=1}^{r}w_{\mathbf{s}^{(i)}}$ represents the joint DQMC weight in the combined configuration space $\{\mathbf{s}^{(1)}\cdots\mathbf{s}^{(r)}\}$. 
We refer to $\mathbf{s}^{(i)}$ as the configuration at the $i$-th replica. 
When $\mathcal{E}_r$ is large, the Grover determinant, $\det g_r$, is typically exponentially small. 
The appearance of sampling spikes leads to the infinite variance problem~\cite{Phys.Rev.E2016Shi}, significantly affecting the evaluation of $\mathcal{E}_r$ through Eq.~\eqref{equ:expRenyiNk-naive}. 
Further studies on REE unveiled that the distribution of the Grover determinant defined by the reduced Green's function is non-Gaussian, whereas the logarithm and the $N_{\mathrm{inc}}$-th root of the Grover determinant follow a Gaussian distribution~\cite{npjQuantumInf2025Liao}. 
These findings motivate a similar incremental algorithm for untwisted RN, 
\begin{equation}
    e^{(1-r)\mathcal{E}_{r}}=\frac{Z_{N_{\mathrm{inc}}}}{Z_{N_{\mathrm{inc}}-1}}\cdots\frac{Z_{k+1}}{Z_{k}}\cdots\frac{Z_{1}}{Z_{0}},
\end{equation}
with the partition function defined as
\begin{equation}
    Z_{k}=\sum_{\mathbf{s}^{(1)}\cdots\mathbf{s}^{(r)}}W_{\mathbf{s}^{(1)}\cdots\mathbf{s}^{(r)}}\left(\det g_{r,\bar{\boldsymbol{s}}}\right)^{\frac{k}{N_{\mathrm{inc}}}}.
\end{equation}
Each intermediate incremental process (denoted by $k$) is an average of an ensemble existing in the replicated configuration space: 
\begin{equation}
    \begin{aligned}
        \frac{{Z}_{k+1}}{{Z}_{k}}&=\frac{\sum_{\mathbf{s}^{(1)}\cdots\mathbf{s}^{(r)}}{W}_{\mathbf{s}^{(1)}\cdots\mathbf{s}^{(r)}}^{\left(k\right)}\left(\det{{g}}_{r,\bar{\boldsymbol{s}}}\right)^{\frac{1}{N_{\mathrm{inc}}}}}{\sum_{\mathbf{s}^{(1)}\cdots\mathbf{s}^{(r)}}{W}_{\mathbf{s}^{(1)}\cdots\mathbf{s}^{(r)}}^{\left(k\right)}}\\&=\left\langle \left(\det{{g}}_{r,\bar{\boldsymbol{s}}}\right)^{\frac{1}{N_{\mathrm{inc}}}}\right\rangle _{{W}^{\left(k\right)}},
    \end{aligned}
\end{equation}
where $W_{\mathbf{s}^{(1)}\cdots\mathbf{s}^{(r)}}^{\left(k\right)}=W_{\mathbf{s}^{(1)}\cdots\mathbf{s}^{(r)}}\left(\det g_{r,\bar{\boldsymbol{s}}}\right)^{\frac{k}{N_{\mathrm{inc}}}}$. 
For a sign-problem-free Monte Carlo sampling of each factor, it is essential to prove that $\det g_{r,\bar{\boldsymbol{s}}}$ is real and positive for any rank $r$. 
In previous work~\cite{Nat.Commun.2025Wanga}, we proved that for two classes of sign-problem-free models, represented by the Hubbard model and $t$-$V$ model, the Grover determinant of any rank is indeed real and positive. 
We also provide the proof in Appendix~\ref{app:sign_detGrvoer}. 

\subsection{Local update scheme for untwisted R\'{e}nyi negativity}\label{sec:scheme3}
\begin{figure}
    \center
    \includegraphics[width=0.35\textwidth]{fig3.pdf}
    \caption{Illustration of the (forward) local sweep in each incremental process of the incremental algorithm for rank-$r$ RN. }
    \label{fig:localsweep}
\end{figure}
As illustrated in Fig.~\ref{fig:localsweep}, we sequentially update the $\mathbf{s}$-configurations, starting from the first replica $\mathbf{s}^{(1)}$ and proceeding to the last replica $\mathbf{s}^{(r)}$. 
Within each replica, updates occur from the first spacetime coordinate to the last. 
When updating the auxiliary field at a specific spacetime coordinate in replica-$i$, the acceptance ratio is given by
\begin{equation}
    R_{i}^{\left(k\right)}\equiv\frac{W_{\mathbf{s}^{(1)}\cdots\mathbf{s}^{(i)\prime}\cdots\mathbf{s}^{(r)}}^{\left(k\right)}}{W_{\mathbf{s}^{(1)}\cdots\mathbf{s}^{(i)}\cdots\mathbf{s}^{(r)}}^{\left(k\right)}}=\frac{w_{\mathbf{s}^{(i)\prime}}}{w_{\mathbf{s}^{(i)}}}\left(\frac{\det g_{r,\bar{\boldsymbol{s}}^\prime}}{\det g_{r,\bar{\boldsymbol{s}}}}\right)^{k/N_{\mathrm{inc}}}. 
\end{equation}
In practice, each calculation task associated with the intermediate factor labeled by $k$ is assigned to a single CPU, allowing these factors to be computed in parallel. 
One essential task for every CPU is to compute the ratio of Grover determinants before and after the update, ${\det g_{r,\bar{\boldsymbol{s}}^\prime}}/{\det g_{r,\bar{\boldsymbol{s}}}}$. 
The only difference between the CPUs lies in the root order of the Grover determinant ratio when evaluating the total ratio $R_i^{(k)}$. 

To design a fast update scheme for Monte Carlo sampling, the initial step is to select one or a few central matrices to update, which will be referred to as the update objects in the following discussion.  
A practical Monte Carlo sampling scheme involves both the update formula and the stabilization formula for these update objects. 
In normal DQMC, the update object is the Green's function matrix $G_{\mathbf{s},ij}$ in Eq.~\eqref{equ:DQMC_equaltimeG}. 
The inverse determinant ratio ($\det\left[G_{\mathbf{s}^{(i)}}\right]/\det\left[G_{\mathbf{s}^{(i)\prime}}\right]$) corresponds to the ``quantum part'' of ratio $w_{\mathbf{s}^{(i)\prime}}/w_{\mathbf{s}^{(i)}}$, excluding the ratio of the coupling coefficients $\alpha[\mathbf{s}^{(i)}]=\prod_l\alpha[s^{(i)}_l]$ (see Eq.~\eqref{equ:Z_DQMC}). 
A straightforward update scheme for an incremental algorithm could involve splitting the total ratio $R^{(k)}_i$ into three components, $R_{i}^{\left(k\right)}=r_{1}(r_{2}r_{3})^{k/N_\mathrm{inc}}$, where
\begin{equation}\label{equ:r1r2r3}
    \begin{aligned}
    &r_{1}=\frac{w_{\mathbf{s}^{(i)\prime}}}{w_{\mathbf{s}^{(i)}}},\ r_{2}=\frac{\det G_{\mathbf{s}^{(i)\prime}}^{T_{2}^{f}}}{\det G_{\mathbf{s}^{(i)}}^{T_{2}^{f}}},
     \\
%    \text{and }
    &r_{3}=\frac{\det\left[I+\mathbb{B}_{r}\cdots\mathbb{B}_{i}^{\prime}\cdots\mathbb{B}_{1}\right]}{\det\left[I+\mathbb{B}_{r}\cdots\mathbb{B}_{i}\cdots\mathbb{B}_{1}\right]}.
    \end{aligned}
\end{equation}
It is important to note that as one sweeps through the local auxiliary fields back and forth, the Grover determinant in the weight should be fixed at a specific imaginary time to ensure the weight is well-defined~\cite{Phys.Rev.B2023Pan,npjQuantumInf2025Liao,ArXiv2023Liao,Phys.Rev.B2024Zhang,Phys.Rev.Lett.2024DEmidio}. 
This point is typically set to $\tau=0$ in finite-temperature DQMC, so all instances of $G_{\mathbf{s}^{(i)}}^{T_{2}^{f}}$ appearing in $r_2$ and $r_3$ in Eq.~\eqref{equ:r1r2r3}, including those inside $\mathbb{B}_i$, are actually $G_{\mathbf{s}^{(i)}}^{T_{2}^{f}}(0,0)$. 

The calculation of $r_1$ follows standard DQMC procedures.  
In the following discussion, we introduce an update scheme for calculating the remaining part of the ratio which demonstrates sufficient numerical stability for the spinless $t$-$V$ model. 
We also discuss two alternative schemes with simpler update objects but lower numerical stability in Appendix~\ref{app:schemes12}. 
For convenience, Table~\ref{tab:dimension} summarizes the dimensions of all key matrices across these schemes.
\begin{table}[h]
\centering
\begin{tabular}{|c|c|}
\hline
Dimension & Intermediate matrices \\
\hline
$N\times m$ & $U,\mathcal{U},\mathscr{U}$ \\
$m\times N$ & $V,\mathcal{V},\mathbb{V},\mathcal{V}^{a},\mathcal{V}^{b},\mathcal{V}^{c}$ \\
$m\times m$ & $I_m,D$ \\
\hline
\end{tabular}
\caption{
    Table for the dimensions of the intermediate matrices in the key steps in the update schemes. 
    $m$ is the number of non-zero entries in the deviation matrix $\Delta$, see the explanation below Eq.~\eqref{equ:update_G00}.
}
\label{tab:dimension}
\end{table}

Before introducing the actual update object employed in this scheme, we must first derive an essential formula for the local update of the partially transposed Green's function matrix. 
Since the $G^{T_2^f}_{\mathbf{s}^{(i)}}$ in Eq.~\eqref{equ:r1r2r3} originates from the equal-time Green's function at $\tau=0$, we begin with the update formula for $G_{\mathbf{s}^{(i)}}(0,0)$ at time $\tau$\footnote{Technically, due to the symmetric Trotter decomposition, the Green's functions in the update formula have undergone partial propagation, i.e., $G(\tau,\tau)\rightarrow e^{-K}G(\tau,\tau)e^K$, $G(\tau,0)\rightarrow e^{-K}G(\tau,0)$, and $G(0,\tau)\rightarrow G(0,\tau)e^K$. }
\begin{equation}\label{equ:update_G00}
    \begin{aligned}
        G_{\mathbf{s}^{(i)\prime}}\left(0,0\right)&=G_{\mathbf{s}^{(i)}}\left(0,0\right)\\&\quad+G_{\mathbf{s}^{(i)}}\left(0,\tau\right)U(I_{m}+VU)^{-1}V\left(\tau,0\right).
    \end{aligned}
\end{equation}
Here is an explanation of the above formula. 
We first define the deviation matrix between the interaction matrix after and before the update, $\Delta\equiv e^{-V[\mathbf{s}^{(i)}]}e^{V[\mathbf{s}^{(i)\prime}]}-I$, and restrict ourselves to consider only $N\times N$ diagonal $\Delta$ with $m$ non-zeros entries. 
For instance, $m=1$ for the Hubbard model, and $m=2$ for the $t$-$V$ model. 
Next, we define projection matrices $P_{N\times m}$ and $P_{m\times N}$, which are cropped from the identity matrix $I=I_{N}$, with columns or rows corresponding to the non-zero entries of $\Delta$. 
The matrices in the above formula can then be expressed as $\Delta=P_{N\times m}DP_{m\times N}$, $U\equiv P_{N\times m}D$, $V\equiv P_{m\times N}(I-G_{\mathbf{s}^{(i)}}(\tau,\tau))$, and $V(\tau,0)\equiv P_{m\times N}G_{\mathbf{s}^{(i)}}(\tau,0)$. 
For simplicity, we abbreviate the subscript $\mathbf{s}^{(i)}$ in these intermediate matrices, but it is essential to remember that they depend on the $\mathbf{s}$-configuration in the current replica before the update. 
The update formula for $G_{\mathbf{s}^{(i)}}^{T_{2}^{f}}$ can be derived by applying Eq.~\eqref{equ:untwistedPTGreen} to Eq.~\eqref{equ:update_G00}:
\begin{equation}\label{equ:update_GT2f}
    \begin{aligned}
        G_{\mathbf{s}^{(i)\prime}}^{T_{2}^{f}}&=G_{\mathbf{s}^{(i)}}^{T_{2}^{f}}\left(I+\mathcal{U}\mathcal{V}\right)=G_{\mathbf{s}^{(i)}}^{T_{2}^{f}}+\mathscr{U}\mathcal{V},\\\mathscr{U}&\equiv\left(\begin{array}{c}
        \left(G_{\mathbf{s}^{(i)}}(0,\tau)\right)_{A_{1}N}\\
        \text{i}\left(G_{\mathbf{s}^{(i)}}(0,\tau)\right)_{A_{2}N}
        \end{array}\right)P_{N\times m}\equiv G_{\mathbf{s}^{(i)}}^{T_{2}^{f}}\mathcal{U},\\\mathcal{V}&\equiv D\left(I_{m}+VU\right)^{-1}\left(\begin{array}{cc}
        V_{mA_{1}}(\tau,0) & \text{i}V_{mA_{2}}(\tau,0)\end{array}\right).
    \end{aligned}
\end{equation}
We assume that the linear indices of the system's degrees of freedom are divided into two continuous subsystems $A_1$ and $A_2$, allowing for convenient expression of the block form of $\mathscr{U}$ and $\mathcal{V}$. 
This approach is easily generalizable to any bipartition, and for different bipartition geometries, the formula above is the only aspect that changes. 

Building upon the established update formula, we now define our primary update object for replica-$i$,
\begin{equation}\label{equ:scheme4ffen}
    \begin{aligned}
        \mathbb{F}_{i}&=\left[G_{\mathbf{s}^{(i)}}^{T_{2}^{f}}\mathbb{B}_{\bcancel{i}}^{-1}G_{\mathbf{s}^{(i-1)}}^{T_{2}^{f}}\mathbb{B}_{i-1}+G_{\mathbf{s}^{(i)}}^{T_{2}^{f}}\mathbb{B}_{i}G_{\mathbf{s}^{(i-1)}}^{T_{2}^{f}}\mathbb{B}_{i-1}\right]^{-1}\\&=\left[G_{\mathbf{s}^{(i)}}^{T_{2}^{f}}\mathbb{F}_{\bcancel{i}}^{-1}+I-G_{\mathbf{s}^{(i-1)}}^{T_{2}^{f}}\right]^{-1},
    \end{aligned}
\end{equation}
where $\mathbb{B}_{\bcancel{i}}\equiv\mathbb{B}\left(i-1,0\right)\mathbb{B}\left(r,i\right)=\mathbb{B}_{i-1}\cdots\mathbb{B}_{1}\mathbb{B}_{r}\cdots\mathbb{B}_{i+1}$ and 
\begin{equation}\label{equ:FcanceliInv}
    \begin{aligned}
        \mathbb{F}_{\bcancel{i}}^{-1}&\equiv\left(\mathbb{B}_{\bcancel{i}}^{-1}-I\right)G_{\mathbf{s}^{(i-1)}}^{T_{2}^{f}}\mathbb{B}_{i-1}\\&=\mathbb{B}_{i+1}^{-1}\cdots\mathbb{B}_{r}^{-1}\mathbb{B}_{1}^{-1}\cdots\mathbb{B}_{i-2}^{-1}G_{\mathbf{s}^{(i-1)}}^{T_{2}^{f}}-\left(I-G_{\mathbf{s}^{(i-1)}}^{T_{2}^{f}}\right).
    \end{aligned}
\end{equation}
It can be verified that the update ratio of $\det \mathbb{F}_i^{-1}$ is equivalent to the update ratio of $\det g_{r,\bar{\boldsymbol{s}}}$, since 
\begin{equation}\label{equ:scheme4detffiInv}
    \begin{aligned}
        \det\mathbb{F}_{i}^{-1}=&\det\left[G_{\mathbf{s}^{(i)}}^{T_{2}^{f}}\right]\det\left[\mathbb{B}_{\bcancel{i}}^{-1}+\mathbb{B}_{i}\right]\det\left[G_{\mathbf{s}^{(i-1)}}^{T_{2}^{f}}\mathbb{B}_{i-1}\right]\\=&\det\left[G_{\mathbf{s}^{(i)}}^{T_{2}^{f}}\right]\det\left[I+\mathbb{B}\left(i,0\right)\mathbb{B}\left(r,i\right)\right]\\&\times\det\left[\mathbb{B}_{\bcancel{i}}^{-1}\right]\det\left[G_{\mathbf{s}^{(i-1)}}^{T_{2}^{f}}\mathbb{B}_{i-1}\right],
    \end{aligned}
\end{equation}
where the third line is independent of $\mathbf{s}^{(i)}$ and cancels out when considering the update ratio.
This definition is based on two basic ideas. 
First, instead of using $\mathbb{B}_i$, we use its inverse, $\mathbb{B}_i^{-1}$, which provides better stability during numerical calculations. 
Second, we explicitly incorporate matrices from neighboring replicas, $G^{T_2^f}_{\mathbf{s}^{(i-1)}}$ and $\mathbb{B}_{i-1}$, into the definition of $\mathbb{F}_i$. 
We note that for high-rank Grover determinants, using the property $\det AB=\det A \det B$, the matrices $G^{T_2^f}_{\mathbf{s}^{(j)}}$ and $\mathbb{B}_j$ can be combined for at most two different $j\in[1,r]$, as the matrices are non-commutative. 
When $r=2$, $\mathbb{F}_i$ reduces to the inverse of the total Grover matrix:
\begin{equation}
    \mathbb{F}_{i}^{-1}=G_{\mathbf{s}^{(i)}}^{T_{2}^{f}}G_{\mathbf{s}^{(i-1)}}^{T_{2}^{f}}+\left(I-G_{\mathbf{s}^{(i)}}^{T_{2}^{f}}\right)\left(I-G_{\mathbf{s}^{(i-1)}}^{T_{2}^{f}}\right)=g_{2,\bar{\boldsymbol{s}}}
    %,i=1,2.
\end{equation}
Updating the inverse of the rank-2 Grover matrix is stable and commonly used in practice, as evidenced by previous studies on the incremental algorithm for rank-2 REE~\cite{Phys.Rev.Lett.2024DEmidio,npjQuantumInf2025Liao,ArXiv2023Liao}. 

Updating $\mathbb{F}_i$ in Eq.~\eqref{equ:scheme4ffen} follows from Eq.~\eqref{equ:update_GT2f}:
\begin{equation}\label{equ:scheme4update}
    \mathbb{F}_{i}^{\prime}=\mathbb{F}_{i}\left(I-\mathscr{U}\left(I_m+\mathcal{V}^{c}\mathscr{U}\right)^{-1}\mathcal{V}^{c}\right),
\end{equation}
with the ratio given by
\begin{equation}\label{equ:scheme4r2r3}
    r_{2}r_{3}=\frac{\det g_{r,\bar{\boldsymbol{s}}^\prime}}{\det g_{r,\bar{\boldsymbol{s}}}}=\frac{\det\mathbb{F}_{i}}{\det\mathbb{F}_{i}^{\prime}}=\det\left(I_m+\mathcal{V}^{c}\mathscr{U}\right),
\end{equation}
where $\mathcal{V}^{c}\equiv\mathcal{V}\mathbb{F}_{\bcancel{i}}^{-1}\mathbb{F}_{i}$.  
However, the numerical stabilization of $\mathbb{F}_i$ is more tedious and does not simply reuse the numerical stabilization routines from DQMC (listed in Appendix~\ref{app:numerical_stabilization_DQMC}). 
Nevertheless, the core principle of numerical stabilization remains the same: perform matrix decomposition and separate small and large scales as much as possible. 
A detailed discussion and code implementation are provided in Appendix \ref{app:numerical_stabilization_scheme4}.
Finally, the measurement of $\det g_{r,\bar{\boldsymbol{s}}}$ is given by:
\begin{equation}
    \ln\det g_{r,\bar{\boldsymbol{s}}}=-\ln\det\mathbb{F}_{i}+\sum_{j\neq i,i-1}\ln\det\left[I-G_{\mathbf{s}^{(j)}}^{T_{2}^{f}}\right],
\end{equation}
as derived from Eq.~\eqref{equ:scheme4detffiInv}. 

\subsection{Incremental algorithm for twisted R\'{e}nyi negativity}
The incremental algorithm for rank-$r$ twisted RN can be implemented in a similar manner. 
We start with Eq.~\eqref{equ:expRenyiNk_twisted_DQMC2} and define the following partition function: 
\begin{equation}
    \tilde{Z}_{k}=\sum_{\mathbf{s}^{(1)}\cdots\mathbf{s}^{(r)}}W_{\mathbf{s}^{(1)}\cdots\mathbf{s}^{(r)}}\left(\det\underline{\tilde{g}}_{r,\bar{\boldsymbol{s}}}\right)^{\frac{k}{N_{\mathrm{inc}}}},
\end{equation}
which results in an incremental expression for rank-$r$ twisted RN:
\begin{equation}
    e^{(1-r)\mathcal{\tilde{E}}_{r}}=\frac{\tilde{Z}_{N_{\mathrm{inc}}}}{\tilde{Z}_{N_{\mathrm{inc}}-1}}\cdots\frac{\tilde{Z}_{k+1}}{\tilde{Z}_{k}}\cdots\frac{\tilde{Z}_{1}}{\tilde{Z}_{0}}.
\end{equation}
Each intermediate factor's calculation is assigned to an incremental process (denoted by $k$), with
\begin{equation}
    \begin{aligned}
        \frac{\tilde{Z}_{k+1}}{\tilde{Z}_{k}}&=\frac{\sum_{\mathbf{s}^{(1)}\cdots\mathbf{s}^{(r)}}\tilde{W}_{\mathbf{s}^{(1)}\cdots\mathbf{s}^{(r)}}^{\left(k\right)}\left(\det\underline{\tilde{g}}_{r,\bar{\boldsymbol{s}}}\right)^{\frac{1}{N_{\mathrm{inc}}}}}{\sum_{\mathbf{s}^{(1)}\cdots\mathbf{s}^{(r)}}\tilde{W}_{\mathbf{s}^{(1)}\cdots\mathbf{s}^{(r)}}^{\left(k\right)}}\\&=\left\langle \left(\det\underline{\tilde{g}}_{r,\bar{\boldsymbol{s}}}\right)^{\frac{1}{N_{\mathrm{inc}}}}\right\rangle _{\tilde{W}^{\left(k\right)}},
    \end{aligned}
\end{equation}
where $\tilde{W}_{\mathbf{s}^{(1)}\cdots\mathbf{s}^{(r)}}^{\left(k\right)}=W_{\mathbf{s}^{(1)}\cdots\mathbf{s}^{(r)}}(\det\underline{\tilde{g}}_{r,\bar{\boldsymbol{s}}})^{\frac{k}{N_{\mathrm{inc}}}}$. 
As demonstrated in Appendix \ref{app:sign_detGrvoer}, analogous to the untwisted case, the twisted Grover determinant exhibits real and non-negative properties ($\det \tilde{g}_{r,\bar{\boldsymbol{s}}}\geq 0$) for two classes of sign-problem-free models. 
However, the normalization factor $\mathcal{Z}_{\tilde{T}_{2}^{f},\mathbf{s}^{(i)}}$ can be negative for certain ranks and subsystem sizes. 
In particular, for the first class of models, such as the Hubbard model, the normalization factor $\mathcal{Z}_{\tilde{T}_{2}^{f},\mathbf{s}^{(i)}}$ is proportional to $\left(-1\right)^{N_{A_{2}}}$, where $N_{A_{2}}$ represents the number of sites in subsystem $A_2$ (see Eq.~\eqref{equ:sign_Dirac_ZT2f}). 
Consequently, for odd ranks and odd values of $N_{A_{2}}$, the twisted RN is ill-defined (see Fig.~\ref{fig:Hubbard_chain_ED_DQMC} for an example with $r=3$). 
In what follows, we will focus exclusively on cases where both the modified Grover determinant $\det\underline{\tilde{g}}_{r,\bar{\boldsymbol{s}}}$ and the weight $\tilde{W}^{(k)}_{\mathbf{s}^{(1)}\cdots\mathbf{s}^{(r)}}$ are sign-problem-free.

\subsection{Local update scheme for twisted R\'{e}nyi negativity}

Consider an update occurring in replica-$i$. The update ratio is given by
\begin{equation}
    \tilde{R}_{i}^{\left(k\right)}\equiv\frac{\tilde{W}_{\mathbf{s}^{(1)}\cdots\mathbf{s}^{(i)\prime}\cdots\mathbf{s}^{(r)}}^{\left(k\right)}}{\tilde{W}_{\mathbf{s}^{(1)}\cdots\mathbf{s}^{(i)}\cdots\mathbf{s}^{(r)}}^{\left(k\right)}}=r_{1}\left(r_{2}\tilde{r}_{3}\right)^{\frac{k}{N_{\mathrm{inc}}}},
\end{equation}
where $r_1$ and $r_2$ are defined in Eq.~\eqref{equ:r1r2r3}, and 
\begin{equation}
    \tilde{r}_{3}=\frac{\det\left[I+\tilde{\mathbb{B}}_{r}\cdots\tilde{\mathbb{B}}_{i}^{\prime}\cdots\tilde{\mathbb{B}}_{1}\right]}{\det\left[I+\tilde{\mathbb{B}}_{r}\cdots\tilde{\mathbb{B}}_{i}\cdots\tilde{\mathbb{B}}_{1}\right]}. 
\end{equation}
Here, $\tilde{\mathbb{B}}_{i}=\mathbb{B}_{i}U_{2}=[(G^{T_{2}^{f}})^{-1}-I]U_{2}$. 
Similar to the untwisted case in Sec.~\ref{sec:scheme3}, we implement a scheme to directly calculate $r_2\tilde{r}_3$, offering numerical stability. 
The primary adjustment involves incorporating $U_2$ into the update object and intermediate matrices.
Eqs.~\eqref{equ:scheme4ffen} and \eqref{equ:FcanceliInv} are thus replaced by
\begin{equation}\label{equ:scheme4ffen_twisted}
    \underline{\tilde{\mathbb{F}}}_{i}=\left[G_{\mathbf{s}^{(i)}}^{T_{2}^{f}}\underline{\tilde{\mathbb{F}}}_{\bcancel{i}}^{-1}+U_{2}\left(I-G_{\mathbf{s}^{(i-1)}}^{T_{2}^{f}}\right)U_{2}\right]^{-1}
\end{equation}
and
\begin{equation}\label{equ:FcanceliInv_twisted}
    \underline{\tilde{\mathbb{F}}}_{\bcancel{i}}^{-1}=\left(\tilde{\mathbb{B}}_{\bcancel{i}}^{-1}-U_{2}\right)\left(I-G_{\mathbf{s}^{(i-1)}}^{T_{2}^{f}}\right)U_{2}, 
\end{equation}
respectively. 
The update formula for Eq.~\eqref{equ:scheme4ffen_twisted} and the calculation of the ratio $r_2\tilde{r}_3$ are analogous to Eqs.~\eqref{equ:scheme4update} and \eqref{equ:scheme4r2r3}, provided that one replaces ${{\mathbb{F}}}_{i}$ and $\underline{\tilde{\mathbb{F}}}_{i}$ therein with ${{\mathbb{F}}}_{\bcancel{i}}^{-1}$ and $\underline{\tilde{\mathbb{F}}}_{\bcancel{i}}^{-1}$, respectively. 
The numerical stabilization of $\underline{\tilde{\mathbb{F}}}_{i}$ is also similar, as discussed in Appendix \ref{app:numerical_stabilization_scheme4}. 
Finally, the measurement of $\det\underline{\tilde{g}}_{r}$ is given by
\begin{equation}
    \ln\det\underline{\tilde{g}}_{r}=-\ln\det\underline{\tilde{\mathbb{F}}}_{i}^{-1}+\sum_{j\neq i,i-1}\ln\det\left[\left(I-G_{\mathbf{s}^{(j)}}^{T_{2}^{f}}\right)U_{2}\right]. 
\end{equation}

\subsection{Incremental algorithm for R\'{e}nyi negativity ratio}
The rank-$r$ RNR defined in Eq.~\eqref{equ:def_RNR} can be efficiently calculated using an incremental algorithm in a single simulation, rather than computing RN and thermal R\'{e}nyi entropy separately. 
This approach offers an additional computational advantage: since RNR is intrinsically smaller than the corresponding RN of the same rank, fewer incremental processes are required for small statistical errors. 
When dividing two exponential observables, the result remains an exponential observable, yielding the following expression for untwisted RNR:
\begin{equation}
    \begin{aligned}
        e^{\left(1-r\right)R_{r}}&=\frac{e^{\left(1-r\right){\cal E}_{r}}}{e^{\left(1-r\right)S_{r}^{{\rm th}}}}=\frac{\sum_{\mathbf{s}^{(1)}\cdots\mathbf{s}^{(r)}}W_{\mathbf{s}^{(1)}\cdots\mathbf{s}^{(r)}}\det g_{r,\bar{\boldsymbol{s}}}}{\sum_{\mathbf{s}^{(1)}\cdots\mathbf{s}^{(r)}}W_{\mathbf{s}^{(1)}\cdots\mathbf{s}^{(r)}}\det g_{r,\bar{\boldsymbol{s}}}^{\mathrm{th}}}\\&=\frac{\sum_{\mathbf{s}^{(1)}\cdots\mathbf{s}^{(r)}}W_{\mathbf{s}^{(1)}\cdots\mathbf{s}^{(r)}}\det g_{r,\bar{\boldsymbol{s}}}^{\mathrm{th}}\left(\det g_{r,\bar{\boldsymbol{s}}}/\det g_{r,\bar{\boldsymbol{s}}}^{\mathrm{th}}\right)}{\sum_{\mathbf{s}^{(1)}\cdots\mathbf{s}^{(r)}}W_{\mathbf{s}^{(1)}\cdots\mathbf{s}^{(r)}}\det g_{r,\bar{\boldsymbol{s}}}^{\mathrm{th}}}\\&=\frac{Z_{N_{\mathrm{inc}}}}{Z_{N_{\mathrm{inc}}-1}}\cdots\frac{Z_{k+1}}{Z_{k}}\cdots\frac{Z_{1}}{Z_{0}},
    \end{aligned}
\end{equation}
where the partition function is defined as
\begin{equation}
    Z_{k}=\sum_{\mathbf{s}^{(1)}\cdots\mathbf{s}^{(r)}}W_{\mathbf{s}^{(1)}\cdots\mathbf{s}^{(r)}}\det g_{r,\bar{\boldsymbol{s}}}^{\mathrm{th}}\left(\det g_{r,\bar{\boldsymbol{s}}}/\det g_{r,\bar{\boldsymbol{s}}}^{\mathrm{th}}\right)^{k/N_{\mathrm{inc}}}.
\end{equation}
Here, we use $g_{r,\bar{\boldsymbol{s}}}^{\mathrm{th}}$ to denote the Grover determinant corresponding to the rank-$r$ R\'{e}nyi entropy, i.e., by setting $A_1=A$ and $G^{T_2^f}=G$ in Eq.~\eqref{equ:detGrover_general}. 
In each incremental CPU, the Monte-Carlo involves the following factor:
\begin{equation}
    \frac{Z_{k+1}}{Z_{k}}=\frac{\sum_{\mathbf{s}^{(1)}\cdots\mathbf{s}^{(r)}}W_{\mathbf{s}^{(1)}\cdots\mathbf{s}^{(r)}}^{\left(k\right)}\left(\det g_{r,\bar{\boldsymbol{s}}}/\det g_{r,\bar{\boldsymbol{s}}}^{\mathrm{th}}\right)^{1/N_{\mathrm{inc}}}}{\sum_{\mathbf{s}^{(1)}\cdots\mathbf{s}^{(r)}}W_{\mathbf{s}^{(1)}\cdots\mathbf{s}^{(r)}}^{\left(k\right)}}, 
\end{equation}
where the total weight is given by
\begin{equation}
    W_{\mathbf{s}^{(1)}\cdots\mathbf{s}^{(r)}}^{\left(k\right)}=W_{\mathbf{s}^{(1)}\cdots\mathbf{s}^{(r)}}\left(\det g_{r,\bar{\boldsymbol{s}}}^{\mathrm{th}}\right)^{1-k/N_{\mathrm{inc}}}\left(\det g_{r,\bar{\boldsymbol{s}}}\right)^{k/N_{\mathrm{inc}}}. 
\end{equation}
In practice, it is essential to perform the same procedures for $\det g_{r,\bar{\boldsymbol{s}}}^{\mathrm{th}}$ as for $\det g_{r,\bar{\boldsymbol{s}}}$, including both updates and stabilization. 
When evaluating the update ratio, particular attention is required to properly account for the distinct root orders specified in the preceding equations.
These formulas address only the untwisted case; similar principles apply to the twisted case.

%%%%%%%%%%%%%%%%%%%%%%%%%%%%%%%%%%%%%%%%%%%%%%%%%%%%%
\section{Model simulation: benchmarks and results}\label{sec:examples}
%%%%%%%%%%%%%%%%%%%%%%%%%%%%%%%%%%%%%%%%%%%%%%%%%%%%%
\subsection{Quantum-classical crossover in the Hubbard chain}\label{sec:Hubbard_chain}

\begin{figure*}
    \center
    \includegraphics[width=0.45\textwidth]{fig2.pdf}
    \caption{
        The variations of high-rank RNs as functions of subsystem length $L_{A_1}$ for a half-filled Hubbard chain with $U/t=1$ and $L_A=6$ is calculated using both DQMC and exact diagonalization (ED). 
        The bipartition geometry follows the inset of Fig.~\ref{fig:free_fermion_standard}. 
        For the untwisted RNs, the ED and DQMC results are obtained using Eq.~\eqref{equ:FPT} and Eq.~\eqref{equ:detGrover_Drut_method}, respectively. Analogous formulas are applied for the calculation of the twisted RNs. 
        We note that for $\beta t=2.0$, the DQMC data can alternatively be calculated directly from Eq.~\eqref{equ:logdetGrover_loh}. 
        Additionally, in the zero-temperature limit ($\beta t=\infty$), employing Eq.~\eqref{equ:logdetGrover_loh} with slightly regularized $G^{T_2^f}$ and $G^{\tilde{T}_2^f}$ matrices ($\Lambda=10^{-13}$) yields results that are consistent within the error bars. 
        For odd $L_{A_1}$, $\tilde{\mathcal{E}}_3$ is ill-defined, as $\mathrm{Tr}[(\rho^{\tilde{T}_2^f})^3]$ is negative (see Eq.~\eqref{equ:sign_Dirac_ZT2f} and the related discussion on the sign problem of the Grover determinants). 
    }
    \label{fig:Hubbard_chain_ED_DQMC}
\end{figure*}

\begin{figure*}[htbp]
    \centering
    \subfloat[]{
        \includegraphics[width=0.45\textwidth]{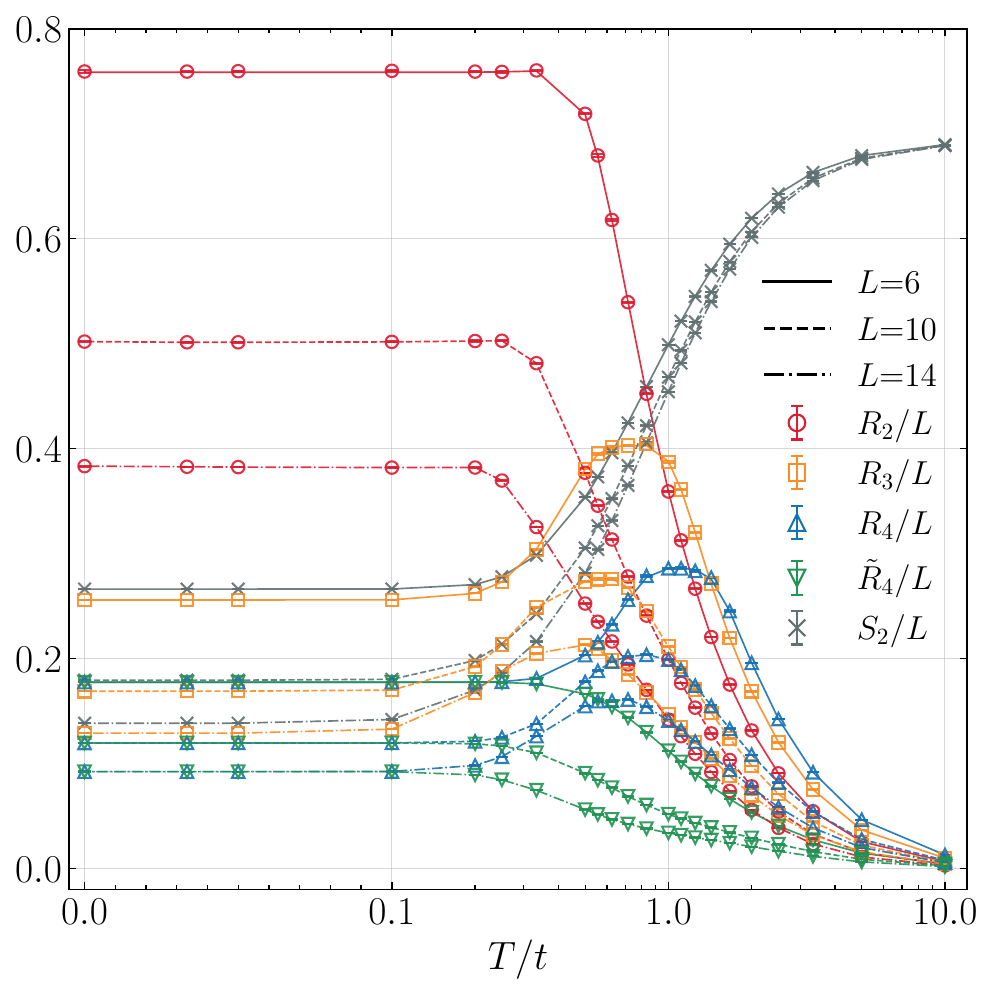}
    }
    \hfill
    \subfloat[]{
        \includegraphics[width=0.45\textwidth]{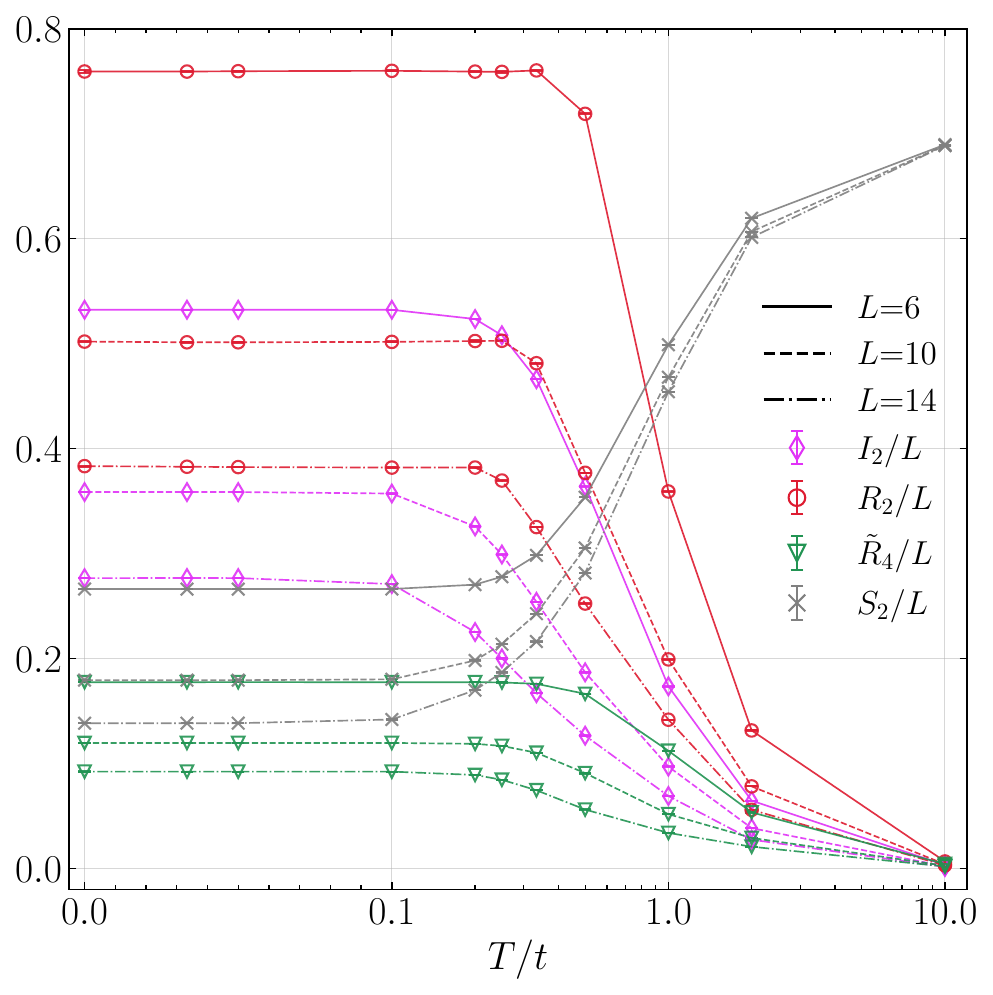}
    }
    \caption{
        The quantum-classical crossover in the Hubbard chain.
        (a) The rescaled rank-2 REE $S_2/L$, rank-2 to rank-4 untwisted RNRs $R_{2,3,4}/L$, and rank-4 twisted RNR $\tilde{R}_4/L$ as functions of temperature.
        The rank-2 untwisted and rank-4 twisted RNRs decrease with increasing temperature, whereas the rank-3 and rank-4 untwisted RNRs exhibit a local maximum at intermediate temperatures.
        (b) Comparison between the rescaled rank-2 R\'{e}nyi mutual information $I_2/L$, rank-2 REE $S_2/L$, rank-2 untwisted RNR $R_2/L$, and rank-4 twisted RNR $\tilde{R}_4/L$ for the same grand-canonical simulations.
        The quantities $I_2/L$, $R_2/L$, and $\tilde{R}_4/L$ all decrease monotonically and approach zero as temperature increases, whereas $S_2/L$ is dominated by thermal entropy.
        The subsystem length is set to $L_{A_1}=L/2$.
        Error bars are smaller than the symbol size.
        The solid lines for $L=6$ are derived from ED calculations. For larger $L$, the dashed and dash-dot lines simply connect the DQMC data points.
    }
    \label{fig:Hubbard_chain_crossover}
\end{figure*}

As a first example, we apply the methods in Sec.~\ref{sec:numerical_stable} within the framework of normal DQMC (without using incremental algorithm) to a half-filled Hubbard chain, whose Hamiltonian is expressed as
\begin{equation}
    \hat{H}=-t\sum_{i,\sigma=\uparrow,\downarrow}(c^\dagger_{i,\sigma}c_{i+1,\sigma}+\text{H.c.})+\frac{U}{2}\sum_{i}(\hat{n}_{i}-1)^2,
\end{equation}
where $\hat{n}_i=c^\dagger_{i\uparrow}c_{i\uparrow}+c_{i\downarrow}^\dagger c_{i\downarrow}$. 
As shown in Fig.~\ref{fig:Hubbard_chain_ED_DQMC}, both the untwisted and twisted RNs have been successfully validated against the exact diagonalization (ED) results at both finite and zero temperatures. 
The singularity of the partially transposed Green's functions depends on the model and the HS transformation. 
For the Hubbard model with $U/t=1.0$, practical experience indicates that Eq.~\eqref{equ:logdetGrover_loh} remains stable up to approximately $\beta t=20.0$.
At lower temperatures, unstable large values of the Grover determinant may appear randomly, depending on the specific $\mathbf{s}$-configuration. 
However, either approach—using Eq.~\eqref{equ:logdetGrover_loh} with a regularized $G^{T_2^f}$ or using Eq.~\eqref{equ:detGrover_Drut}—yields the same accurate results within the error bars.
For the former approach, a regularization parameter of $\Lambda=10^{-13}$ suffices to achieve the desired accuracy\footnote{While the real part of Monte-Carlo average $\langle \det g_r\rangle$ is sufficiently accurate, there may be an unexpected non-zero imaginary part on the order of $10^{-6}$ to $10^{-8}$, which should be disregarded. }.
% \fhw{For temperatures higher than $\beta t=4.0$, the RNs are large, causing the Grover determinants to become exponentially small, which makes accurate measurement challenging. 
% This issue represents the second challenge in computing high-rank RNs, which we will address using the incremental algorithm in the next section. }
We also note that, similarly to the free-fermion results in Fig.~\ref{fig:free_fermion_standard}, the ground-state values of untwisted and twisted rank-4 RNs, exhibited in Fig.~\ref{fig:Hubbard_chain_ED_DQMC}(b), are exactly identical, which aligns with the analytical result in Eq.~\eqref{equ:moment_4k_with_parity}. 

Next, we investigate the finite-temperature dependence of the RNRs in the Hubbard chain. 
In this paper, we further reveal that the rank-4 twisted RNR exhibits behavior similar to the rank-2 untwisted RNR studied in Ref.~\cite{Nat.Commun.2025Wanga}. 
Fig.~\ref{fig:Hubbard_chain_crossover}(a) shows the rescaled behaviors of rank-2 REE ($S_2/L$), rank-2 to rank-4 untwisted RNRs ($R_{2,3,4}/L$), and rank-4 twisted RNR ($\tilde{R}_4/L$) as functions of temperature. 
The monotonic decrease of both $R_2/L$ and $\tilde{R}_4/L$ with increasing temperature highlights the quantum-classical crossover in the Hubbard chain~\cite{Nat.Commun.2025Wanga,J.Stat.Mech.2019Shapourian}. 
This finding not only provides a physical interpretation for the definition of RNR but also illustrates why the RNR may serve as a more reliable indicator of mixed-state entanglement compared to the REE, which tends to follow a volume law as temperature rises. 
Remarkably, $R_3$ and $R_4$ exhibit a local maximum at intermediate temperatures, contrasting with the monotonic decrease observed for $R_2$ and $\tilde{R}_4$. 
These results highlight the nuanced definition of untwisted RNR and suggest that certain ranks may not serve as reliable indicators of entanglement.

\subsection{Finite-temperature phase transition in the spinless $t$-$V$ model}\label{sec:tVmodel}

\begin{figure*}[htbp]
    \center
    \includegraphics[width=\textwidth]{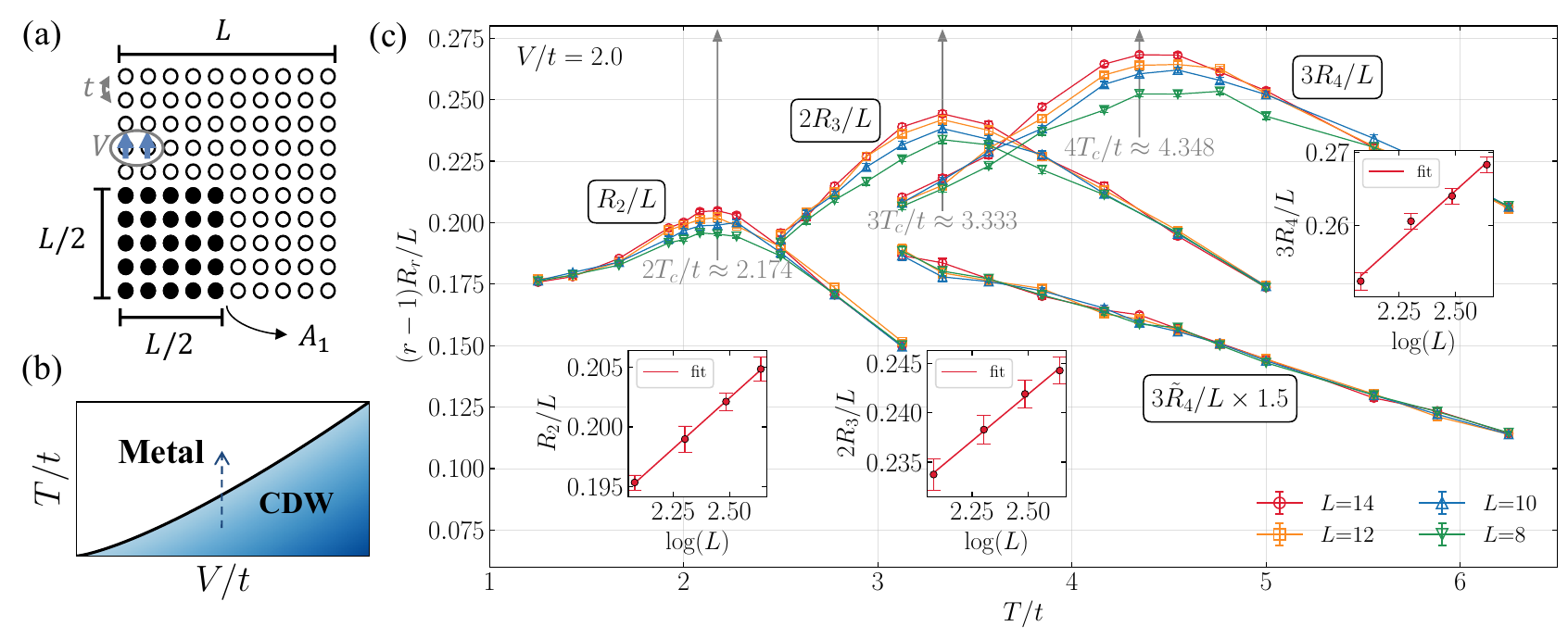}
    \caption{
        The behavior of different RNRs across the finite-temperature transition point of the spinless $t$-$V$ model. 
        (a) Illustration of the spinless $t$-$V$ model on a square lattice and our chosen bipartite geometry. 
        The model includes hopping and interaction terms between nearest neighboring sites.  
        We select the subsystem $A_1$ as the lower-left corner of the lattice, with a boundary length $L$ shared with $A_2$. 
        (b) Illustration of the phase diagram for the $t$-$V$ model. 
        The blue dashed arrow indicates our parameter path in (c). 
        (c) The temperature dependence of the area-law coefficients for rank-2, 3, and 4 untwisted RNRs, as well as the rank-4 twisted RNR, is illustrated. 
        To facilitate a direct comparison of values across different ranks and enhance the clarity of the plot, we present $(r-1)R_r=-\ln\{\mathrm{Tr}[(\rho^{T_2^f})^r]/\mathrm{Tr}[\rho^r]\}$ instead. 
        Additionally, the twisted RNR is scaled by a factor of 1.5 for improved visualization. 
        The gray arrows indicate the finite-temperature transition points identified by the untwisted RNRs at various ranks, with insets displaying the finite-size scaling of their respective area-law coefficients.
    }
    \label{fig:tVmodel}
\end{figure*}

Now we apply our incremental algorithm for arbitrary-rank RNRs (detailed in Sec.~\ref{sec:incremental}) to the half-filled spinless $t$-$V$ model on a square lattice with periodic boundary conditions, as illustrated in Fig.~\ref{fig:tVmodel}(a)~\cite{Phys.Rev.B1984Scalapino,Phys.Rev.B1985Gubernatis,Phys.Rev.Lett.2014Wang}, 
\begin{equation}
    H=-t\sum_{\langle i,j\rangle}(c_{i}^{\dagger}c_{j}+c_{j}^{\dagger}c_{i})+V\sum_{\langle i,j\rangle}\left(n_{i}-\frac{1}{2}\right)\left(n_{j}-\frac{1}{2}\right), 
\end{equation}
where both the hopping $t$ and the interaction $V$ involve only nearest neighbors. 
This model has been proved to be sign-problem-free model~\cite{Phys.Rev.B2014Huffman,Phys.Rev.B2015Li,Phys.Rev.Lett.2015Wang,Phys.Rev.Lett.2016Wei,Phys.Rev.Lett.2016Li}, and its Grover determinants of various ranks are also sign-problem-free~\cite{Nat.Commun.2025Wanga}. 
When a finite coupling $V$ is present, this model exhibits a charge density wave (CDW) ground state and undergoes a transition from the CDW phase to a metallic phase at a finite temperature, as illustrated in Fig.~\ref{fig:tVmodel}(b). 
The critical behavior of this transition is governed by the 2D Ising universality class~\cite{Phys.Rev.B1985Gubernatis,Phys.Rev.B2016Hesselmann}.
Here, we concentrate on a specific coupling strength of $V/t=2$, where the critical temperature has been estimated to be $T_c/t\approx 1.0$~\cite{Phys.Rev.B1985Gubernatis}. 
This temperature falls within the range where the singularity of the partially transposed Green's function can be handled stably, thus ensuring the feasibility of our simulation. 

As depicted in Fig.~\ref{fig:tVmodel}(c), the untwisted RNRs of different ranks exhibit similar behavior as temperature increases. 
Remarkably, the area-law coefficient of rank-$r$ untwisted RNR displays a peak at the replicated critical temperature $T=rT_c$ for various system sizes.
This phenomenon arises from the replica trick, $\rho^r\sim e^{-rH/T}$, which effectively elevates the phase transition temperature to $r$ times its original value.
The logarithmic divergence of $R_r/L$ with respect to the linear size $L$, as unveiled in Ref.~\cite{Nat.Commun.2025Wanga} for rank 2, is also observed for higher ranks.
This result qualitatively aligns with findings from studies on the RNR of bosonic systems near finite-temperature transition points~\cite{Phys.Rev.Lett.2020Wu,Phys.Rev.B2019Lu,Phys.Rev.Research2020Lu,Phys.Rev.B2025Ding}, in that singular behaviors for both bosons and fermions are observed at $rT_c$.
However, the quantitative behavior of the singularity differs significantly.
In the 2D transverse field Ising model, which features a finite-temperature transition also belonging to the 2D Ising universality class, the RNR monotonically decreases, and the singularity was identified in its temperature derivative~\cite{Phys.Rev.Lett.2020Wu,Phys.Rev.Research2020Lu}.
In contrast, in the present fermionic case, the singularity of untwisted RNR is clearly observed directly, without the need for a derivative.
Furthermore, perhaps the most intriguing observation in Fig.~\ref{fig:tVmodel}(c) is that the rank-4 twisted RNR $\tilde{R}_4$ more closely resembles the bosonic case: it monotonically decreases and obeys the area-law within error bars~\cite{Phys.Rev.Lett.2020Wu,Phys.Rev.B2025Ding}.

\begin{figure*}[htbp]
    \centering
    \includegraphics[width=0.7\textwidth]{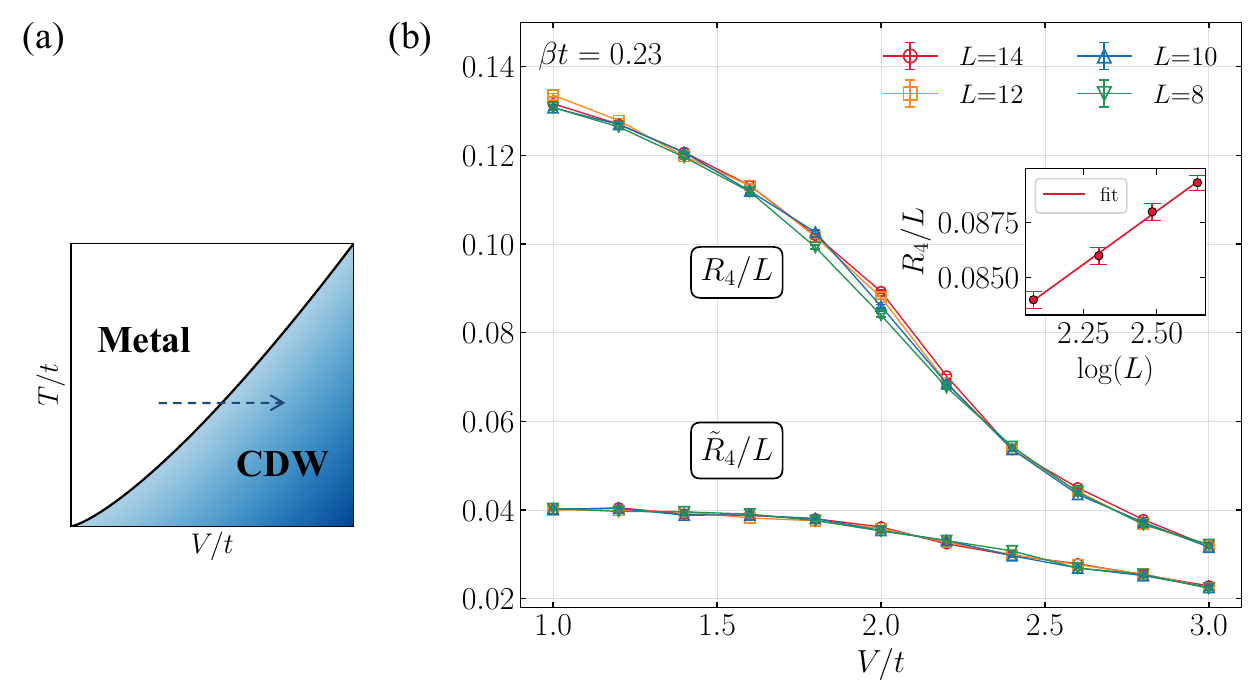}
    \caption{
        Complementary fixed-temperature interaction scan in the spinless $t$-$V$ model.
        (a) Schematic horizontal cut across the phase boundary in the $T/t$-$V/t$ phase diagram.
        (b) Rank-4 untwisted and twisted RNRs at fixed $\beta t=0.23$ ($T/t\approx 4.35$), scanning $V/t$ from 1 to 3 for $L=8,10,12,14$.
        The inset is an enlarged view of the untwisted $R_4/L$ data at $V/t=2$.
    }
    \label{fig:tVmodel_horizontalcut}
\end{figure*}

To complement the temperature scans at fixed $V/t=2$, we also consider a horizontal cut of the phase diagram by fixing $\beta t=0.23$ ($T/t\approx 4.35$) and scanning $V/t$ from 1 to 3 for $L=8,10,12,14$.
As shown in Fig.~S9, the untwisted rank-4 RNR $R_4$ exhibits evident beyond-area-law behavior near the phase boundary along this horizontal cut, and the inset at $V/t=2$ highlights the corresponding $L\ln L$ behavior at criticality.
By contrast, the twisted rank-4 RNR $\tilde{R}_4/L$ decreases with increasing $V/t$, while the data for different system sizes nearly collapse within error bars, indicating that the area law remains valid for the twisted RNR along the same cut.
This untwisted-versus-twisted contrast is consistent with the direction-independent critical picture of the underlying classical Ising transition.
Notably, increasing $V$ is not equivalent to lowering $T$ at fixed $t$, because $V$ and $T$ are independent energy scales and larger $V$ drives the system from Fermi surface with $L\ln L$ entanglement~\cite{Phys.Rev.Lett.2006Gioev} toward the CDW/product-state limit, thereby reducing the bipartite entanglement.

\subsection{Discussion on R\'{e}nyi proxy}

\begin{table*}[htbp]
    \centering
    \caption[Renyi Proxy]
        {Comparison between untwisted and twisted rank-$r$ RNRs as R\'{e}nyi proxies for LN. 
        The second row summarizes temperature dependence in the Hubbard chain [see Fig.~\ref{fig:Hubbard_chain_crossover}]. 
        Subsequent rows labeled ``spinless $t$-$V$ model'' summarize behaviors across the finite-temperature transition in the 2D Ising universality class at $T_c$ [see Fig.~\ref{fig:tVmodel}]. 
        The symbols $\uparrow$ and $\downarrow$ indicate that the RNR monotonically increases and decreases, respectively, with increasing temperature.}
    \label{tab:renyi_proxy_summary}
    \begin{tabular}{C{0.12\textwidth}C{0.28\textwidth}C{0.25\textwidth}C{0.25\textwidth}}
    \toprule
    \multicolumn{2}{c}{Scope} & Untwisted RNR & Twisted RNR \\
    \midrule
    \multirow{2}{*}{\parbox[c]{0.40\textwidth}{\centering Relation to LN by definition}} & & \multirow{2}{*}{Do not relate to LN explicitly} & Analytically continued to LN: \\
    & & & $\mathcal{E}=\lim_{r\to 1/2}(1-2r)\tilde{\mathcal{E}}_{2r}$ \\
    \midrule
    \multirow{2}{*}{Hubbard chain} & \multirow{2}{*}{Variation as $T$} & $r=2$: $\downarrow$ & \multirow{2}{*}{$r=4$: $\downarrow$} \\
    & & $r=3,4$: non-monotonic &  \\
    \midrule
    \multirow{2}{*}{Spinless $t$-$V$ model} & Variation as $T$ across $rT_c$ & \multicolumn{1}{c}{$\uparrow\,\rightarrow\,\text{peak at } rT_c\,\rightarrow\,\downarrow$} & $\downarrow$ \\
    & Finite-size scaling across $rT_c$ & \multicolumn{1}{c}{area-law $\rightarrow$ beyond-area-law $\rightarrow$ area-law} & area-law \\
    \bottomrule
    \end{tabular}%
\end{table*}%

The above findings offer insights into the nature of untwisted and twisted RNRs, neither of which has been rigorously established as an entanglement measure.
Table~\ref{tab:renyi_proxy_summary} summarizes the key differences in definition and temperature dependence between untwisted and twisted RNRs. 
The primary concern with untwisted RNR stems from the non-Hermitian nature of the untwisted PTDM $\rho^{T_2^f}$, whose moments lack a clear connection to the LN.
In contrast, the twisted RNR is arguably more physically meaningful, as the twisted PTDM $\rho^{\tilde{T}_2^f}$ is Hermitian, and even-rank twisted RNs can be analytically continued to the LN (see the discussion below Eq.~\eqref{equ:def_Renyinegativity}).
The bosonic RNR, defined using the conventional partial transpose~\cite{Phys.Rev.Lett.2012Calabrese,Phys.Rev.B2014Chunga,Phys.Rev.Lett.2020Wu}, shares this analytical continuation feature with twisted RNR, and their behaviors across the 2D Ising transition are indeed consistent.
Moreover, the rank-3 and -4 untwisted RNRs in the Hubbard chain exhibit anomalous temperature dependence.
Collectively, both definitional considerations and numerical evidence support the view that twisted RNR serves as a more appropriate R\'{e}nyi proxy for LN, analogous to how REE relates to EE.

\subsection{Comparison with mutual information}

It is also instructive to compare R\'{e}nyi negativities with the R\'{e}nyi mutual information, defined as $I_n(A_1:A_2) = S_n(A_1) + S_n(A_2) - S_n(A_1\cup A_2)$, a well-established probe of finite-temperature correlations.
Table~\ref{tab:mi_rnr_comparison} summarizes the temperature dependence and finite-size scaling of three diagnostics: rank-2 untwisted RNR, rank-4 twisted RNR, and rank-2 R\'{e}nyi mutual information.

In the Hubbard chain, the rank-2 R\'{e}nyi mutual information serves as a more appropriate finite-temperature reference than $S_2$, as it removes the bulk thermal entropy.
In the grand-canonical ensemble, the mutual information is expected to decrease monotonically with increasing temperature~\cite{PRB2013L.Bonnes}, which also reproduced in our calculations in Fig.~\ref{fig:Hubbard_chain_crossover}(b).
This behavior, shared with the twisted RNR, captures the overall thermal trend in the quantum-classical crossover scenario.

In the spinless $t$-$V$ model, the distinction becomes pronounced near the finite-temperature Ising transition.
The area-law coefficient of R\'{e}nyi mutual information features a well-known finite-size crossing pattern~\cite{PRL2011RRP.Singh,PRB2013J.Iaconis,Phys.Rev.Lett.2014Wang}: curves of $I_2/L$ versus $V$ for different system sizes $L$ cross at the critical point $T_c$ and the replicated critical point $rT_c$, where the area law holds exactly with vanishing constant term, while away from these points a non-zero constant term emerges and changes sign across them.
This behavior differs from both negativity variants studied this work.
The twisted RNR remains area-law-like and decreases monotonically across the transition, whereas the untwisted RNR develops a replicated-critical peak accompanied by a beyond-area-law $L\ln L$ enhancement.

\begin{table*}[htbp]
    \centering
    \caption[Mutual Information vs RNR]
        {Comparison of two candidate R\'enyi proxies for logarithmic negativity, rank-2 untwisted RNR and rank-4 twisted RNR, with the rank-2 R\'{e}nyi mutual information.
        The Hubbard-model column summarizes the temperature dependence in grand-canonical simulations [see Fig.~\ref{fig:Hubbard_chain_crossover}(b) and Ref.~\cite{PRB2013L.Bonnes}].
        The spinless $t$-$V$ columns summarize the behavior across the finite-temperature transition [see Fig.~\ref{fig:tVmodel}], along with the corresponding mutual-information literature~\cite{PRL2011RRP.Singh,PRB2013J.Iaconis,Phys.Rev.Lett.2014Wang}.
        The symbols $\uparrow$ and $\downarrow$ indicate monotonic increase and decrease, respectively, with increasing temperature.}
    \label{tab:mi_rnr_comparison}
    \begin{tabular}{C{0.17\textwidth}C{0.14\textwidth}C{0.2\textwidth}C{0.40\textwidth}}
    \toprule
    \multirow{2}{*}{\parbox[c]{0.17\textwidth}{\centering Observable}} & Hubbard model & \multicolumn{2}{c}{Spinless $t$-$V$ model} \\
    \cmidrule(lr){2-2}\cmidrule(lr){3-4}
    & Variation as $T$ & Variation as $T$ across $rT_c$ & Finite-size scaling across $rT_c$ \\
    \midrule
    Untwisted RNR ($r=2$) & $\downarrow$ & $\uparrow\,\rightarrow\,\text{peak at } rT_c\,\rightarrow\,\downarrow$ & area-law $\rightarrow$ beyond-area-law $\rightarrow$ area-law \\
    \midrule
    Twisted RNR ($r=4$) & $\downarrow$ & $\downarrow$ & area-law \\
    \midrule
    R\'{e}nyi mutual information ($r=2$) & $\downarrow$~\cite{PRB2013L.Bonnes} & $\downarrow$~\cite{PRB2013J.Iaconis} & constant term: negative $\rightarrow$ zero $\rightarrow$ positive (crossing of the area-law coefficient)~\cite{PRL2011RRP.Singh,PRB2013J.Iaconis,Phys.Rev.Lett.2014Wang} \\
    \bottomrule
    \end{tabular}%
\end{table*}%

%%%%%%%%%%%%%%%%%%%%%%%%%%%%%%%%%%%%%%%%%%%%%%%%%%%%%
\section{Conclusions and Outlook}\label{sec:conclusion}
%%%%%%%%%%%%%%%%%%%%%%%%%%%%%%%%%%%%%%%%%%%%%%%%%%%%%
In summary, to facilitate the calculation of high-rank RNs in interacting fermionic systems using DQMC, we have (i) identified stable formulas for calculating high-rank Grover determinants and (ii) developed stable update schemes for incremental algorithms for untwisted and twisted RNs of arbitrary rank.
We applied these methods to study the Hubbard model and spinless $t$-$V$ model, and presented the temperature-dependent behavior of high-rank RNRs, thereby expanding upon the previous results presented in Ref.~\cite{Nat.Commun.2025Wanga}.
We found that high-rank RNRs, whether untwisted or twisted, can exhibit significantly different behavior compared to the rank-2 untwisted RNR~\cite{Nat.Commun.2025Wanga}.
Specifically, rank-3 and -4 untwisted RNRs fail to accurately represent the quantum-classical crossover in the Hubbard chain.
Furthermore, the rank-4 twisted RNR displays behavior distinct from untwisted RNRs in the spinless $t$-$V$ model, with the former adhering to the area law while the latter exhibits beyond-area-law scaling around the critical points.
These results illuminate the nature of untwisted and twisted RNRs.
Due to drawbacks in its definition, the untwisted RNR exhibits unusual non-monotonic behaviors in certain cases, whereas the twisted RNR may serve as a more suitable R\'{e}nyi proxy for LN.

Several potential avenues exist for future research.
First and foremost, the nature of untwisted and twisted RNRs remains to be fully elucidated. 
Although lacking a clear connection to the LN, the untwisted RNR also captures the singularity at the critical point and displays intriguing beyond-area-law behavior in its vicinity.
A theoretical investigation of the rank-2 untwisted RN, $\mathcal{E}_2=-\ln\mathrm{Tr}[\rho X_2 \rho X_2]$, may provide a productive starting point.
Secondly, more case studies are needed to examine whether the twisted RNR faithfully captures mixed-state entanglement.
Our stable formulas for Grover determinants (e.g., Eqs.~\eqref{equ:logdetGrover_loh} and \eqref{equ:detGrover_Drut_method}) enable calculation of high-rank RNs for 1D interacting fermionic models across various bipartitions using direct sampling DQMC, thereby facilitating examination of theoretical predictions.
For instance, one could compare the scaling of twisted RN with $L_{A_1}$ for the Hubbard chain at finite temperature against conformal field theory predictions~\cite{J.Stat.Mech.2019Shapourian}.
Additionally, our stable incremental algorithms enable investigation of high-temperature behavior of high-rank RN(R)s in 2D interacting fermionic models, provided they remain free from the sign problem.
Finally, to extend these calculations to lower temperatures, a valuable technical advancement would be the development of an incremental algorithm, based on Drut's formula for high-rank Grover determinants, that completely avoids inverting the Green's function matrix, thereby enhancing computational stability.
Complementary to this, examining the validity of Green's function regularization within the incremental algorithm, particularly the systematic error it introduces, would offer another promising direction for low-temperature investigations.

%%%%%%%%%%%%%%%%%%%%%%%%%%%%%%%%%%%%%%%%%%%%%%%%%%%%%
\bibliography{../reference_increm}
%%%%%%%%%%%%%%%%%%%%%%%%%%%%%%%%%%%%%%%%%%%%%%%%%%%%%

\appendix

%%%%%%%%%%%%%%%%%%%%%%%%%%%%%%%%%%%%%%%%%%%%%%%%%%%%%
\section{Proof of the trace properties of the partially transposed density matrices}\label{app:prood_in_Fock_space}
%%%%%%%%%%%%%%%%%%%%%%%%%%%%%%%%%%%%%%%%%%%%%%%%%%%%%
Consider the $r$-moment of the twisted PTDM $\rho^{\tilde{T}_{2}^{f}}$: 
\begin{widetext}
    \begin{equation}
        \begin{aligned}
            &{\rm Tr}\left[\left(\rho^{\tilde{T}_{2}^{f}}\right)^{r}\right]={\rm Tr}\left[\rho^{T_{2}^{f}}X_{2}\right]\\=&\sum_{n^{1},\bar{n}^{1},n^{2},\bar{n}^{2},\dots,n^{r},\bar{n}^{r}}\left\langle n^{1}\left|\rho\right|\bar{n}^{1}\right\rangle \cdots\left\langle n^{r}\left|\rho\right|\bar{n}^{r}\right\rangle \left(-1\right)^{\phi\left(n^{1},\bar{n}^{1}\right)}\cdots\left(-1\right)^{\phi\left(n^{r},\bar{n}^{r}\right)}\\&\times\sum_{n^{0}}\langle n^{0}|n_{A_{1}}^{1},\bar{n}_{A_{2}}^{1}\rangle\langle\bar{n}_{A_{1}}^{1},n_{A_{2}}^{1}|X_{2}|n_{A_{1}}^{2},\bar{n}_{A_{2}}^{2}\rangle\langle\bar{n}_{A_{1}}^{2},n_{A_{2}}^{2}|X_{2}\cdots X_{2}|n_{A_{1}}^{r},\bar{n}_{A_{2}}^{r}\rangle\langle\bar{n}_{A_{1}}^{r},n_{A_{2}}^{r}|X_{2}|n^{0}\rangle\\=&\sum_{n^{1},n^{2},\dots,n^{r}}\left\langle n^{1}\left|\rho\right|n_{A_{1}}^{2},n_{A_{2}}^{r}\right\rangle \left\langle n^{2}\left|\rho\right|n_{A_{1}}^{3},n_{A_{2}}^{1}\right\rangle \cdots\left\langle n^{r}\left|\rho\right|n_{A_{1}}^{1},n_{A_{2}}^{r-1}\right\rangle \\&\times\left(-1\right)^{\phi\left(n^{1},\left\{ n_{A_{1}}^{2},n_{A_{2}}^{r}\right\} \right)}\left(-1\right)^{\phi\left(n^{2},\left\{ n_{A_{1}}^{3},n_{A_{2}}^{1}\right\} \right)}\cdots\left(-1\right)^{\phi\left(n^{r},\left\{ n_{A_{1}}^{1},n_{A_{2}}^{r-1}\right\} \right)}\\&\times\left(-1\right)^{\tau_{A_{2}}^{r}}\left(-1\right)^{\tau_{A_{2}}^{1}}\cdots\left(-1\right)^{\tau_{A_{2}}^{r-1}}.
        \end{aligned}
    \end{equation}
\end{widetext}
As for the untwisted moment, simply delete the $X_2$ phase strings in the last line of the above equation. After doing a change of variables, $n^{(i-1)\text{ mod }r}_{A_2}\rightarrow n^{(i+1)\text{ mod }r}_{A_2}$, we arrive at the desired results in Eq.~\eqref{equ:moment_untwisted} and Eq.~\eqref{equ:moment_twisted}. 

When $\rho$ has a well-defined parity (denoted as $P_\rho$), i.e., $[\rho,\hat{P}]=0$ with $\hat{P}=(-1)^{\sum_j c^\dagger_j c_j}$ the fermionic parity operator, there are additional constraints on the summation over dummy indices $\{n^{1},n^{2},\dots,n^{r}\}$. 
From the sandwich chain in the second line of Eq.~\eqref{equ:moment_untwisted} or \eqref{equ:moment_twisted}, we deduce that, in order to have non-vanishing contribution, $\forall i\in\{1,2,\dots,r\}$, (i) each of $n^i$ must be of parity $P_\rho$; (ii) each of $\tilde{n}^i=n^{i}_{A_1}\cup n^{i+2}_{A_2}$ must also be of parity $P_\rho$. 
Formulate these constraints we obtain, 
\begin{equation}
    \begin{aligned}
        P_{\rho}=&\left(-1\right)^{\tau^{1}}=\left(-1\right)^{\tau^{2}}=\cdots=\left(-1\right)^{\tau^{r}}\\=&\left(-1\right)^{\tau_{1}^{1}+\tau_{2}^{3}}=\left(-1\right)^{\tau_{1}^{2}+\tau_{2}^{4}}=\cdots=\left(-1\right)^{\tau_{1}^{r}+\tau_{2}^{2}}.
    \end{aligned}
\end{equation}
Thus we further obtain the constraints on the parity of subsystems, $\left(-1\right)^{\tau_{b}^{i}}=\left(-1\right)^{\tau_{b}^{\left(i+2\right)\text{ mod }r}}$ for any $b\in\{1,2\}$ and $i\in\{1,2,\dots,r\}$. 
For even-rank moments, the dummy bit strings are naturally split into two groups, $\{n^{2\mathbb{Z}-1}\}$ and $\{n^{2\mathbb{Z}}\}$, each of which is of homogeneous subsystem parity. For odd-rank moments, all the dummy bit strings should have the same subsystem parity. 
\begin{subequations}
\begin{equation}\label{equ:sub_parity_constraints_even}
    \begin{aligned}
        \text{even }k\quad&\begin{cases}
        \left(-1\right)^{\tau_{b}^{1}}=\left(-1\right)^{\tau_{b}^{3}}=\cdots\left(-1\right)^{\tau_{b}^{r-1}}\\
        \left(-1\right)^{\tau_{b}^{2}}=\left(-1\right)^{\tau_{b}^{4}}=\cdots\left(-1\right)^{\tau_{b}^{r}}
        \end{cases},
    \end{aligned}
\end{equation}
\begin{equation}\label{equ:sub_parity_constraints_odd}
    \begin{aligned}
        \text{odd }k\quad&\left(-1\right)^{\tau_{b}^{1}}=\left(-1\right)^{\tau_{b}^{2}}=\cdots\left(-1\right)^{\tau_{b}^{r}},b=1,2.
    \end{aligned}
\end{equation}
\end{subequations}
Specifically, if $\rho$ is even (odd), $P_\rho=1(-1)$, then the parity of $A_1$ is equal (opposite) to the parity of $A_2$ within each group. 
Now we can prove the Eq.~\eqref{equ:moment_4k_with_parity}. Consider a slightly more general case of rank-$2\mathbb{Z}$ moments, we subtract the untwisted moment from the twisted moment,
\begin{equation}
    \begin{aligned}
        &\mathrm{Tr}\left[\left(\rho^{\tilde{T}_{2}^{f}}\right)^{2m}\right]-{\rm Tr}\left[\left(\rho^{{T}_{2}^{f}}\right)^{2m}\right]\\=&\sum_{n^{1},\dots,n^{2m}}\left\langle n_{1}^{1},n_{2}^{3}\left|\rho\right|n^{2}\right\rangle \left\langle n_{1}^{2},n_{2}^{4}\left|\rho\right|n^{3}\right\rangle \cdots\left\langle n_{1}^{2m},n_{2}^{2}\left|\rho\right|n^{1}\right\rangle \\\times&\left(-1\right)^{\sum_{i=1}^{2m}\phi\left(n^{i},n^{\left(i+1\right)\text{ mod }2m}\right)}\\\times&\left[\left(-1\right)^{\sum_{j=1}^{m}\tau_{2}^{2j-1}+\sum_{j=1}^{m}\tau_{2}^{2j}}-1\right].
    \end{aligned}
\end{equation}
Here we abbreviate the subscripts of dummy bit strings, $n^i_b\equiv n^i_{A_b}$. 
When the total parities of the two groups (see Eq.~\eqref{equ:sub_parity_constraints_even}) are the same, then $(-1)^{\sum_j \tau_2^j}=1$, so these kind of terms are canceled out. We may resort to the terms with even (odd) $\sum_{j=1}^{m}\tau_{2}^{2j-1}$ and odd (even) $\sum_{j=1}^{m}\tau_{2}^{2j}$. However, these terms are forbidden by the parity constraints in Eq.~\eqref{equ:sub_parity_constraints_even} when $m$ is even. Thus we have proved the Eq.~\eqref{equ:moment_4k_with_parity}. 
Furthermore, the phase chain in the third line of Eq.~\eqref{equ:moment_untwisted}
\begin{equation}
    \begin{aligned}
        &\left(-1\right)^{\phi\left(n^{1},n^{2}\right)}\left(-1\right)^{\phi\left(n^{2},n^{3}\right)}\cdots\left(-1\right)^{\phi\left(n^{i},n^{i+1}\right)}\cdots\left(-1\right)^{\phi\left(n^{r},n^{1}\right)}\\&=\left(-1\right)^{\frac{\left[\left(\tau_{2}^{1}+\tau_{2}^{2}\right)\bmod2\right]}{2}+\frac{\left[\left(\tau_{2}^{2}+\tau_{2}^{3}\right)\bmod2\right]}{2}+\cdots+\frac{\left[\left(\tau_{2}^{r}+\tau_{2}^{1}\right)\bmod2\right]}{2}}\\&\times\left(-1\right)^{\left(\tau_{1}^{1}+\tau_{1}^{2}\right)\left(\tau_{2}^{1}+\tau_{2}^{2}\right)+\left(\tau_{1}^{2}+\tau_{1}^{3}\right)\left(\tau_{2}^{2}+\tau_{2}^{3}\right)+\cdots+\left(\tau_{1}^{r}+\tau_{1}^{1}\right)\left(\tau_{2}^{r}+\tau_{2}^{1}\right)}
    \end{aligned}
\end{equation}
can be simplified since the second line of the above equation is always unity. Indeed, when $P_\rho=1$, $\tau^i_1$ and $\tau^i_2$ have the same parity, thus $\left(-1\right)^{\sum_{i}\left(\tau_{1}^{i}+\tau_{1}^{i+1}\right)\left(\tau_{2}^{i}+\tau_{2}^{i+1}\right)}=\left(-1\right)^{\sum_{i}\left(\tau_{1}^{i}+\tau_{1}^{i+1}\right)^{2}}=1$. When $P_\rho=-1$, $\tau^i_1$ and $\tau^i_2$ have opposite parities, thus $\left(-1\right)^{\sum_{i}\left(\tau_{1}^{i}+\tau_{1}^{i+1}\right)\left(\tau_{2}^{i}+\tau_{2}^{i+1}\right)}=\left(-1\right)^{\sum_{i}\left(\tau_{1}^{i}+\tau_{1}^{i+1}\right)\left(\tau_{1}^{i}+\tau_{1}^{i+1}+2\right)}=1$. So we finally obtain the Eq.~\eqref{equ:phase_chain_with_parity} since $\left(-1\right)^{\left[\left(\tau+\tau^{\prime}\right)\bmod2\right]/2}=\left(-1\right)^{\left(\tau+\tau^{\prime}\right)^{2}/2}$. 

%%%%%%%%%%%%%%%%%%%%%%%%%%%%%%%%%%%%%%%%%%%%%%%%%%%%%
\section{Proof of Eq.~\eqref{equ:Gamma_W} using the product form of Gaussian states}\label{app:Gamma-W-relation}
%%%%%%%%%%%%%%%%%%%%%%%%%%%%%%%%%%%%%%%%%%%%%%%%%%%%%
For physical Gaussian states of the form in Eq.~\eqref{equ:Gaussian_state_Majorana} with Hermitian $W$, there exists a canonical transformation to bring $\rho_0$ into the product form (or standard form)~\cite{Quant.Inf.Comput.2005Bravyi,ArXiv2005Bravyi}, given by
\begin{equation}\label{equ:Gaussian_product_form}
    \begin{aligned}
        \rho_{0}&=\frac{1}{Z}\exp\left(\sum_{j}\frac{w_{j}}{2}\mathrm{i}\tilde{\gamma}_{2j-1}\tilde{\gamma}_{2j}\right)\\&=\frac{1}{Z}\prod_{j}\left[\cosh\frac{w_{j}}{2}+\sinh\frac{w_{j}}{2}\mathrm{i}\tilde{\gamma}_{2j-1}\tilde{\gamma}_{2j}\right]\\&=\frac{1}{2^{N}}\prod_{j}\left[1+\mathrm{i}\lambda_{j}\tilde{\gamma}_{2j-1}\tilde{\gamma}_{2j}\right],
        \end{aligned}
\end{equation}
where $\lambda_j=\tanh\frac{w_j}{2}$ with $\{\pm w_j\}$ the eigenvalues of the $W$ matrix, and $\{\tilde{\gamma}_j\}$ is related to the original basis via $\tilde{\boldsymbol\gamma}=O^T\boldsymbol{\gamma}$ with $O$ a $2N\times 2N$ real orthogonal matrix satisfying $OO^T=I$. 
Indeed, since a real antisymmetric matrix $W_0$ can always be transformed to its real canonical form via a real orthogonal transformation $O$, 
\begin{equation}\label{equ:real_canonical_form}
    O^TW_{0}O=\bigoplus_{i=1}^{N}\left(\begin{array}{cc}
        0 & w_{i}\\
        -w_{i} & 0
        \end{array}\right),
\end{equation}
the product form in Eq.~\eqref{equ:Gaussian_product_form} is always available for Hermitian $W=\mathrm{i}W_0$, and each $w_j$ is real and positive. 
Each block in Eq.~\eqref{equ:real_canonical_form} can be written as $\mathrm{i}w_i\sigma_y$ and be further diagonalized to obtain $W_0$'s eigenvalues $\{\pm \mathrm{i}w_i\}$, thus $\{\pm w_i\}$ are the eigenvalues of $W$.
To be more specific, let us denote the unitary matrix that diagonalizes the Pauli matrix $\sigma_y=\left(\begin{array}{cc}
    0 & -\mathrm{i} \\
    \mathrm{i} & 0
    \end{array}\right)$ as $V_y\equiv\frac{1}{\sqrt{2}}\left(\begin{array}{cc}
    -\mathrm{i} & -\mathrm{i}\\
    1 & -1
    \end{array}\right)$, then $V_{y}^{\dagger}\sigma_{y}V_{y}=\sigma_z$.
Therefore, $W$ can be diagonalized by
\begin{equation}\label{equ:diagonalize_W}
    U_{W}^{\dagger}WU_{W}=-\bigoplus_{i=1}^{N}\left(\begin{array}{cc}
        w_{i} & 0\\
        0 & -w_{i}
        \end{array}\right)\equiv-\Lambda,
\end{equation}
where $U_{W}\equiv OU_{y}$ and $U_{y}\equiv\oplus_{i=1}^{N}V_{y}$. 
Next, it is easy to evaluate the covariance matrix $\Gamma$ using the product form~\eqref{equ:Gaussian_product_form} under the new basis $\tilde{\gamma}_{k}=\gamma_{l}O_{lk}$: 
\begin{equation}
    \begin{aligned}
        \mathrm{i}\Gamma_{kl}\equiv&\frac{\mathrm{i}}{2}\langle[\gamma_{k},\gamma_{l}]\rangle=\frac{\mathrm{i}}{2}\left\langle \left[O_{kk^{\prime}}\tilde{\gamma}_{k^{\prime}},O_{ll^{\prime}}\tilde{\gamma}_{l^{\prime}}\right]\right\rangle \\=&\sum_{k^{\prime},l^{\prime}=1}^{2N}O_{kk^{\prime}}O_{ll^{\prime}}\frac{\mathrm{i}}{2}\left(\left\langle \tilde{\gamma}_{k^{\prime}}\tilde{\gamma}_{l^{\prime}}\right\rangle -\left\langle \tilde{\gamma}_{l^{\prime}}\tilde{\gamma}_{k^{\prime}}\right\rangle \right)\\=&\sum_{i=1}^{N}\left(O_{k,2i-1}O_{l,2i}-O_{k,2i}O_{l,2i-1}\right)\mathrm{i}\left\langle \tilde{\gamma}_{2i-1}\tilde{\gamma}_{2i}\right\rangle \\=&\left(O\left[\bigoplus_{i=1}^{N}\left(\begin{array}{cc}
        0 & \tanh\left(\frac{w_{i}}{2}\right)\\
        -\tanh\left(\frac{w_{i}}{2}\right) & 0
        \end{array}\right)\right]O^{T}\right)_{kl}.
        \end{aligned}
\end{equation}
We note that only the diagonal correlations have non-zero values under $\tilde{\boldsymbol{\gamma}}$-basis, with $\langle \mathrm{i}\tilde{\gamma}_{2i-1}\tilde{\gamma}_{2i}\rangle =2\langle \tilde{c}_{i}^{\dagger}\tilde{c}_{i}\rangle -1=\lambda_j$. 
Since $\mathrm{i}\Gamma$ has similar real canonical form as $W_0$, one finds that $\Gamma$ can also be diagonalized by $U_{W}$,
\begin{equation}\label{equ:diagonalize_Gamma}
    \begin{aligned}
        U_{W}^{\dagger}\Gamma U_{W}&=\left(-\mathrm{i}\right)U_{y}^{\dagger}\left[\bigoplus_{i=1}^{N}\left(\begin{array}{cc}
        0 & \tanh\frac{w_{i}}{2}\\
        -\tanh\frac{w_{i}}{2} & 0
        \end{array}\right)\right]U_{y}\\&=\bigoplus_{i=1}^{N}\left(\begin{array}{cc}
        \tanh\frac{w_{i}}{2}\\
         & -\tanh\frac{w_{i}}{2}
        \end{array}\right)=\tanh\frac{\Lambda}{2}.
        \end{aligned}
\end{equation}
Combine Eq.~\eqref{equ:diagonalize_W} and Eq.~\eqref{equ:diagonalize_Gamma}, we obtain the relation in Eq.~\eqref{equ:Gamma_W}. 
From this relation, one finds that the $\{\pm\lambda_j\}$ appeared in the product form~\eqref{equ:Gaussian_product_form} of Gaussian states are the eigenvalues of the covariance matrix. 

The essence of the above proof lies in Eq.~\eqref{equ:real_canonical_form} where one finds a real orthogonal basis transformation, which is a canonical transformation preserving the anti-commutation relation between Majorana operators, $\{\gamma_i,\gamma_j\}=2\delta_{ij}$, and the self-adjoint property of Majorana fermions, $\gamma_i^\dagger=\gamma_i$~\cite{Ripka1986QuantumTheoryFinite,J.Stat.Mech.2014Klich}. 
This proof procedure can be straightforwardly generalized to some special non-Hermitian $W$. 
It is easy to realize that all $W=cW_0$ with $c$ an arbitrary complex number can be proved by the same procedure, but with complex $w_j$. 
In fact, for any complex and normal antisymmetric $W$ satisfying $WW^\dagger=W^\dagger W$ (which includes $W=cW_0$), since its real and imaginary parts can be separately transformed to the real canonical form using the same $O$~\cite{Horn1985MatrixAnalysis}, the product form is available but with complex $w_j$. 

For a generic non-Hermitian $W$, the real orthogonal transformation $O$ and thus the product form~\eqref{equ:Gaussian_product_form} are not always available. 
However, the antisymmetric property of $W$ still provides us with some mathematically similar results. 
In particular, it turns out that Eq.~\eqref{equ:Gamma_W} holds for arbitrary diagonalizable $W$~\cite{J.Stat.Mech.2010Fagotti}. 
Below, we outline a brief review of the proof. 
First of all, we note that the eigenvalues of any $2N\times 2N$ antisymmetric matrix $W$ are in pairs of $\pm w_j$ where $j=1,2,\dots,N$. 
Denote the corresponding eigenvector to eigenvalue $w$ as $\vec{v}_{w}$ and for simplicity assume that $W$'s spectrum is non-degenerate\footnote{If $W$ has degenerate eigenvalues, one can always redefine the eigenvectors inside this degenerate subspace to make them ``orthogonal'' to each other.}, 
there is an ``orthogonal relation'' between the eigenvectors, $\vec{v}_{w}\cdot\vec{v}_{w^{\prime}}=\delta_{w+w^{\prime}}\vec{v}_{w}\cdot\vec{v}_{-w}$. 
Both $W$ and functions of $W$ can be decomposed using $W$'s eigenvectors: 
\begin{gather}
    \label{equ:decompose_W}
    W=\sum_{j=1}^{N}w_{j}\frac{\vec{v}_{w_{j}}\otimes\vec{v}_{-w_{j}}-\vec{v}_{-w_{j}}\otimes\vec{v}_{w_{j}}}{\vec{v}_{w_{j}}\cdot\vec{v}_{-w_{j}}},
    \\
    \label{equ:decompose_fW}
    f(W)=\sum_{j=1}^{N}\frac{f(w_{j})\vec{v}_{w_{j}}\otimes\vec{v}_{-w_{j}}+f(-w_{j})\vec{v}_{-w_{j}}\otimes\vec{v}_{w_{j}}}{\vec{v}_{w_{j}}\cdot\vec{v}_{-w_{j}}}.
\end{gather}
Here, we assume the condition $\vec{v}_{w}\cdot\vec{v}_{-w}\neq0$. 
Putting Eq.~\eqref{equ:decompose_W} into Eq.~\eqref{equ:Gaussian_state_Majorana}, we get a similar ``product form'' for Gaussian states with generic $W$: 
\begin{equation}
    \begin{aligned}
        \rho_{0}=&\frac{1}{Z}\exp\left(\sum_{j}\frac{w_{j}}{2}\frac{\left[\vec{v}_{w_{j}}\cdot\boldsymbol{\gamma},\vec{v}_{-w_{j}}\cdot\boldsymbol{\gamma}\right]}{2\vec{v}_{w_{j}}\cdot\vec{v}_{-w_{j}}}\right)\\=&\frac{1}{Z}\prod_{j}\left[\cosh\left(\frac{w_{j}}{2}\right)+\sinh\left(\frac{w_{j}}{2}\right)\frac{\left[\vec{v}_{w_{j}}\cdot\boldsymbol{\gamma},\vec{v}_{-w_{j}}\cdot\boldsymbol{\gamma}\right]}{2\vec{v}_{w_{j}}\cdot\vec{v}_{-w_{j}}}\right]\\=&\frac{1}{2^{N}}\prod_{j}\left[1+\tanh\left(\frac{w_{j}}{2}\right)\frac{\left[\vec{v}_{w_{j}}\cdot\boldsymbol{\gamma},\vec{v}_{-w_{j}}\cdot\boldsymbol{\gamma}\right]}{2\vec{v}_{w_{j}}\cdot\vec{v}_{-w_{j}}}\right].
        \end{aligned}
\end{equation}
We note that the operator $\left[\vec{v}_{w_{j}}\cdot\boldsymbol{\gamma},\vec{v}_{-w_{j}}\cdot\boldsymbol{\gamma}\right]/2\vec{v}_{w_{j}}\cdot\vec{v}_{-w_{j}}$ behaves very similarly to the $\mathrm{i}\tilde{\gamma}_{2i-1}\tilde{\gamma}_{2i}$ in Eq.~\eqref{equ:Gaussian_product_form}. 
Firstly, $\left[\vec{v}_{w_{j}}\cdot\boldsymbol{\gamma},\vec{v}_{-w_{j}}\cdot\boldsymbol{\gamma}\right]$ with different $j$ commutes with each other. 
Secondly, it satisfies $\left[\vec{v}_{w_{j}}\cdot\boldsymbol{\gamma},\vec{v}_{-w_{j}}\cdot\boldsymbol{\gamma}\right]^2=4(\vec{v}_{w_{j}}\cdot\vec{v}_{-w_{j}})^2$. 
The covariance matrix $\Gamma$ can be evaluated as follows, 
\begin{equation}
    \begin{aligned}
        \Gamma_{kl}=&\frac{1}{Z}\mathrm{Tr}\left[\frac{\left[\gamma_{k},\gamma_{l}\right]}{2}\exp\left(\frac{\boldsymbol{\gamma}^{T}W\boldsymbol{\gamma}}{4}\right)\right]\\=&\frac{1}{2^{N}}\sum_{j}\tanh\left(\frac{w_{j}}{2}\right)\mathrm{Tr}\left[\frac{\left[\gamma_{k},\gamma_{l}\right]}{2}\frac{\left[\vec{v}_{w_{j}}\cdot\boldsymbol{\gamma},\vec{v}_{-w_{j}}\cdot\boldsymbol{\gamma}\right]}{2\vec{v}_{w_{j}}\cdot\vec{v}_{-w_{j}}}\right]\\=&\sum_{j}\frac{-\tanh\left(\frac{w_{j}}{2}\right)\vec{v}_{w_{j},k}\vec{v}_{-w_{j},l}+\tanh\left(\frac{w_{j}}{2}\right)\vec{v}_{w_{j},l}\vec{v}_{-w_{j},k}}{\vec{v}_{w_{j}}\cdot\vec{v}_{-w_{j}}}\\=&\left[-\tanh\left(\frac{W}{2}\right)\right]_{kl}, 
        \end{aligned}
\end{equation}
where the last line uses the decomposition of $f(W)$ in Eq.~\eqref{equ:decompose_fW}.

%%%%%%%%%%%%%%%%%%%%%%%%%%%%%%%%%%%%%%%%%%%%%%%%%%%%%
\section{Product rule of Gaussian states}\label{app:product_rule_Gaussian}
%%%%%%%%%%%%%%%%%%%%%%%%%%%%%%%%%%%%%%%%%%%%%%%%%%%%%
To the best of our knowledge, the terminology ``product rule of Gaussian states'' is first mentioned in Ref.~\cite{J.Stat.Mech.2010Fagotti}. 
Since many formulations and statements in Sec.~\ref{sec:Prelimiary} are based on the product rule, including the twisted partially transposed Gaussian state and formulation of DQMC, we review it here in both Majorana basis and (particle-number-conserving) Dirac basis. 

The product rule of Gaussian states means that the product of two Gaussian states is still a Gaussian state. 
Let us denote the Gaussian state as $\rho_W=\frac{1}{Z_{W}}e^{\hat{W}}$. For the general cases in the Majorana basis, $\hat{W}\equiv\frac{1}{4}\boldsymbol{\gamma}^{T}W\boldsymbol{\gamma}$ and $Z_{W}=\pm\sqrt{\det\left(1+e^{W}\right)}$. For the particle-number-conserving cases in the Dirac basis, $\hat{W}\equiv\mathbf{c}^{\dagger}W\mathbf{c}$ and $Z_{W}={\det\left(1+e^{W}\right)}$. 
Then we wonder whether the product of $\rho_W$ and $\rho_{W^\prime}$ is still a Gaussian state.  
The proof roots in the Baker-Campbell-Hausdorff formula,
\begin{equation}
    \begin{aligned}
    e^{\hat{W}^{\prime\prime}}&\equiv e^{\hat{W}}e^{\hat{W}^{\prime}}
    \\&=e^{\hat{W}+\hat{W}^{\prime}+\frac{1}{2}\left[\hat{W},\hat{W}^{\prime}\right]+\frac{1}{12}\left(\left[\hat{W},\left[\hat{W},\hat{W}^{\prime}\right]\right]+\left[\hat{W}^{\prime},\left[\hat{W},\hat{W}^{\prime}\right]\right]\right)+\cdots},
    \end{aligned}
\end{equation}
and the following commuting relations
\begin{subequations}
    \begin{equation}
        \left[\gamma_{l}\gamma_{n},\gamma_{j}\gamma_{k}\right]=-2\gamma_{l}\gamma_{j}\delta_{kn}+2\gamma_{l}\gamma_{k}\delta_{jn}+2\gamma_{n}\gamma_{j}\delta_{lk}-2\gamma_{n}\gamma_{k}\delta_{lj},
    \end{equation}
    \begin{equation}
        \left[c_{l}^{\dagger}c_{n},c_{j}^{\dagger}c_{k}\right]=c_{l}^{\dagger}c_{k}\delta_{nk}-c_{j}^{\dagger}c_{n}\delta_{lk}.
    \end{equation}
\end{subequations}
As a result, although we do not write down a simply closed form of $\hat{W}^{\prime\prime}$ yet, we can assert that it is a fermionic quadratic, thus $\rho_{W^{\prime\prime}}\equiv\rho_{W}\rho_{W^{\prime}}/\mathrm{Tr}\left[\rho_{W}\rho_{W^{\prime}}\right]$ is a Gaussian state. 

Furthermore, using the trace formula for a product of exponential quadratic~\cite{Assaad2008WorldlineDeterminantalQuantum,J.Stat.Mech.2014Klich}, 
\begin{equation}
    \mathrm{Tr}\left[e^{\hat{W}_{1}}e^{\hat{W}_{2}}\cdots e^{\hat{W}_{n}}\right]=\left(\pm\right)\left(\det\left[I+e^{W_{1}}e^{W_{2}}\cdots e^{W_{n}}\right]\right)^{\eta},
\end{equation}
where $\eta=\frac{1}{2}(1)$ for Majorana (Dirac) basis and there is a sign ambiguity for Majorana cases, 
we can deduce more useful relations. Starting with the relation between single-particle Hamiltonians, 
\begin{equation}\label{equ:prodrule_W}
    e^{W^{\prime\prime}}=e^{W}e^{W^{\prime\prime}},
\end{equation}
we obtain the relations between the covariance matrices/Green's functions, 
\begin{subequations}\label{equ:prodrule_GG}
    \begin{equation}\label{equ:prodrule_Gamma}
        \Gamma^{\prime\prime}=\left(I+\Gamma^{\prime}\right)\left(I+\Gamma\Gamma^{\prime}\right)^{-1}\left(\Gamma+\Gamma^{\prime}\right)\left(I+\Gamma^{\prime}\right)^{-1},
    \end{equation}
    \begin{equation}\label{equ:prodrule_Green}
        G^{\prime\prime}=\left[I+\left(G^{-1}-I\right)\left(G^{\prime-1}-I\right)\right]^{-1}.
    \end{equation}
\end{subequations}
As an example, we consider the twisted PTDM $\rho^{\tilde{T}_2^f}=\rho^{T_2^f}X_2$ where $\rho^{T_2^f}$ is a Gaussian state given by the covariance matrix in Eq.~\eqref{equ:PT_Gamma} and $X_2=e^{\mathrm{i}\pi\sum_{j\in A_2}n_j}\equiv\frac{1}{Z_2}e^{\boldsymbol{\gamma}^TW_2}\boldsymbol{\gamma}$ with $e^{W_2}=I_1\oplus (-I_2)\equiv U_2$. 
Using the product rule in Eq.~\eqref{equ:prodrule_Gamma}, one can obtain the twisted covariance matrix in Eq.~\eqref{equ:TPT_Gamma}. 
By the way, we note that casting the decoupled Hamiltonian into one single Gaussian state as in Eqs.~\eqref{equ:rho_MQMC} and \eqref{equ:rho_DQMC} also relies on the product rule. 
Since we only care about the covariance matrices or Green's functions about the final Gaussian state $\rho_\mathbf{s}$, we do not need the relations in Eq.~\eqref{equ:prodrule_GG}.

%%%%%%%%%%%%%%%%%%%%%%%%%%%%%%%%%%%%%%%%%%%%%%%%%%%%%
\section{Sign problem of the Grover determinants}\label{app:sign_detGrvoer}
%%%%%%%%%%%%%%%%%%%%%%%%%%%%%%%%%%%%%%%%%%%%%%%%%%%%%
In this appendix, we prove that both the untwisted and twisted Grover determinants are real and positive for two classes of sign-problem-free models. 
We have detailed the untwisted case in the supplementary of Ref.~\cite{Nat.Commun.2025Wanga}, including how to show that these two model classes are represented by the Hubbard model and the spinless $t$-$V$ model respectively. 
Here, we focus on the mathematical structure and unify the proof for different entanglement Green's functions, including reduced, untwisted, and twisted Green's functions. 
\subsection{Sufficient condition I in the Dirac basis}
\begin{lemma}\label{the:lemma1}
    If two matrices, $G^{\uparrow}$ and $G^{\downarrow}$, are connected by the relation $G^{\downarrow}=U^{\dagger}\left(I-G^{\uparrow}\right)^{\dagger}U$ with $U$ an unitary matrix satisfying $UU^\dagger=I$, then their rank-$r$ Grover determinants defined by
\begin{equation}
    \det \mathbf{g}_{r}\equiv\det\left[\prod_{i=r}^{1}G_{i}\left[I+\prod_{i=r}^{1}\left(G_{i}^{-1}-I\right)\right]\right],
\end{equation}
where $G_i=G^{\uparrow,\downarrow}_i$ and in every replica the relation $G_i^{\downarrow}=U^{\dagger}(I-G_i^{\uparrow})^{\dagger}U$ holds, are complex conjugate to each other, i.e., $\det \mathbf{g}_r^{\uparrow}=(\det \mathbf{g}_r^{\downarrow})^*$. 
\end{lemma}
\begin{proof}[Proof]
    Using $G^{\downarrow}=U^{\dagger}(I-G^{\uparrow})^{\dagger}U$ we further obtain
\begin{equation}
    \left(G^{\downarrow}\right)^{-1}-I=U^{\dagger}\left[\left(\left(G^{\uparrow}\right)^{-1}-I\right)^{-1}\right]^{\dagger}U.
\end{equation}
Then we can prove Lemma~\ref{the:lemma1} by using the fact that unitary transformation does not change determinant
\begin{equation}
    \begin{aligned}
        &\det \mathbf{g}_{r}^{\downarrow}\\=&\det\left\{ \prod_{i=r}^{1}G_{i}^{\downarrow}\left[I+\prod_{i=r}^{1}\left(\left(G_{i}^{\downarrow}\right)^{-1}-I\right)\right]\right\} \\=&\det\left\{ \prod_{i=r}^{1}\left(I-G_{i}^{\uparrow}\right)^{\dagger}\left[I+\prod_{i=r}^{1}\left(\left(G_{i}^{\uparrow}\right)^{-1}-I\right)^{-1\dagger}\right]\right\} \\=&\det\left\{ \left[I+\prod_{i=1}^{r}\left(\left(G_{i}^{\uparrow}\right)^{-1}-I\right)^{-1}\right]\prod_{i=1}^{r}\left(I-G_{i}^{\uparrow}\right)\right\} ^{*}\\=&\det\left\{ \left[\prod_{i=r}^{1}\left(\left(G_{i}^{\uparrow}\right)^{-1}-I\right)+I\right]\prod_{i=1}^{r}G_{i}^{\uparrow}\right\} ^{*}\\=&\left(\det \mathbf{g}_{r}^{\uparrow}\right)^{*}.
        \end{aligned}
\end{equation}
We have use $(I-G_{i}^{\uparrow})=((G_{i}^{\uparrow})^{-1}-I)G_{i}^{\uparrow}$ and $\det AB=\det A\det B$ to come to the last but second line. 
\end{proof}

\paragraph*{Example 1.1 } The Green's function matrices of a class of sign-problem-free models, represented by the half-filled Hubbard model on bipartite lattices such as square and honeycomb lattices, can be decoupled into the two independent spin sectors under appropriate basis, such that $G=G^{\uparrow}\oplus G^{\downarrow}$, and they satisfy the relation $G_{jk}^{\downarrow}=(-)^{j+k}(\delta_{jk}-G_{kj}^{\uparrow*})$~\cite{Nat.Commun.2025Wanga}. 
Here, the superscript $\uparrow$ and $\downarrow$ denotes spin-up and spin-down sectors, respectively. The subscripts $j$ and $k$ run over the lattice sites. 
The matrix defined by $U_{\eta,jk}=\delta_{jk}(-)^{j}$ is diagonal and unitary, which gives $G^{\downarrow}=U_{\eta}^{\dagger}\left(I-G^{\uparrow}\right)^{\dagger}U_{\eta}$. 
Consequently, Lemma~\ref{the:lemma1} applies to this case, and the Grover determinants corresponding to the total or reduced Green's functions are real and positive, i.e., $\det \mathbf{g}_r=\det \mathbf{g}_r^\uparrow \det \mathbf{g}_r^\downarrow \geq 0$~\cite{Phys.Rev.B2024Zhang}. 
% It is important to note that when the $G$ matrix represents a reduced Green's function, the subsystem of interest $A$ should be selected to be bipartite to ensure that $G^{\downarrow}_A=U_{\eta}^{\dagger}\left(I-G^{\uparrow}_A\right)^{\dagger}U_{\eta}$. 

\paragraph*{Example 1.2} Still consider the class of models in Example 1.1. 
Consider a continuous bipartite geometry $A=A_1\cup A_2$ and using Eq.~\eqref{equ:untwistedPTGreen} to do untwisted partial transpose, we find that 
\begin{equation}\label{equ:sign_dirac_untwisted}
    G^{\downarrow,T_{2}^{f}}=V^{\dagger}U_{\eta}^{\dagger}\left(I-G^{\uparrow,T_{2}^{f}}\right)^{\dagger}U_{\eta}V
\end{equation}
where $V\equiv(\mathrm{i}I_{A_{1}})\oplus(-\mathrm{i}I_{A_{2}})\equiv\mathrm{i}U_2$ is block-diagonal, unitary and anti-Hermitian, i.e., $V^\dagger=V^{-1}=-V$. 
Thus the untwisted Green's function still satisfies the condition for Lemma~\ref{the:lemma1} and the untwisted Grover determinant (see Eq.~\eqref{equ:gk_untwisted_DQMC}) is real and positive, $\det g_r=\det g_r^{\uparrow}\det g_r^{\downarrow}\geq 0$.  

\paragraph*{Example 1.3} Still consider the class of models in Example 1.1. 
Now we further consider twisted partial transpose given by Eq.~\eqref{equ:twistedPTGreen}. 
From Eq.~\eqref{equ:sign_dirac_untwisted} we have 
\begin{equation}
    \left(\left(G^{\downarrow,T_{2}^{f}}\right)^{-1}-I\right)=V^{\dagger}U_{\eta}^{\dagger}\left(\left(G^{\uparrow,T_{2}^{f}}\right)^{-1}-I\right)^{-1\dagger}U_{\eta}V. 
\end{equation}
Next, using the definition of twisted partial transpose and the product rule of Gaussian state in Eq.~\eqref{equ:prodrule_Green} we have
\begin{equation}
    \begin{aligned}
        &\left(\left(G^{\downarrow,\tilde{T}_{2}^{f}}\right)^{-1}-I\right)=\left(\left(G^{\downarrow,T_{2}^{f}}\right)^{-1}-I\right)U_{2}\\&=V^{\dagger}U_{\eta}^{\dagger}\left(\left(G^{\uparrow,T_{2}^{f}}\right)^{-1}-I\right)^{-1\dagger}U_{\eta}VU_{2}\\&=V^{\dagger}U_{\eta}^{\dagger}\left[\left(\left(G^{\uparrow,\tilde{T}_{2}^{f}}\right)^{-1}-I\right)U_{2}^{\dagger}\right]^{-1\dagger}U_{\eta}VU_{2}\\&=U_{2}^{\dagger}V^{\dagger}U_{\eta}^{\dagger}\left(\left(G^{\uparrow,\tilde{T}_{2}^{f}}\right)^{-1}-I\right)^{-1\dagger}U_{\eta}VU_{2}. 
        \end{aligned}
\end{equation}
To go to the last line, we use the fact that $U_2\equiv I_1\oplus (-I_2)$ commutes with any block-diagonal matrices, including $V^\dagger$ and $U^\dagger_\eta$. 
Actually, due to $VU_2=\mathrm{i}U_2^2=\mathrm{i}$, the above formula can be further simplified. 
As a result, the twisted Green's function still satisfies the condition for Lemma~\ref{the:lemma1}, so the twisted Grover determinant (see Eq.~\eqref{equ:gk_twisted_DQMC}) is real and positive, $\det \tilde{g}_r=\det \tilde{g}_r^{\uparrow}\det \tilde{g}_r^{\downarrow}\geq 0$. 

By the way, we discuss the sign of the normalization factor $\mathcal{Z}_{\tilde{T}_2^f}\equiv \det(I_{2}-2G_{22})$ for this class of models. 
It turns out that $\mathcal{Z}_{\tilde{T}_2^f}$ may be negative, depending on the number of sites in subsystem $A_2$ (denoted as $N_2$): 
\begin{equation}\label{equ:sign_Dirac_ZT2f}
    \begin{aligned}
        \mathcal{Z}_{\tilde{T}_{2}^{f}}&=\det\left(I_{2}-2G_{22}\right)\\&=\det\left(I_{2}^{\uparrow}-2G_{22}^{\uparrow}\right)\det\left(I_{2}^{\downarrow}-2G_{22}^{\downarrow}\right)\\&=\det\left(-I_{2}^{\uparrow}+2U_{\eta,2}^{\dagger}\left(G_{22}^{\downarrow}\right)^{\dagger}U_{\eta,2}\right)\det\left(I_{2}^{\downarrow}-2G_{22}^{\downarrow}\right)\\&=\left(-1\right)^{N_{2}}\left|\det\left(I_{2}^{\downarrow}-2G_{22}^{\downarrow}\right)\right|^{2}.
        \end{aligned}
\end{equation}
We emphasize that the final step relies on the diagonal nature of $U_{\eta}$, making the formula not generally applicable to all cases in Lemma~\ref{the:lemma1}. 
By combining Eq.~\eqref{equ:sign_Dirac_ZT2f} with $\det \tilde{g}_r \geq 0$, we conclude that, for the models in Example 1.1, the modified twisted Grover determinants (namely, $\det \underline{\tilde{g}}_r=(\prod_{i=1}^r\mathcal{Z}_{\tilde{T}_2^f,\mathbf{s}^{(i)}})\det \tilde{g}_r$) are real and positive except when both $r$ and $N_2$ are odd.

\subsection{Sufficient condition II in the Majorana basis}
The second sufficient condition for real and positive Grover determinants is proved in the Majorana basis, where the density matrix is decomposed as in Eq.~\eqref{equ:rho_MQMC}. 
And we note that the $W_\mathbf{s}$ matrix is related to the covariance matrix $\Gamma_\mathbf{s}$ via $\Gamma_\mathbf{s}=-\tanh (W_\mathbf{s}/2)$. 

\begin{lemma}\label{the:lemma2}
If two matrices, $\Gamma^{(1)}$ and $\Gamma^{(2)}$, are connected by the relation $\Gamma^{(2)}=U^{\dagger}\Gamma^{(1)*}U$ with $U$ an unitary matrix satisfying $UU^\dagger=I$, then their rank-$r$ Grover determinants defined by
\begin{equation}
    \det\mathbf{g}_{r}=\left(\prod_{i=r}^{1}\det\left[I+e^{W_{i}}\right]^{-1/2}\right)\det\left[I+\prod_{i=r}^{1}e^{W_{i}}\right]^{1/2},
\end{equation}
where $\tanh (-W_i/2)=\Gamma_i=\Gamma^{(1),(2)}_i$ and in every replica the relation $\Gamma_i^{(2)}=U^{\dagger}\Gamma^{(1)*}_iU$ holds, are complex conjugate to each other, i.e., $\det \mathbf{g}_r^{(1)}=(\det \mathbf{g}_r^{(2)})^*$. 
\end{lemma}

\begin{proof}[Proof]
Since $\tanh (W_i/2)=\Gamma_i$, we know that $W_i^{(1)}$ and $W_i^{(2)}$ also fulfill the relation $W_i^{(2)}=U^{\dagger}W^{(1)*}_iU$. 
Then we can prove Lemma~\ref{the:lemma2} by using the fact that unitary transformation does not change determinant
\begin{equation}
    \begin{aligned}
        &\det\mathbf{g}_{r}^{\left(1\right)}\\=&\left(\prod_{i=r}^{1}\det\left[I+e^{W_{i}^{\left(1\right)}}\right]^{-1/2}\right)\det\left[I+\prod_{i=r}^{1}e^{W_{i}^{\left(1\right)}}\right]^{1/2}\\=&\left(\prod_{i=r}^{1}\det\left[I+e^{W_{i}^{\left(2\right)}}\right]^{-1/2}\right)^{*}\left(\det\left[I+\prod_{i=r}^{1}e^{W_{i}^{\left(2\right)}}\right]^{1/2}\right)^{*}\\=&\left(\det\mathbf{g}_{r}^{\left(2\right)}\right)^{*}.
        \end{aligned}
\end{equation}
\end{proof}

\paragraph*{Example 2.1}The covariance matrices of a class of sign-problem-free models, represented by the spinless $t$-$V$ model on bipartite lattices such as square and honeycomb lattices, satisfy $\Gamma_{jk}^{(2)}=(-)^{j+k}\Gamma^{(1)*}_{jk}$~\cite{Nat.Commun.2025Wanga}. 
Here, the superscript $(1)$ and $(2)$ corresponds to two Majorana species, i.e., $\gamma_i^{(1)}$ and $\gamma_i^{(2)}$ in Eq.~\eqref{equ:MajoranaOp}. 
The total covariance matrix in this class of models can be decoupled into the two independent specie sectors under an appropriate basis where $\Gamma=\Gamma^{(1)}\oplus \Gamma^{(2)}$. 
The matrix defined by $U_{\eta,jk}=\delta_{jk}(-)^{j}$ is diagonal and unitary. 
Thus Lemma~\ref{the:lemma2} applies to this case, and the Grover determinants corresponding to the total or reduced covariance matrices are real and positive, i.e., $\det \mathbf{g}_r=\det \mathbf{g}_r^{(1)} \det \mathbf{g}_r^{(2)} \geq 0$. 

\paragraph*{Example 2.2} Consider the class of models in Example 2.1 and the bipartition geometry in Example 1.2. 
Using Eq.~\eqref{equ:PT_Gamma} to do untwisted partial transpose, we find that
\begin{equation}\label{equ:sign_Majo_untwisted}
    \Gamma^{(2),T_{2}^{f}}=U_{2}^{\dagger}U_{\eta}^{\dagger}\left(\Gamma^{\left(1\right),T_{2}^{f}}\right)^{*}U_{\eta}U_{2}.
\end{equation}
Thus Lemma~\ref{the:lemma2} applies to $\Gamma^{(1),T_2^f}$ and $\Gamma^{(2),T_2^f}$, and the untwisted Grover determinant is real and positive, $\det g_r=\det g_r^{(1)}\det g_r^{(2)}\geq 0$. 

\paragraph*{Example 2.3} Still consider the class of models in Example 2.1. 
Now we further consider twisted partial transpose given by Eq.~\eqref{equ:TPT_Gamma}. 
From Eq.~\eqref{equ:sign_Majo_untwisted} and $\tanh(-W^{T_2^f})=\Gamma^{T_2^f}$ we learn that
\begin{equation}
    W^{(2),T_{2}^{f}}=U_{2}^{\dagger}U_{\eta}^{\dagger}\left(W^{\left(1\right),T_{2}^{f}}\right)^{*}U_{\eta}U_{2}.
\end{equation}
Further, using the definition of twisted partial transpose and the product rule of Gaussian state in Eq.~\eqref{equ:prodrule_W} we have 
\begin{equation}
    \begin{aligned}
        &e^{W^{\left(2\right),\tilde{T}_{2}^{f}}}=e^{W^{\left(2\right),T_{2}^{f}}}U_{2}\\=&U_{2}^{\dagger}U_{\eta}^{\dagger}\left(e^{W^{\left(1\right),T_{2}^{f}}}\right)^{*}U_{\eta}U_{2}U_{2}\\=&U_{2}^{\dagger}U_{\eta}^{\dagger}\left(e^{W^{\left(1\right),\tilde{T}_{2}^{f}}}U_{2}^{\dagger}\right)^{*}U_{\eta}U_{2}U_{2}\\=&U_{2}^{\dagger}U_{\eta}^{\dagger}\left(e^{W^{\left(1\right),\tilde{T}_{2}^{f}}}\right)^{*}U_{\eta}U_{2},
        \end{aligned}
\end{equation}
where we have used $U_2^*=U_2=U_2^\dagger$, $U_2^2=I$, and the fact that $U_2$ commutes with $U_\eta$ to go to the last line. 
As a result, according to Lemma \ref{the:lemma2}, the twisted Grover determinant is real and positive, $\det \tilde{g}_r=\det \tilde{g}_r^{(1)}\det \tilde{g}_r^{(2)}\geq 0$. 

\textbf{Important Note}: The proof presented in this subsection regarding the Majorana basis, while providing valuable intuitive and heuristic insights, lacks mathematical rigor. 
The trace of Gaussian states in the Majorana basis would involve a sign ambiguity~\cite{J.Stat.Mech.2014Klich}, which has not been accounted for in our current proof. 
In our numerical investigations of the 2D $t$-$V$ model, we consistently observe positive values for both untwisted and twisted Grover determinants. 
However, we note that the normalization factor $\mathcal{Z}_{\tilde{T}_2^f}$ does not follow Eq.~\eqref{equ:sign_Dirac_ZT2f} and can assume negative values for even $N_2$. 
Importantly, the rank-4 modified twisted Grover determinant $\det \underline{\tilde{g}}_4$ that we consider in the main text is always real and positive. 

%%%%%%%%%%%%%%%%%%%%%%%%%%%%%%%%%%%%%%%%%%%%%%%%%%%%%
\section{Discussion on other local update schemes for untwisted R\'{e}nyi negativity}\label{app:schemes12}
%%%%%%%%%%%%%%%%%%%%%%%%%%%%%%%%%%%%%%%%%%%%%%%%%%%%%
In this appendix, we present two alternative local update schemes for untwisted R\'{e}nyi negativity. Compared to the approach presented in the main text (hereafter referred to as "Scheme 3"), these methods offer greater conceptual simplicity and mathematical elegance, though at the cost of reduced numerical stability. 
\subsection{Scheme 1: Calculate $r_2$ and $r_3$ separately}
\paragraph*{Calculate $r_2$. } As a result of Eq.~\eqref{equ:update_GT2f}, we immediately obtain $r_{2}=\det\left(I_{m}+\mathcal{V}\mathcal{U}\right)$ by applying Sylvester's determinant identity. 
Considering the definition of $\mathbb{B}_i$, we deduce that the appropriate central matrix to be updated is the inverse of $G_{\mathbf{s}^{(i)}}^{T_{2}^{f}}$, 
\begin{equation}
    \begin{aligned}
        \left(G_{\mathbf{s}^{(i)\prime}}^{T_{2}^{f}}\right)^{-1}&=\left(I+\mathcal{U}\mathcal{V}\right)^{-1}\left(G_{\mathbf{s}^{(i)}}^{T_{2}^{f}}\right)^{-1}\\&=\left(G_{\mathbf{s}^{(i)}}^{T_{2}^{f}}\right)^{-1}-\mathcal{U}(I_{m}+\mathcal{V}\mathcal{U})^{-1}\mathcal{V}\left(G_{\mathbf{s}^{(i)}}^{T_{2}^{f}}\right)^{-1}.
    \end{aligned}
\end{equation}
In the second line, we utilize the Woodbury matrix identity, which reduces to the Sherman-Morrison formula when $m=1$. 
Its numerical stabilization directly follows from the stabilization results of $G_{\mathbf{s}^{(i)}}$, followed by performing the FPT and taking the inverse.

\paragraph*{Calculate $r_3$. } Since $r_3$ shares a mathematical structure similar to the DQMC ratio $r_1$, with $\mathbb{B}_i$ corresponding to $B_{l}$, it can be evaluated using formulas analogous to those employed in DQMC. 
When an update occurs in replica-$i$, the modification of $\mathbb{B}_i$ is described by:
\begin{equation}
    \begin{aligned}
        \mathbb{B}_{i}^{\prime}&=(I-\mathcal{U}\mathcal{V}^{a})\mathbb{B}_{i},\\
        \mathcal{V}^{a}&=(I_{m}+\mathcal{V}\mathcal{U})^{-1}\mathcal{V}\left(G_{\mathbf{s}^{(i)}}^{T_{2}^{f}}\right)^{-1}\mathbb{B}_{i}^{-1}.
    \end{aligned}
\end{equation}
Here, the term $(I-\mathcal{U}\mathcal{V}^{a})$ is analogous to the matrix $\Delta$ in DQMC. 
The calculation of $r_3$ then closely follows the DQMC ratio: 
\begin{equation}
    \begin{aligned}
        r_{3}=&\frac{\det\left[I+\mathbb{B}\left(r,i\right)(I-\mathcal{U}\mathcal{V}^{a})\mathbb{B}\left(i,0\right)\right]}{\det\left[I+\mathbb{B}\left(r,0\right)\right]}\\=&\det\left[I-\mathcal{U}\mathcal{V}^{a}(I-\mathbb{G}(i,i))\right]=\det\left[I_{m}-\mathbb{V}\mathcal{U}\right],
    \end{aligned}
\end{equation}
where $\mathbb{V}\equiv\mathcal{V}^{a}\left(I-\mathbb{G}\left(i,i\right)\right)$. 
We define an equal-replica ``Green's function" analogous to $G(\tau,\tau)$, 
\begin{equation}\label{equ:getgii}
    \mathbb{G}(i,i)\equiv\left(I+\mathbb{B}(i,0)\mathbb{B}(r,i)\right)^{-1},
\end{equation}
which generalize $\mathbb{G}(0,0)=\mathbb{G}(r,r)$ as previously defined in Eq.~\eqref{equ:Getg00}. 
The central matrices to be updated are $\mathbb{B}_i^{-1}$ and $\mathbb{G}$, with update formulas given by:
\begin{gather}
    (\mathbb{B}_{i}^{-1})^{\prime}=\mathbb{B}_{i}^{-1}\left(I+\mathcal{U}(I_{m}-\mathcal{V}^{a}\mathcal{U})^{-1}\mathcal{V}^{a}\right),\\
    \mathbb{G}^{\prime}(i,i)=\mathbb{G}\left[I+\mathcal{U}\left(I_{m}-\mathbb{V}\mathcal{U}\right){}^{-1}\mathbb{V}\right].
\end{gather}
The stabilization of $\mathbb{B}_{i}$ and $\mathbb{B}_{i}^{-1}$ directly derives from $G_{\mathbf{s}^{(i)}}$, while the stabilization routines for $\mathbb{G}(i,i)$ are similar to those for $G_{\mathbf{s}^{(i)}}$, with all instances of $B(\tau,\tau^\prime)\rightarrow \mathbb{B}(i,i^\prime)$ (see Eq.~\eqref{equ:Getg00} and similar formulas in Appendix~\ref{app:numerical_stabilization_DQMC}). 
When the local updates in one replica are completed and we need to move to the next replica, $\mathbb{G}(i,i)$ also follows a propagation formula similar to $G(\tau,\tau)$, for example: 
\begin{equation}
    \begin{aligned}
        \mathbb{G}(i+1,i+1)&=\mathbb{B}_{i+1}\mathbb{G}(i,i)\mathbb{B}_{i+1}^{-1},\\\mathbb{G}(i-1,i-1)&=\mathbb{B}_{i}^{-1}\mathbb{G}(i,i)\mathbb{B}_{i}.
    \end{aligned}
\end{equation}
\paragraph*{Measure $\det g_{r,\bar{\boldsymbol{s}}}$. }The measurement of the Grover determinant, and consequently its $N_{\mathrm{inc}}$-th root, can be conducted after numerical stabilization, where $\ln\det G_{\mathbf{s}^{(i)}}^{T_{2}^{f}}$ and $\ln\det\mathbb{G}(i,i)$ have been calculated stably. 
By retrieving $\ln\det G_{\mathbf{s}^{(j)}}^{T_{2}^{f}}$ from other replicas ($j\neq i$), which remain unchanged during updates in replica-$i$, we can calculate $\det g_{r,\bar{\boldsymbol{s}}}$ as follows:
\begin{equation}
    \ln\det g_{r,\bar{\boldsymbol{s}}} = \sum_{j=1}^r \ln\det G_{\mathbf{s}^{(j)}}^{T_{2}^{f}}-\ln\det\mathbb{G}(i,i).
\end{equation}

\subsection{Scheme 2: Calculate $r_2r_3$ together}
In this approach, we focus on a more stable update object that combines $G^{T_2^f}_{\mathbf{s}^{(i)}}$ and $\left(I+\mathbb{B}(i,0)\mathbb{B}(r,i)\right)$: 
\begin{equation}\label{equ:scheme2ffen}
    \mathbb{F}_{i}=\left[G_{\mathbf{s}^{(i)}}^{T_{2}^{f}}\left(I+\mathbb{B}\left(i,0\right)\mathbb{B}\left(r,i\right)\right)\right]^{-1}=\mathbb{G}\left(i,i\right)\left(G_{\mathbf{s}^{(i)}}^{T_{2}^{f}}\right)^{-1}.  
\end{equation}
This is because $G_{\mathbf{s}^{(i)}}^{T_{2}^{f}}\mathbb{B}(i,0)$ inside the square bracket results in the cancellation of one inverse of $G_{\mathbf{s}^{(i)}}^{T_{2}^{f}}$ (note that $G_{\mathbf{s}^{(i)}}^{T_{2}^{f}}\mathbb{B}_i=I-G_{\mathbf{s}^{(i)}}^{T_{2}^{f}}$). 
The update formula for $\mathbb{F}_i$ can be derived from the update of $G^{T_2^f}_{\mathbf{s}^{(i)}}$ in Eq.~\eqref{equ:update_GT2f}
\begin{equation}
    \mathbb{F}_{i}^{\prime}=\mathbb{F}_{i}\left(I-\mathscr{U}\left(I_m+\mathcal{V}^{b}\mathscr{U}\right)^{-1}\mathcal{V}^{b}\right),
\end{equation}
where we define
\begin{equation}
    \mathcal{V}^{b}\equiv\mathcal{V}\mathbb{G}_{\bcancel{i}}^{-1}\mathbb{F}_{i}\text{ with }\mathbb{G}_{\bcancel{i}}^{-1}\equiv I-\mathbb{B}\left(i-1,0\right)\mathbb{B}\left(r,i\right).
\end{equation}
Thus, the ratio $r_2r_3$ is given by
\begin{equation}
    r_{2}r_{3}=\frac{\det g_{r,\bar{\boldsymbol{s}}^\prime}}{\det g_{r,\bar{\boldsymbol{s}}}}=\frac{\det\mathbb{F}_{i}}{\det\mathbb{F}_{i}^{\prime}}=\det\left(I_m+\mathcal{V}^{b}\mathscr{U}\right).
\end{equation}
Here, the subscript $\bcancel{i}$ of $\mathbb{G}_{\bcancel{i}}$ indicates that it does not depend on the auxiliary fields in replica-$i$. 
This choice of updating object (i.e., $\mathbb{F}_i$) was also recommended by a former work on high-rank REE calculations~\cite{Phys.Rev.B2024Zhang}. 
The stabilization of $\mathbb{F}_i$ directly follows from $\mathbb{G}(i,i)$ and $G^{T_2^f}_{\mathbf{s}^{(i)}}$. 
The measurement of $\det g_{r,\bar{\boldsymbol{s}}}$ is similar to Scheme 1, given by
\begin{equation}
    \ln\det g_r = \sum_{j\neq i}^r \ln\det G_{\mathbf{s}^{(j)}}^{T_{2}^{f}}-\ln\det\mathbb{F}_i.
\end{equation}

\subsection{Practical experience and discussion}

We have tested Schemes 1 and 2 on the 1D/2D half-filled Hubbard model, and the high-rank RN behaves well at moderate temperatures. 
However, for the spinless $t$-$V$ model (see Sec.~\ref{sec:tVmodel} for details), updating $\left(G^{T_2^f}_{\mathbf{s}^{(i)}}\right)^{-1}$, $\mathbb{B}_i^{-1}$, and $\mathbb{F}_i$ is unstable for a generic bipartition\footnote{An exception occurs in the special case of an equal bipartition, where the Green's functions exhibit much weaker singularities (see the discussion in Sec.~\ref{sec:numerical_stable}); in this scenario, Schemes 1 and 2 can be effective.}.
Specifically, the total phase of $W^{(k)}_{\mathbf{s}^{(1)}\cdots\mathbf{s}^{(r)}}$ may deviate from unity during updates. 
This issue motivates us to update a more complex and stable object (see Eq.~\eqref{equ:scheme4ffen}), which is a variant of the total Grover matrix in Eq.~\eqref{equ:gk_untwisted_DQMC}, maintaining the same determinant ratio.
Compared to Eq.~\eqref{equ:scheme2ffen}, Eq.~\eqref{equ:scheme4ffen} cancels one additional inverse of $G^{T_2^f}$, as shown in the first term of the second line of Eq.~\eqref{equ:FcanceliInv}. 

%%%%%%%%%%%%%%%%%%%%%%%%%%%%%%%%%%%%%%%%%%%%%%%%%%%%%
\section{Numerical stabilization routines in DQMC}\label{app:numerical_stabilization_DQMC}
%%%%%%%%%%%%%%%%%%%%%%%%%%%%%%%%%%%%%%%%%%%%%%%%%%%%%
In this appendix, we briefly review the numerical stabilization routines used in DQMC. 
They aim at stably dealing with matrices with exponentially large conditional numbers, including the $B(\tau,\tau^\prime)$ matrices defined below Eq.~\eqref{equ:DQMC_equaltimeG} and also the $\mathbb{B}(i,j)$ matrices defined in Eq.~\eqref{equ:mathbbBii}. 
For numerical stable matrix operations, these matrices are decomposed into $UDV$ forms~\cite{Loh1989StableMatrixMultiplicationAlgorithms,SciPostPhys.Core2020Bauer}. 
The real diagonal matrix $D$ stores scale information of $B(\tau,\tau^\prime)$ or $\mathbb{B}(i,j)$, and for matrices with large conditional numbers, there are both exponentially large and small number among the entries of $D$. 
In most numerically stable routines, we further factorize it to become $D=D_+D_-$ with $D_{+}=\max (D, 1)$ and $D_{-}=\min (D, 1)$, and avoid adding or subtracting these different scales during computation as much as possible. 
In practice, there are two kinds of UDV decomposition, one is singular value decomposition (SVD) and the other is a QR decomposition followed by extracting the absolute value of diagonal elements of $R$ as $D$. 
One can see a detailed comparison and benchmark of them in Ref.~\cite{SciPostPhys.Core2020Bauer}. 
In both choices, $U$ is a unitary matrix that will be used to simplify some formulas. 

In this and the next appendices, we use $A$ (and other variant fonts like $\mathbb{A}$) to represent the matrix decomposition form of some matrix $B$, and their relation is denoted as $A:=B=UDV$. 
Since we store $B(\beta,\tau)^\dagger=U_LD_LV_L=:A_L$ and $B(\tau,0)=U_RD_RV_R=:A_R$ in practice, the following stable inverse routines are used for calculating Green's functions stably. 
For equal-time Green's functions, we need: 
\begin{equation}
    \begin{aligned}\label{equ:stab_Gtt}
        G\left(\tau,\tau\right)&=\left[1+A_{R}A_{L}^{\dagger}\right]^{-1}\\&=U_{L}D_{L+}^{-1}M^{-1}D_{R+}^{-1}U_{R}^{\dagger}\\\text{with }M&\equiv D_{R+}^{-1}U_{R}^{\dagger}U_{L}D_{L+}^{-1}+D_{R-}V_{R}V_{L}^{\dagger}D_{L-}.
        \end{aligned}
\end{equation}
In particular, we have simplified forms at $\tau=0$ and $\beta$
\begin{gather}
    \begin{aligned}\label{equ:stab_G00}
        G\left(0,0\right)&=\left[I+A_{L}^{\dagger}\right]^{-1}\\&=U_{L}D_{L+}^{-1}\left[U_{L}^{\dagger}D_{L+}^{-1}+T_{L}^{\dagger}D_{L-}\right]^{-1},
    \end{aligned}\\
    \begin{aligned}\label{equ:stab_Gbb}
        G\left(\beta,\beta\right)&=\left[I+A_{R}\right]^{-1}\\&=\left[D_{R+}^{-1}U_{R}^{\dagger}+D_{R-}V_{R}\right]^{-1}D_{R+}^{-1}U_{R}^{\dagger}.
    \end{aligned}
\end{gather}
For time-displaced Green's functions we need: 
\begin{equation}\label{equ:stab_Gt0}
    \begin{aligned}
        G\left(\tau,0\right)&=\left[A_{R}^{-1}+A_{L}^{\dagger}\right]^{-1}\\&=U_{L}D_{L+}^{-1}M^{-1}D_{R-}V_{R},\\\text{with }M&\equiv D_{R+}^{-1}U_{R}^{\dagger}U_{L}D_{L+}^{-1}+D_{R-}V_{R}V_{L}^{\dagger}D_{L-},
        \end{aligned}
\end{equation}
\begin{equation}\label{equ:stab_G0t}
    \begin{aligned}
        -G\left(0,\tau\right)&=\left[A_{L}^{\dagger-1}+A_{R}\right]^{-1}\\&=V_{R}^{-1}D_{R+}^{-1}M^{-1}D_{L-}U_{L}^{\dagger}\\\text{with }M&\equiv D_{L+}^{-1}\left(V_{L}^{\dagger}\right)^{-1}V_{R}^{-1}D_{R+}^{-1}+D_{L-}U_{L}^{\dagger}U_{R}D_{R-},
        \end{aligned}
\end{equation}
and we can also evaluate $G(0,0)$ at time $\tau$
\begin{equation}\label{equ:stab_G00_t}
    \begin{aligned}
        G\left(0,0\right)&=\left[1+A_{L}^{\dagger}A_{R}\right]^{-1}\\&=V_{R}^{-1}D_{R+}^{-1}M^{-1}D_{L+}^{-1}\left(V_{L}^{\dagger}\right)^{-1}\\\text{with }M&\equiv D_{L+}^{-1}\left(V_{L}^{\dagger}\right)^{-1}V_{R}^{-1}D_{R+}^{-1}+D_{L-}U_{L}^{\dagger}U_{R}D_{R-}.
        \end{aligned}
\end{equation}
Similarly, to calculate $\mathbb{G}(i,i)$ defined in Eq.~\eqref{equ:getgii}, just let $A_L:=\mathbb{B}(r,i)^\dagger$ and $A_R:=\mathbb{B}(i,0)$ in Eq.~\eqref{equ:stab_Gtt}. 

%%%%%%%%%%%%%%%%%%%%%%%%%%%%%%%%%%%%%%%%%%%%%%%%%%%%%
\section{Numerical stabilization in the Scheme 3 of incremental algorithm}\label{app:numerical_stabilization_scheme4}
%%%%%%%%%%%%%%%%%%%%%%%%%%%%%%%%%%%%%%%%%%%%%%%%%%%%%
In this appendix, we detailed the numerical stabilization of the update object $\mathbb{F}_i$ defined in Eq.~\eqref{equ:scheme4ffen}, 
\begin{equation}\label{equ:appE_FFi}
    \mathbb{F}_{i}=\left[G_{\mathbf{s}^{(i)}}^{T_{2}^{f}}\mathbb{F}_{\bcancel{i}}^{-1}+I-G_{\mathbf{s}^{(i-1)}}^{T_{2}^{f}}\right]^{-1}, 
\end{equation}
where $\mathbb{F}_{\bcancel{i}}^{-1}$ is given by Eq.~\eqref{equ:FcanceliInv}
\begin{equation}\label{equ:appE_FFiInv}
    \mathbb{F}_{\bcancel{i}}^{-1}\equiv\underbrace{\mathbb{B}_{i+1}^{-1}\cdots\mathbb{B}_{r}^{-1}}_{\mathbb{A}_{L}\left(i\right)}\underbrace{\mathbb{B}_{1}^{-1}\cdots\mathbb{B}_{i-2}^{-1}}_{\mathbb{A}_{R}\left(i-2\right)}G_{\mathbf{s}^{(i-1)}}^{T_{2}^{f}}-\left(I-G_{\mathbf{s}^{(i-1)}}^{T_{2}^{f}}\right).
\end{equation}
Here we defined two matrix decompositions, namely, $\mathbb{A}_{L}\left(i\right):=\left(\mathbb{B}\left(r,i\right)\right)^{-1}$, $\mathbb{A}_{R}\left(i\right):=\left(\mathbb{B}\left(i,0\right)\right)^{-1}$, where $0\leq i\leq r$. 
It reduces to simpler formulas in special replicas close to the boundary: 
\begin{subequations}
    \begin{equation}\label{equ:appE_FFiInv1}
        \begin{aligned}
            \mathbb{F}_{\bcancel{1}}^{-1}&\equiv\mathbb{B}_{2}^{-1}\cdots\mathbb{B}_{r-1}^{-1}G_{\mathbf{s}^{(r)}}^{T_{2}^{f}}-\left(I-G_{\mathbf{s}^{(r)}}^{T_{2}^{f}}\right)\\&=\underbrace{\mathbb{B}_{2}^{-1}\cdots\mathbb{B}_{r}^{-1}}_{\mathbb{A}_{L}(1)}\left(I-G_{\mathbf{s}^{(r)}}^{T_{2}^{f}}\right)-\left(I-G_{\mathbf{s}^{(r)}}^{T_{2}^{f}}\right),
        \end{aligned}
    \end{equation}
    \begin{equation}\label{equ:appE_FFiInv2}
        \mathbb{F}_{\bcancel{2}}^{-1}\equiv\underbrace{\mathbb{B}_{3}^{-1}\cdots\mathbb{B}_{r}^{-1}}_{\mathbb{A}_{L}(2)}G_{\mathbf{s}^{(1)}}^{T_{2}^{f}}-\left(I-G_{\mathbf{s}^{(1)}}^{T_{2}^{f}}\right),
    \end{equation}
    \begin{equation}\label{equ:appE_FFiInvr}
        \mathbb{F}_{\bcancel{r}}^{-1}\equiv\underbrace{\mathbb{B}_{1}^{-1}\cdots\mathbb{B}_{r-2}^{-1}}_{\mathbb{A}_{R}(r-2)}G_{\mathbf{s}^{(r-1)}}^{T_{2}^{f}}-\left(I-G_{\mathbf{s}^{(r-1)}}^{T_{2}^{f}}\right)
    \end{equation}
\end{subequations}
Both $\mathbb{F}_{\bcancel{i}}^{-1}$ and $G_{\mathbf{s}^{(i-1)}}^{T_{2}^{f}}$ are unchanged during sweeping inside replica-$i$. 
When going from one replica to the next, we need to update these two matrices. 
For the best stability, all the $G^{T_2^f}_{\mathbf{s}^{(i)}}$ and $\mathbb{B}_i^{-1}$ are decomposed into UDV form, 
\begin{gather}
    G_{\mathbf{s}^{(i)}}^{T_{2}^{f}}=U_{i}D_{i}V_{i}=:A_{i},
    \\
    I-G_{\mathbf{s}^{(i)}}^{T_{2}^{f}}=\bar{U}_{i}\bar{D}_{i}\bar{V}_{i}=:\bar{A}_{i},
    \\
    \begin{aligned}
        \\\mathbb{B}_{i}^{-1}&=\left(A_{i}^{-1}-I\right)^{-1}\\&=U_{i}D_{i-}\left(\underbrace{V_{i}^{-1}D_{i+}^{-1}-U_{i}D_{i-}}_{U^{\prime}D^{\prime}V^{\prime}}\right)^{-1}\\&=U_{i}\underbrace{D_{i-}V^{\prime-1}D^{\prime-1}}_{U^{\prime\prime}D^{\prime\prime}V^{\prime\prime}}U^{\prime}\\&=\underbrace{U_{i}U^{\prime\prime}}_{\mathbb{U}_{i}}\underbrace{D^{\prime\prime}}_{\mathbb{D}_{i}}\underbrace{V^{\prime\prime}U^{\prime}}_{\mathbb{V}_{i}}=:\mathbb{A}_{i}.
    \end{aligned}
\end{gather}
Thus we also have $\mathbb{A}_L(i)=\prod_{j=i+1}^r\mathbb{A}_j$ and $\mathbb{A}_R(i)=\prod_{j=1}^i\mathbb{A}_i$. 
Next, all the matrix multiplications in the first term of Eq.~\eqref{equ:appE_FFiInv} are done in decomposition form, i.e., 
\begin{equation}\label{equ:UDVxUDV}
    \begin{aligned}
        A_{1}A_{2}&=U_{1}\underbrace{D_{1}V_{1}U_{2}D_{2}}_{U^{\prime}D^{\prime}V^{\prime}}V_{2}\\&=\underbrace{U_{1}U^{\prime}}_{U_{12}}\underbrace{D^{\prime}}_{D_{12}}\underbrace{V^{\prime}V_{2}}_{V_{12}}=A_{12}.
        \end{aligned}
\end{equation}
After added with $(I-G^{T_2^f}_{\mathbf{s}^{(i)}})$, we obtain an accurate $\mathbb{F}_{\bcancel{i}}^{-1}$. 
Finally, the evaluation of Eq.~\eqref{equ:appE_FFi} is done using normal matrix operations.
\footnote{One could have tried an even more stable approach: keep the $G_{\mathbf{s}^{(i)}}^{T_{2}^{f}}\mathbb{F}_{\bcancel{i}}^{-1}$ and $(I-G_{\mathbf{s}^{(i-1)}}^{T_{2}^{f}})$ in Eq.~\eqref{equ:appE_FFi} stay in UDV decomposition forms, and then use a stable inverse routine (which is similar to Eq.~\eqref{equ:stab_Gt0} but without any inverse or Hermitian conjugate) to calculate $\mathbb{F}_i$. We find that normal matrix operations are stable enough. }

For the twisted case, the numerical stabilization requires only minor adjustments. 
First, store $\tilde{\mathbb{B}}_{i}^{-1}=U_2\mathbb{B}_i^{-1}=\left(U_{2}\mathbb{U}_{i}\right)\mathbb{D}_{i}\mathbb{V}_{i}=:\tilde{\mathbb{A}}_i$ instead of ${\mathbb{A}}_i$.
Second, the stable calculation of $\underline{\tilde{\mathbb{F}}}_{\bcancel{i}}^{-1}$ (see Eq.~\eqref{equ:FcanceliInv_twisted}) involves additional $U_2$ matrices during UDV multiplication:
\begin{equation}
    \begin{aligned}
        \underline{\tilde{\mathbb{F}}}_{\bcancel{i}}^{-1}\equiv&\underbrace{\tilde{\mathbb{B}}_{i+1}^{-1}\cdots\tilde{\mathbb{B}}_{r}^{-1}}_{\tilde{\mathbb{A}}_{L}\left(i\right)}\underbrace{\tilde{\mathbb{B}}_{1}^{-1}\cdots\tilde{\mathbb{B}}_{i-2}^{-1}}_{\tilde{\mathbb{A}}_{R}\left(i-2\right)}U_{2}G_{\mathbf{s}^{(i-1)}}^{T_{2}^{f}}U_{2}\\&-U_{2}\left(I-G_{\mathbf{s}^{(i-1)}}^{T_{2}^{f}}\right)U_{2}.
    \end{aligned}
\end{equation}

The pseudo codes in Algorithms~\ref{alg:scheme4forward} and \ref{alg:scheme4backward} detail the forward (from replica-$1$, time-$0$ to replica-$n$, time-$\beta$, see Fig.~\ref{fig:localsweep} for illustration) and backward procedures. 
We maintain two arrays consisting of UDV decompositions for conducting numerical stabilization for $G_{\mathbf{s}^{(i)}}$ and $\mathbb{F}_i$, namely, the $A_{\mathrm{stab}}$ of shape $(r,N_{\mathrm{stab}}+1)$ and $\mathbb{A}_{\mathrm{stab}}$ of shape $(r+1)$, respectively. 
The workflow with $A_{\mathrm{stab}}$ is the same as normal DQMC implementation~\cite{Assaad2008WorldlineDeterminantalQuantum} but now there are $r$ replicas. 
Within each replica, we conduct stable calculations of $G_{\mathbf{s}^{(i)}}$ at regular intervals of $\tau_{\mathrm{stab}}$, where $\tau_{\mathrm{stab}} N_{\mathrm{stab}}=\beta$.
The workflow with $\mathbb{A}_{\mathrm{stab}}$ is also similar: every time before going from replica-$i$ to replica-$(i+1)$ (replica-$(i-1)$), we calculate $\mathbb{A}_R(i)$ ($\mathbb{A}_L(i-1)$) by right (left) multiplying $\mathbb{A}_{i}$ to it according to Eq.~\eqref{equ:UDVxUDV}, and then set it to become $\mathbb{A}_\mathrm{stab}[i]$ ($\mathbb{A}_\mathrm{stab}[i-1]$). 

Below we list other routines we need.  
The first one is the update routines for Green's functions:
\begin{equation}\label{equ:update_Greens_funcs}
    \begin{aligned}
        G_{\mathbf{s}^{\prime}}\left(\tau,\tau\right)&=G_{\mathbf{s}}-G_{\mathbf{s}}U(I_{m}+VU)^{-1}V,\\G_{\mathbf{s}^{\prime}}\left(\tau,0\right)&=G_{\mathbf{s}}\left(\tau,0\right)\\&\quad-G_{\mathbf{s}}U(I_{m}+VU)^{-1}P_{k\times N}G_{\mathbf{s}}(\tau,0),\\G_{\mathbf{s}^{\prime}}\left(0,\tau\right)&=G_{\mathbf{s}}\left(0,\tau\right)-G_{\mathbf{s}}\left(0,\tau\right)U(I_{m}+VU)^{-1}V,\\G_{\mathbf{s}^{\prime}}\left(0,0\right)&=G_{\mathbf{s}}\left(0,0\right)\\&\quad+G_{\mathbf{s}}\left(0,\tau\right)U(I_{m}+VU)^{-1}P_{k\times N}G_{\mathbf{s}}\left(\tau,0\right).
        \end{aligned}
\end{equation}
Again, we note that if the symmetric Trotter decomposition is utilized, one should apply the above update formula to partially propagated Green's function (see the footnote before Eq.~\eqref{equ:update_G00}) instead of the true Green's functions. 
Secondly, two routines for multiplying a normal matrix with a UDV decomposition are needed:
\begin{subequations}
    \begin{equation}\label{equ:appE_lmulAF}
        BA=\underbrace{BUD}_{U_{0}D_{0}V_{0}}V=\underbrace{U_{0}}_{U^{\prime}}\underbrace{D_{0}}_{D^{\prime}}\underbrace{V_{0}V}_{V^{\prime}}=A^{\prime},
    \end{equation}
    \begin{equation}\label{equ:appE_rmulAF}
        AB=U\underbrace{DVB}_{U_{0}D_{0}V_{0}}=\underbrace{UU_{0}}_{U^{\prime}}\underbrace{D_{0}}_{D^{\prime}}\underbrace{V_{0}}_{V^{\prime}}=A^{\prime}.
    \end{equation}
\end{subequations}

\clearpage

\begin{algorithm*}
    \small % 设置小号字体
    \linespread{0.9}\selectfont % <--- 将行距设置为默认的 0.9 倍
	\caption{Forward propagation framework for Scheme 3 of incremental algorithm}
    \label{alg:scheme4forward}

	\SetKwData{Left}{left}\SetKwData{This}{this}\SetKwData{Up}{up}
	\SetKwFunction{Union}{Union}\SetKwFunction{FindCompress}{FindCompress}
	\SetKwInput{Input}{Input}
    \SetKwInOut{Output}{Output}
	
	\Input{$\{G_{\mathbf{s}^{(i)}}(0,0)\mid i=1,2,\dots,r\}$,$\{G_{\mathbf{s}^{(i)}}^{T_2^f}\mid i=1,2,\dots,r\}$}
    \Input{$\{A_i := G_{\mathbf{s}^{(i)}}^{T_2^f}\mid i=1,2,\dots,r\}$,$\{\bar{A}_i := I-G_{\mathbf{s}^{(i)}}^{T_2^f}\mid i=1,2,\dots,r\}$}
    \Input{${A}_{\mathrm{stab}}[1:r,0:N_{\mathrm{stab}}]=\{A_{L}(i,l_{\mathrm{stab}}):=B_{\mathbf{s}^{(i)}}(\beta,l_{\mathrm{stab}}\tau_{\mathrm{stab}})^\dagger\mid i=1,2,\dots,r;l_{\mathrm{stab}}=0,1,2,\dots,N_{\mathrm{stab}}\}$}
    \Input{$\mathbb{A}_{\mathrm{stab}}[0:r]=\{\mathbb{A}_L(i):=\mathbb{B}(r,i)^{-1}\mid i=0,1,2,\dots,r\}$}

    \BlankLine
	Set $\mathbb{A}_{\mathrm{stab}}[0]=\mathbb{A}_R(0):=I$\;
    \For{$i\in 1,2,\dots,r$}{
        \uIf{$i==1$}{
            Retrieve $\mathbb{A}_L(1)={\mathbb{A}}_{\mathrm{stab}}[1]$ and $\bar{A}_r$\;
            Calculate $\mathbb{F}_{\bcancel{1}}^{-1}$ using Eq.~\eqref{equ:appE_FFiInv1}\;
        }
        \uElseIf{$i==2$}{
            Retrieve $\mathbb{A}_L(2)=\mathbb{A}_{\mathrm{stab}}[2]$ and $A_1$\;
            Calculate $\mathbb{F}_{\bcancel{2}}^{-1}$ using Eq.~\eqref{equ:appE_FFiInv2}\;
        }
        \uElseIf{$i==r$}{
            Retrieve $\mathbb{A}_R(r-2)=\mathbb{A}_{\mathrm{stab}}[r-2]$ and $A_{r-1}$\;
            Calculate $\mathbb{F}_{\bcancel{r}}^{-1}$ using Eq.~\eqref{equ:appE_FFiInvr}\;
        }
        \Else{
            Retrieve $\mathbb{A}_L(i)=\mathbb{A}_{\mathrm{stab}}[i]$, $\mathbb{A}_R(i-2)=\mathbb{A}_{\mathrm{stab}}[i-2]$ and $A_{i-1}$\;
            Calculate $\mathbb{F}_{\bcancel{i}}^{-1}$ using Eq.~\eqref{equ:appE_FFiInv}\;
        }
        Calculate $\mathbb{F}_{i}$ using Eq.~\eqref{equ:appE_FFi}\;
        Set $A_{\mathrm{stab}}[i,0]=A_R(i,0):=I$\;
        \For{$l\in 1,2,\dots,L_\tau$}{
            Propagate to $G_{\mathbf{s}^{(i)}}(\tau,\tau),G_{\mathbf{s}^{(i)}}(\tau,0),G_{\mathbf{s}^{(i)}}(0,\tau)$\;
            Update $G_{\mathbf{s}^{(i)}}(\tau,\tau),G_{\mathbf{s}^{(i)}}(\tau,0),G_{\mathbf{s}^{(i)}}(0,\tau),G_\mathbf{s}(0,0)$ using Eq.~\eqref{equ:update_Greens_funcs} and update $\mathbb{F}_i$ using Eq.~\eqref{equ:scheme4update}\;
            \If{$\tau==l_{\mathrm{stab}}\tau_{\mathrm{stab}}$}{
                Retrieve $A_L(i,l_{\mathrm{stab}})=A_{\mathrm{stab}}[i,l_{\mathrm{stab}}]$\;
                Calculate $B_{\mathbf{s}^{(i)}}(\tau,\tau-\tau_{\mathrm{stab}})$\;
                \uIf{$l_{\mathrm{stab}}==1$}{
                    Calculate $A_R(i,1):=B_{\mathbf{s}^{(i)}}(\tau,\tau-\tau_{\mathrm{stab}})$\;
                }
                \Else{
                    Retrieve $A_R(i,l_{\mathrm{stab}}-1)=A_{\mathrm{stab}}[i,l_{\mathrm{stab}}-1]$\;
                    Calculate $A_R(i,l_{\mathrm{stab}})=B_{\mathbf{s}^{(i)}}(\tau,\tau-\tau_{\mathrm{stab}})A_R(i,l_{\mathrm{stab}}-1)$ using Eq.~\eqref{equ:appE_lmulAF}\;
                }
                Set $A_{\mathrm{stab}}[i,l_{\mathrm{stab}}]=A_R(i,l_{\mathrm{stab}})$\;
                \uIf{$l==L_\tau$}{
                    Calculate $G_{\mathbf{s}^{(i)}}(\beta,\beta)=[I+A_R(i,N_{\mathrm{stab}})]^{-1}$ using Eq.~\eqref{equ:stab_Gbb}\;
                    Calculate $G_{\mathbf{s}^{(i)}}(\beta,0)=I-G_{\mathbf{s}^{(i)}}(\beta,\beta)$\;
                    Calculate $G_{\mathbf{s}^{(i)}}(0,\beta)=-G_{\mathbf{s}^{(i)}}(\beta,\beta)$\;
                    Calculate $G_{\mathbf{s}^{(i)}}(0,0)=G_{\mathbf{s}^{(i)}}(\beta,\beta)$\;
                }
                \Else{
                    Calculate $G_{\mathbf{s}^{(i)}}(\tau,\tau)=[I+A_R(i,l_{\mathrm{stab}})A_L^\dagger(i,l_{\mathrm{stab}})]^{-1}$ using Eq.~\eqref{equ:stab_Gtt}\;
                    Calculate $G_{\mathbf{s}^{(i)}}(\tau,0)=[A_R^{-1}(i,l_{\mathrm{stab}})+A_L^\dagger(i,l_{\mathrm{stab}})]^{-1}$ using Eq.~\eqref{equ:stab_Gt0}\;
                    Calculate $G_{\mathbf{s}^{(i)}}(0,\tau)=-[A_L^{\dagger -1}(i,l_{\mathrm{stab}})+A_R(i,l_{\mathrm{stab}})]^{-1}$ using Eq.~\eqref{equ:stab_G0t}\;
                    Calculate $G_{\mathbf{s}^{(i)}}(0,0)=[I+A_L^{\dagger}(i,l_{\mathrm{stab}})A_R(i,l_{\mathrm{stab}})]^{-1}$ using Eq.~\eqref{equ:stab_G00_t}\;
                }
                Calculate $G_{\mathbf{s}^{(i)}}^{T_2^f}=(G_{\mathbf{s}^{(i)}}(0,0))^{T_2^f}$ using Eq.~\eqref{equ:untwistedPTGreen}\;
                Calculate $\mathbb{F}_i$ using Eq.~\eqref{equ:appE_FFi}\;
                Measure $\det g_r$\;
            }
        }
        Calculate $A_i:=G_{\mathbf{s}^{(i)}}^{T_2^f}$, $\bar{A}_i:=I-G_{\mathbf{s}^{(i)}}^{T_2^f}$, and $\mathbb{A}_i:=\mathbb{B}_i^{-1}$\;
        \uIf{$i==1$}{
            Set $\mathbb{A}_R(i)=\mathbb{A}_i$\;
        }
        \Else{
            Retrieve $\mathbb{A}_R(i-1)$\;
            Calculate $\mathbb{A}_R(i)=\mathbb{A}_R(i-1)\mathbb{A}_i$ using Eq.~\eqref{equ:UDVxUDV}\;
        }
        Set $\mathbb{A}_{\mathrm{stab}}[i]=\mathbb{A}_R(i)$\;
    }
	\Output{${A}_{\mathrm{stab}}[1:r,0:N_\mathrm{stab}]=\{A_{R}(i,l_{\mathrm{stab}}):=B_{\mathbf{s}^{(i)}}(l_{\mathrm{stab}}\tau_{\mathrm{stab}},0)\mid i=1,2,\dots,r;l_{\mathrm{stab}}=0,1,2,\dots,N_{\mathrm{stab}}\}$}
    \Output{$\mathbb{A}_{\mathrm{stab}}[0:r]=\{\mathbb{A}_R(i):=\mathbb{B}(i,0)^{-1}\mid i=0,1,2,\dots,r\}$}
\end{algorithm*}

\begin{algorithm*}
    \small % 设置小号字体
    \linespread{0.9}\selectfont % <--- 将行距设置为默认的 0.9 倍
	\caption{Backward propagation framework for Scheme 3 of incremental algorithm}
    \label{alg:scheme4backward}

	\SetKwData{Left}{left}\SetKwData{This}{this}\SetKwData{Up}{up}
	\SetKwFunction{Union}{Union}\SetKwFunction{FindCompress}{FindCompress}
	\SetKwInput{Input}{Input}
    \SetKwInOut{Output}{Output}
	
	\Input{$\{G_{\mathbf{s}^{(i)}}(\beta,\beta)\mid i=1,2,\dots,r\}$,$\{G_{\mathbf{s}^{(i)}}^{T_2^f}\mid i=1,2,\dots,r\}$}
    \Input{$\{A_i:=G_{\mathbf{s}^{(i)}}^{T_2^f}\mid i=1,2,\dots,r\}$,$\{\bar{A}_i:=I-G_{\mathbf{s}^{(i)}}^{T_2^f}\mid i=1,2,\dots,r\}$}
    \Input{${A}_{\mathrm{stab}}[1:r,0:N_\mathrm{stab}]=\{A_{R}(i,l_{\mathrm{stab}}):=B_{\mathbf{s}^{(i)}}(l_{\mathrm{stab}}\tau_{\mathrm{stab}},0)\mid i=1,2,\dots,r;l_{\mathrm{stab}}=0,1,2,\dots,N_{\mathrm{stab}}\}$}
    \Input{$\mathbb{A}_{\mathrm{stab}}[0:r]=\{\mathbb{A}_R(i):=\mathbb{B}(i,0)^{-1}\mid i=0,1,2,\dots,r\}$}
    
	\BlankLine
    Set $\mathbb{A}_{\mathrm{stab}}[r]=\mathbb{A}_L(r):=I$\;
    \For{$i\in r,r-1,\dots,1$}{
        \uIf{$i==1$}{
            Retrieve $\mathbb{A}_L(1)={\mathbb{A}}_{\mathrm{stab}}[1]$ and $\bar{A}_r$\;
            Calculate $\mathbb{F}_{\bcancel{1}}^{-1}$ using Eq.~\eqref{equ:appE_FFiInv1}\;
        }
        \uElseIf{$i==2$}{
            Retrieve $\mathbb{A}_L(2)=\mathbb{A}_{\mathrm{stab}}[2]$ and $A_1$\;
            Calculate $\mathbb{F}_{\bcancel{2}}^{-1}$ using Eq.~\eqref{equ:appE_FFiInv2}\;
        }
        \uElseIf{$i==r$}{
            Retrieve $\mathbb{A}_R(r-2)=\mathbb{A}_{\mathrm{stab}}[r-2]$ and $A_{r-1}$\;
            Calculate $\mathbb{F}_{\bcancel{r}}^{-1}$ using Eq.~\eqref{equ:appE_FFiInvr}\;
        }
        \Else{
            Retrieve $\mathbb{A}_L(i)=\mathbb{A}_{\mathrm{stab}}[i]$, $\mathbb{A}_R(i-2)=\mathbb{A}_{\mathrm{stab}}[i-2]$ and $A_{i-1}$\;
            Calculate $\mathbb{F}_{\bcancel{i}}^{-1}$ using Eq.~\eqref{equ:appE_FFiInv}\;
        }
        Calculate $\mathbb{F}_{i}$ using Eq.~\eqref{equ:appE_FFi}\;
        Set $\mathbb{A}_{\mathrm{stab}}[N_{\mathrm{stab}}]=\mathbb{A}_L(N_{\mathrm{stab}}):=I$\;
        \For{$l\in L_\tau, L_\tau-1,\dots,1$}{
            Update $G_{\mathbf{s}^{(i)}}(\tau,\tau),G_{\mathbf{s}^{(i)}}(\tau,0),G_{\mathbf{s}^{(i)}}(0,\tau),G_\mathbf{s}(0,0)$ using Eq.~\eqref{equ:update_Greens_funcs} and update $\mathbb{F}_i$ using Eq.~\eqref{equ:scheme4update}\;
            Propagate to $G_{\mathbf{s}^{(i)}}(\tau-\Delta_\tau,\tau-\Delta_\tau),G_{\mathbf{s}^{(i)}}(\tau-\Delta_\tau,0),G_{\mathbf{s}^{(i)}}(0,\tau-\Delta_\tau)$\;
            \If{$\tau-\Delta_\tau==l_{\mathrm{stab}}\tau_{\mathrm{stab}}$}{
                Retrieve $A_R(i,l_{\mathrm{stab}})=A_{\mathrm{stab}}[i,l_{\mathrm{stab}}]$\;
                Calculate $B_{\mathbf{s}^{(i)}}(\tau-\Delta_\tau+\tau_{\mathrm{stab}},\tau-\Delta_\tau)$\;
                \uIf{$l_{\mathrm{stab}}+1==N_{\mathrm{stab}}$}{
                    Calculate $A_L(i,l_{\mathrm{stab}}):=B_{\mathbf{s}^{(i)}}(\tau-\Delta_\tau+\tau_{\mathrm{stab}},\tau-\Delta_\tau)$\;
                }
                \Else{
                    Retrieve $A_L(i,l_{\mathrm{stab}}+1)=A_{\mathrm{stab}}[i,l_{\mathrm{stab}}+1]$\;
                    Calculate $A_L(i,l_{\mathrm{stab}})=A_L(i,l_{\mathrm{stab}}+1)B_{\mathbf{s}^{(i)}}(\tau-\Delta_\tau+\tau_{\mathrm{stab}},\tau-\Delta_\tau)$ using Eq.~\eqref{equ:appE_rmulAF}\;
                }
                Set $A_{\mathrm{stab}}[i,l_{\mathrm{stab}}]=A_L(i,l_{\mathrm{stab}})$\;
                \uIf{$l==1$}{
                    Calculate $G_{\mathbf{s}^{(i)}}(0,0)=[I+A_L^\dagger(i,0)]^{-1}$ using Eq.~\eqref{equ:stab_G00}\;
                    Calculate $G_{\mathbf{s}^{(i)}}(0,0^+)=G_{\mathbf{s}^{(i)}}(0,0)-I$\;
                }
                \Else{
                    Calculate $G_{\mathbf{s}^{(i)}}(\tau-\Delta_\tau,\tau-\Delta_\tau)=[I+A_R(i,l_{\mathrm{stab}})A_L^\dagger(i,l_{\mathrm{stab}})]^{-1}$ using Eq.~\eqref{equ:stab_Gtt}\;
                    Calculate $G_{\mathbf{s}^{(i)}}(\tau-\Delta_\tau,0)=[A_R^{-1}(i,l_{\mathrm{stab}})+A_L^\dagger(i,l_{\mathrm{stab}})]^{-1}$ using Eq.~\eqref{equ:stab_Gt0}\;
                    Calculate $G_{\mathbf{s}^{(i)}}(0,\tau-\Delta_\tau)=-[A_L^{\dagger -1}(i,l_{\mathrm{stab}})+A_R(i,l_{\mathrm{stab}})]^{-1}$ using Eq.~\eqref{equ:stab_G0t}\;
                    Calculate $G_{\mathbf{s}^{(i)}}(0,0)=[I+A_L^{\dagger}(i,l_{\mathrm{stab}})A_R(i,l_{\mathrm{stab}})]^{-1}$ using Eq.~\eqref{equ:stab_G00_t}\;
                }
                Calculate $G_{\mathbf{s}^{(i)}}^{T_2^f}=(G_{\mathbf{s}^{(i)}}(0,0))^{T_2^f}$ using Eq.~\eqref{equ:untwistedPTGreen}\;
                Calculate $\mathbb{F}_i$ using Eq.~\eqref{equ:appE_FFi}\;
            }
        }
        Calculate $A_i:=G_{\mathbf{s}^{(i)}}^{T_2^f}$, $\bar{A}_i:=I-G_{\mathbf{s}^{(i)}}^{T_2^f}$, and $\mathbb{A}_i:=\mathbb{B}_i^{-1}$\;
        \uIf{$i==r$}{
            Set $\mathbb{A}_L(i-1)=\mathbb{A}_i$\;
        }
        \Else{
            Retrieve $\mathbb{A}_L(i)$\;
            Calculate $\mathbb{A}_L(i-1)=\mathbb{A}_i\mathbb{A}_L(i)$ using Eq.~\eqref{equ:UDVxUDV}\;
        }
        Set $\mathbb{A}_{\mathrm{stab}}[i-1]=\mathbb{A}_L(i-1)$\;
    }
    \Output{${A}_{\mathrm{stab}}[1:r,0:N_\mathrm{stab}]=\{A_{L}(i,l_{\mathrm{stab}}):=B_{\mathbf{s}^{(i)}}(\beta,l_{\mathrm{stab}}\tau_{\mathrm{stab}})^\dagger\mid i=1,2,\dots,r;l_{\mathrm{stab}}=0,1,2,\dots,N_{\mathrm{stab}}\}$}
    \Output{$\mathbb{A}_{\mathrm{stab}}[0:r]:=\{\mathbb{A}_L(i)=\mathbb{B}(r,i)^{-1}\mid i=0,1,2,\dots,r\}$}
\end{algorithm*}
%%%%%%%%%%%%%%%%%%%%%%%%%%%%%%%%%%%%%%%%%%%%%%%%%%%%%